\documentclass[a4paper,11pt]{article}
\usepackage[authoryear,round]{natbib}
\usepackage{hyperref}
\usepackage{todonotes}
\usepackage{comment}
\usepackage{algorithm}
\usepackage{subcaption}
\usepackage{svg}

\usepackage[normalem]{ulem}

\usepackage[margin=3cm,footskip=1cm]{geometry}  
\usepackage[T1]{fontenc}   
\usepackage{xr}

\usepackage{amsmath,amssymb,amsthm,bm,mathtools,xspace,booktabs,url}
\usepackage{graphicx,etoolbox}
\usepackage[shortlabels]{enumitem}
\usepackage{mathrsfs}
\usepackage{diagbox}
\usepackage{makecell}
\usepackage{appendix}

\usepackage{lmodern} 
\usepackage{dsfont}
\usepackage{microtype}
\usepackage{syntonly}
\usepackage{float}
\usepackage{xcolor}
\usepackage{amsfonts}
\usepackage{color}

\usepackage{etoc}
\usepackage{multirow} \usepackage[above,below,section]{placeins}

\usepackage{bibunits} \defaultbibliography{Refs}
\defaultbibliographystyle{dcu}
\newcommand{\de}{\mathrm{d}}

\newcommand{\nn}{\mathbf{N}}
\newcommand{\M}{{\mathcal M}}
\newcommand{\N}{{\mathcal N}}
\newcommand{\x}{{\mathbf x}}

\newcommand{\y}{{\mathbf y}}
\newcommand{\z}{{\mathbf z}}
\newcommand{\vv}{{\mathbf v}}
\newcommand{\R}{\mathbb R}

\newcommand{\1}{{\mathbf 1}}

\newcommand{\Sm}{{\mathcal S}}
\newcommand{\T}{{\mathcal T}}
\newcommand{\Loss}{{\mathcal L}}
\newcommand{\I}{{\mathcal I}}
\def\P{\mathbb P}
\def\E{\mathbb E}

\def\P{\mathbb P}
\def\E{\mathbb E}
\DeclareMathOperator{\e}{e}
\DeclareMathOperator{\Var}{Var}

\DeclareMathOperator{\bias}{Bias}

\newcommand{\beann}{\begin{eqnarray*}}
\newcommand{\eeann}{\end{eqnarray*}}

\newtheorem{theorem}{Theorem}\newtheorem{cor}{Corollary}\newtheorem{definition}{Definition}\newtheorem{lemma}{Lemma}

\graphicspath{{./Figures/}}

\setlist{
listparindent=\parindent,
parsep=0pt,
}
      
\AtBeginEnvironment{thebibliography}{\small \setlength\itemsep{0.75em plus 0.5em minus 0.75em}}
    
\usepackage{caption}
\usepackage{authblk}

\makeatletter
\newcommand\CorrespondingAuthor[1]{\def\@makefnmark{}\footnotetext{Corresponding author: #1}}
\makeatother

\makeatletter
\renewenvironment{abstract}{\small \providecommand\keywords{\par\medskip\noindent\textit{Keywords:}\xspace}\begin{center}\bfseries \abstractname\vspace{-.5em}\vspace{\z@}\end{center}\quote }{\endquote}
\makeatother

\graphicspath{{Figures/}}

\numberwithin{equation}{section}

\begin{document}

\title{Comparison of Point Process Learning and its special case Takacs-Fiksel estimation}

\author[1]{Julia Jansson\footnote{juljans@chalmers.se}}
\affil[1]{
Department of Mathematical Sciences, Chalmers University of Technology \& University of Gothenburg, Sweden
}
\author[1]{Ottmar Cronie\footnote{ottmar@chalmers.se, corresponding author}}

\date{}

\maketitle

\begin{bibunit}

\begin{abstract}

Recently, \citet{cronie2023cross} introduced the notion of cross-validation for point processes and, with it, a new statistical methodology called Point Process Learning (PPL). 
In PPL one splits a point process/pattern into a training and a validation set, and then predicts the latter from the former through a parametrised Papangelou conditional intensity. 
The model parameters are estimated by minimizing a point process prediction error; this notion was introduced as the second building block of PPL.
It was
shown that PPL outperforms the state-of-the-art in both kernel intensity estimation and 
estimation
of the 
parameters of the Gibbs hard-core process. 
In the latter case, the state-of-the-art was represented by pseudolikelihood estimation. 
In this paper we study PPL in relation to Takacs-Fiksel estimation, 
of which pseudolikelihood is a special case.
We show that Takacs-Fiksel estimation is a special case of PPL in the sense that PPL with a 
specific
loss function asymptotically reduces to Takacs-Fiksel estimation if we let the cross-validation regime tend to leave-one-out cross-validation. 
Moreover, PPL involves a certain type of hyperparameter given by a weight function which ensures that the prediction errors have expectation zero if and only if we have the correct parametrisation. We show that the weight function takes an explicit but intractable form for general Gibbs models. Consequently, we propose different approaches to estimate the weight function in practice. 
In order to assess how the general PPL setup performs in relation to its special case Takacs-Fiksel estimation, we conduct a simulation study where we find that for common Gibbs models, including Poisson, hard-core, Strauss and Geyer saturation processes, we can find loss functions and hyperparameters so that PPL typically outperforms Takacs-Fiksel estimation significantly in terms of mean square error.
Here, the hyperparameters are the cross-validation parameters and the weight function estimate.


\end{abstract}

\section{Introduction}

Point process learning (PPL), recently proposed by \citet{cronie2023cross}, is a new prediction-based statistical theory for point processes inspired by previous work by \citet{moradi2019resample, cronie2018bandwidth} and Takacs-Fiksel estimation \citep{takacs1986interaction, Coeurjolly2019understanding}.
Motivated by the complexity of modern large datasets, PPL introduces statistical learning concepts for generalized random samples, i.e. point processes.
PPL is based on the combination of two concepts that are novel for general point processes: cross-validation and prediction errors. 
The cross-validation approach uses thinning to split a point process/pattern into pairs of training and validation sets, while the prediction errors measure discrepancy between two point processes via a parametrised Papangelou conditional intensity. 
Both Takacs-Fiksel estimation and the prediction errors of PPL 
are closely linked to the innovations of \citet{baddeley2005residual}, which in turn 
are based on the Georgii-Nguyen-Zessin formula \citep{georgii1976canonical,nguyen1979integral}.

In \citet{cronie2023cross} it is shown that PPL outperforms the state-of-the-art in kernel intensity estimation, i.e. the \citet{cronie2018bandwidth} approach.
Further, \citet{jansson2024gibbs} show that PPL outperforms pseudolikelihood estimation for the Gibbs hard-core process. 
Pseudolikelihood estimation is a special case of Takacs-Fiksel estimation, but it is not necessarily the optimal case \citep{coeurjolly2016towards}.
In this paper, we study PPL in relation to Takacs-Fiksel estimation and show that PPL produces Takacs-Fiksel estimation as a limiting case, when the cross-validation regime tends to leave-one-out cross-validation.
More specifically, in Theorem \ref{thm:TF} and Theorem \ref{thm:TF_Block} we show that Takacs-Fiksel estimation is a special limit case of PPL, 
when using 
Monte-Carlo cross-validation 
or 
block cross-validation, respectively. 
Notably, this shows the generality and flexibility of PPL.

\citet{cronie2023cross} show that the expectation of the prediction errors in PPL is zero if and only if the so-called PPL-weight is of a certain form. This result is reformulated and presented in Theorem \ref{thm:ExpectationPredErrs}. In Theorem \ref{thm:PPLweights} we provide general expressions for the PPL-weight for Gibbs models. Specifically, we investigate PPL-weight expressions for Poisson, Gibbs hard-core, Strauss and Geyer saturation processes. For a Poisson process the weight takes a simple form but for the other models the weight is intractable. Therefore we discuss different approaches to estimate PPL-weights from data.

We compare PPL with Takacs-Fiksel estimation through a simulation study for four common Gibbs models, namely Poisson, hard-core, Strauss and Geyer saturation processes. 
In contrast to its special case Takacs-Fiksel estimation, the general PPL formulation contains several hyperparameters to be specified, namely the cross-validation regime/parameters, the PPL-weight and a certain test function which determines how much a validation point contributes to the total prediction error. 
Motivated by e.g.\ \citet{kresin2023parametric}, in our simulations the test function is fixed to be the so-called Stoyan-Grabarnik test function \citep{stoyan1991second}. 
The results of the simulation study show that 
there are 
hyperparameter choices such that PPL outperforms Takacs-Fiksel estimation in terms of mean square error.

The paper is structured as follows. In Section \ref{sec:Background}, some preliminaries for point processes are recalled, and an overview of PPL and Takacs-Fiksel estimation is given. 
Section \ref{sec:asymptotic} contains our main results, which show that Takacs-Fiksel estimation is a special case of PPL. Next, Section \ref{s:Gibbs} introduces Gibbs models, and we derive explicit expressions for the PPL-weight for the Poisson, Gibbs hard-core, Strauss and Geyer saturation processes. Additionally, different approaches to estimate PPL-weights in practice are proposed. Also, results of our simulation study, comparing PPL to Takacs-Fiksel estimation for the four aforementioned processes, are presented. Finally, Section \ref{sec:discuss} consists of a discussion.

\section{Preliminaries}
\label{sec:Background}

Let $\Sm$ be a general (complete separable metric) space with metric $d(\cdot,\cdot)$, which is equipped with a suitable (non-atomic $\sigma$-finite and locally finite) reference measure $A\mapsto|A|=\int_A\de x$, $A\subseteq \Sm$. 
All sets considered in this paper are Borel sets and a closed ball around $u\in\Sm$, with radius $r>0$, will be denoted by $b(u,r)=\{v\in\Sm:d(u,v)\leq r\}$. 
Examples of general spaces include (compact subsets of) the Euclidean space $\mathbb{R}^d$, $d\geq1$, e.g.\ $\Sm=[0,1]^2\subseteq\R^2$ as in our simulation study, spheres and linear networks \citep{cronie2020inhomogeneous,BRT15}.

A (simple) point process $X=\{x_i\}_{i=1}^N$ in $\Sm$ is a random mechanism whose outcomes are collections of points, 
so-called point patterns. 
Hence, we may view $X$ as a generalisation of a classical random sample, where we allow the sample size $N$ to be random and the sample points $x_i$ to be dependent random variables. 
Formally, 
$X=\{x_i\}_{i=1}^N$, $0\leq N\leq \infty$, 
is defined as a measurable mapping 
from a probability space $(\Omega, \mathcal{F}, \mathbb{P})$ to the measurable space $(\nn, \N)$ \citep{VanLieshoutBook, MW04}. 
Here, $\nn$ is the collection of point patterns/configurations $\x=\{x\}_{i=1}^n\subseteq\Sm$, $0\leq n\leq\infty$, 
which are locally finite, i.e.\ those satisfying that the cardinality $\#(\x\cap A)$ is almost surely (a.s.) finite for any bounded $A\subseteq\Sm$. 
Further, the $\sigma$-algebra $\N$ is the Borel $\sigma$-algebra generated by a Prohorov-type metric on $\nn$ \citep[Section A2.5-A2.6]{DVJ1}. 
Note that, by construction, $X$ is simple, which means that a.s.\ no two points of $X$ have the same location. 
The point process $X$ induces a distribution $P_X$ on $(\nn,\N)$, which is governed by its finite dimensional distributions. 
Moreover, we refer to $X$ as a finite point process if $\#(X\cap\Sm) < \infty$ a.s.. 
Hence, if $\Sm$ is a bounded set then $X$ is automatically a finite point process, due to the local finiteness. 

Given a further general space $\M$, there are two ways of constructing a marked point process $\breve X=\{(x_i,m_i)\}_{i=1}^N = \{(x_i,m(x_i))\}_{i=1}^N$. It can either be constructed as a point process on $\Sm\times\M$ such that $X=\{x_i\}_{i=1}^N$ is a well-defined point process \citep{VanLieshoutBook}, or by attaching an additional component $m(x)\in \M$, a so-called mark, to each point of the point process $X=\{x_i\}_{i=1}^N$ on $\Sm$,  
according to some (random) marking function/field $m:\Sm \to \M$. 
When the mark assignment is done independently, conditionally on $X$, we say that $\breve X$ is an independent marking of $X$ 
in the sense of \citet{Cronie2016,DVJ2}.

\subsection{Distributional characteristics}
\label{sec:dist_char}
The distribution of a point process $X$ is completely characterized by its (Papangelou) conditional intensity $\lambda_X$, which can be interpreted as follows: $\lambda_X(u|\x)\de u$ is the probability of finding a point of the point process in an infinitesimal neighbourhood $du$ of $u\in\Sm$, given that the point process agrees with the configuration $\x$ outside $du$. 
It satisfies the Georgii-Nguyen-Zessin (GNZ) formula (see e.g.\ \citet[Section 1.8.2]{VanLieshoutBook}), which states that 
\begin{align}
\label{e:GNZ}
\E\left[\sum_{u\in X} h(u,X\setminus\{u\})\right] = \int_\Sm \E[h(u,X)\lambda_X(u|X)]\de u
\end{align}
for any non-negative (potentially infinite) measurable function $h$ on $\Sm\times\nn$, whereby it also holds for any integrable function $h$. 
A model, or the point process $X$ it generates, is called \textit{attractive} if $\lambda_X(u|\x)\leq \lambda_X(u|\y)$ and \textit{repulsive} if $\lambda_X(u|\x)\geq \lambda_X(u|\y)$ whenever $\x\subseteq \y$ \citep[Section 6.1.1]{MW04}. In addition, it is called locally stable if there exists a $|\cdot|$-integrable function $\phi$ such that $\lambda_X(u|\cdot)\leq \phi(u)<\infty$, $u\in\Sm$.

Conditional intensities have a central role in the study of point processes, e.g.\ $X$ has the intensity function $\rho_X(u)=\E[\lambda_X(u|X)]$, $u\in\Sm$. Moreover, for a finite point process $X$, the (Janossy) density $f_X$ on $\nn$ can be straightforwardly obtained from its conditional intensity \citep[Lemma 3.1]{betsch2023structural} and the conditional intensity satisfies \citep[Section 15.5]{DVJ2}
$$
\lambda_X(u|\x) 
= 
\frac{f_X((\x\setminus\{u\})\cup\{u\})}{f_X(\x\setminus\{u\})}, \qquad u\in\Sm, \quad \x\in\nn
.
$$

It is common to refer to any point process with existing conditional intensity as a Gibbs process \citep{betsch2023structural}. 
The reason is that the distribution $P_X$ on $(\nn,\N)$, which yields expectations with respect to $X$, most notably those in the GNZ formula, is a so-called Gibbs measure if and only if the GNZ formula is satisfied with a conditional expectation of exponential form \citep{dereudre2019introduction}.
Hereby, it is often convenient to express the conditional intensity of a parametrised Gibbs model $P_{\theta}$, $\theta\in\Theta$, as 
\begin{align}
\label{e:GibbsPI}
\lambda_{\theta}(u|\x) = \e^{\Phi_1(u;\theta)+\Phi_2(u,\x;\theta)}
=\widetilde\rho_{\theta}(u)\e^{\Phi_2(u,\x;\theta)}
, 
\end{align}
for specific functions $\Phi_1$ and $\Phi_2$, where $P_X=P_{\theta_0}$ and $\lambda_X=\lambda_{\theta_0}$ for some $\theta_0\in\Theta$. 
Here, $\widetilde\rho_{\theta_0}(u)$ controls the propensity of $X$ to place a point at location $u\in\Sm$, in absence of interaction with other points, while $\e^{\Phi_2(u,\x;\theta_0)}$ scales up/down $\widetilde\rho_{\theta_0}(u)$, depending on the strength and type of interaction between points that the distribution of $X$ exhibits. 
Note in particular that we here have $\rho_X(u) = \widetilde\rho_{\theta_0}(u)\E[\e^{\Phi_2(u,X;\theta_0)}]$, which is typically not known explicitly. One exception is a Poisson process (see Section \ref{s:Poisson}), where $\Phi_2(\cdot)=0$, which implies that $\rho_X(u) = \lambda_X(u|\x)
=\widetilde\rho_{\theta_0}(u)$.
Many Gibbs processes further belong to the sub-class of Markov point processes, where $\lambda_{\theta}(u|\x)=\lambda_{\theta}(u|\partial_{\x}(u))$, given the 'neighbourhood' $\partial_{\x}(u)=\{x\in\x:u\sim x\}$, defined in terms of a symmetric and reflexive relation $\sim$ on $\Sm$ \citep{baddeley1995area, VanLieshoutBook}.

\subsection{Point process thinning}\label{s:Marking}

Historically, point process thinning was defined as applying some mechanism to a point process/pattern so that each point is retained or deleted with a certain probability \citep{chiu2013stochastic}. In order to formalise this, \citet{cronie2023cross} introduced a definition in terms of binary marking. More specifically, the mark space is given by 
$\M=\{0,1\}$, so that each point obtains either 0 or 1 as mark, 
and the reference measure on $\M$ is the counting measure. 

\begin{definition}[Thinning]
\label{def:thinning}
Let $\breve{X}$ be a binary marking of $X$. Then the associated \textit{thinning} of $X$ is defined as the marginal point process 
$$
X^V = \{x: (x,m) \in \breve{X}, m = 1\} = \{x \in X: m(x) = 1\}\subseteq X. 
$$
When this marking is done independently, we obtain an independent thinning $X^V$ in the sense of \citet[Section 11.3]{DVJ2}. 
The definition of a thinning $\x^V$ of a point pattern $\x$ is analogous. 
\end{definition}

It may be worthwhile to emphasise the dual operation of superpositioning; the superposition of $X^V$ and the remainder, $X^T=X\setminus X^V$, gives us $X$. In other words, we may split a point process/pattern into two parts through thinning or, reversed, obtain a superposition as the union of two point processes/patterns. 

In an independent thinning, the associated retention probability function $p(u)=\P(m(u)=1)$, $u\in\Sm$, governs whether 
each point $x\in\x$ is deleted, which occurs with probability $1-p(x)$,  or retained, which occurs with probability $p(x)$, independently of all other deletions. 
Consequently, $X^T = X\setminus X^V$ is also an independent thinning, but with retention probability function $1-p(u)$, $u\in\Sm$. 
A particular instance of independent thinning is $p$-thinning \citep{chiu2013stochastic}, where we have $p(u)=p\in(0,1)$, $u\in\Sm$, whereby we consider an independent binary marking according to a Bernoulli distribution with parameter $p$. 
A particularly interesting example here is a Poisson process $X$, since also $X^V$ is a Poisson process, but with the intensity $\rho_{X^V}(\cdot) = p(\cdot)\rho_X(\cdot)$ \citep[p. 161]{chiu2013stochastic}. Such distributional closedness, i.e.\ that the independent thinning belongs to the same model family as the original point process, does not hold in general. See \citet[Counterexample 2]{baddeley1996markov} for an example of a Markov point process which is non-Markovian under independent thinning.

There are some interesting distributional properties for independently thinned point processes, 
which we next state for completeness. 
Given an independent thinning $X^V$ of a point process $X$, based on a retention probability function $p(u)\in(0,1)$, $u\in\Sm$, the main observation here is that \citep{cronie2023cross, decreusefond2018stein}
\begin{align}
\label{e:ThinnedPap}
\lambda_{X^T}(u|X^T)
&=
(1-p(u))
\E[\lambda_X(u|X)|X^T]
=(1-p(u))
\left.
\E\left[
\frac{\breve\lambda((u,1)|\breve{X})}{p(u)}
\right|X^T\right]
\end{align}
a.e.\ and a.s., where we recall the marked point process representation $\breve{X}$ in Definition \ref{def:thinning}, which has conditional intensity $\breve\lambda:(\Sm\times\{0,1\})\times\breve\nn\to[0,\infty)$. 
Note that this implies that the intensity function is given by $\rho_{X^T}(u)=\E[\lambda_{X^T}(u|X^T)]=p(u)\rho_X(u)$. 
Clearly, if we replace $X^V$ by $X^T$ we also need to replace $p(\cdot)$ by $1 - p(\cdot)$ above. 
Similar results for non-independent thinnings are, in general, not available, thus making independent thinnings particularly interesting.

\subsection{Cross-validation}
\label{s:CV}

The general idea of cross-validation is to repeatedly split a dataset into a training set and a validation set (see e.g.\ \citet{arlot2010survey}). 
In the case of point processes/patterns, following \citet{cronie2023cross}, such operations may be formalised through thinning.

\begin{definition}[Point process cross-validation]
\label{def:PPCV}
Given a point process $X\subseteq S$ and $k\geq1$ thinnings 
$Z_1,\ldots,Z_k$ of $X$, we refer to the collection of pairs $(X_i^T,X_i^V)=(Y_i,Z_i)$, $Y_i=X\setminus Z_i$, $i=1,\ldots,k$, as a cross-validation  splitting/partitioning.
Analogously, for a point pattern $\x$ we may generate a cross-validation splitting/partitioning $(\x_i^T,\x_i^V)$, $i=1,\ldots,k$ by thinning $\x$.
\end{definition}

Common cross-validation methods in classical statistics are leave-one-out cross-validation and $k$-fold cross-validation, where the former is an instance of the latter, with $k$ given by the sample size $\#\x$ \citep{arlot2010survey}. 
As these are instances of dependent thinning-generated splitting schemes \citep{cronie2023cross}, they become mathematically intractable for point processes. More specifically, for general thinning schemes, we do not have access to statements like the ones found in \eqref{e:ThinnedPap}. 
Instead, \citet{cronie2023cross} advocate for the 
use of independent thinning to split the data, in particular $p$-thinning, because of the associated mathematical tractability. 

Although one can specify an infinite number of methods (thinning schemes) to generate independent thinning-based training-validation pairs $(\x_i^T,\x_i^V)$, $i=1,\ldots,k$, some present themselves as more natural. 
The most straightforward method here is Monte-Carlo cross-validation, where each point $x\in\x$ is assigned to $\x_i^V$ with a fixed common probability $p\in(0,1)$, independently of the other assignments. The procedure is repeated $k\geq1$ times in order to obtain $\{(\x_i^T,\x_i^V)\}_{i=1}^k$. Note that we here may have that $\x_i^V\cap\x_j^V\neq\emptyset$, $i\neq j$, when $k\geq2$. 
As an alternative, where $\x_i^V\cap\x_j^V=\emptyset$, \citet{cronie2023cross} introduced the concept of multinomial $k$-fold cross-validation, which has a hierarchical structure. Here we independently attach marks $m(x)$ to all $x\in\x$, where the mark distribution is given by a multinomial distribution with $p_i(x)=\P(m(x)=i)=1/k$, $i=1,\ldots,k$. We then let $\x_i^V=\{x\in\x:m(x)=i\}$ and $\x_i^T=\x\setminus\x_i^V$, $i=1,\ldots,k$.  
A further alternative mentioned by \citet{cronie2023cross} is block cross-validation, where we instead let $p_i(x)=\P(m(x)=i)=\1\{x\in \Sm_i\}$, $i=1,\ldots,k$, for a fixed partition $\{\Sm_i\}_{i=1}^k$ of $\Sm$.

\subsection{Parametric estimation of Gibbs point processes}
Consider the typical setting where we observe a point pattern $\x\in\nn$, which is a realisation of a point process $X$ on $\Sm$, with unknown conditional intensity $\lambda_X$. 
As is commonly the case in parametric Gibbs process modelling, we do not deal with the setting where we observe $X$ restricted to some bounded sub-domain $W\subseteq\Sm$, forcing us to take edge effects into account. 
Given a model family $\Lambda_{\Theta}=\{\lambda_{\theta}: \theta\in\Theta\}$ with Euclidean parameter space $\Theta$, 
we here assume that 
the unknown conditional intensity $\lambda_X$ is given by $\lambda_{\theta_0}$, for some unknown $\theta_0\in\Theta$. 
Our objective is now to find the member $\lambda_{\theta}$ in $\Lambda_{\Theta}$ which is ``closest to'' $\lambda_{\theta_0}$ in some suitable sense. 
To do so, we need some criterion $\mathcal{L}(\theta;\x)$, $\theta\in\Theta$, to optimise in order to obtain an estimate $\widehat\theta=\widehat\theta(\x)\in\Theta$, which in turn yields our estimate $\lambda_{\widehat\theta}$ of 
$\lambda_{\theta_0}$. 
When we compare the performance of different estimation approaches, we choose the approach for which the estimator $\widehat\theta(X)$ yields the smallest mean square error, 
$$
\mathrm{MSE}(\widehat\theta(X))=\E[(\widehat\theta(X)-\theta_0)^2] = \Var(\widehat\theta(X)) + \bias(\widehat\theta(X))^2, 
$$
where $\bias(\widehat\theta(X))=\E[\widehat\theta(X)]-\theta_0$. This is thus our way of measuring the above-mentioned 'closeness' when modelling with a parametric family where the true conditional intensity is assumed to belong to this family.

For the statistical methodologies we present here, the starting point is, in the words of \citet{cronie2023cross}, a parameterised estimator family 
$$
\Xi_{\Theta}=\{\xi_{\theta}: \theta\in\Theta\}, 
\qquad
\Theta\subseteq\R^d,
$$ 
which consists of non-negative or integrable functions 
$\xi_{\theta}:\Sm\times\nn\to\R$, 
where a member $\xi_{\theta}$ typically is a function of $\lambda_{\theta}$. 
The different statistical methodologies are expressed through different loss functions $\Loss(\theta;\x)=\Loss(\xi_{\theta};\x)$.
It turns out that different such specifications yield estimators $\widehat\theta(X)$ with different performance in terms of MSE.

\subsubsection{Point process prediction errors}
\label{sec:PredErr}

The basis of the statistical methodologies presented here is the notion of point process prediction errors, originally introduced in \citet{cronie2023cross}.

\begin{definition}
    Given two parameterised estimator families, $\Xi_{\Theta}=\{\xi_{\theta}: \theta\in\Theta\}$ and  $\mathcal{H}_{\Theta}=\{h_{\theta}:\theta\in\Theta\}$, where the latter is called a (parametrised) test function, the collection of functions
    \begin{align}
\label{e:PredError}
\I_{\xi_{\theta}}^{h_{\theta}}(A;\z,\y)
=
\sum_{x\in \z \cap A}
h_{\theta}(x;\y\setminus\{x\})
-
\int_{A}
h_{\theta}(u;\y)
\xi_{\theta}(u;\y)
\de u,  
\end{align}
where $\y,\z\in\nn$, 
$A\subseteq \Sm$, $\theta\in\Theta$, 
constitute the $\mathcal{H}$-weighted prediction errors associated to $\Xi_{\Theta}$.
\end{definition}

The key here is the term $\sum_{x\in \z \cap A}
h_{\theta}(x;\y\setminus\{x\})$, which has the role of predicting the points of $\z$ from the points of $\y$, through the function $h_{\theta}(u;\y\setminus\{x\})$, $u\in A$. It should be clear that different choices for the test function and the parameter $\theta$ yield different predictions. 
Given any two point processes $Y$ and $Z$ with superposition $X=Y\cup Z$, 
the integral term 
is a form of compensator which, by \citet[Theorem 1]{cronie2023cross}, ensures that the expectation of $\I_{\xi_{\theta}}^{h_{\theta}}(A;Z,Y)$ is 0 if and only if $\theta$ is the parameter which corresponds to the distribution of $X$.
It should further be emphasised that we may also consider vector-valued test functions.
Note finally that for any two (dependent) point processes $Y$ and $Z$, the mapping $A\mapsto\I_{\xi_{\theta}}^{h_{\theta}}(A;Z,Y)$, $A\subseteq\Sm$, is a random signed measure on $\Sm$ for any $\theta\in\Theta$ \citep{cronie2023cross}.

\subsubsection{Auto-prediction, Innovations and Takacs-Fiksel estimation}
\label{s:auto_innov_TF}
By letting $\y=\z=\x\in\nn$ in \eqref{e:PredError}, we obtain the special case  
\begin{align}
\label{e:Innovation}
\I_{\xi_{\theta}}^{h_{\theta}}(A;\x,\x)
=
\sum_{x\in \x \cap A}
h_{\theta}(x;\x\setminus\{x\})
-
\int_{A}
h_{\theta}(u;\x)
\xi_{\theta}(u;\x)
\de u,
\end{align}
which is referred to as auto-prediction \citep{cronie2023cross}. When we let $\Xi_{\Theta}=\Lambda_{\Theta}$ here, i.e.\ $\xi_{\theta}(u;\x)=\lambda_{\theta}(u|\x)$, the prediction errors in \eqref{e:Innovation} reduce to the innovations of \citet{baddeley2005residual}, $\I_{\lambda_{\theta}}^{h_{\theta}}(A;\x,\x)$, which form the basis of both point process residuals and Takacs-Fiksel estimation (TF); see e.g.\ \citet{Coeurjolly2019understanding} for details.

In Takacs-Fiksel estimation, the estimate $\widehat\theta(\x)$ is obtained as the $\theta\in\Theta$ which either i) minimises 
$\I_{\lambda_{\theta}}^{h_{\theta}}(\Sm;\x,\x)$ or ii) yields that $\I_{\lambda_{\theta}}^{h_{\theta}}(\Sm;\x,\x)=0$; \citet{coeurjolly2016towards} 
consider the latter. The motivation here is that by the GNZ formula we have that $\E[\I_{\lambda_{\theta}}^{h_{\theta}}(\Sm;X,X)]=0$ when $\theta=\theta_0$. When the aim is to solve $\I_{\lambda_{\theta}}^{h_{\theta}}(\Sm;\x,\x)=0$, in order to have a solvable system of equations, one considers a vector-valued test function $h_{\theta}(\cdot) = (h_{\theta}^1(\cdot),\ldots,h_{\theta}^d(\cdot))^T\in\R^d$, where $d\geq1$ is equal to the dimension of $\Theta$. Depending on the model $\Lambda_{\Theta}$ and the test function $\mathcal{H}_{\Theta}$, this mathematically driven approach is, however, not always simple to carry out \citep{coeurjolly2016towards}. 
Hereby, numerically minimising $\theta\mapsto\|\I_{\lambda_{\theta}}^{h_{\theta}}(\Sm;\x,\x)\|$, with respect to $\theta\in\Theta$, may be a more appealing choice, in particular if the model has many parameters or a complex interaction structure.

A prominent special case of Takacs-Fiksel estimation is pseudolikelihood estimation. Here, the vectorial test function used in the innovation is given by the normalised gradient $h_{\theta}(\cdot) = \nabla_{\theta}\lambda_{\theta}(\cdot)/\lambda_{\theta}(\cdot)\in\R^d$, where $d$ is the dimension of $\Theta$. Setting the resulting innovation to $0$ is equivalent to maximising the log-pseudolikelihood function
\[
\theta
\mapsto
\sum_{x\in \x}
\log\lambda_{\theta}(x|\x\setminus\{x\})
-
\int_{\Sm}
\lambda_{\theta}(u|\x)
\de u,
\qquad
\theta\in\Theta.
\]
In particular, for a Poisson process model with intensity $\rho_{\theta}(u) = \lambda_{\theta}(u|\cdot)$, $u\in\Sm$, the above is in fact the actual log-likelihood function and the innovation is the associated score function to be minimised. 
It further holds that the popular maximum logistic regression likelihood method of \citet{baddeley2014logistic} is a numerically stable approximation of pseudolikelihood estimation \citep{coeurjolly2016towards}. 
As appealing as pseudolikelihood estimation may be, it has its issues, most notably poor performance when there are strong interactions present \citep{BRT15}. 
In addition, it suffers from identifiability issues, even in the context of rather basic models. It is shown in \citet{jansson2024gibbs} that a range of parameter choices for the hard-core distance in a Gibbs hard-core process give the same value of the log-pseudolikelihood function. This stems from the fact that the conditional intensity of a Gibbs hard-core process is not differentiable with respect to the hard-core distance. In practice, estimation is done by the function \verb|ppm| in the \textsf{R} package \textsc{spatstat} \citep{BRT15}, which uses pseudolikelihood for the intensity parameter and a plug-in approach for the hard-core distance.

That the test function is allowed to depend on the model parameters highlights the fact that it is natural to let $h_\theta(\cdot)=f(\lambda_{\theta}(\cdot))$ for some suitable function $f:\R\to\R$. 
One interesting such choice is 
\begin{align}
\label{e:SGtest}
h_\theta(\cdot) = 1/\xi_{\theta}(\cdot)^{\alpha}
, 
\qquad \alpha\geq0,
\end{align}
which becomes $h_\theta(\cdot)=1/\lambda_{\theta}(\cdot)^{\alpha}$ in the case of innovations and Takacs-Fiksel estimation \citep{cronie2023cross}. 
The test function in \eqref{e:SGtest} with $\alpha=0$ results in the raw test function while $\alpha=1/2$ renders the Pearson test function \citep{baddeley2005residual}.
The choice $\alpha=1$ corresponds to the \citet{stoyan1991second} test function, 
also called the inverse test function \citep{baddeley2005residual}, which has been a common component in various statistical approaches found in the literature \citep{cronie2018bandwidth, cronie2023cross, Cronie2016, kresin2023parametric, stoyan2000improving}. A particular appeal of the Stoyan-Grabarnik test function is that it sets the integral in the prediction error (recall \eqref{e:PredError})
to $|A|$, either if $\xi_{\theta}$ is strictly positive a.e.\ or under the convention that $0/0=1$. Moreover, given the omnipresent variance-bias trade-off in estimation, it mainly seems to target the variance of the estimator \citep{moradi2019resample}.

\subsubsection{Point Process Learning}
\label{s:PPL}
\citet{cronie2023cross} combined their thinning-based notion of cross-validation (Section \ref{s:CV}) and their prediction errors (Section \ref{sec:PredErr}) to propose a novel statistical theory for point processes, which they referred to as Point Process Learning (PPL).
More specifically, given a training-validation pair $(X^V,X^T)$ in accordance with Definition \ref{def:PPCV}, 
the heuristic idea of PPL is to find a parameter choice $\theta\in\Theta$ such that $\xi_{\theta}(x,X^T)$ provides the ``best'' prediction of all $x\in X^V$. 
What ``best'' means here is made precise by Theorem \ref{thm:ExpectationPredErrs} below, which is a slight reformulation of \citet[Theorem 2 and Corollary 1]{cronie2023cross}. The proof can be found in the supplementary material document of \citet{cronie2023cross}. In essence, Theorem \ref{thm:ExpectationPredErrs} provides an equivalence between $\I_{\xi_{\theta}}^{h_{\theta}}(A;X^T,X^V)$, $A\subseteq\Sm$, having expectation 0 and the parameter $\theta$ being set to $\theta_0$, provided that $\xi_{\theta_0}$ has a particular form.

\begin{theorem}\label{thm:ExpectationPredErrs}
Given a point process $X$ with conditional intensity $\lambda_{\theta_0}(u| X)$, $u\in\Sm$, let $X^V$ be a thinning of $X$ and let $X^T=X\setminus X^V$. In addition, consider the marked point process $\breve X = \{(x,0):x\in X^T\}\cup\{(x,1):x\in X^V\}$ with mark space $\M=\{0,1\}$ 
and conditional intensity $\breve\lambda_{\theta_0}\{(u,m)|\breve X\}$, $(u,m)\in\Sm\times\M$. 
For any parametrised estimator family $\Xi_{\Theta}$, test function $\mathcal{H}_{\Theta}$ and $A\subseteq\Sm$, it follows that 
\begin{align}
    \label{e:ExpectationInnovation}
    &\E\lbrace\I_{\xi_{\theta}}^{h_{\theta}}(A;X^T,X^V)\rbrace
      = 
\int_A
       \E\left[
       h_{\theta}(u;X^T)
       \left\lbrace 
    \breve\lambda_{\theta_0}\lbrace (u,1)|\breve{X}\rbrace
       -
       \xi_{\theta}(u;X^T)
       \right\rbrace
       \right]
       \de u
       .
\end{align}
If further $\E[\breve\lambda_{\theta_0}\lbrace (u,1)|\breve{X}\rbrace^2]$ and $\E\{h_{\theta}(u;X^T)^2\}$, $\theta\in\Theta$, are finite for $|\cdot|$-almost every $u\in\Sm$, then $\E\lbrace\I_{\xi_{\theta}}^{h_{\theta}}(A;X^T,X^V)\rbrace=0$ for any $A\subseteq\Sm$ if and only if we let $\theta=\theta_0$ and  
\begin{align}
\label{eq:theta0}
\xi_{\theta_0}(u;X^T)
=&
V_{\theta_0}(u) \lambda_{\theta_0}(u| X^T)
,
\\
V_{\theta_0}(u) =&
\left.
\E\left[
\frac{\breve\lambda_{\theta_0}\lbrace (u,1)|\breve{X}\rbrace}{\lambda_{\theta_0}(u| X^T)}
\right|
X^T
\right]\notag.
\end{align}
In particular, when $X^V$ is an independent thinning of $X$, based on the retention probability function $p(u)\in(0,1)$, $u\in\Sm$, then 
\[
V_{\theta_0}(u) =
p(u)
\left.
\E\left[
\frac{\lambda_{\theta_0}(u|X)}{\lambda_{\theta_0}(u| X^T)}
\right|
X^T
\right],
\]
where $V_{\theta_0}(u) = p(u)$ if  $X$ is a Poisson process, $V_{\theta_0}(u) \geq p(u)$ if $X$ is attractive and $V_{\theta_0}(u) \leq p(u)$ if $X$ is repulsive. Under locally stability we also have that $\xi_{\theta_0}(u;\cdot)\leq p(u)\phi(u)$. 

\end{theorem}

Recall that we want to model the true (unknown) conditional intensity $\lambda_{\theta_0}$ using a parametrised conditional intensity $\Lambda_{\Theta}=\{\lambda_{\theta}: \theta\in\Theta\}$, under the assumption that $\theta_0\in\Theta$. 
Given a training-validation pair $(X^T,X^V)$, let 
\begin{align}
\label{eq:xi_Vlambda}
\xi_{\theta}(u;X^T)
=&
V_{\theta}(u) \lambda_{\theta}(u| X^T),
\\
\label{e:PPLweights}
V_{\theta}(u) 
=& 
V_{\theta}(u,X^T,X^V) 
=
\left.
\E\left[
\frac{\breve\lambda_{\theta}\lbrace (u,1)|\breve{X}\rbrace}{\lambda_{\theta}(u| X^T)}
\right|
X^T
\right],
\quad u\in\Sm,
\end{align}
where $\E[\lambda_{\theta}(u|X)^2]<\infty$ for almost all $u\in\Sm$ and all $\theta\in\Theta$, with the convention that $0/0=0$. 
This results in the prediction errors
\begin{align*}
\I_{\xi_{\theta}}^{h_{\theta}}(A;X^T,X^V)
= 
\sum_{x\in X^V \cap A}
h_{\theta}(x;X^T)
-
\int_{A}
h_{\theta}(u;X^T)
V_{\theta}(u)
\lambda_{\theta}(u|X^T)
\de u.
\end{align*}
Assuming that the chosen test function satisfies $\E\{h_{\theta}(u;X^T)^2\}<\infty$, $\theta\in\Theta$, for almost all $u\in\Sm$, Theorem \ref{thm:ExpectationPredErrs} tells us that the expectation of 
$\I_{\xi_{\theta}}^{h_{\theta}}(A;X^T,X^V)$, $A\subseteq\Sm$, 
is 0 
if and only if we set $\theta = \theta_0$ in Equation \eqref{eq:xi_Vlambda}. 
Hence, we let the generalised estimator family be given by \eqref{eq:xi_Vlambda}.

The random field $V_{\theta}(u)$, $u\in\Sm$, in \eqref{e:PPLweights}, which will be referred to as the associated PPL-weight (function), varies both depending on the chosen model $\Lambda_{\theta}$, the true distribution of $X$, i.e.\ $\theta_0$, and the type of thinning employed for the cross-validation. 
Note, in particular, that it 
takes the form 
\begin{align}
\label{e:PPLweightsIndThin}
V_{\theta}(u) =
V_{\theta}(u,X^T,X^V) 
=
p(u)
\left.
\E\left[
\frac{\lambda_{\theta}(u|X^T\cup X^V)}{\lambda_{\theta}(u| X^T)}
\right|
X^T
\right]
\end{align}
when $X^V$ is an independent thinning of $X=X^T\cup X^V$, based on the retention probability function $p(\cdot)$.

\citet{cronie2023cross} argued that cross-validation offers a form of conditional i.i.d.\ sampling which allows us to estimate the expectation of $\I_{\xi_{\theta}}^{h_{\theta}}(A;X^T,X^V)$ using a cross-validation round $\{(\x_i^T,\x_i^V)\}_{i=1}^k$ of $\x$. Based on these heuristics, they introduced 
\begin{eqnarray}
\label{e:L1_2}
\Loss_j(\theta)
&=& 
\frac{1}{\#\T_k}
\sum_{i\in\T_k}
\left|\I_{\xi_{\theta}}^{h_{\theta}}(A;\x_i^T,\x_i^V)\right|^j,
\qquad j=1,2,
\\
\label{e:L3}
\Loss_3(\theta)
&=& 
\left(
\frac{1}{\#\T_k}
\sum_{i\in\T_k}
\I_{\xi_{\theta}}^{h_{\theta}}(A;\x_i^T,\x_i^V)\right)^2,
\end{eqnarray}
as potential loss functions to be minimised in order to find an estimate $\widehat\theta(\x)\in\Theta$ of $\theta_0$. Given $\mathbf{X}=\{(\x_i^T,\x_i^V)\}_{i=1}^k$, 
\begin{align}
\label{e:IndexSet}
\T_k = \T(\mathbf{X})
=
\{i \in \{1,\ldots,k\}: \x_i^T\neq \emptyset, \x_i^V\neq \emptyset\}
\end{align}
ensures that we exclude pairs where either the training 
or the validation set is empty. 
The reason for this exclusion is that it makes no sense to predict the empty set from all of $\x$, or all of $\x$ from the empty set. 
As \citet{cronie2023cross} remarked,
one is of course free to let $\T_k
=\{1,\ldots,k\}$ in \eqref{e:L1_2} and \eqref{e:L3} if this makes sense.

An advantage of PPL, with respect to state-of-the-art methods, is that it seems to suffer less from identifiability issues   
\citep{cronie2023cross,jansson2024gibbs}. 
Turning to its numerical performance, \citet{cronie2023cross} showed that it outperforms the state-of-the-art in kernel intensity estimation \citep{cronie2018bandwidth}, in terms of integrated MSE. Moreover, 
\citet{jansson2024gibbs} provided a simulation study to show that PPL with Monte-Carlo cross-validation and the Stoyan-Grabarnik test function outperforms pseudolikelihood estimation in the case of a Gibbs hard-core process (Section \ref{s:Hard-core}), in terms of MSE. 
Therefore, one might ask whether PPL outperforms Takacs-Fiksel estimation, in terms of MSE. 
Since Takacs-Fiksel estimation seems to represent a sort of leave-one-out cross-validation version of general PPL, comparison of the two methods would reduce to comparing general PPL
to PPL with the specific hyperparameter choice represented by Takacs-Fiksel estimation. We next set out to address these questions, as well as some other related issues.

\section{Takacs-Fiksel estimation as a limiting case of PPL}
\label{sec:asymptotic}
We now arrive at the main results of the paper, which essentially state that an innovation, an instance of autoprediction, constitutes a limit of averages of scaled prediction errors, under specific assumptions on the cross-validation regime employed. Specifically, the cross-validation regime should tend to leave-one-out cross-validation. In other words, they tell us that an innovation (Takacs-Fiksel estimation) is a limiting special case of an average of prediction errors (PPL), which is reflected by the $\Loss_3$ loss function. Hence, comparison of Takacs-Fiksel estimation and the general PPL setup is simply a comparison of PPL with a specific hyperparameter choice and PPL with a free choice of hyperparameters. 
The question that then still remains is how these two setups compare in terms of MSE, and this is investigated numerically in Section \ref{s:Simulations}. 

Theorem \ref{thm:TF}, whose proof can be found in Appendix \ref{s:proofTF},  addresses the convergence when we apply Monte-Carlo cross-validation. 
Heuristically, it  tells us that 
employing the $\Loss_3$ loss function we approximately have autoprediction, when $k$ is large, $p=1/\sqrt{k}$ and $\T(\mathbf{X})
=\{1,\ldots,k\}$. Hence, in practice, to preform Takacs-Fiksel estimation one could carry out PPL with such a setup for the cross-validation regime.

\begin{theorem}
\label{thm:TF}
Assume that $\lambda_{\theta}(u|\x)$ and $h_{\theta}(u;\x)$ are bounded for any $\theta\in\Theta$, $u\in\Sm$ and $\x\in\nn$. 
Moreover, for any $k\geq2$, let $\{(X_i^T(p_k),X_i^V(p_k))\}_{i=1}^k$ be a Monte-Carlo cross-validation of $X$, based on a retention probability 
$p_k\in(0,1)$. 
If $p_k=1/\sqrt{k}$, then 
    \[
    p_k
    \sum_{i=1}^k
\I_{\xi_{\theta}}^{h_{\theta}}(A;X_i^V(p_k),X_i^T(p_k))
-
\I_{\lambda_{\theta}}^{h_{\theta}}(A;X,X)
\stackrel{p}{\longrightarrow}
0
\]
as $k\to\infty$, for any bounded $A\subseteq\Sm$.  
    \end{theorem}
We next provide a similar result under block-cross-validation, where the proof can be found in Appendix \ref{s:proofTF_Block}. Here, we consider partitions which become finer and finer, with the size of any member of the partition rendering a fold-retention probability tending to 0.

\begin{theorem}
\label{thm:TF_Block}
Assume that $\lambda_{\theta}(u|\x)$ and $h_{\theta}(u;\x)$ are bounded for any $\theta\in\Theta$, $u\in\Sm$ and $\x\in\nn$. 
Given a bounded $A\subseteq\Sm$, let $\{(X_{ik}^T,X_{ik}^V)\}_{i=1}^k$, $k\geq2$, be block cross-validations of $X\cap A$, based on partitions $\{A_{ik}\}_{i=1}^k$ of $A$, with associated retention probabilities $p_{ik}(u)=\1\{u\in A_{ik}\}$, $i=1,\ldots,k$. 
Assume further that the partition sizes satisfy $\max_{i=1,\ldots,k}|A_{ik}|\to0$ as $k\to\infty$ and that for any $i=1,\ldots,k$ there exists only one $j=1,\ldots,k+1$ such that $A_{ik}\subseteq A_{j(k+1)}$, i.e.\ we have a refinement. 
Then 
\[
\sum_{i=1}^k
\I_{\xi_{\theta}}^{h_{\theta}}(A;X_{ik}^V,X_{ik}^T)
-
\I_{\lambda_{\theta}}^{h_{\theta}}(A;X\cap A, X\cap A)
\stackrel{p}{\longrightarrow}
0
\]
as $k\to\infty$. 
    \end{theorem}

It should be emphasised that the boundedness conditions imposed in Theorem \ref{thm:TF} and Theorem \ref{thm:TF_Block} do not guarantee that the convergence holds for any model and test function combination.  
However, we note that the imposed conditional intensity boundedness implies that the result holds for at least all locally stable models. 
Turning to the test function, 
in Corollary \ref{cor:ConditionsThm} below, which we prove in Appendix \ref{s:proofCor}, we show that the boundedness conditions hold for a large class of Gibbs models in combination with the test function in \eqref{e:SGtest}.

\begin{cor}
\label{cor:ConditionsThm}
Let $\lambda_{\theta}(u|\x) = \e^{\Phi_1(u;\theta)+\Phi_2(u,\x;\theta)}$, $\theta\in\Theta$, be a Gibbs model where $\Theta$ is bounded and $|\Phi_1(u;\theta)|, |\Phi_2(u,\x;\theta)|$ are bounded for any $u\in A$ and any $\x\in\nn$ such that   $\x\subseteq A$.
Assume further that the test function is given by \eqref{e:SGtest} and that the conditional intensity of $X$ is given by $\lambda_{\theta_0}$,
$\theta_0\in\Theta$. 
It then follows that the conditions of Theorem \ref{thm:TF} and Theorem \ref{thm:TF_Block} are satisfied. 

\end{cor}

There is, however, an immediate example where 
Corollary \ref{cor:ConditionsThm}
is not satisfied when using the test function in \eqref{e:SGtest}, with $\alpha>0$, and that is the Gibbs hard-core process, defined in Section \ref{s:Hard-core}. 
Here, $\lambda_{\theta}(u|\x) = \e^{\Phi_1(u;\theta)+\Phi_2(u,\x;\theta)}$, $\theta=(\beta,R)\in\Theta=(0,\infty)^2$,  $\Phi_1(u;\theta)=\log \beta$ and $\Phi_2(u,\x;\theta) = \log\1\{u\notin \bigcup_{x\in\x}b(x,R)\}$.
The support of $\lambda_{\theta}(\cdot|\x)$ 
is here given by $\Sm\setminus\bigcup_{x\in\x}b(x,R)$, 
whereby 
$h_{\theta}(u;\x)$ is infinite for $u\in \bigcup_{x\in\x}b(x,R)$. A way to solve this issue, in general, is to replace the test function by a new truncated version of the test function, given by $\min\{h_\theta(\cdot),C\}$, for some large constant $C>0$.

These results and observations are satisfactory for our purposes. However, we do conjecture that both Theorem \ref{thm:TF} and Theorem \ref{thm:TF_Block} can be proved under less restrictive boundedness conditions. In addition, we conjecture that a result of a similar spririt could also be obtained for the following generalisation of multinomial $k$-fold and block cross-validation, which we refer to as generalised multinomial cross-validation.

\begin{definition}
\label{def:Kfold}
Given $k\geq2$, consider a collection of $k\geq2$ measurable functions $p_i(u)\in[0,1]$, $u\in\Sm$, $i=1,\ldots,k$, which satisfy that $\sum_{i=1}^k p_i(u)\de u = 1$ for any $u\in\Sm$. Then, attach iid marks $m(x)$ to all $x\in\x$, according to $\P(m(x)=i)=p_i(x)$, $i=1,\ldots,k$. We define {\em generalised multinomial ($k$-fold) cross-validation} 
by letting $\x_i^V=\{x\in\x:m(x)=i\}$ and $\x_i^T=\x\setminus\x_i^V$, $i=1,\ldots,k$. The construction is analogous when we start from a point process $X$ and generate $(X_i^T,X_i^V)$, $i=1,\ldots,k$.
\end{definition}

Given the results in \citet{jansson2024gibbs}, one may hope to obtain an understanding of whether general PPL renders lower estimator variances than Takacs-Fiksel estimation, by looking closer at (limits of) the prediction error variance expressions in \citet{cronie2023cross}. 
Recalling the notation in Theorem \ref{thm:ExpectationPredErrs}, from 
\citet{cronie2023cross} we obtain that the variance is given by
    \begin{align*}
    &\Var\lbrace \I_{\xi_{\theta}}^{h_{\theta}}(A;X^T,X^V)\rbrace
    = 
    \\
    =
    \int_{A} \int_{A}
   \E\Big[
     &h_{\theta}(u;X^T)
     \breve\lambda_{\theta_0}\lbrace (u,1);\breve{X}\rbrace
     \Big(
     h_{\theta}(v;X^T)
    \breve\lambda_{\theta_0}((v,1);\breve{X}\cup (u,1))
    \\
    &
    -
     h_{\theta}(v;X^T\cup\{u\})
    V_{\theta}(v,X^T\cup\{u\},X^V) \lambda_{\theta}(v;X^T\cup\{u\})
    \Big)\Big]
     \\
     +\E\Big[
     &h_{\theta}(u;X^T)
     \Big(
    h_{\theta}(v;X^T)
    V_{\theta}(u,X^T,X^V) \lambda_{\theta}(u;X^T)
     V_{\theta}(v,X^T,X^V) \lambda_{\theta}(v;X^T)
    \\
    &
    -
     h_{\theta}(v;X^T\cup\{u\})
    V_{\theta}(v,X^T\cup\{u\},X^V) \lambda_{\theta}(v;X^T\cup\{u\})
    \breve\lambda_{\theta_0}\lbrace (u,1);\breve{X}\rbrace
    \Big)
     \Big]
     \de u
     \de v
     \\
     +\int_{A}
  \E\Big[
     &
     h_{\theta}(u;X^T)^2
    \breve\lambda_{\theta_0}\lbrace 
    (u,1);\breve{X}
    \rbrace 
    \Big]\de u
    \\
    - \Bigg(\int_A
       \E\Big[
       &
       h_{\theta}(u;X^T)
       \left( 
    \breve\lambda_{\theta_0}\lbrace (u,1);\breve{X}\rbrace
       -
       V_{\theta}(u,X^T,X^V) \lambda_{\theta}(u;X^T)
       \right)
       \Big]
       \de u
       \Bigg)^2.
\end{align*}
Using the limits in Theorem \ref{thm:TF} and Theorem \ref{thm:TF_Block}, as well as the expectation in Theorem \ref{thm:ExpectationPredErrs}, one may hope to shed some light on how the bias, the variance and the MSE of the general PPL setup relates to its limiting case Takacs-Fiksel estimation. 
Based on this approach, we have, however, not been 
able to theoretically show that either of the two methods results in a lower bias, variance or MSE for the estimators than the other method.
Hence, 
we have to resort to a simulation study.

\section{PPL-weights and simulation study for specific models}

\label{s:Gibbs}

In 
our simulation study, which can be found in Section \ref{s:Simulations}, 
we will 
consider 
a range of common Gibbs models, which all belong to the exponential family \citep{Coeurjolly2019understanding}. Specifically, these model families 
are Poisson, hard-core, Strauss and Geyer saturation processes. 
Moreover, in Section \ref{s:wGibbs} we derive explicit results for the PPL-weights of general Gibbs processes, characterised by \eqref{e:GibbsPI}, and in turn use these results to obtain expressions for the PPL-weights of the four model families we consider in the simulation study. 
Section \ref{sec:w_choice} discusses how to estimate or approximate PPL-weights in practice.

\subsection{PPL-weights}
\label{s:wGibbs}
Recall that for a given model $\Lambda_{\theta}$, in order to specify the generalised estimator family in \eqref{eq:xi_Vlambda}, we need to find the associated PPL-weights, i.e.\ \eqref{e:PPLweights}. 
The result below, which is proved in Section \ref{sec:proofPPLweights}, provides more explicit PPL-weight expressions for different settings, in particular for Gibbs processes, when  combined with independent thinning-based cross-validation. In particular, for Gibbs models we see that the PPL-weights are governed by the interaction function $\Phi_2$ in \eqref{e:GibbsPI}.

\begin{theorem}
\label{thm:PPLweights}
Consider a point process $X$ and a conditional intensity model $\Lambda_{\Theta}=\{\lambda_{\theta}: \theta\in\Theta\}$. 
Under auto-prediction, i.e.\ when we consider the prediction error $\I_{\xi_{\theta}}^{h_{\theta}}(A;X,X)$, it follows that $\xi_{\theta}(\cdot;X) = \lambda_{\theta}(\cdot| X)$, i.e.\ $V_{\theta}(\cdot,X,X) = 1$, a.s.. 
Moreover, 
for a Gibbs model with $\Lambda_{\Theta}$ specified by \eqref{e:GibbsPI}, i.e.\ $\lambda_{\theta}(u|\x) = \e^{\Phi_1(u;\theta)+\Phi_2(u,\x;\theta)}$, we have 
\begin{align}
V_{\theta}(u,X^T,X^V) 
&= 
p(u)
\E[
\e^{\Phi_2(u,X^T\cup X^V;\theta)-\Phi_2(u,X^T;\theta)}
|
X^T
]\notag
\\
&= \label{eq:PPLweight_nonlinear}
p(u)
\e^{-\Phi_2(u,X^T;\theta)}
\E[
\e^{\Phi_2(u,X^T\cup X^V;\theta)}
|
X^T
]
.
\end{align}
If further  $\Phi_2(\cdot;\theta)$ is linear in the sense that $\Phi_2(\cdot,\x\cup\y;\theta) = \Phi_2(\cdot,\x;\theta) + \Phi_2(\cdot,\y;\theta)$, $\x,\y\in\nn$, 
\begin{equation}
V_{\theta}(u,X^T,X^V) 
= \label{eq:PPLweight_linear}
p(u)
\E[
\e^{\Phi_2(u,X^V;\theta)}
|
X^T
]
.  
\end{equation}

\end{theorem}
That $\Phi_2$ is linear is e.g.\ the case for a pairwise interaction model, where $\Phi_2(u,\x;\theta)=\sum_{x\in\x}g_{\theta}(u,x)$, for some pairwise interaction function $g_{\theta}$, and here $V_{\theta}(u,X^T,X^V)/p(u)$ becomes the Laplace functional \citep{DVJ2} of $X^V|X^T$ evaluated in $g_{\theta}$. 
Note further that $p(u)=V_{\theta}(u,X_i^T,X_i^V)/\E[\e^{\Phi_2(u,X_i^T\cup X_i^V;\theta)-\Phi_2(u,X_i^T;\theta}) | X^T]$ is given by $p$, $1/k$ and $\1\{u\in\Sm_i\}$ for Monte-Carlo, multinomial $k$-fold and block cross-validation, respectively.
Note also that, for $i = 1,\ldots,k$, the superposition $X_i^T\cup X_i^V$ is equal to $X$.

Given a particular model, we have knowledge of $\Phi_2(\cdot;\theta)$, whereby for a particular cross-validation regime we know what the associated PPL-weight looks like in theory. However, as we shall see, the associated conditional expectation is typically intractable or is not of much direct statistical use. 
Hence, we need to find a way to estimate PPL-weights. A challenge here, however, is that we usually only have access to one realisation $\x$ of the point process, which makes estimating the conditional expectation difficult. 
Next, after looking closer at the PPL-weights for the four aforementioned Gibbs models, to understand their theoretical properties, in Section \ref{sec:w_choice} we discuss how to estimate/approximate the weights.

\subsubsection{Poisson processes}
\label{s:Poisson}
We have already noted that a Poisson process $X$ is obtained by letting $\Phi_2(\cdot)=0$ in \eqref{e:GibbsPI} and 
that $\rho_X(u) = \lambda_X(u|\x)
=\widetilde\rho_{\theta_0}(u)=e^{\Phi_1(u,\theta_0)}$. 
In Section \ref{s:Marking} we also remarked that an independent thinning $X^V$ of $X$ is also a Poisson process, but with intensity $\rho_{X^V}(\cdot) = p(\cdot)\rho_X(\cdot)$. 
Further, from Theorem \ref{thm:ExpectationPredErrs} 
we have that under independent thinning-based cross-validation the PPL-weight for a Poisson process is given by 
$V_\theta(u,X^T,X^V) = p(u)$.

\subsubsection{Hard-core processes}
\label{s:Hard-core}
One of the simplest kinds of interaction between points is where they are forbidden to come too close together, which is the case for a hard-core process \citep{BRT15}.
The conditional intensity is here given by $\lambda_{\theta_0}(u|\x)  
= \beta \1\left\{u\notin \bigcup_{x \in \x} b(x,R)\right\}$
where $\beta>0$ and 
$R>0$ is referred to as the hard-core distance; note that $\lambda_{\theta_0}(\cdot)\leq\beta$. 
Lemma \ref{lemma:HardCore}, which is included for completeness and proved in Section \ref{sec:lemmas}, verifies that hard-core models are indeed repulsive and gives an expression for the PPL-weight.

\begin{lemma}\label{lemma:HardCore}
Hard-core models are repulsive and under independent thinning-based cross-validation 
their 
PPL-weights 
are given by  
$$
V_{\theta}(u,X^T,X^V)
=
p(u)
\P\bigg(\bigcap_{x \in X^V} \{u\notin b(x,R)\}\bigg| X^T\bigg)
=
p(u)
\P\bigg(u\notin\bigcup_{x \in X^V} b(x,R)\bigg| X^T\bigg)
.
$$
\end{lemma}
Note that the PPL-weight $V_{\theta}(u,X^T,X^V)$ for the hard-core process is at most $p(u)$. This follows from 
Theorem \ref{thm:ExpectationPredErrs},
since the hard-core process is repulsive, but it can also be seen by the form of the PPL-weight in Lemma \ref{lemma:HardCore}, as probabilities range between 0 and 1. In particular, the special case $R=0$, which corresponds to a Poisson process with intensity $\beta>0$, yields the upper bound 
$p(u)$.

\subsubsection{Strauss processes}
\label{s:Strauss}
A Strauss process \citep{strauss1975model} is intermediate between a hard core process and 
a Poisson process, and its behaviour depends on the value of its interaction parameter, 
$\gamma \in [0, 1]$. 
Its conditional intensity 
is given by 
\begin{align}
    \label{eq:DR}
\lambda_{\theta_0}(u|\x) &= \beta\gamma^{D_R(u,\x)}
\leq \beta,
\nonumber
\\
    D_R(u; \x) &= \sum_{x \in \x \setminus \{u\}} \1\{d(u, x) \leq R\} = \sum_{x \in \x \setminus \{u\}} \1\{u \in b(x,R)\},
\end{align}
where $R>0$ is called the interaction radius. 
When $\gamma=0$, using the convention that $0/0=1$, we obtain the hard-core model, while $\gamma=1$ results in a Poisson process with intensity $\beta>0$. 
Just as in the hard-core case, for completeness, we here show that Strauss models are repulsive and obtain an expression for the PPL-weights; see Appendix \ref{sec:lemmas} for a proof of Lemma \ref{lemma:Strauss}.

\begin{lemma}\label{lemma:Strauss}
Strauss models are repulsive and the PPL-weight for a Strauss model is
$$V_{\theta}(u,X^T,X^V) 
=
p(u)
\E[\gamma^{D_R(u,X^V)}|X^T].$$
\end{lemma}
Again, we note that PPL-weights 
for 
Strauss processes are smaller than $p(u)$, as a consequence of Theorem \ref{thm:ExpectationPredErrs}.

\subsubsection{Geyer-saturation processes}
\label{s:Geyer}
A Geyer saturation process \citep{geyer1999likelihood} is an extension of a Strauss process where $\gamma$ is allowed to be greater than $1$. Intuitively, this assumption should promote that points tend to aggregate around each other. 
However, in a Strauss process this assumption can lead to a non-integrable density function \citep[Section 6.2.2]{MW04},
so therefore in a Geyer saturation process, the overall contribution from each point is trimmed to never exceed a maximum value. 
The conditional intensity of a Geyer saturation process is given by
$$   
\lambda_{\theta_0}(u|\x) = \beta\gamma^{\min(s,D_R(u,\x))+\sum_{y\in\x}(\min(s,D_R(y,\x\cup \{u\})-\min(s,D_R(y,\x)))},$$ 
where $\beta,\gamma,R,s$ are parameters, and $D_R(y,\x)$, which is given in \eqref{eq:DR}, is the number of other data points $z\in\x$ lying within a distance $R$ of the point $y$. 
The parameter $s\geq 0$ is a saturation threshold, which ensures that each term in the product is never larger than $\gamma^s$, so that the product is never larger than $\gamma^{s\#\x}$. 
For $s = 0$ the a Geyer stauration process reduces to a Poisson process while $s = \infty$ results in a Strauss process with interaction parameter $\gamma^2$. Note that, in contrast to a Strauss process, a Geyer saturation process is not a pairwise interaction model \citep{BRT15}. Moreover, by \citet[Proposition 3.5]{geyer1999likelihood} it is locally stable.

The general claim in the literature is that when $\gamma > 1$, a Geyer model is attractive, and when $\gamma < 1$, it is repulsive \citep{BRT15,geyer1999likelihood}. 
However, 
personal communication with Marie-Colette van Lieshout (CWI, the Netherlands) 
verified our suspicion, namely 
that this is not true. 
The counterexample provided by Marie-Colette van Lieshout is given in the proof of Lemma \ref{lem:counter}, which is found in Appendix \ref{sec:lemmas}.
\begin{lemma}
\label{lem:counter}
It does not hold that in general that a Geyer model is attractive  when $\gamma > 1$, nor that it is repulsive  when $\gamma < 1$. 

\end{lemma}
 We next provide an expression for  
 PPL-weights 
 for Geyer model saturation processes; see Appendix \ref{sec:lemmas} for the proof of Lemma \ref{lem:geyer}.
\begin{lemma}
\label{lem:geyer}
The PPL-weight for a Geyer saturation model is given by 
\begin{align*}
&V_{\theta}(u,X^T,X^V) 
=
p(u)
\E\left[
\e^{\Phi_2(u,X;\theta)
-
\Phi_2(u,X^T;\theta)
}
|X^T
\right].
\end{align*}
where
\begin{align*}
\Phi_2(u,X;\theta)
=&
    \log(\gamma)(\min(s,D_R(u,X))
    \\
&+
\sum_{y\in X}
\1\{y \in b(u,R)\}
\1\{1\leq D_R(y,X\cup \{u\})\leq s\}),
\\
\Phi_2(u,X^T;\theta)
=&
\log(\gamma)(\min(s,D_R(u,X^T))
\\
&+
\sum_{y\in X^T}
\1\{y \in b(u,R)\}
\1\{1\leq D_R(y,X^T\cup \{u\})\leq s\}).
\end{align*}
\end{lemma}

As for Strauss processes, the form of the PPL-weight here is not of much use in statistical settings. Moreover, since a Geyer saturation process is 
not guaranteed to be attractive when $\gamma > 1$ and repulsive when $\gamma < 1$,
we can not be certain about the aforementioned bounds for the PPL-weights. However, one could likely obtain some form of upper bound by exploiting the local stability of the model; recall Theorem \ref{thm:ExpectationPredErrs}.

\subsection{Estimation of PPL-weights}
\label{sec:w_choice}
Assuming independent thinning based cross-validation, 
from Section \ref{s:Poisson} we know that the PPL-weight is $p(u)$ for a Poisson process. 
Unfortunately, as we have seen, for other models 
it is hard to obtain practically useful expressions for the PPL-weights. 
More specifically, in Lemma \ref{lemma:HardCore} and \ref{lemma:Strauss} we obtained closed form expressions for the PPL-weights for hard-core and Strauss processes. However, it does not seem to be an easy task to compute them explicitly. 
Moreover, in Lemma \ref{lem:geyer} we find the PPL-weight for a Geyer saturation process, which clearly 
is a very complicated expression. 
Hence, in practice, we have to either approximate them in some sensible way \citep[cf.][]{cronie2023cross} or estimate them in some suitable way. 
What we do know from Section \ref{s:Hard-core} and Section \ref{s:Strauss} is that PPL-weights for hard-core and Strauss processes are smaller than $p(u)$. 
Further, since it is not known whether a certain Geyer saturation model is attractive or repulsive, it is also not known how the PPL-weight relates to the bound $p(u)$. 
Throughout, recall that in the particular case of Monte-Carlo cross-validation we have $p(\cdot)=p\in(0,1)$.

As a first suggestion, we approximate/estimate $V_{\theta}(u)$ by the retention probability $p(u)$ itself. For repulsive models, such as hard-core and Strauss processes, this 
corresponds to the upper bound of the PPL-weight, i.e.\ we know that $V_{\theta}(u)\in[0,p(u)]$, and for Poisson processes this coincides with the actual PPL-weight. This was also the choice made in the simulation study of \citet{jansson2024gibbs}, who considered Monte-Carlo cross-validation-based PPL in the context of hard-core processes; note that $\widehat V_{\theta}(u)=p$ here. 
However, for repulsive models 
a better choice would in fact be something slightly smaller than $p(u)$. Yet, when $p(u)$ is small the difference between the true weight and $p(u)$ is also small for repulsive models. 
On the other hand, for attractive models, in particular those which promote strong interactions, this PPL-weight estimate is likely far from optimal since Theorem \ref{thm:ExpectationPredErrs} only tells us that $V_{\theta}(u)\in[p(u),\infty)$. 
A slight refinement on the upper bound may be obtained for a locally stable model, where $\lambda_{\theta}(u|\cdot)\leq\phi_{\theta}(u)$, $\theta\in\Theta$. Recall that $\phi_\theta$ is assumed to be a $|\cdot|$-integrable function such that $\phi_\theta(u)<\infty$, $u\in\Sm$.
For 
a locally stable
model, regardless of whether it is attractive or repulsive, by the law of total expectation we have
\begin{align*}
\frac{p(u)\E[\lambda_{\theta}(u|X)|X^T]}{\phi_{\theta}(u)}
\leq
V_{\theta}(u) 
\leq 
\frac{p(u)\phi_{\theta}(u)}{\lambda_{\theta}(u| X^T)}
.
\end{align*}
When the model is attractive we further have that the left-hand side has the lower bound $p(u)\lambda_{\theta}(u|X^T)/\phi_{\theta}(u)$, by the fact that $X^T\subseteq X$, the attractiveness itself and the law of total expectation.
If we were to let the weight estimate be given by the average of the upper and lower bounds in the attractive setting, we would obtain
\begin{align*}
\widehat V_{\theta}(u)
&= p(u)\left(\lambda_{\theta}(u|X^T)/\phi_{\theta}(u) + \phi_{\theta}(u)/\lambda_{\theta}(u| X^T)\right)/2
,
\end{align*}
where we e.g.\ have $\lambda_{\theta}(u|X^T)/\phi_{\theta}(u)= \gamma^{D_R(u,X^T)}$, $0\leq\gamma\leq1$, for a Strauss process and $\lambda_{\theta}(u|X^T)/\phi_{\theta}(u)= 1$ for a Poisson process; for repulsive Gibbs processes where $\Phi_2(\cdot;\theta)\leq1$ we have that $\phi_{\theta}(u)=\Phi_1(u;\theta)$ and $\lambda_{\theta}(u|X^T)/\phi_{\theta}(u)=\Phi_2(u,X^T;\theta)$. 
This estimator is approximately equal to $p(u)$ if $\lambda_{\theta}(u| X^T)/\phi_{\theta}(u) \approx 1$, which goes back to not having too strong interactions. 
Hence, $\widehat V_{\theta}(u) = p(u)$ seems to be a sensible choice also in the attractive case, provided that the model is not too clustered.
In the case of Monte-Carlo cross-validation we thus let $\widehat V_{\theta}(\cdot) = p$.

Another option is to approximate the PPL-weight by $p(u)/(1-p(u))$; cf.\ the kernel intensity estimation in \citet{cronie2023cross}, where both Monte-Carlo cross-validation (where $p(u)=p$) and multinomial $k$-fold cross-validation (where $p(u)=1/k$) were employed. 
Given the bound $p(u)$ in Theorem \ref{thm:ExpectationPredErrs}, the choice $p(u)/(1-p(u))$ would make sense for attractive processes since here the weight is always larger than $p(u)$ as $0<p(u)<1$; for $p(u) = 1/2$ we have $p(u)/(1-p(u)) = 1$ 
and 
for $p(u) > 1/2$ the weight is strictly larger than 1. On the other hand, given Theorem \ref{thm:ExpectationPredErrs}, the choice $p(u)/(1-p(u))$ is clearly not a good approximation for the PPL-weight for repulsive models, such as hard-core and Strauss processes. In the case of Monte-Carlo cross-validation we here let $\widehat V_{\theta}(\cdot) = p/(1-p)$.

Another approach is to estimate the weights from the point pattern at hand.
We here consider weight estimation in the setting where we only have one realisation $\x$ of $X$. 
More specifically, when we employ independent thinning-based cross-validation, based on some retention probability function $p(\cdot)$, we
start from 
\begin{align*}
\widetilde V_{\theta}(u;\x) =
\frac{p(u)}{k'}
\sum_{i=1}^{k'}
\frac{\lambda_{\theta}(u|\x)}{\lambda_{\theta}(u| \z_i^T)}
\approx
p(u)
\left.
\E\left[
\frac{\lambda_{\theta}(u|X)}{\lambda_{\theta}(u| X^T)}
\right|
X^T
\right]
=V_{\theta}(u)
,
\end{align*}
where $\z_i^T$, $i=1,\ldots,k'$, is a set of different
independent thinnings of $\x$, based on the retention probability function 
$1-p(u)$. To obtain numerical stability and to avoid having a spatially varying estimate, in particular if the retention probability is not spatially varying as in the case of Monte-Carlo cross-validation, we then consider a spatial average of this function, i.e.
\begin{align}
\label{e:WeightEst}
\widehat V_{\theta}(\cdot) 
= 
\frac{1}{|\Sm|}
\int_{\Sm} \widetilde V_{\theta}(v;\x) \de v
.
\end{align}
With the general Gibbs formulation we have adopted, we have that
\[
\widetilde V_{\theta}(u;\x) =
\frac{p(u)}{k'}
\sum_{i=1}^{k'}
\e^{\Phi_2(u,\x;\theta)-\Phi_2(u, \z_i^T;\theta)}
,
\]
and if $\Phi_2(\cdot,\x\cup\y;\theta) = \Phi_2(\cdot,\x;\theta) + \Phi_2(\cdot,\y;\theta)$, $\x,\y\in\nn$, 
\[
\widetilde V_{\theta}(u;\x) =
\frac{p(u)}{k'}
\sum_{i=1}^{k'}
\e^{\Phi_2(u,\x\setminus\z_i^T;\theta)}
.
\]
It should be noted that this idea is also applicable outside independent thinning-based cross-validation, as long as we have an explicit form for the marked conditional intensity $\breve\lambda_{\theta}$.

In conclusion, we here propose, and will numerically evaluate in the context of Monte-Carlo cross-validation, the three PPL-weight estimators $\widehat V_{\theta}(\cdot)=p$, $\widehat V_{\theta}(\cdot)=p/(1-p)$ and $\widehat V_{\theta}(\cdot)$ given by 
\eqref{e:WeightEst}.

\subsection{Simulation study}
\label{s:Simulations}

\begin{figure}
    \centering
    \begin{subfigure}{0.4\textwidth}
        \includegraphics[width = \textwidth]{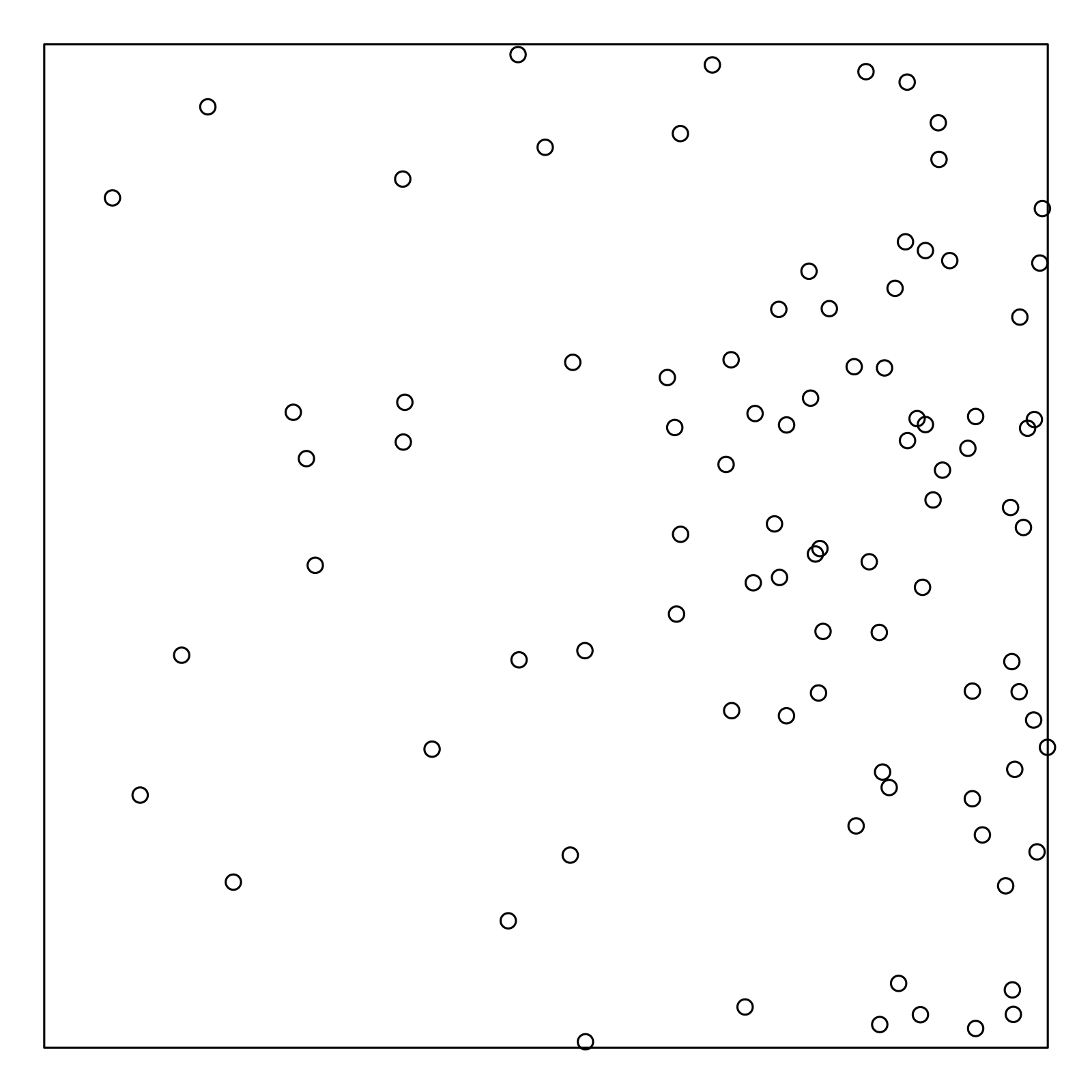}
        \caption{Poisson process with log-linear intensity function $\rho(u) = \e^{2+4 u_x}$, where $u = (u_x,u_y)$ is a point in the window $[0,1]^2$.}
        \label{fig:PoissonPP}
    \end{subfigure}
    \hfill
    \begin{subfigure}{0.4\textwidth}
        \includegraphics[width = \textwidth]{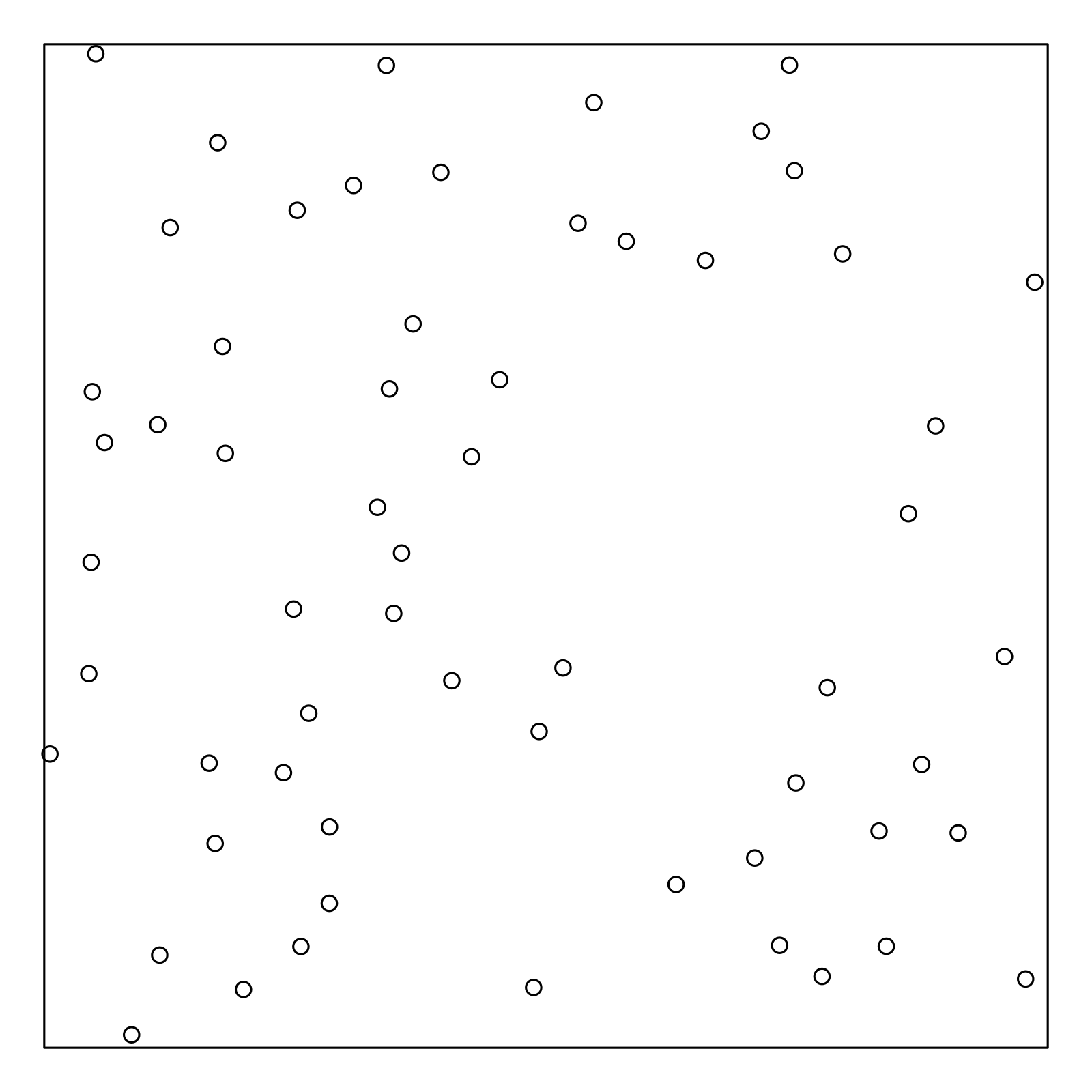}
        \caption{Hard-core process with hard-core distance $R = 0.05$ and intensity parameter $\beta = 100$.}
        \label{fig:Hard-corePP}
    \end{subfigure}
    \begin{subfigure}{0.4\textwidth}
        \includegraphics[width = \textwidth]{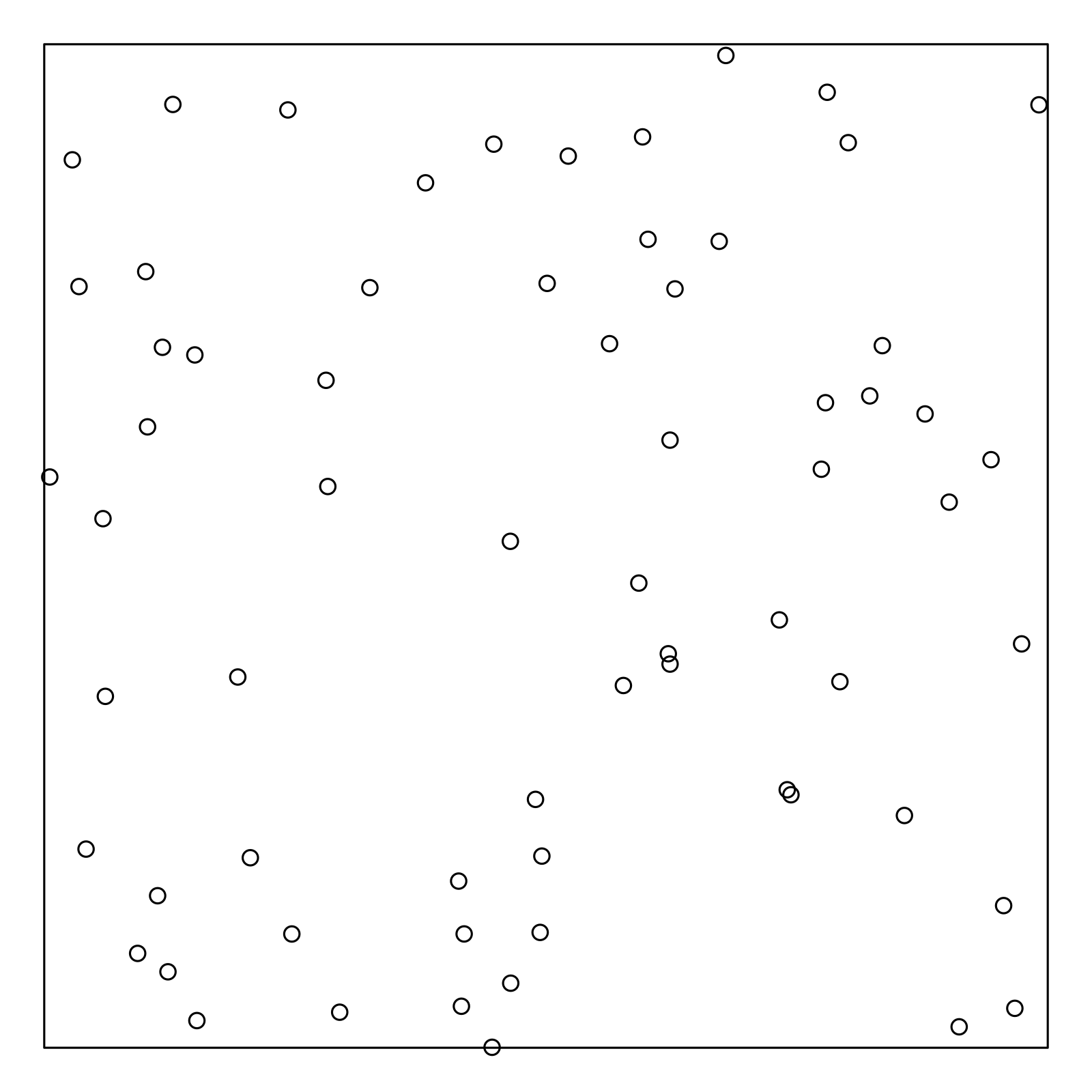}
        \caption{Strauss process with parameters $R = 0.05, \beta = 100$ and $\gamma = 0.5$.}
        \label{fig:StraussPP}
    \end{subfigure}
    \hfill
    \begin{subfigure}{0.4\textwidth}
        \includegraphics[width = \textwidth]{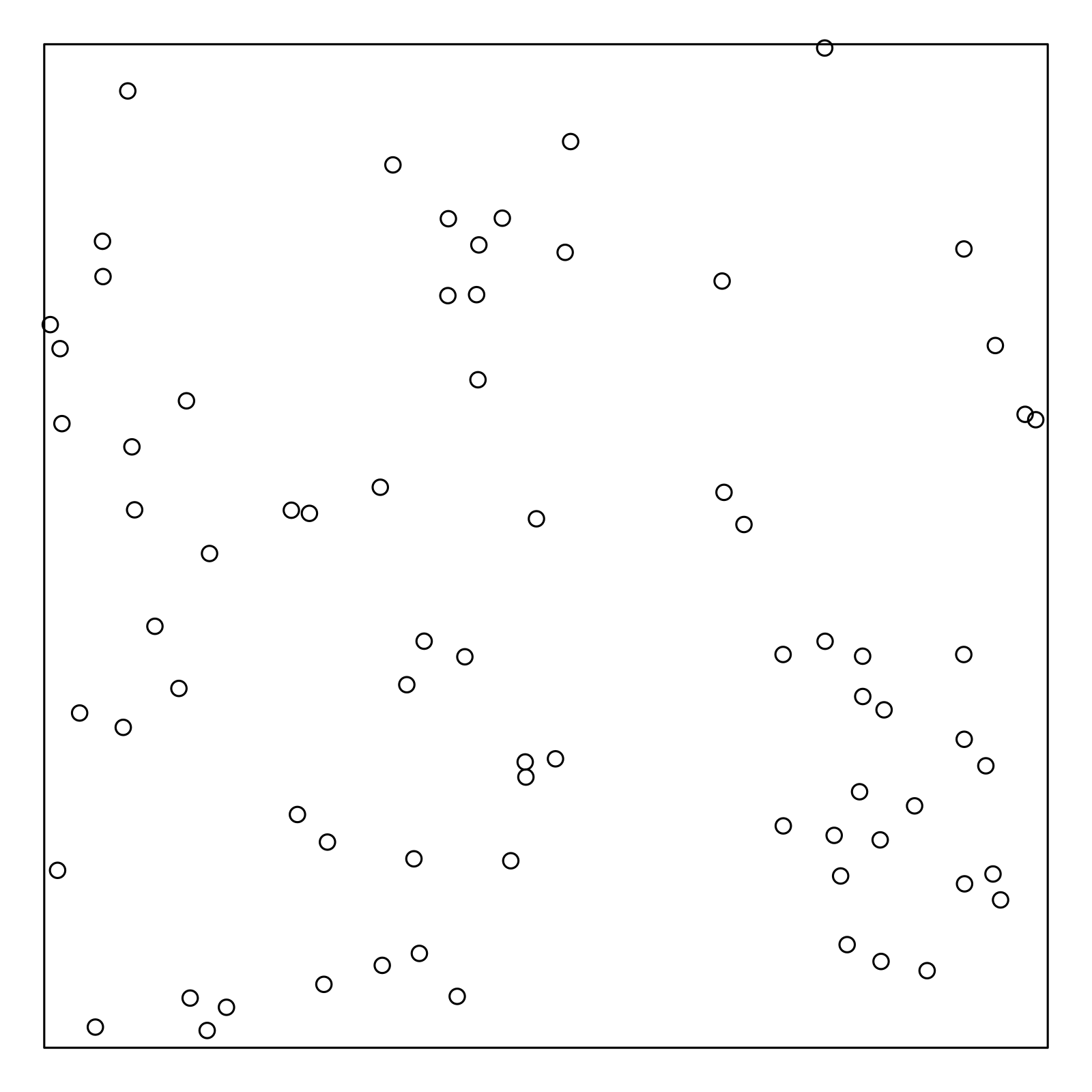}
        \caption{Geyer saturation process with parameters $R = 0.05, \beta = 60, \gamma = \sqrt{1.5}$ and $s = 2$.}
        \label{fig:GeyerPP}
    \end{subfigure}
    \caption{Examples of point patterns from the simulation study for the four considered models.}
    \label{fig:4PPs}
\end{figure}

Recall that we were unable to theoretically deduce whether the general PPL setup performs better than its limiting case Takacs-Fiksel estimation. 
Consequently, we next perform a simulation study which compares the performance of PPL and Takacs-Fiksel estimation for Poisson, hard-core, Strauss and Geyer processes. 
More specifically, in each scenario we consider a conditional intensity model $\lambda_{\theta}(u|\x)$, $u=(u_x,u_y)\in\Sm=[0,1]^2$, $x\in\nn$, and in the case of PPL we consider
Monte-Carlo cross-validation, where we let $k=100$ and $p=0.1,0.2,\ldots,0.9$, following recommendations in \citet{cronie2023cross}. This we compare with Takacs-Fiksel estimation, which was effectively implemented as PPL using leave-one-out cross-validation. 
Here, for both methods, we use grid search to minimise the prediction errors.
Our simulation study further looks into the performance of the aforementioned PPL-weight estimation/approximation schemes, i.e.\ $\widehat V_{\theta}(u;\x)=p$, $\widehat V_{\theta}(u;\x)=p/(1-p)$ and $\widehat V_{\theta}(u;\x)$ given by \eqref{e:WeightEst}, except for in the Poisson process case, where we only use the correct PPL-weight, $p$. 
To limit the scope of the simulation study, 
for both PPL and Takacs-Fiksel estimation we employed the Stoyan-Grabarnik test function.
This 
test function was chosen 
because 
it is well-established in the literature and 
because 
it simplifies the 
integral in 
the prediction error, which decreases the computational burden; recall Section \ref{s:auto_innov_TF}.
For each model we considered $N = 100$ simulated realisations, which we in turn used to 
estimate the bias, variance and MSE for each of the parameters of the models. In the case of PPL, where we have the three different loss functions $\Loss_1$, $\Loss_2$ and $\Loss_3$, we thus obtain three sets of bias, variance and MSE. 
The key question here is how the MSE is influenced by different choices of hyperparameters. 

The computational procedure for PPL is discussed in Section \ref{sec:implementation} of the Appendix.
All simulations were run on the Vera cluster, provided by Chalmers e-Commons at Chalmers University of Technology\footnote{https://www.c3se.chalmers.se/about/Vera/}. In Table \ref{tab:timing} we provide approximate runtimes for Takacs-Fiksel estimation and PPL 
for each of the four models. We stress that our code is not optimised for performance.
\begin{table}[!htb]
    \centering
    \begin{tabular}{c|c c c c}
         & Poisson & hard-core & Strauss & Geyer \\
         \hline
        Takacs-Fiksel estimation & 1 & 1.5 & 6 & 120\\
        PPL & 21 & 24 & 120 & 280
    \end{tabular}
    \caption{Average 
run-time  in seconds per realisation  
for 
    the Vera cluster, for each of the models in the simulation study. These results illustrate the performance for PPL when 
$k = 100$, $p$ is given by a fixed value and the PPL-weight is estimated by $p$. 
}
    \label{tab:timing}
\end{table}

\subsubsection{Poisson process}\label{s:SimPoisson}
We start with a Poisson process with log-linear intensity function $\lambda_{\theta}(u|\cdot)=\rho(u) = \e^{\alpha+\beta u_x}$, where $u = (u_x,u_y)$ is a point in the window $[0,1]^2$. Specifically, we choose $\alpha = 2$ and $\beta = 4$, and the grids where we searched for the parameters were $-1.0,-0.9,\ldots,3.0$ for $\alpha$ and $-2.0,-1.7,\ldots,10.0$ for $\beta$. Recall that we simulate $N = 100$ realisations of the process. A realisation of this model, where the expected number of points in $\Sm$ is approximately 100, can be found in Figure \ref{fig:PoissonPP}.

In Figure \ref{fig:poisson3-L1} we see that, for any $p$, the MSE for PPL is lower than the MSE for Takacs-Fiksel estimation, for both of the parameters, when using the $\Loss_1$ loss function. 
The results for the $\Loss_2$ and $\Loss_3$ loss functions can be found in Appendix \ref{sec:app_poi}, Figure \ref{fig:poisson3-L2L3}. The message is the same for $\Loss_2$ while $\Loss_3$ only performs better than Takacs-Fiksel estimation for large values of $p$. 

Note that, when using the \texttt{ppm} function of the \textsc{R} package \textsc{spatstat} \citep{BRT15} for this model we obtain an MSE of $0.16$ for $\alpha$ and $0.24$ for $\beta$.
This is because, here,  
\texttt{ppm} is based on the Poisson process likelihood function, particularly optimised to work well for Poisson processes, which unsurprisingly outperforms  both Takacs-Fiksel and PPL estimation with the Stoyan-Grabarnik test function.

\begin{figure}[!htb]
    \centering
    \includegraphics[width = 0.48\textwidth]{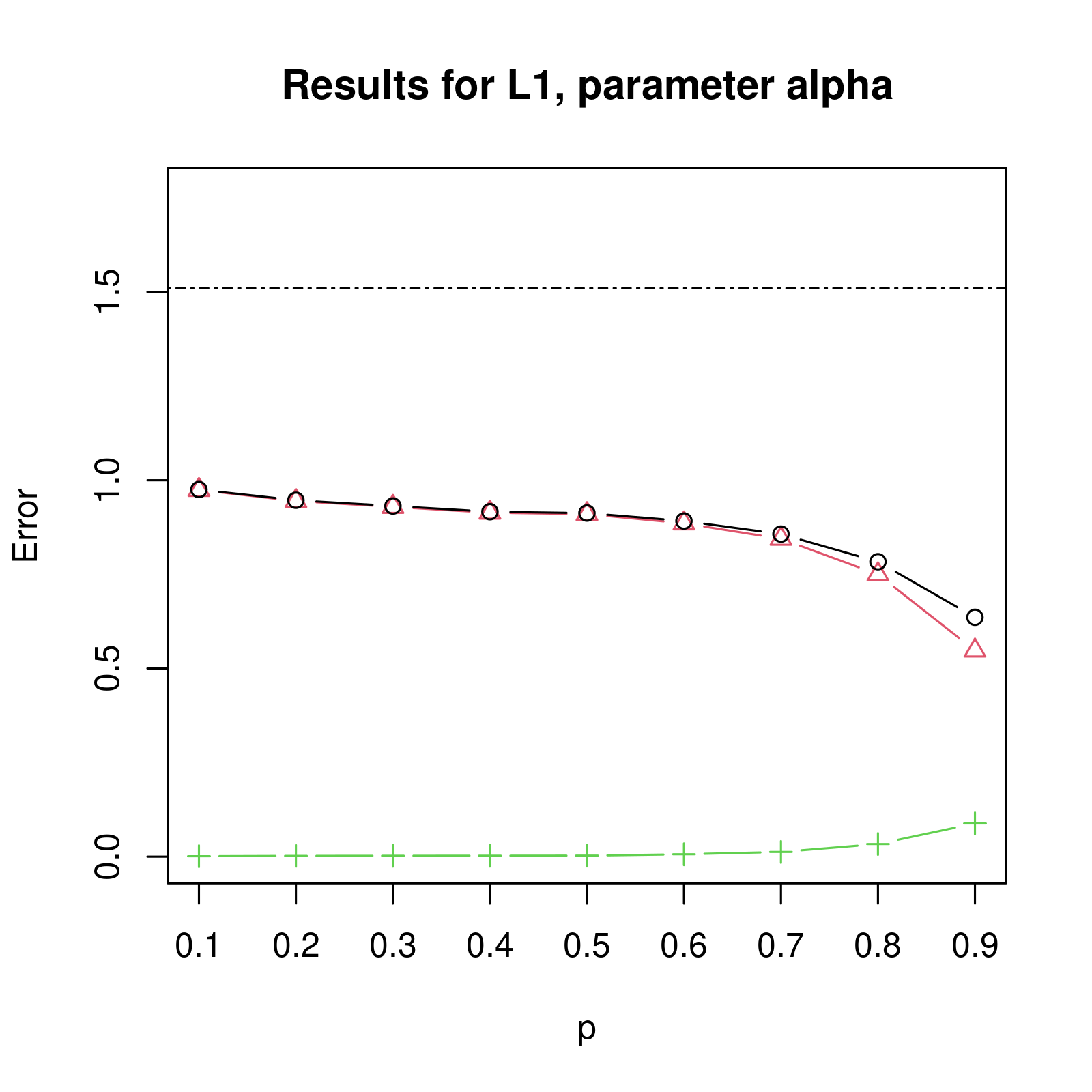}
    \includegraphics[width = 0.48\textwidth]{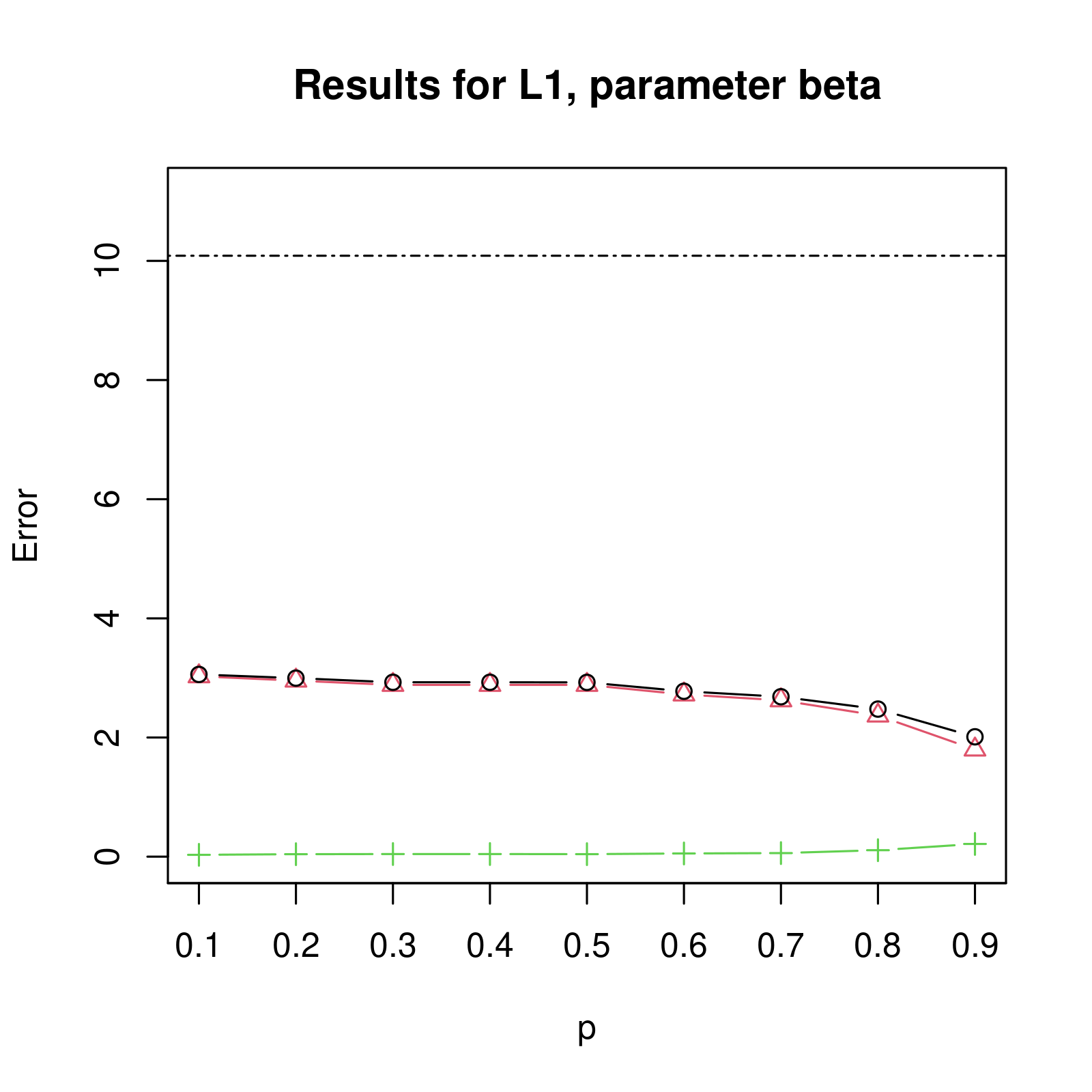}
    \caption{MSE, squared bias and variance for the Poisson process, using PPL with the loss function $\Loss_1$, for the parameters $\alpha$ and $\beta$. Here $k = 100$, $N = 100$, the PPL-weight is set to $p$ and $p = 0.1,0.2,\ldots,0.9$. The black lines with circles correspond to MSE, the red lines with triangles correspond to the squared bias and the green lines with the plus signs correspond to the variance. The horizontal black dotted lines correspond to MSE for Takacs-Fiksel estimation.}
    \label{fig:poisson3-L1}
\end{figure}

\subsubsection{Hard-core process}
\label{s:hc_sims}
For the hard-core process we use the parameters $R = 0.05$ and $\beta = 100$, 
which results in an average of approximately 60 points per realisation; see Figure \ref{fig:Hard-corePP} for an example. 
Here the grid in which we searched for estimates for $\beta$ and $R$ varied from realisation to realisation. 
The grid for $\beta$ consisted of 41 evenly spaced values, starting from $\#\x/|\Sm|$ and ending in $\#\x/|\Sm\setminus\bigcup_{x\in\x}b(x,R_0)|$, where $R_0$ is the smallest distance between any two points in $\x$. To search for an estimate for $R$, we considered an evenly spaced grid of 41 values between $R_0/2$ and $R_0$. These limits were motivated by findings in \citet{jansson2024gibbs}. 
We here present the results for the $\Loss_1$ loss function, while the results for the $\Loss_2$ and $\Loss_3$ loss functions can be found in Appendix \ref{sec:app_hc}, Figure \ref{fig:hard-core-p-L2L3}
--
\ref{fig:hard-core-est-L2L3}. 
The results for the $\Loss_1$ and $\Loss_2$ loss functions are very similar, and the results for the $\Loss_3$ loss function are generally worse than for both $\Loss_1$ and $\Loss_2$. 

First, we look at the case when the PPL-weight is $p$, which is illustrated in Figure \ref{fig:hard-core-p-L1}. The performance in terms of MSE for $\beta$ is better than that of Takacs-Fiksel estimation for $p\leq0.4$. Note in particular that the MSE for $\beta$ using PPL is roughly half of that of Takacs-Fiksel estimation when $p\in[0.1,0.2]$.
This seems to be largely due to the fact that the squared bias is very low for these values of $p$. 
Note further that as we know that the weight for this repulsive model is at most $p$, smaller values of $p$ here result in PPL-weight approximations which are closer to the true PPL-weight (recall Lemma \ref{lemma:HardCore}), while the weight approximation $p$ might be far from the actual weight when $p$ is large. 
Hence, simply setting the weight to $p$ for a hard-core process is likely not an optimal solution, in particular for large $p$, whereby we conclude that the weight estimation should be explored more. 
The MSE for the parameter $R$ is virtually 0 
using PPL, which, relatively speaking, is lower than the MSE for Takacs-Fiksel estimation. However, it should be noted that the MSE for Takacs-Fiksel estimation is still very low compared to the magnitude of the parameter ($0.3\%$).
Therefore, it is perhaps not extremely meaningful to compare the methods too much 
in terms of estimation of $R$, 
and to conclude, we consider them both to perform well.

\begin{figure}[!htb]
    \centering
    \includegraphics[width = 0.45\textwidth]{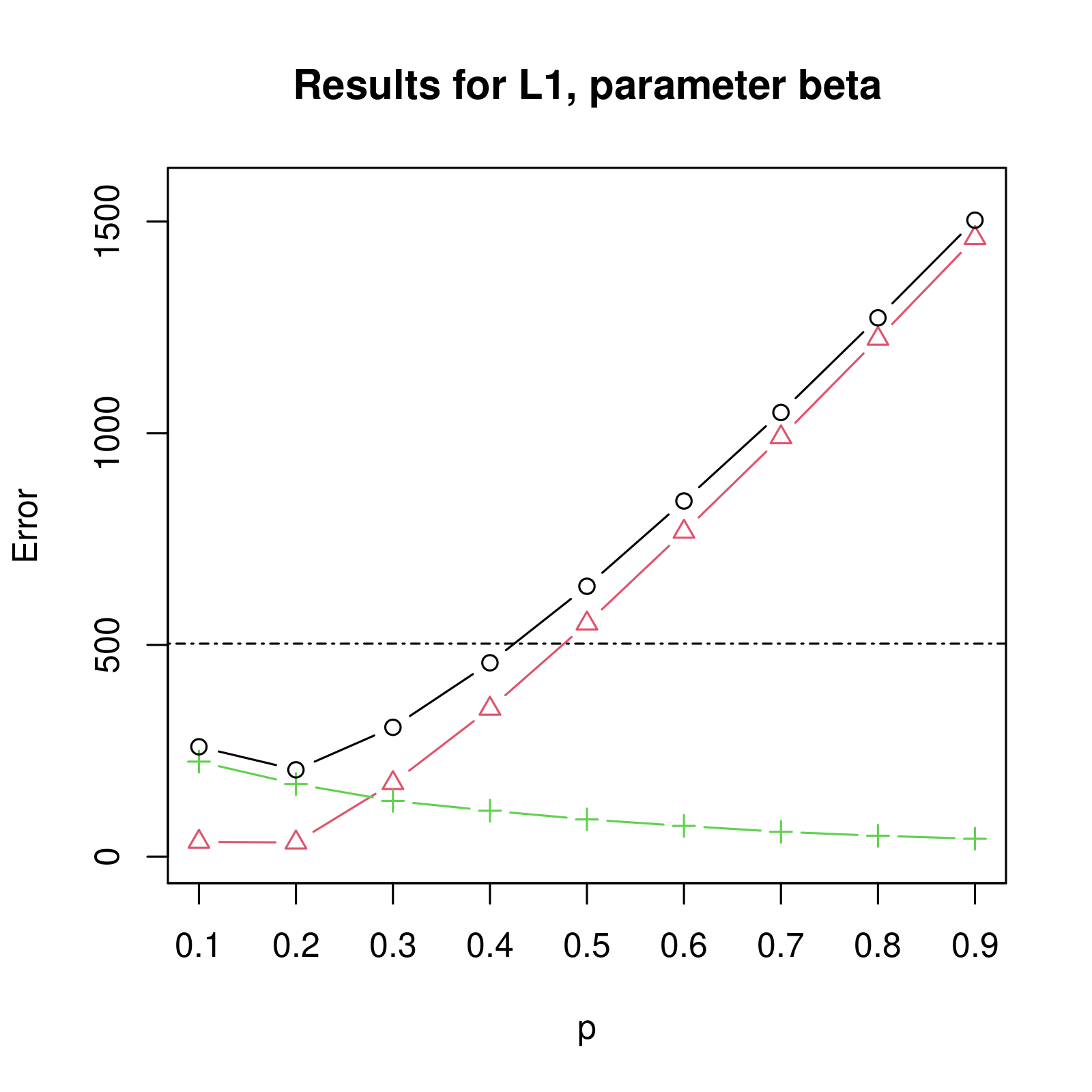}
    \includegraphics[width = 0.45\textwidth]{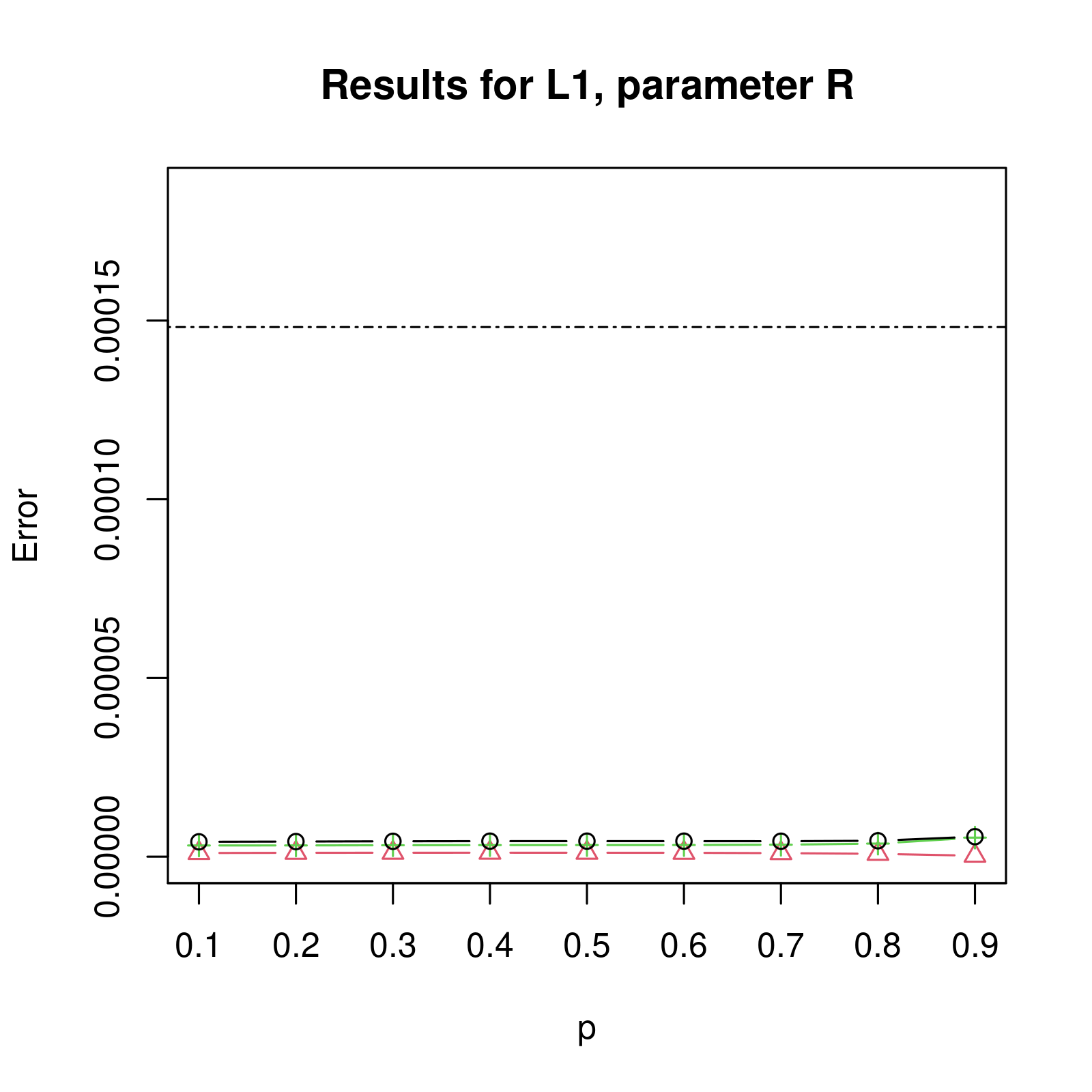}
    \caption{MSE, squared bias and variance for the hard-core model using PPL with the $\Loss_1$ loss function, when estimating the parameters $\beta$ and $R$. Here $k = 100$, $N = 500$,  $p = 0.1,0.2,\ldots,0.9$ and the PPL-weight is set to $p$. The black lines with circles correspond to MSE, the red lines with triangles correspond to squared bias and the green lines with plus signs correspond to variance. The black dotted lines correspond to the Takacs-Fiksel estimates.}
    \label{fig:hard-core-p-L1}
\end{figure}

In Figure \ref{fig:hard-core-(1-p)-L1} we find the results when using the PPL-weight approximation $p/(1-p)$. As expected, the results here are worse than when using the weight $p$. This is reasonable, since 
$p/(1-p)>p$ when $p\in(0,1)$, 
i.e.\ it is larger than its theoretical bound. That PPL performs best when $p=0.1$ here is not surprising, since this is the choice that ensures that $p/(1-p)$ is the closest to $p$.
Note that PPL in the context of a hard-core process has already been investigated by \citet{cronie2021statistical}, who used the weight $p/(1-p)$, and by \citet{jansson2024gibbs}, who used the weight $p$. 
However, in both cases the results were compared to pseudolikelihood estimation, when it would have been more fair to instead compare the results to Takacs-Fiksel estimation with the Stoyan-Grabarnik test function, as we do here. Yet, both of these studies showed that PPL outperformed pseudolikelihood estimation. 

We further emphasise that the implementation of pseudolikelihood estimation for hard-core processes in the \texttt{ppm} function of \textsc{spatstat} uses a plug-in estimate for $R$, in order to handle the associated identifiability issues and thereby have a decent estimate of $\beta$. Interestingly, as we also see here, \citet{cronie2021statistical,jansson2024gibbs} found that PPL does not seem to suffer from such issues. More specifically, we here do not need to plug an external estimate of, say, $R$ into the loss function in order to obtain a good estimate of $\beta$, but we may here jointly minimise the loss function with respect to $R$ and $\beta$,
in order to obtain the parameter estimates. 
It should be emphasised that we are using grid search for the optimisation here, which may have an influence on the performance.

\begin{figure}[!htb]
    \centering
    \includegraphics[width = 0.45\textwidth]{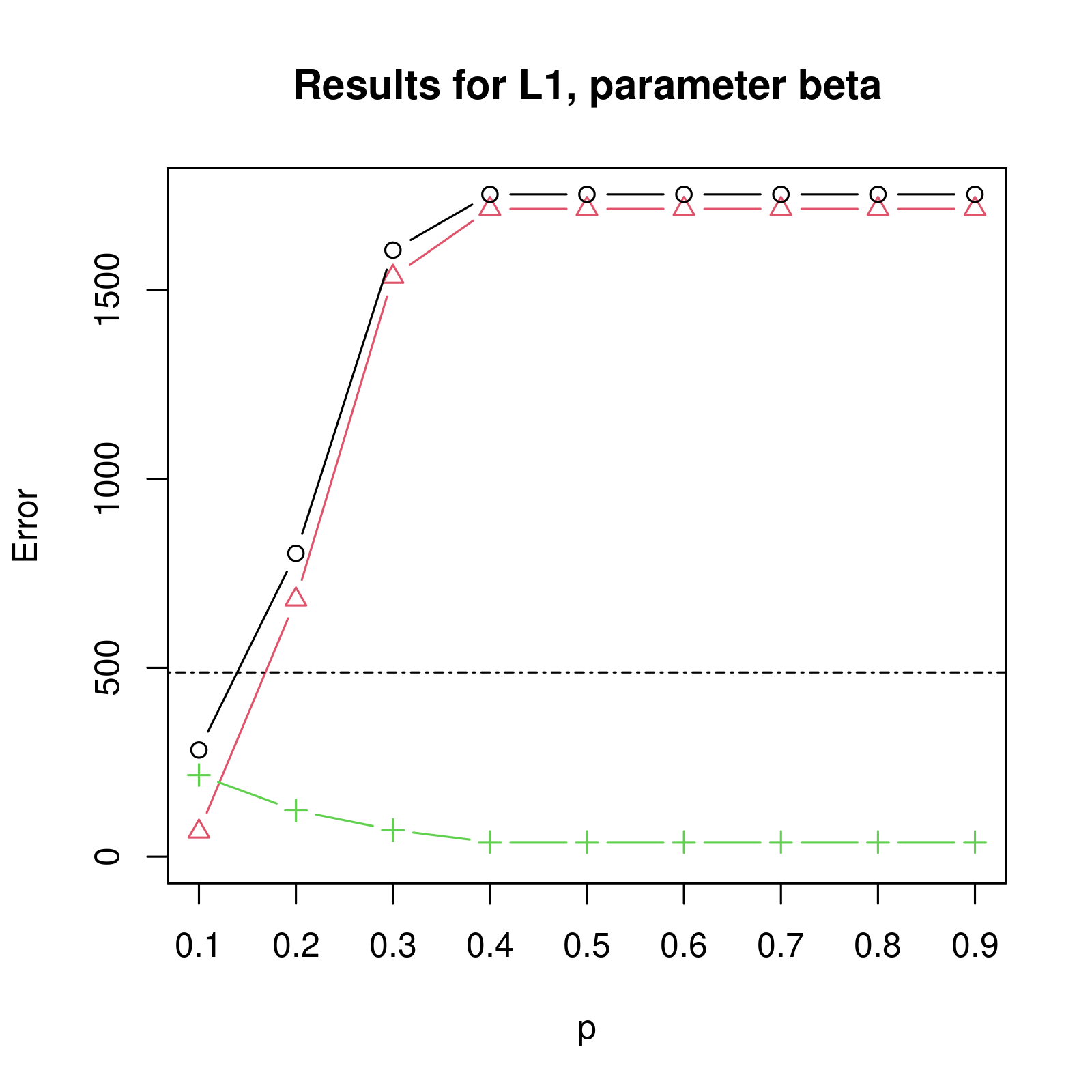}
    \includegraphics[width = 0.45\textwidth]{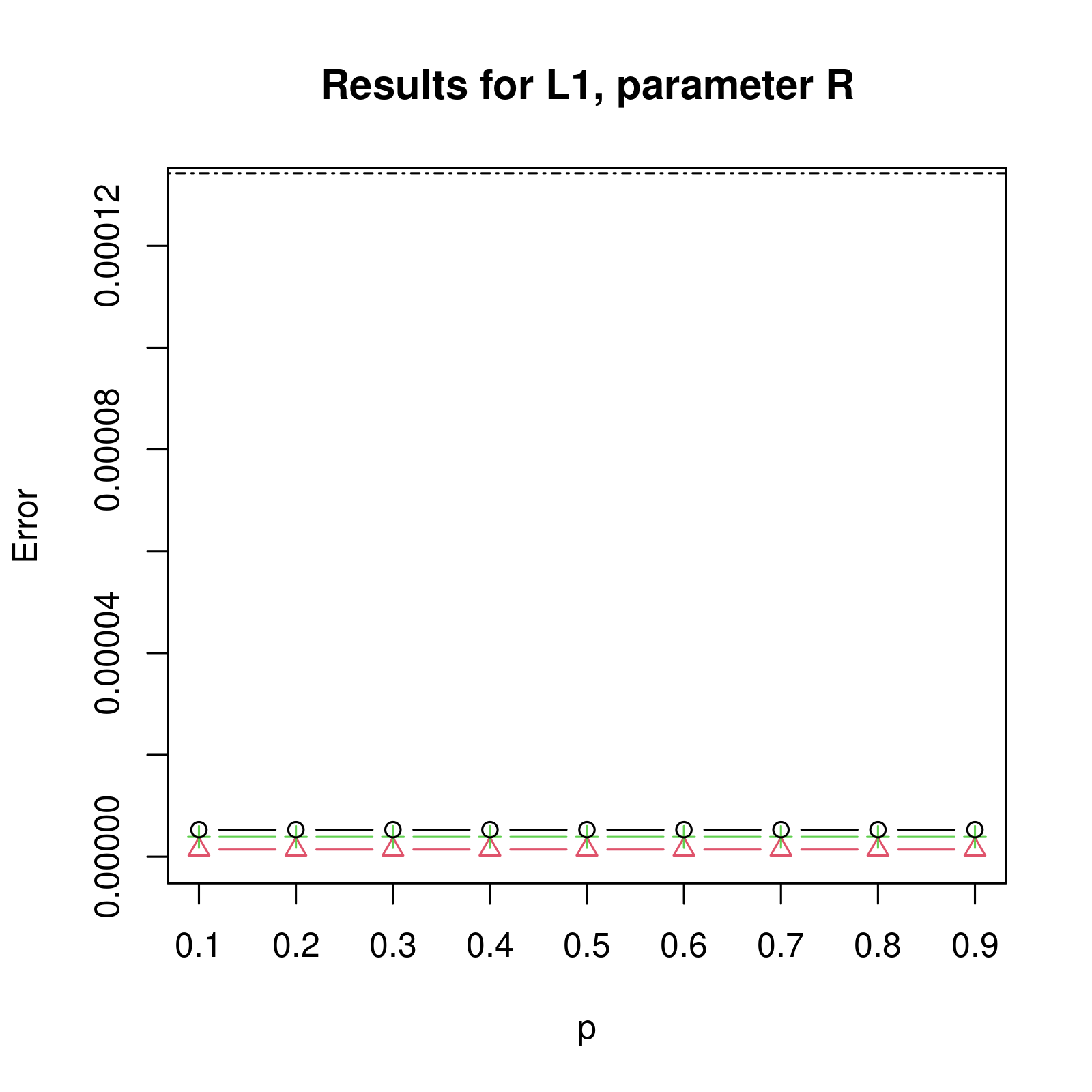}
    \caption{
    MSE, squared bias and variance for the hard-core model using PPL with the $\Loss_1$ loss function, when estimating the parameters $\beta$ and $R$. Here $k = 100$, $N = 100$,  $p = 0.1,0.2,\ldots,0.9$ and the PPL-weight is set to $p/(1-p)$. The black lines with circles correspond to MSE, the red lines with triangles correspond to squared bias and the green lines with plus signs correspond to variance. The black dotted lines correspond to the Takacs-Fiksel estimates.
}
    \label{fig:hard-core-(1-p)-L1}
\end{figure}

Lastly, in Figure \ref{fig:hard-core-est-L1} we find the results when we estimate the PPL-weight in accordance with \eqref{e:WeightEst}. These results are the best so far, as here the MSE for PPL is below the MSE of Takacs-Fiksel estimation for all $p$. The MSE values for $\beta$ for PPL are 
approximately half of the MSE value for 
Takacs-Fiksel estimation.
This seems to be due to a much smaller bias, compared to the bias for the weight $p$ 
in Figure \ref{fig:hard-core-p-L1}, and the bias for the weight $p/(1-p)$ 
in Figure \ref{fig:hard-core-(1-p)-L1}. 
It should be stated though that when $p=0.2$, both the current weight estimation scheme and using the fixed weight $p$ have about the same performance, but the former performs roughly equally well for all $p\geq0.2$.

To summarise the 
observations made 
for the hard-core process, by using the $\Loss_1$ loss function, the weight estimation scheme  \eqref{e:WeightEst}, and $p\geq0.2$ in the Monte-Carlo cross-validation, we have that PPL significantly outperforms Takacs-Fiksel estimation. Moreover, even by using the fixed weight $p\in[0.1,0.4]$ we obtain better results than Takacs-Fiksel estimation, noting that $p=0.2$ seems to be the optimal choice.

\begin{figure}[!htb]
    \centering
    \includegraphics[width = 0.45\textwidth]{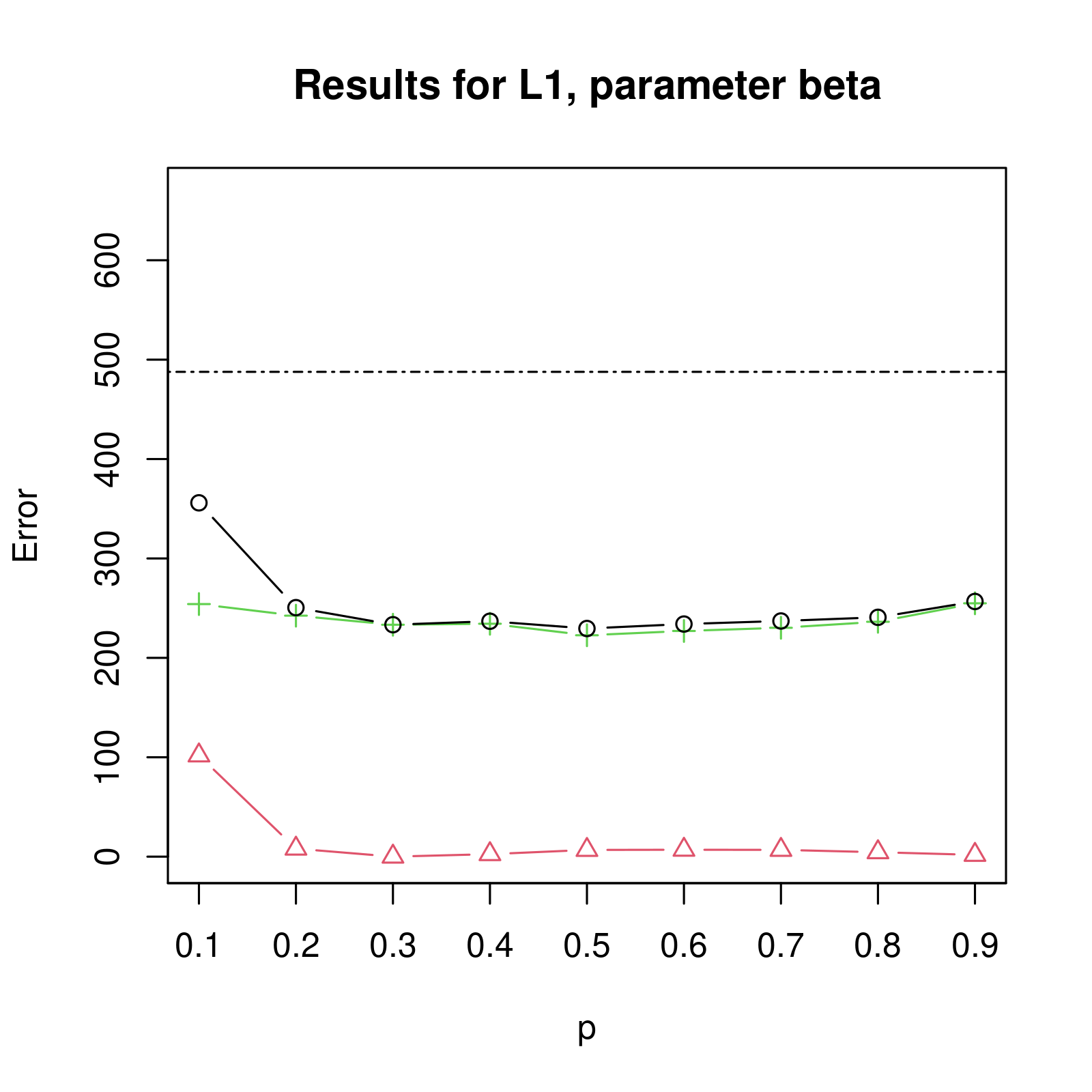}
    \includegraphics[width = 0.45\textwidth]{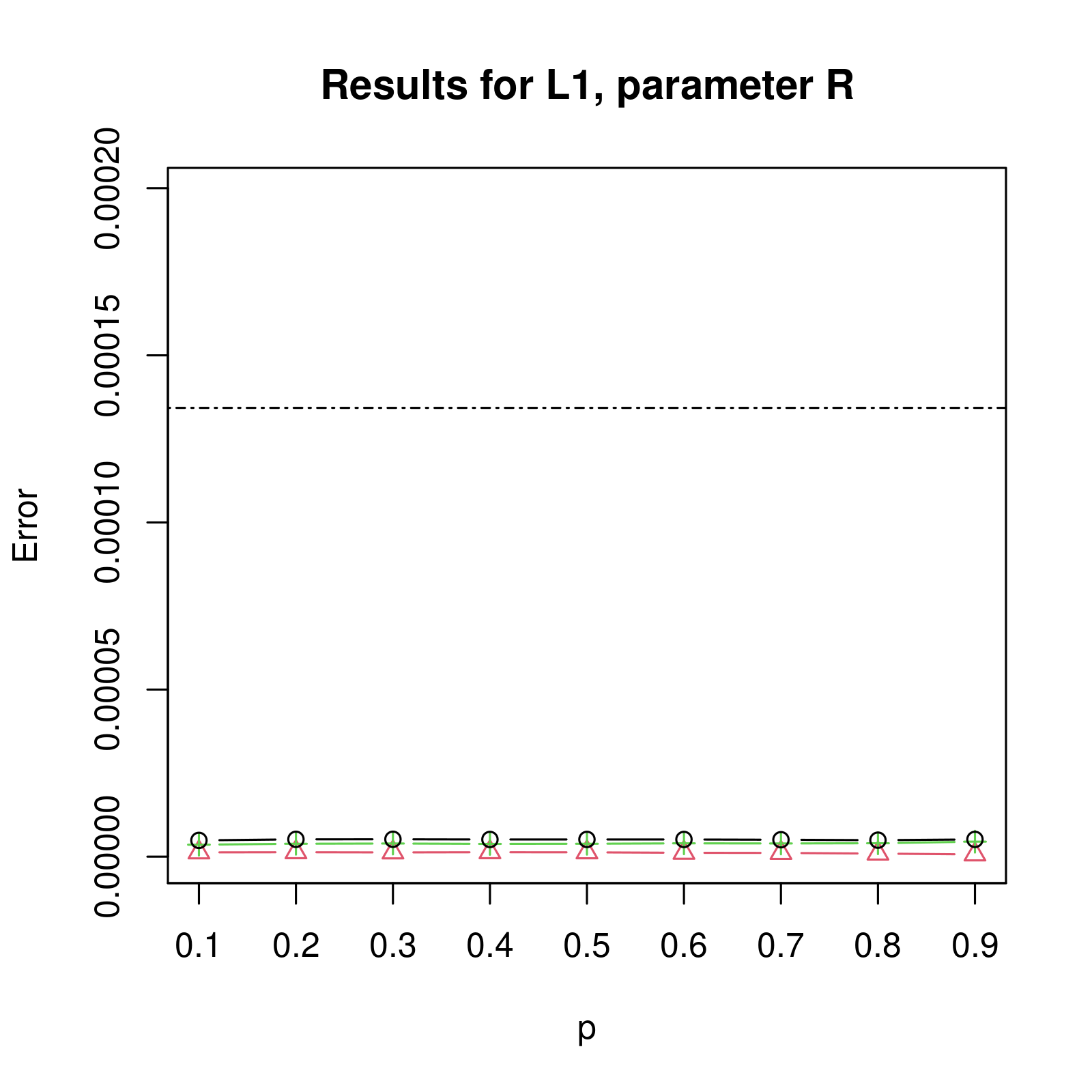}
    \caption{
    MSE, squared bias and variance for the hard-core model using PPL with the $\Loss_1$ loss function, when estimating the parameters $\beta$ and $R$. Here $k = 100$, $N = 100$,  $p = 0.1,0.2,\ldots,0.9$ and the PPL-weight is estimated in accordance with \eqref{e:WeightEst}. The black lines with circles correspond to MSE, the red lines with triangles correspond to squared bias and the green lines with plus signs correspond to variance. The black dotted lines correspond to the Takacs-Fiksel estimates.
}
    \label{fig:hard-core-est-L1}
\end{figure}

\subsubsection{Strauss process}
\label{sec:strauss_sims}
For the Strauss process we use the parameters $R = 0.05$, $\beta = 100$ and $\gamma = 0.5$,
which on average yields around 75 points per point pattern; see Figure \ref{fig:StraussPP} for an example of a realisation with this parametrisation. The grid where we searched for the parameters was $50,55,\ldots,150$ for $\beta$, $0.0350,0.0365,\ldots,0.0650$ for $R$ and $0.10,0.14,\ldots,0.90$ for $\gamma$. We here present the results for the $\Loss_1$ loss function while the results for the $\Loss_2$ and $\Loss_3$ loss functions can be found in Figure \ref{fig:strauss-p-L2L3}
--
\ref{fig:strauss-est-L2L3} 
in Section \ref{sec:app_strauss} of the Appendix. The results for the $\Loss_1$ and $\Loss_2$ loss functions are very similar, and the results for the $\Loss_3$ loss function are generally worsethan the results for both $\Loss_1$ and $\Loss_2$.

First, in Figure \ref{fig:strauss-p-L1} we consider the case where the PPL-weight is set to $p$. We observe that 
PPL's MSE performance for
$\beta$ is better than the performance for Takacs-Fiksel estimation, for all $p$ except for $p = 0.8,0.9$. More specifically, for intermediate $p\in[0.3,0.7]$ the MSE is approximately half of that of Takacs-Fiksel estimation. 
Hence, the general PPL performance significantly outperforms the performance of Takacs-Fiksel estimation when $p\in[0.3,0.7]$.
For the $R$ parameter 
the MSE is on par with Takacs-Fiksel estimation for $p=0.1$ but is then increasing almost linearly from around 0.00013 ($p=0.1$) to around 0.00017 ($p=0.9$). It is further interesting that we here see a very clear bias-variance trade-off as $p$ varies.
Note again here that the error for both methods is much lower than the magnitude of the parameter $R$ (0.3\%) so performance difference for $R$ between the two methods is considered negligible.
Moving on to $\gamma$, the MSE is high 
for small $p$ but quickly decreases as $p$ increases, with MSE values being smaller than the MSE value for Takacs-Fiksel estimation when $p>0.3$. In fact, the MSE for Takacs-Fiksel estimation is more than twice that of PPL when $p\in[0.6,0.8]$. 
That there is not a universally optimal value for $p$ to be found here is very likely an effect of non-identifiability. 
Hence, we are forced to make a choice such that we prioritise a good performance for either $R$ (small $p$) or $\gamma$ (large $p$). If we were to recommend a fixed value for $p$ here, our suggestion would be to let $p\in[0.4,0.7]$, as we here obtain good estimates for $\beta$ and $\gamma$ but pay an acceptable MSE price in terms of $R$.

\begin{figure}[!htb]
    \centering
    \includegraphics[width = 0.3\textwidth]{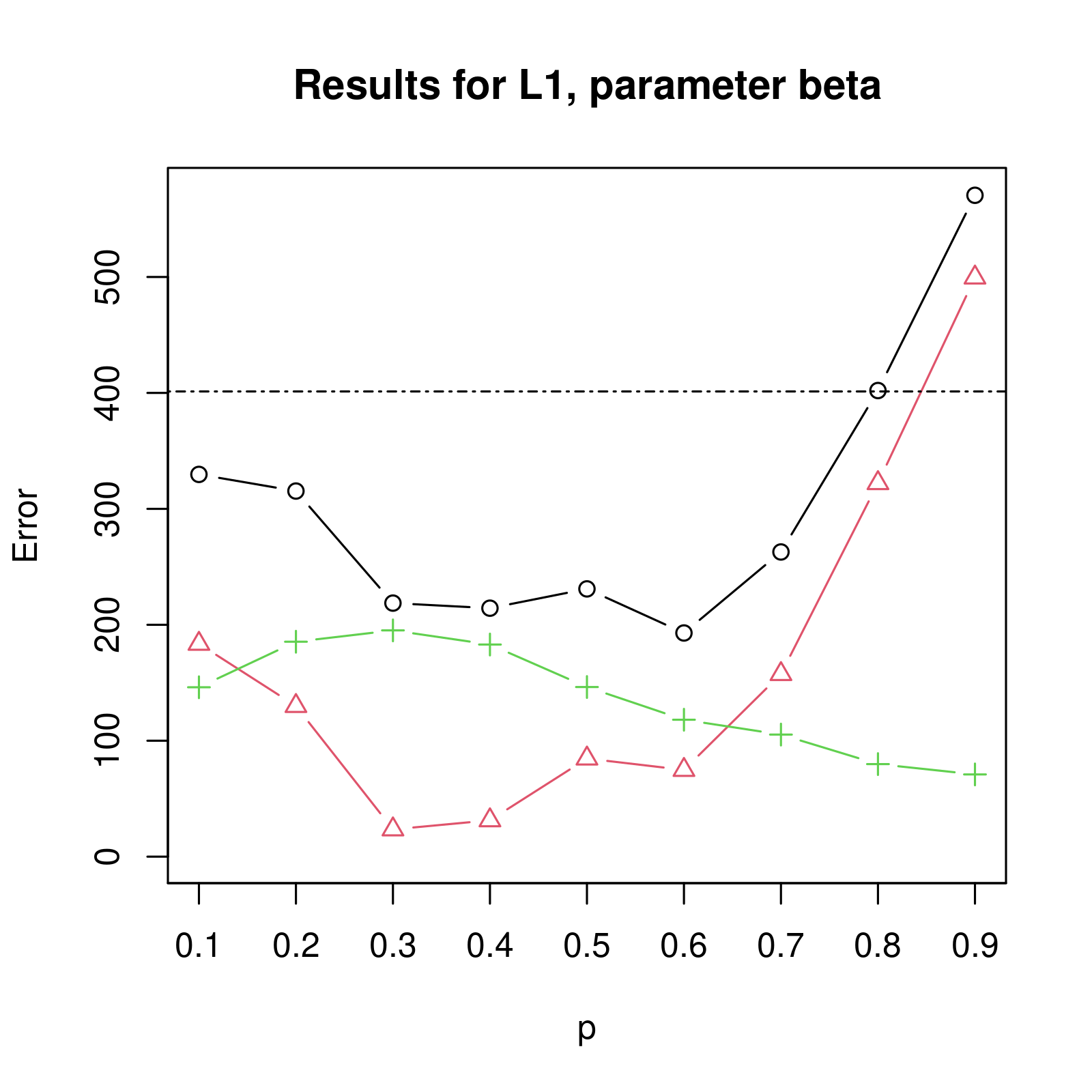}
    \includegraphics[width = 0.3\textwidth]{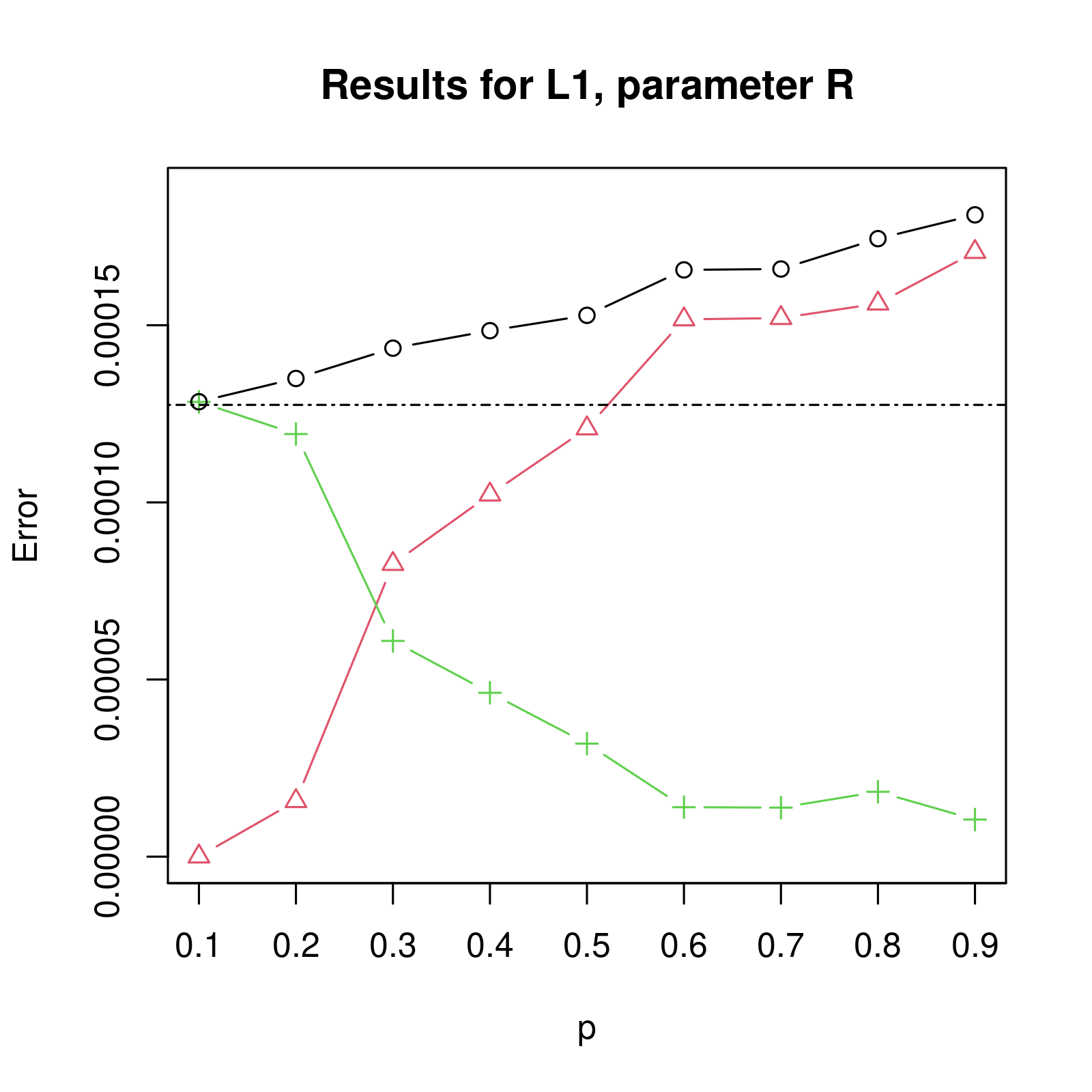}
    \includegraphics[width = 0.3\textwidth]{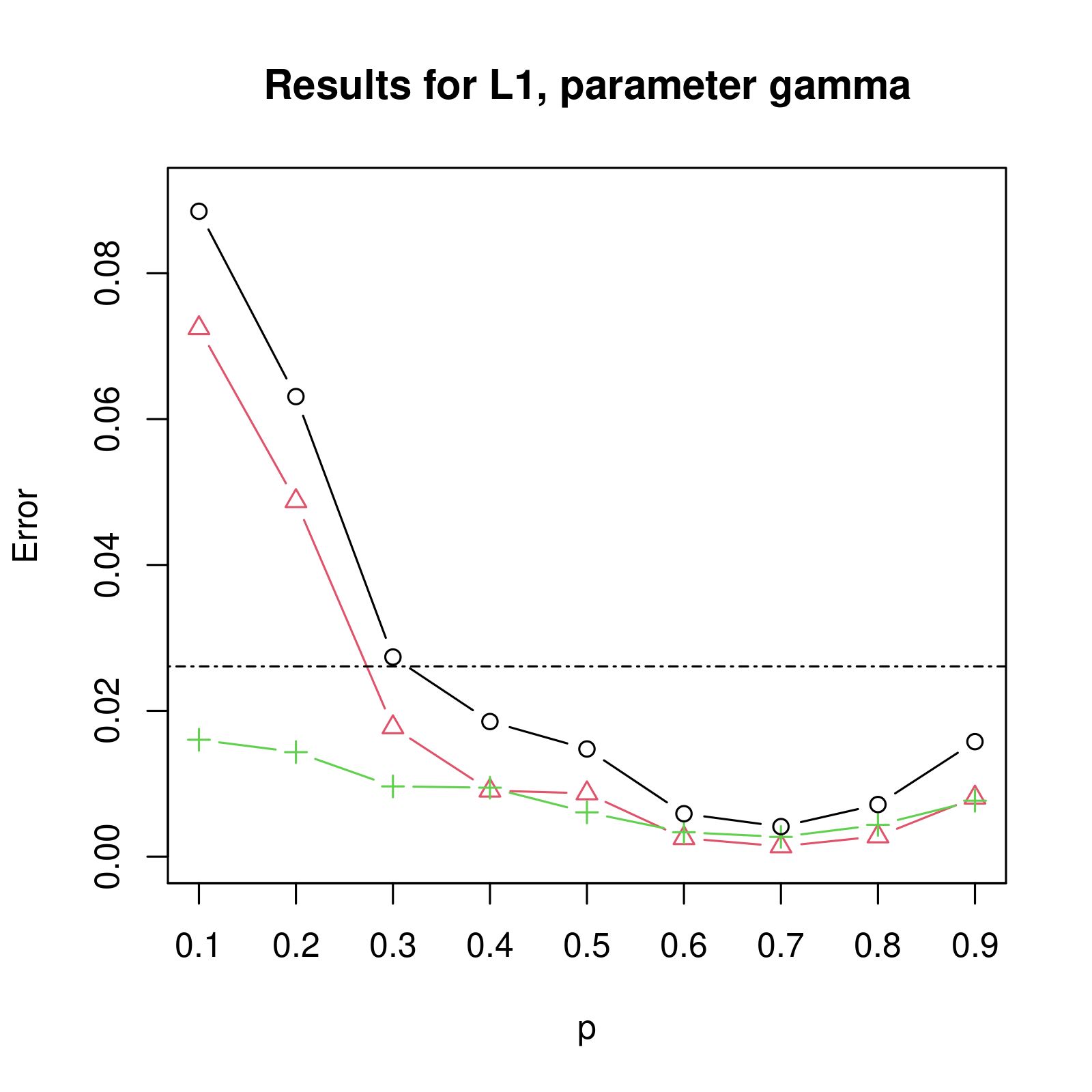}
    \caption{
    MSE, squared bias and variance for the hard-core model using PPL with the $\Loss_1$ loss function, when estimating the parameters 
$\beta$, $R$ and $\gamma$. 
Here $k = 100$, $N = 100$, $p = 0.1,0.2,\ldots,0.9$ and the PPL-weight is set to $p$. 
    The black lines with circles correspond to MSE, the red lines with triangles correspond to squared bias and the green lines with plus signs correspond to variance. The black dotted lines correspond to the Takacs-Fiksel estimates.
}
    \label{fig:strauss-p-L1}
\end{figure}

In Figure \ref{fig:strauss-(1-p)-L1} we find the results obtained when employing the PPL-weight estimate $p/(1-p)$. The results are in general worse than when using the weight estimate $p$, which is to be expected since a Strauss process is repulsive, and thereby has a PPL-weight which is at most $p<p/(1-p)$.

\begin{figure}[!htb]
    \centering
    \includegraphics[width = 0.3\textwidth]{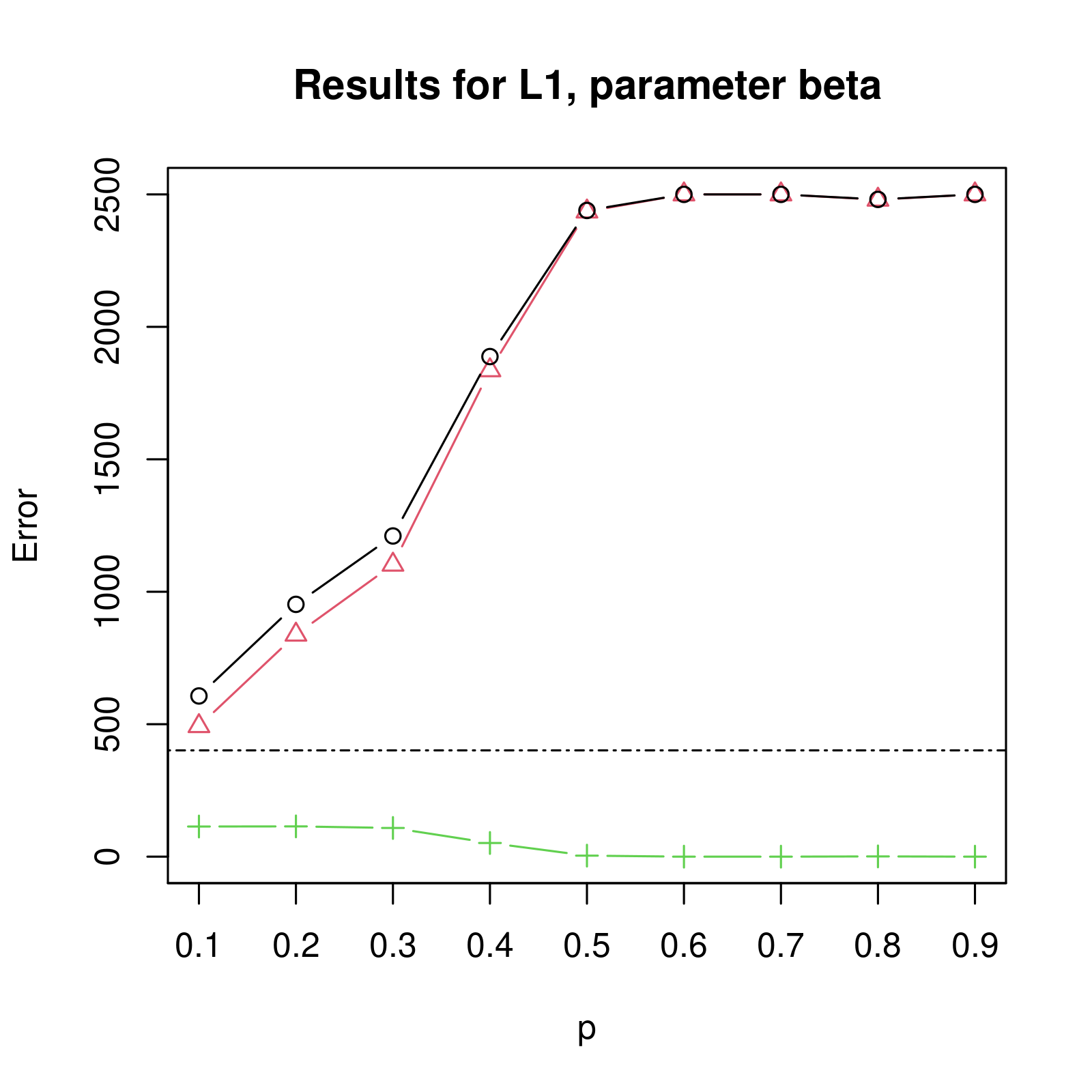}
    \includegraphics[width = 0.3\textwidth]{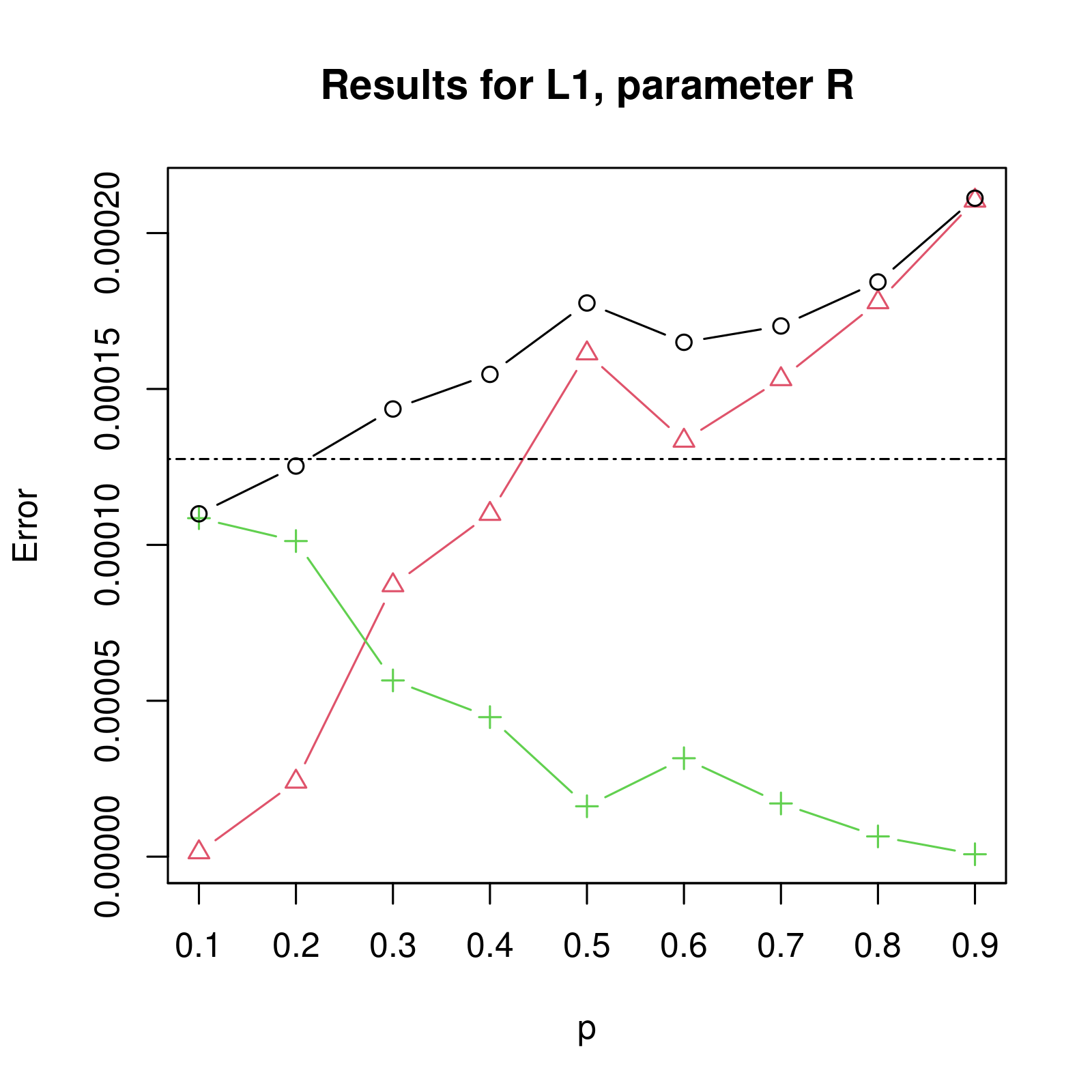}
    \includegraphics[width = 0.3\textwidth]{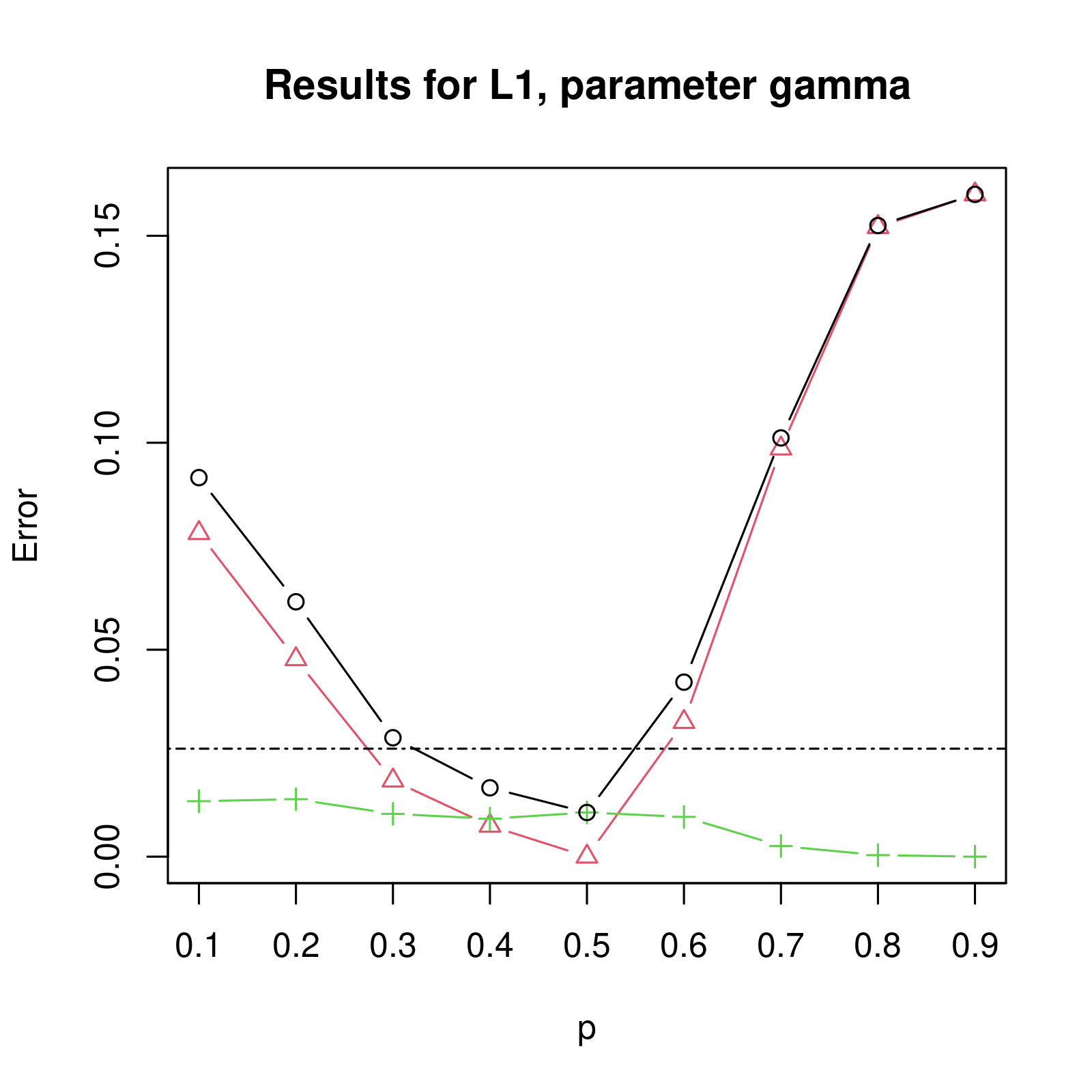}
    \caption{
    MSE, squared bias and variance for the hard-core model using PPL with the $\Loss_1$ loss function, when estimating the parameters 
    $\beta$, $R$ and $\gamma$. 
    Here $k = 100$, $N = 100$, $p = 0.1,0.2,\ldots,0.9$ and the PPL-weight is set to $p/(1-p)$. 
    The black lines with circles correspond to MSE, the red lines with triangles correspond to squared bias and the green lines with plus signs correspond to variance. The black dotted lines correspond to the Takacs-Fiksel estimates.
}
    \label{fig:strauss-(1-p)-L1}
\end{figure}

Lastly, in Figure \ref{fig:strauss-est-L1} we see the results when the PPL-weight is estimated in accordance with \eqref{e:WeightEst}. 
Also here we can find values for $p$ such that the estimators of either $(\beta,R)$ or $(\beta,\gamma)$ perform better than the corresponding Takacs-Fiksel estimators. Hence, again PPL is the preferred choice. Yet, in terms of MSE, the results are not generally better than when using $p$ for the PPL-weight. 
One main difference is that the values for $p$ for which the MSE for $\beta$ with PPL is higher than with Takacs-Fiksel estimation are different compared to when we use the weight $p$ (recall Figure \ref{fig:strauss-p-L1}). 
Further, when using weight estimation the minimum MSE is given by $324$ (attained when $p = 0.1$) while the minimum MSE when fixing the weight to $p$ is given by $214$ (attained when $p = 0.4$). 
The difference in MSE seems to be due to a higher variance when using weight estimation than when using the fixed weight $p$. 
We thus conclude that using $p$ as weight estimate is a sufficiently good approximation of the actual weight in this case.

\begin{figure}[!htb]
    \centering
    \includegraphics[width = 0.3\textwidth]{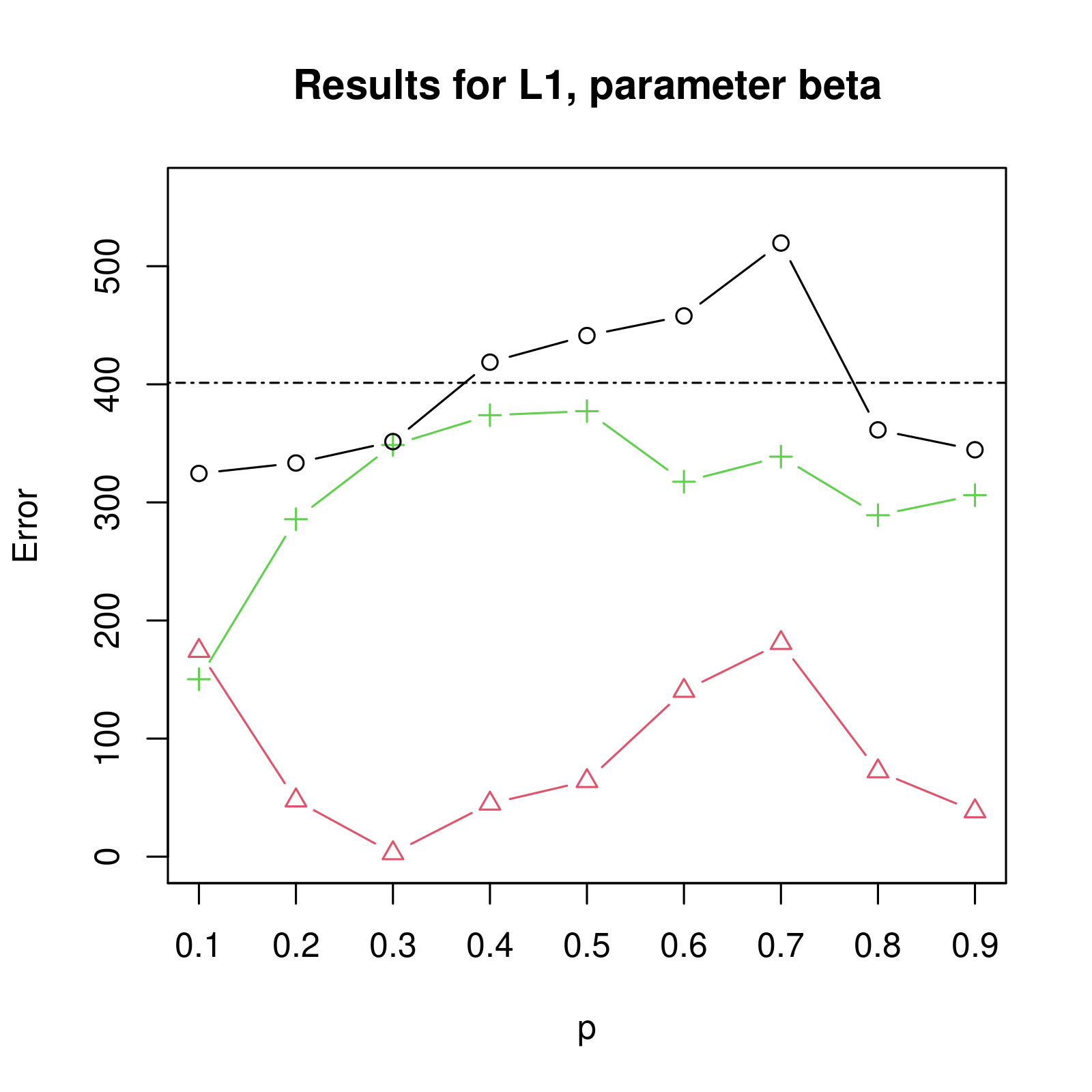}
    \includegraphics[width = 0.3\textwidth]{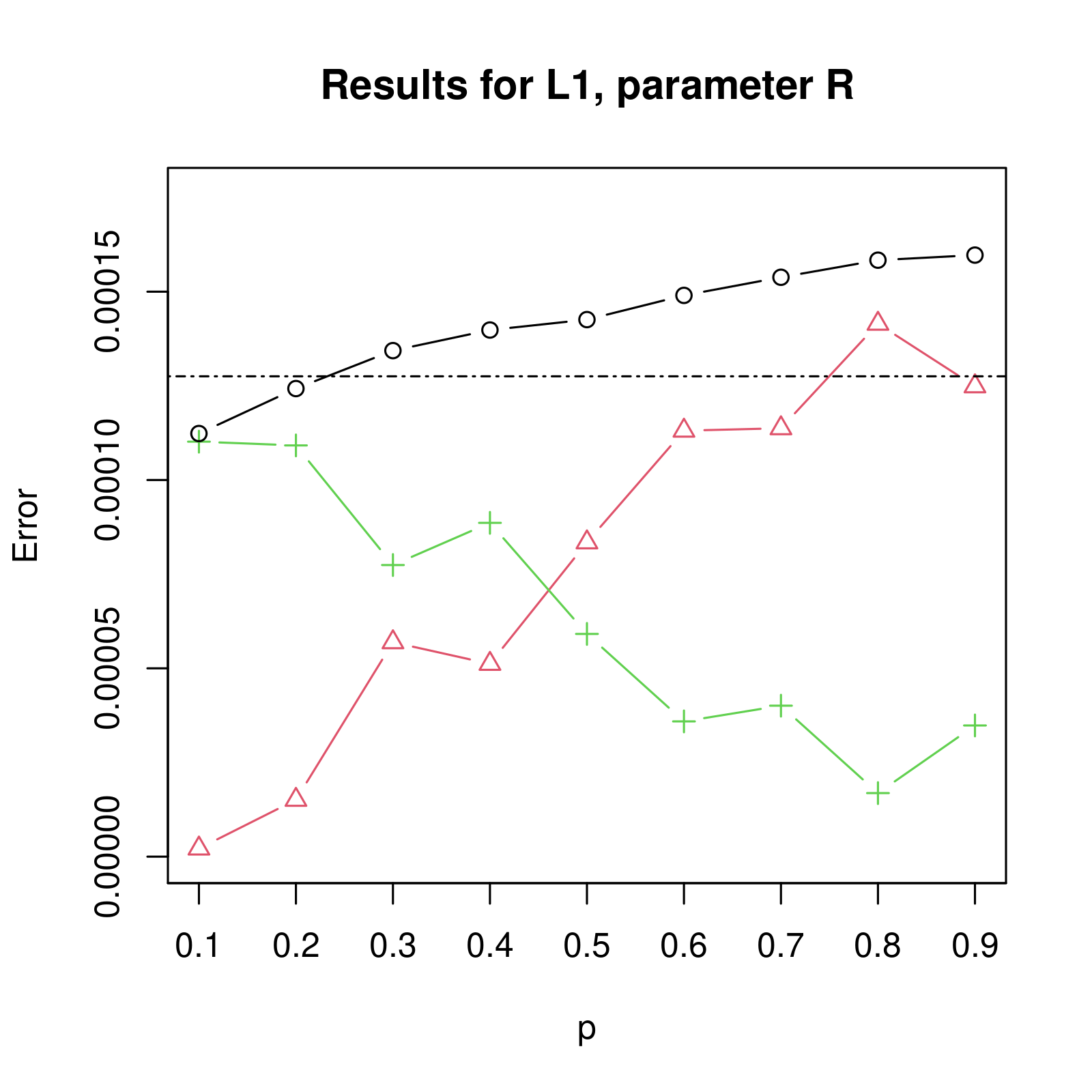}
    \includegraphics[width = 0.3\textwidth]{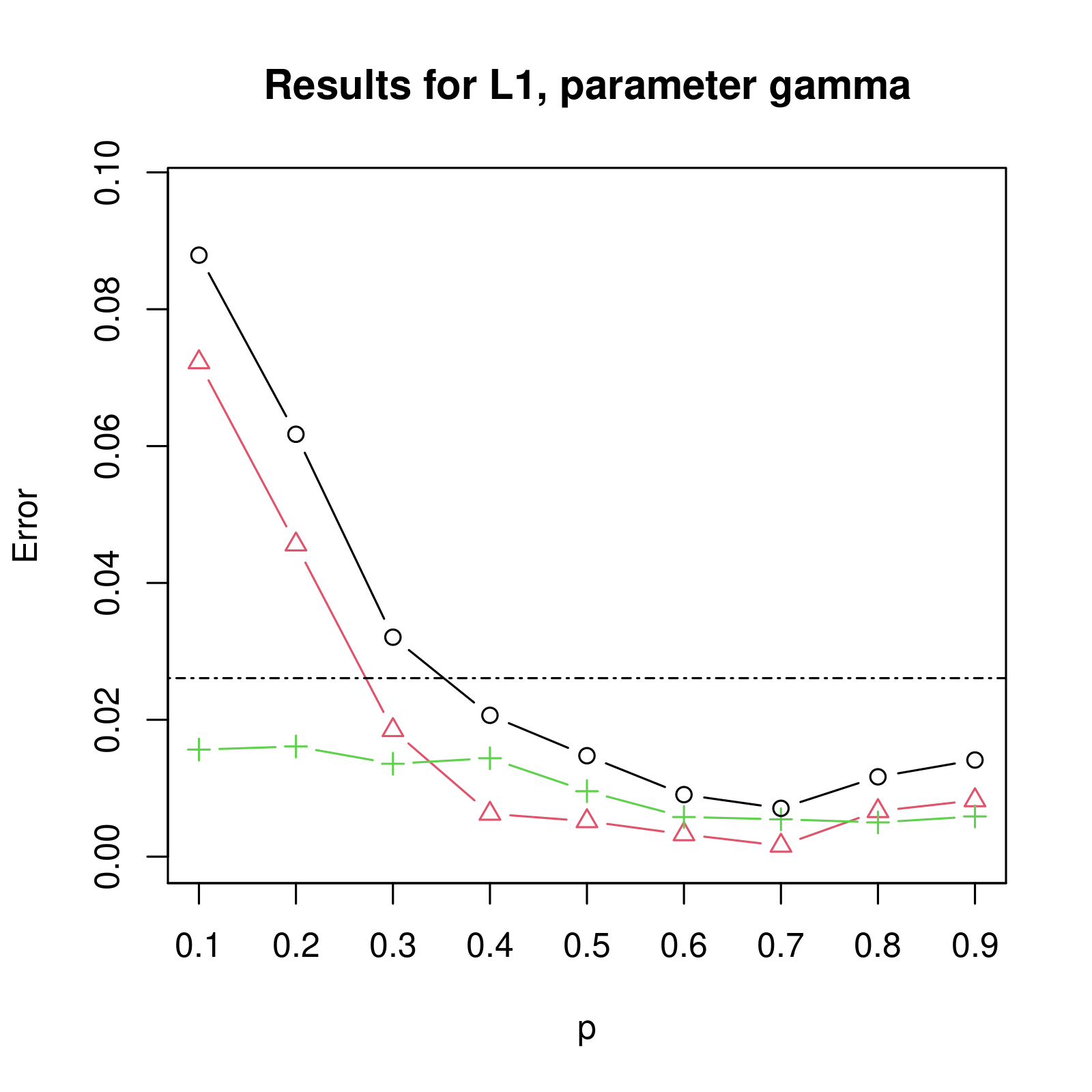}
    \caption{
    MSE, squared bias and variance for the hard-core model using PPL with the $\Loss_1$ loss function, when estimating the parameters 
    $\beta$, $R$ and $\gamma$. 
    Here $k = 100$, $N = 50$, $p = 0.1,0.2,\ldots,0.9$ and the PPL-weight is estimated in accordance with \eqref{e:WeightEst}. 
    The black lines with circles correspond to MSE, the red lines with triangles correspond to squared bias and the green lines with plus signs correspond to variance. The black dotted lines correspond to the Takacs-Fiksel estimates.
}
    \label{fig:strauss-est-L1}
\end{figure}

\subsubsection{Geyer saturation process}
\label{sec:geyer_sims}
For the Geyer saturation process we use the parameters $R = 0.05$, $\beta = 60$, $\gamma = \sqrt{1.5}$ and $s = 2$, where, on average, we have approximately 75 points per point pattern; see Figure \ref{fig:GeyerPP} for an example of a realisation. The grid where we searched for the parameters was $40,45,\ldots,80$ for $\beta$, $0.03500, 0.03875, \ldots, 0.06500$ for $R$, $0.5000, 0.6875, \ldots, 2.0000$ for $\gamma$ and $1.00, 1.25, \ldots, 3.00$ for $s$. In contrast to the previous models, this example represents what is considered an attractive model (although it is formally not; recall the discussion in Section \ref{s:Geyer}). 
We here present the results for the $\Loss_3$ loss function, while the results for the $\Loss_1$ and $\Loss_2$ loss functions can be found in Figure \ref{fig:geyer_p-L1L2}
--
\ref{fig:geyer_w-L1L2} in Appendix \ref{sec:app_geyer}. 
We find that the $\Loss_3$ loss function performs better than both the $\Loss_1$ and the $\Loss_2$ loss functions, and that the results for  $\Loss_1$ and $\Loss_2$ are very similar. 

First we look at the case when the PPL-weight is set to $p$, which is illustrated in Figure \ref{fig:geyer_p-L3}. For the parameter $\beta$ the MSE for PPL is slightly lower than the MSE for Takacs-Fiksel estimation for most values of $p$, while for the parameter $R$ the MSE for PPL is instead slightly 
higher for most values of $p$; the values are of similar scale though. As seen before for the hard-core and Strauss processes, the errors for both methods are much lower than the magnitude of the $R$ parameter (in this case 0.2\%) so differences between the methods are considered negligible.
For the parameter $\gamma$, for all values of $p$, the MSE for PPL is lower than for Takacs-Fiksel estimation. Lastly, for the parameter $s$ the MSE for PPL is next to indistinguishable from that of Takacs-Fiksel estimation for any $p$.
We observe that for $\beta$, $R$ and $s$, the bias is close to zero, whereby the MSE values consist of mostly variance contributions. This differs from the results for the Poisson, hard-core and Strauss processes, using
the weight $p$ and the loss function $\Loss_1$ (recall Figure \ref{fig:poisson3-L1}, \ref{fig:hard-core-p-L1} and \ref{fig:strauss-p-L1}), where the bias usually played a larger role than the variance. 
It is not completely clear to us why this is the case.

\begin{figure}[!htb]
    \centering
    \includegraphics[width = 0.4\textwidth]{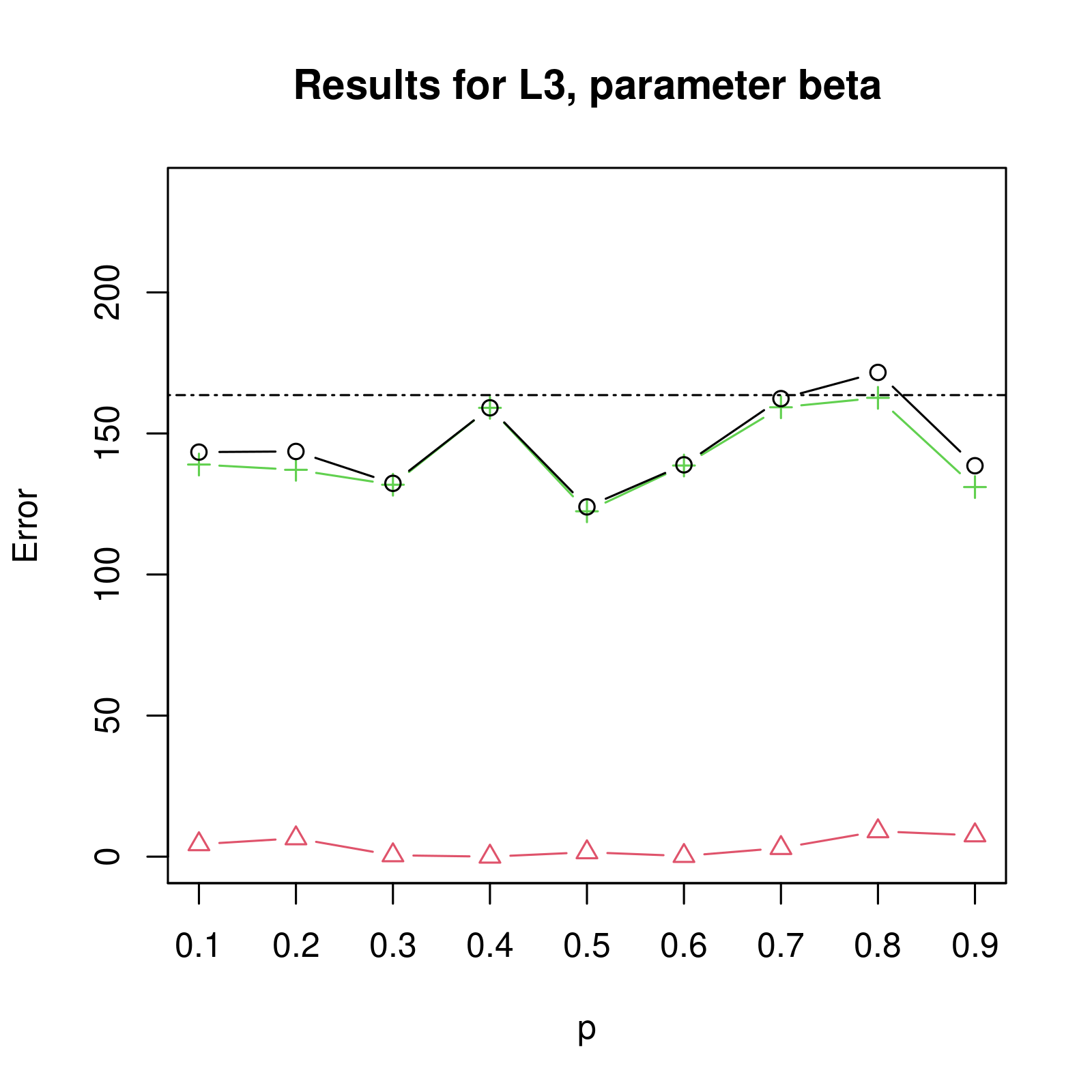}
    \includegraphics[width = 0.4\textwidth]{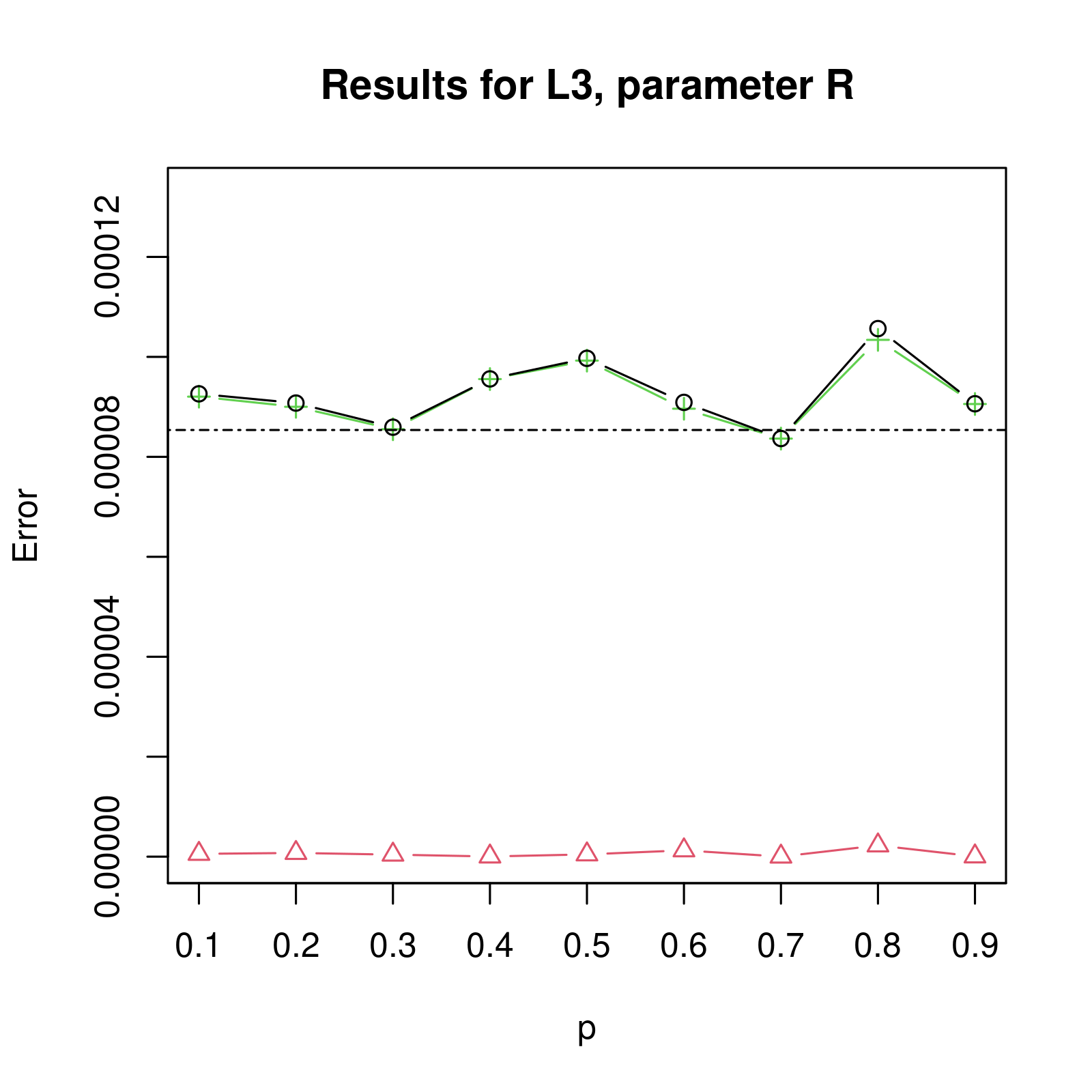}
    
    \includegraphics[width = 0.4\textwidth]{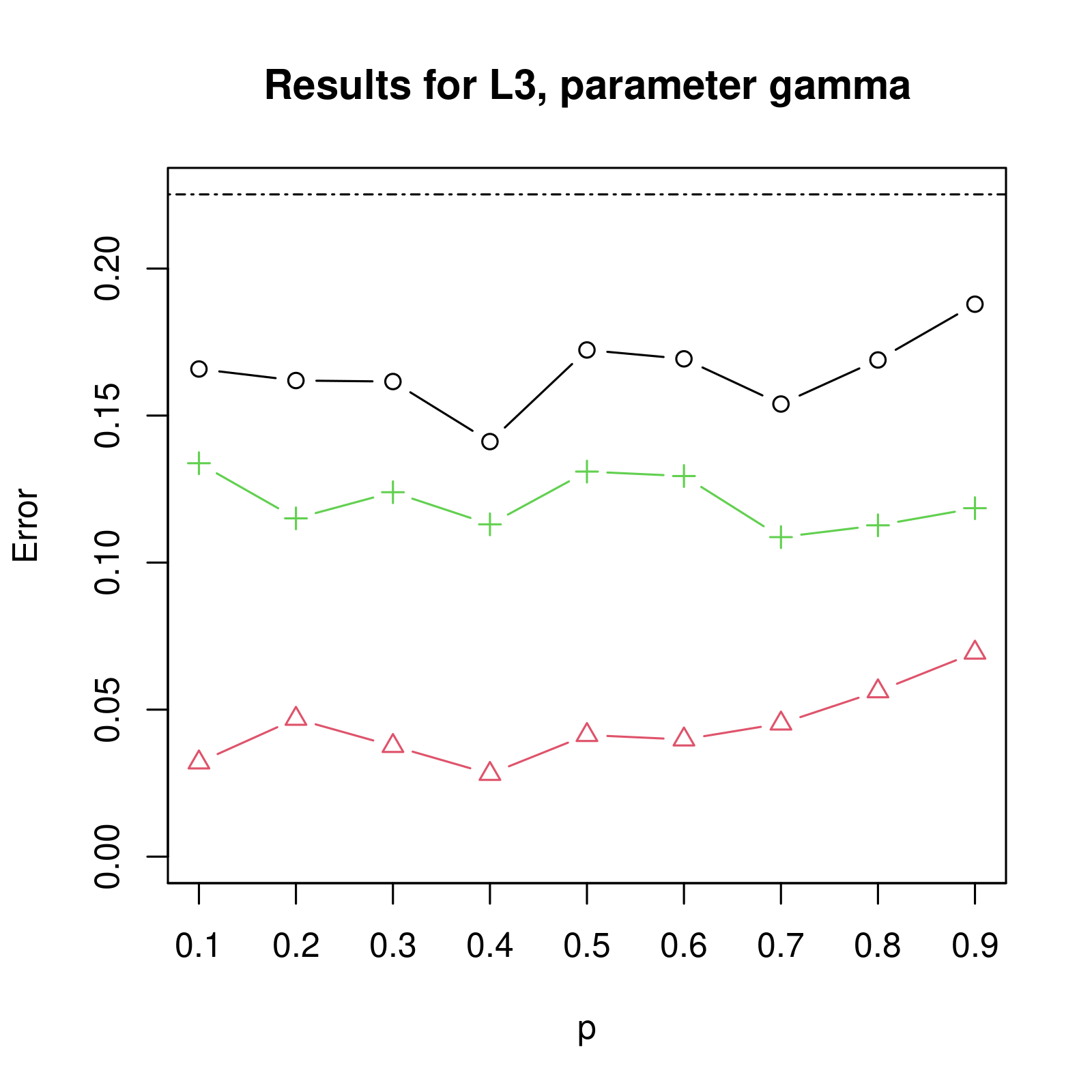}
    \includegraphics[width = 0.4\textwidth]{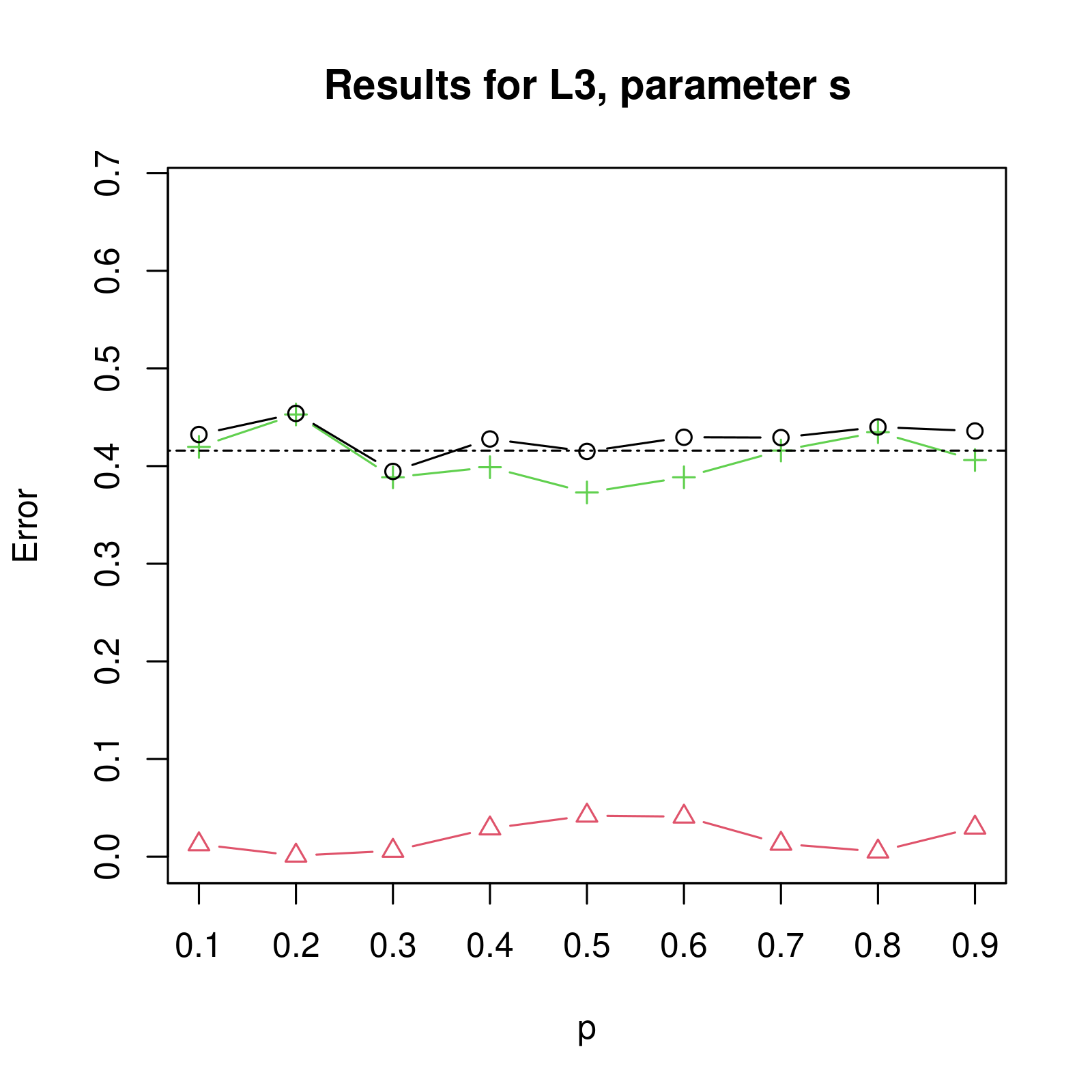}
    \caption{
    MSE, squared bias and variance for the Geyer saturation model using PPL with the $\Loss_3$ loss function, when estimating the parameters 
    $\beta$, $R$, $\gamma$ and $s$. 
    Here $k = 100$, $N = 100$, $p = 0.1,0.2,\ldots,0.9$ and the PPL-weight is set to $p$. 
    The black lines with circles correspond to MSE, the red lines with triangles correspond to squared bias and the green lines with plus signs correspond to variance. The black dotted lines correspond to the Takacs-Fiksel estimates.
}
    \label{fig:geyer_p-L3}
\end{figure}

In Figure \ref{fig:geyer_p/(1-p)-L3} we see the results when employing the PPL-weight estimate $p/(1-p)$. The MSE values for PPL for the estimators for $\beta$, $R$ and $s$ are similar to those obtained through Takacs-Fiksel estimation. For $\gamma$, however, the MSE for PPL is higher than for Takacs-Fiksel estimation when $p\geq0.5$, likely since the PPL-weight is larger than 1 in this case. Note also that comparing with the bias for the $\gamma$ estimator using the weight $p$ (see Figure \ref{fig:geyer_p-L3}), the bias rendered by the weight choice $p/(1-p)$ is much larger for $p\geq 0.5$, which is causing the increase in MSE.

\begin{figure}[!htb]
    \centering
    \includegraphics[width = 0.4\textwidth]{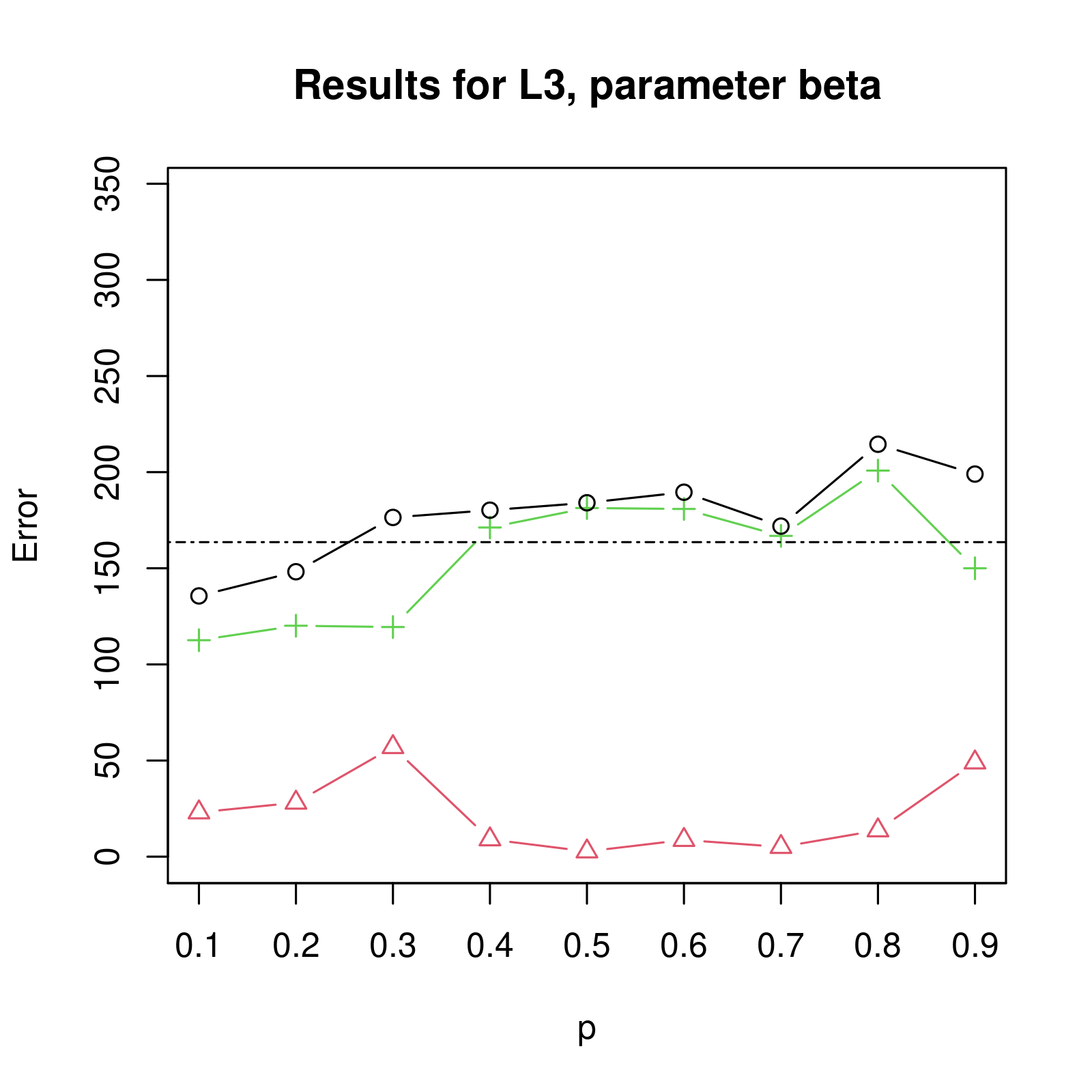}
    \includegraphics[width = 0.4\textwidth]{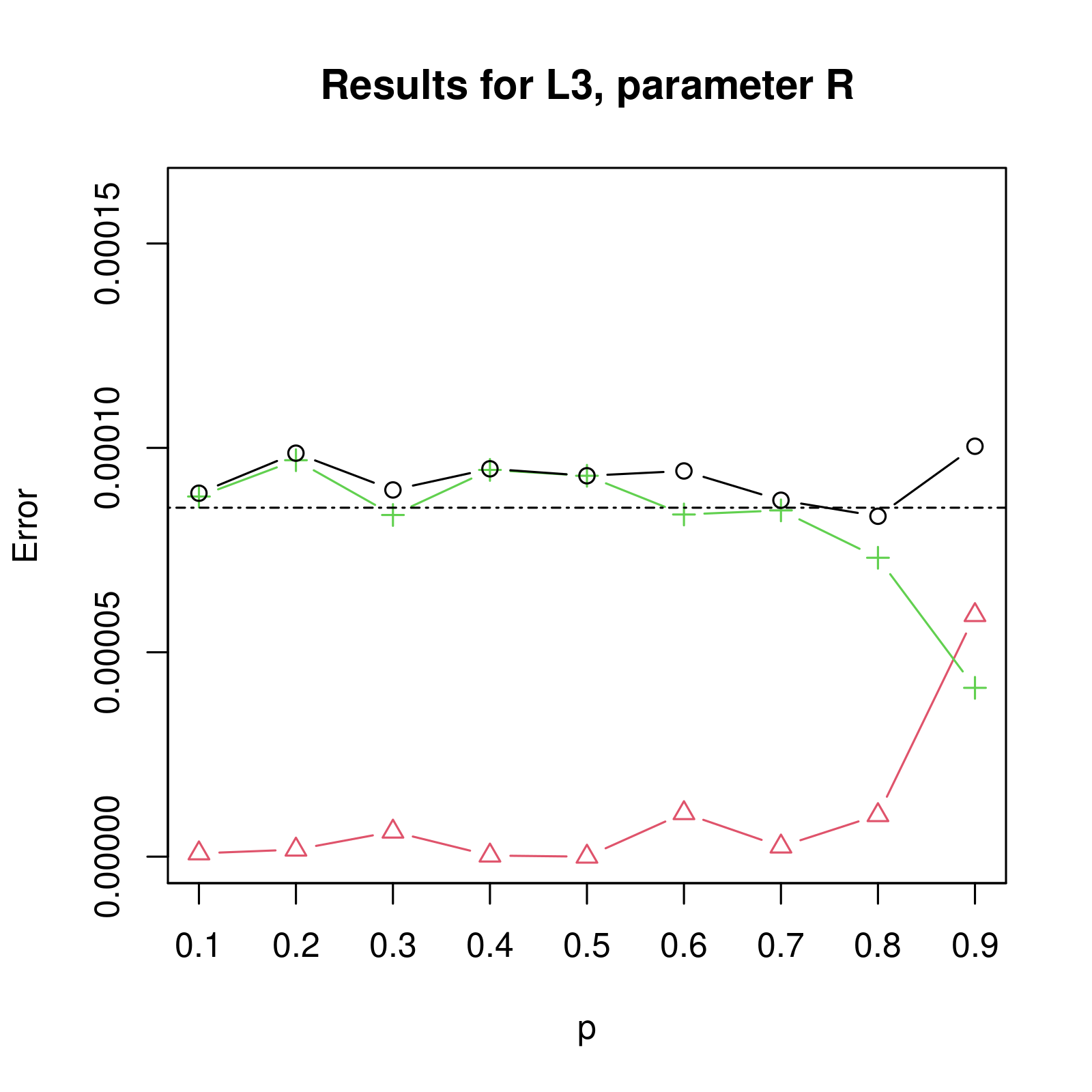}
    
    \includegraphics[width = 0.4\textwidth]{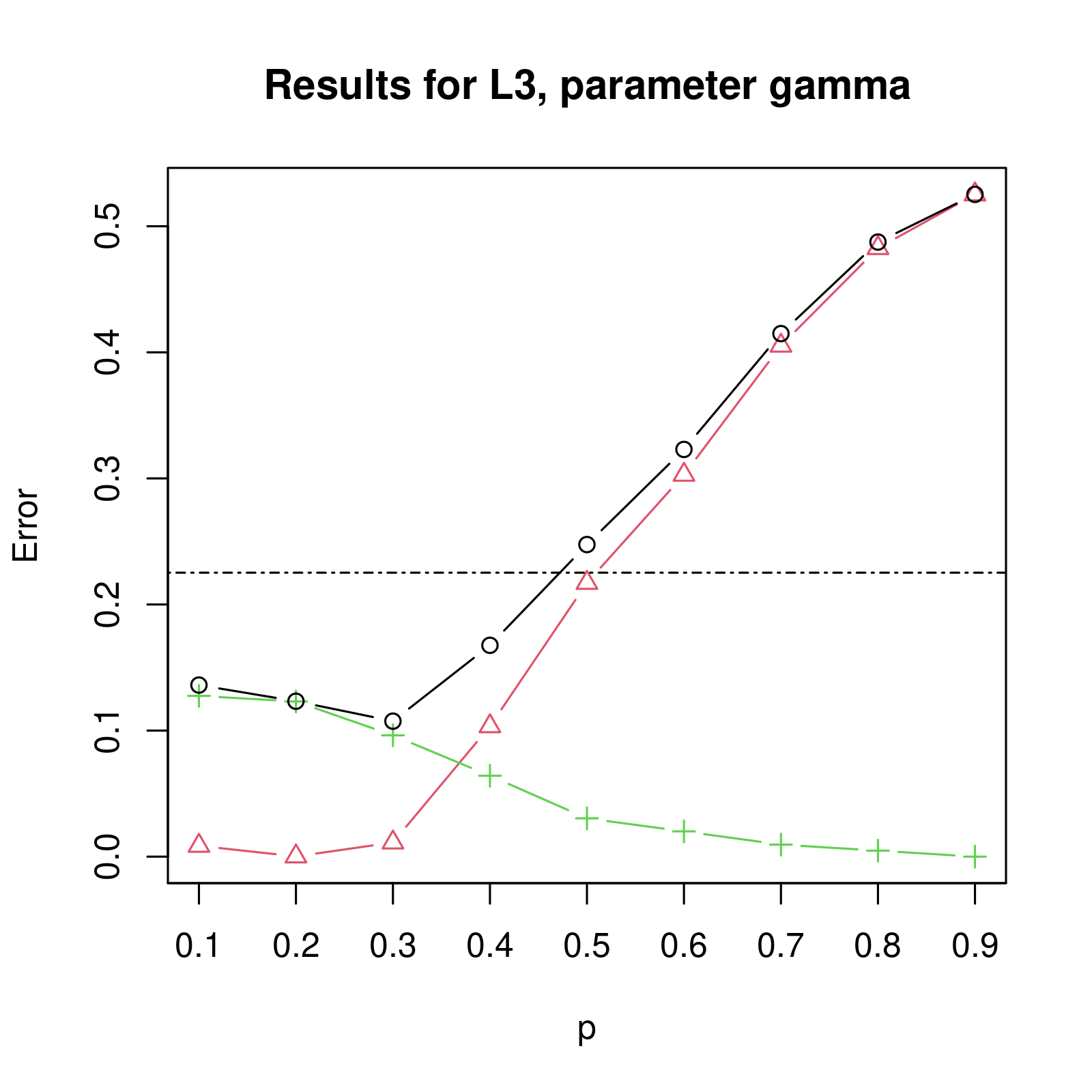}
    \includegraphics[width = 0.4\textwidth]{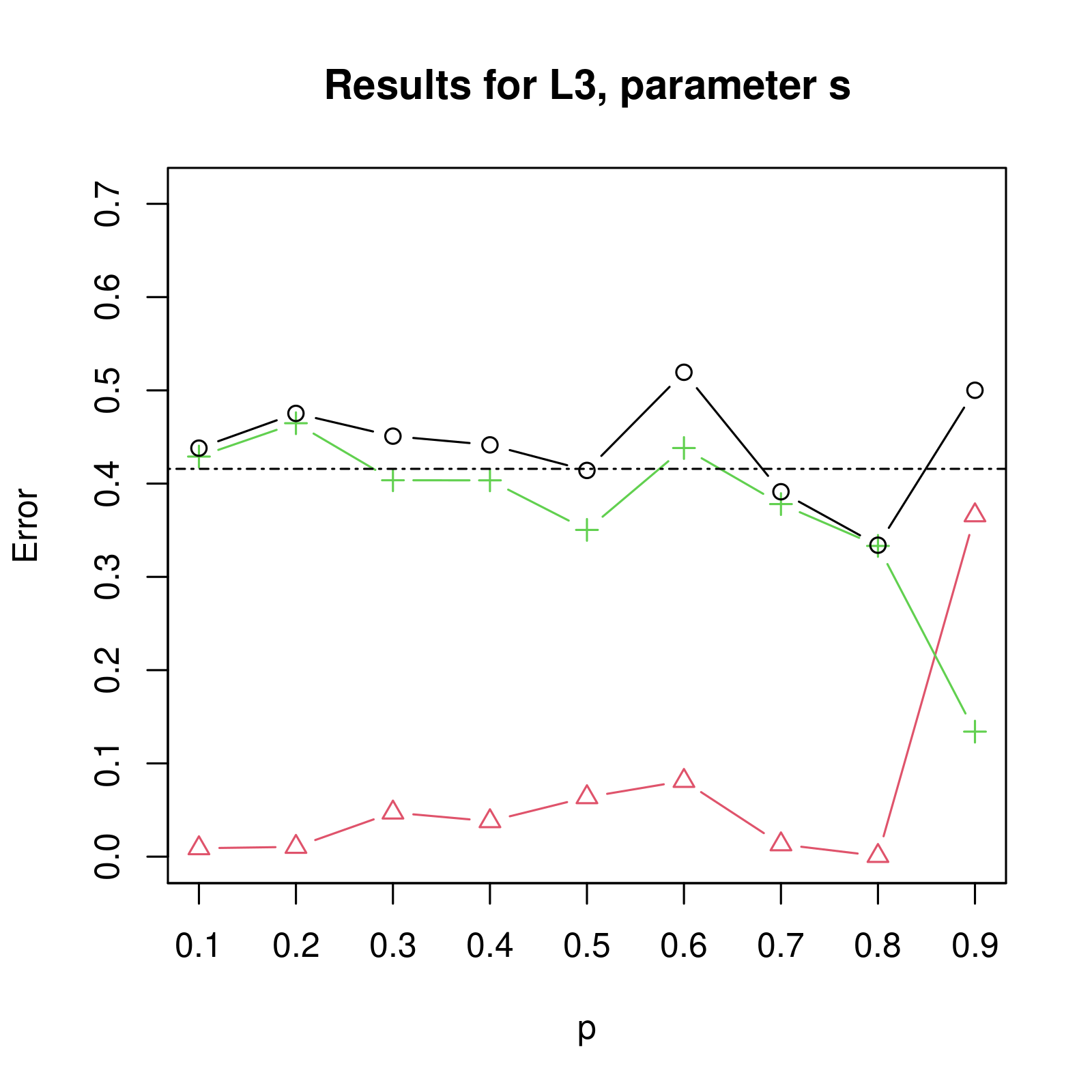}
    \caption{
    MSE, squared bias and variance for the Geyer saturation model using PPL with the $\Loss_3$ loss function, when estimating the parameters 
    $\beta$, $R$, $\gamma$ and $s$. 
    Here $k = 100$, $N = 100$, $p = 0.1,0.2,\ldots,0.9$ and the PPL-weight is set to $p/(1-p)$. 
    The black lines with circles correspond to MSE, the red lines with triangles correspond to squared bias and the green lines with plus signs correspond to variance. The black dotted lines correspond to the Takacs-Fiksel estimates.
}
    \label{fig:geyer_p/(1-p)-L3}
\end{figure}

Lastly, in Figure \ref{fig:geyer_w-L3} we see the results for when the PPL-weight is estimated based on \eqref{e:WeightEst}. Here the 
performance is 
similar to when the weight is set to $p$. Compared to the hard-core and Strauss processes, the Geyer process seems to be the least affected by the choice of weight. The main difference is that the MSE for $\gamma$ is even lower in this case. Here, the bias is small for all four parameters, so the MSE values consist of mostly variance contributions. 

To conclude, the MSE for PPL is similar to or better than Takacs-Fiksel estimation for the Geyer saturation process when using the $\Loss_3$ loss function; for both $\beta, R$ and $s$, the two methods perform identically but with general PPL we seem to generally obtain better estimates of $\gamma$. This holds for when we fix the weight to $p$ or estimate the weight, where we further note that the MSE values seem quite stable when varying $p$. This is not the case for the hard-core and Strauss processes, where the MSE is more affected by the choice of $p$.

\begin{figure}[!htb]
    \centering
    \includegraphics[width = 0.4\textwidth]{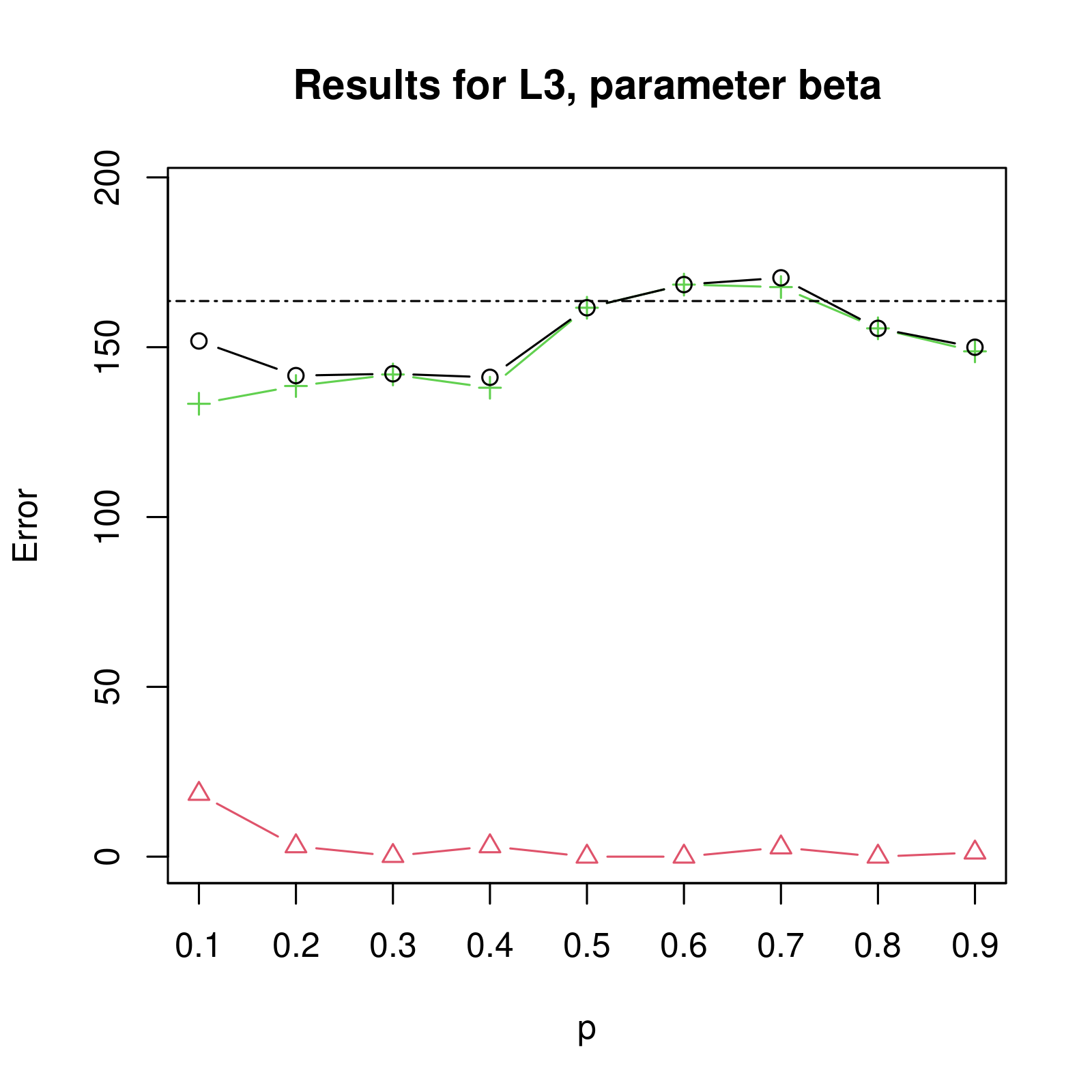}
    \includegraphics[width = 0.4\textwidth]{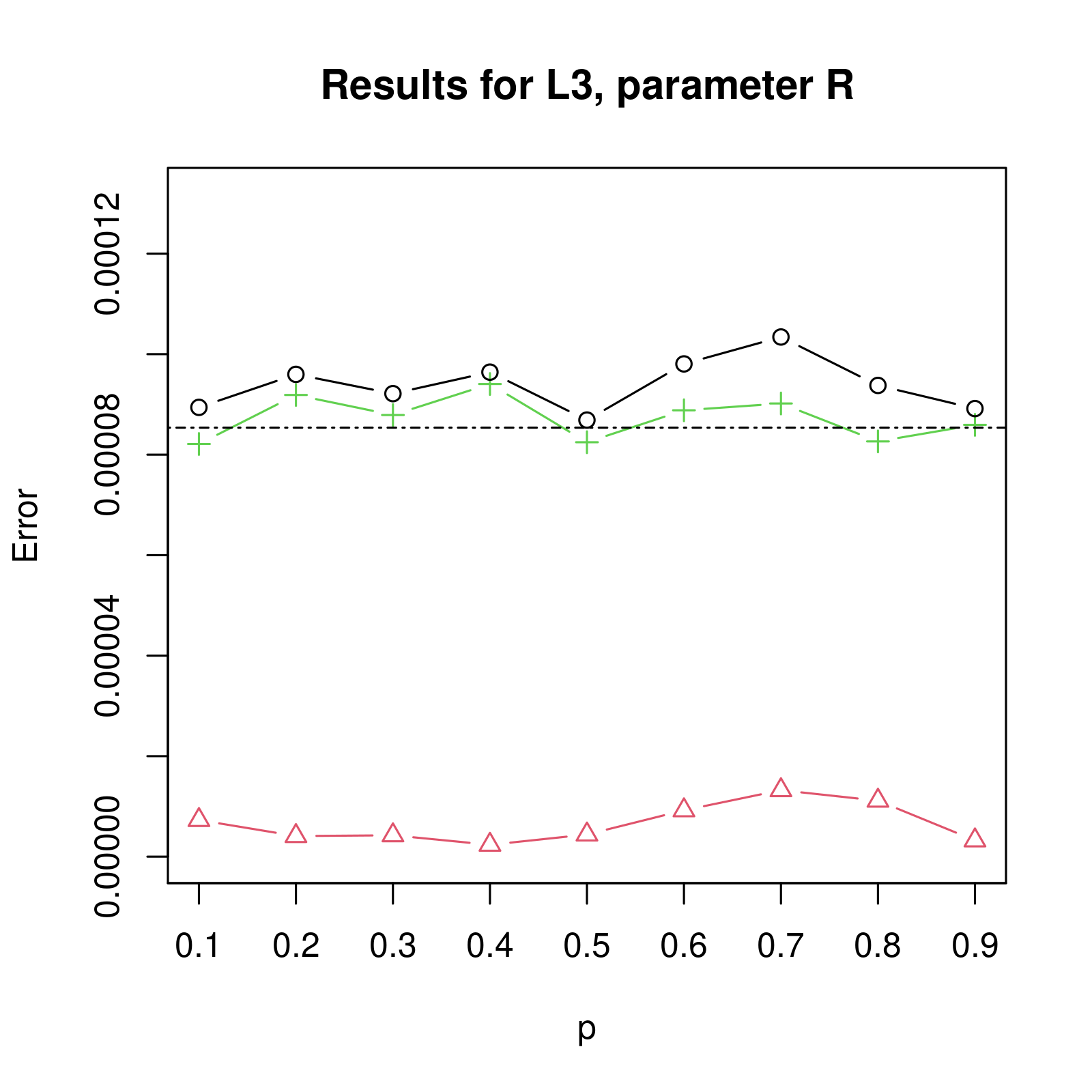}
    
    \includegraphics[width = 0.4\textwidth]{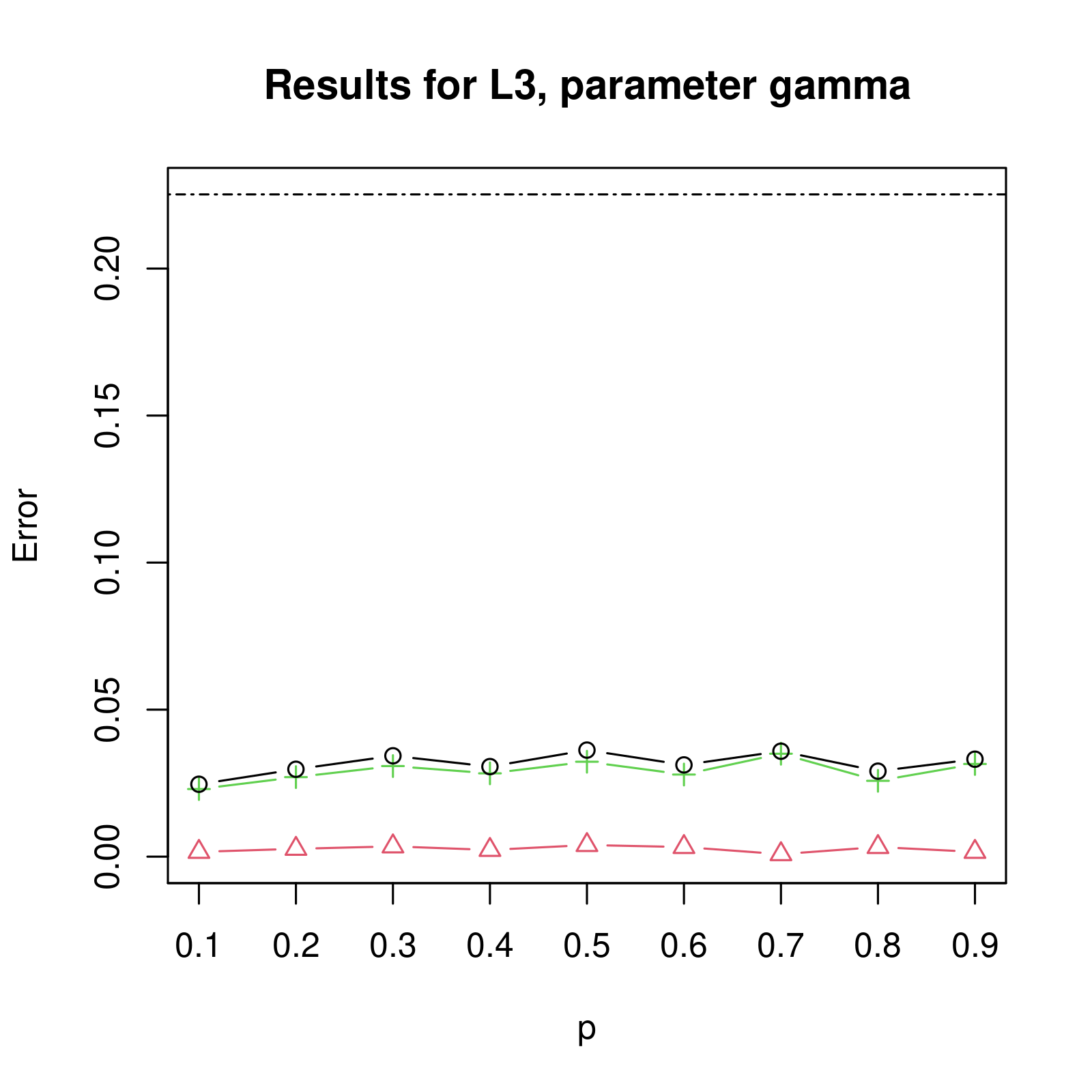}
    \includegraphics[width = 0.4\textwidth]{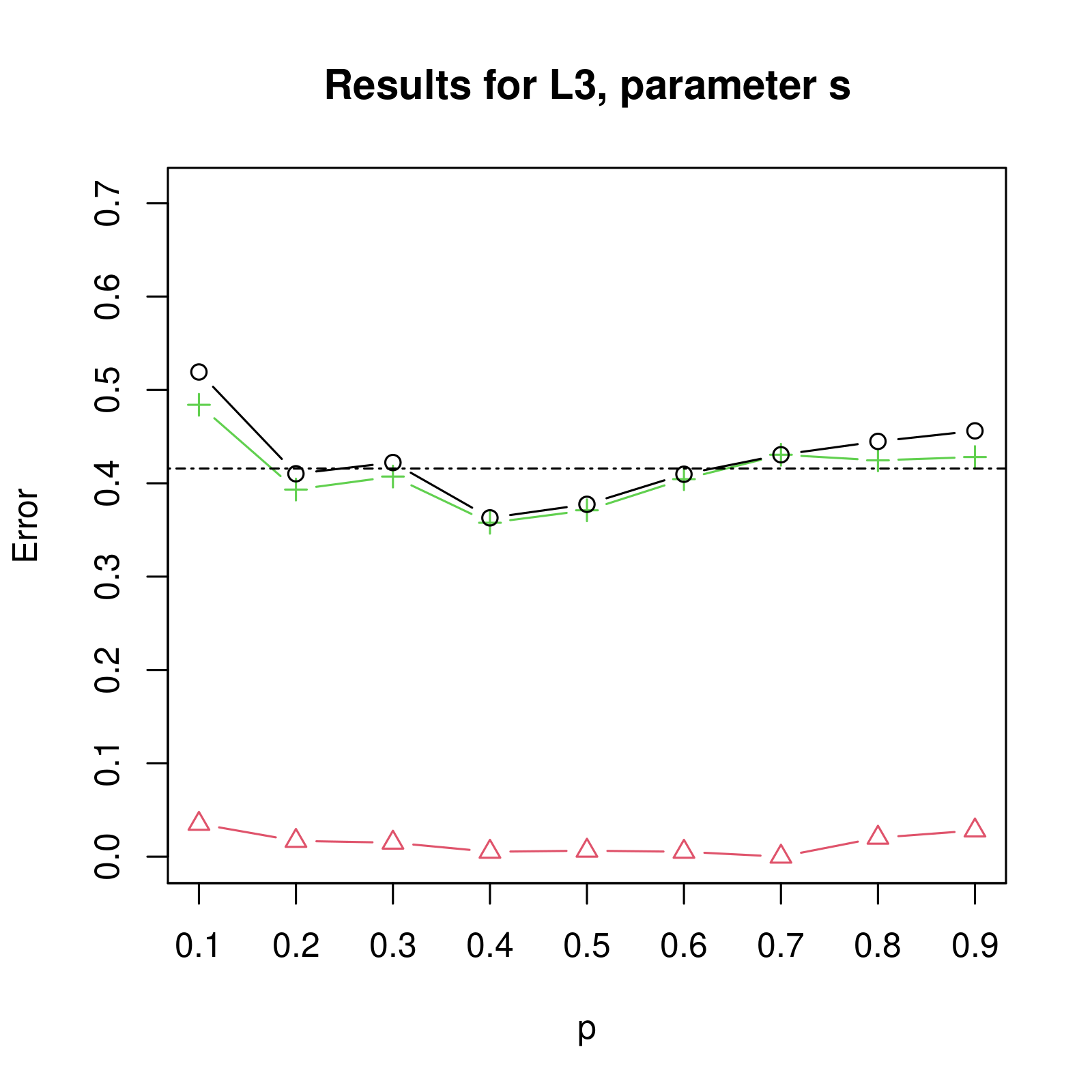}
    \caption{
    MSE, squared bias and variance for the Geyer saturation model using PPL with the $\Loss_3$ loss function, when estimating the parameters 
    $\beta$, $R$, $\gamma$ and $s$. 
    Here $k = 100$, $N = 100$, $p = 0.1,0.2,\ldots,0.9$ and the PPL-weight is estimated in according to \eqref{e:WeightEst}. 
    The black lines with circles correspond to MSE, the red lines with triangles correspond to squared bias and the green lines with plus signs correspond to variance. The black dotted lines correspond to the Takacs-Fiksel estimates.
}
    \label{fig:geyer_w-L3}
\end{figure}

\subsubsection{Overall conclusions}

To summarise the observations made in our simulation study, we start with the Poisson process. Here, we theoretically know that the weight is $p$ and it seems that the $\Loss_1$ loss function is the optimal choice. 
For the hard-core process, the $\Loss_1$ loss function 
in combination with weight estimation based on \eqref{e:WeightEst}
seems to be the optimal setup. 
In both of these cases, PPL outperforms Takacs-Fiksel estimation for both of the model parameters and all values of $p$. 
For the Strauss process the $\Loss_1$ loss function with the fixed weight $p$ seems to be the optimal combination. 
Here PPL outperforms Takacs-Fiksel estimation for both of the parameters $\beta$ and $\gamma$, when choosing $p$ to be between $0.4$ and $0.7$.
For the parameter $R$, the results of PPL and Takacs-Fiksel estimation are similar.
Lastly, for the Geyer saturation process, the best 
combination is 
the $\Loss_3$ loss function and weight estimation using \eqref{e:WeightEst}. 
Then, for all values of $p$, PPL is similar to Takacs-Fiksel estimation for $\beta$, $R$ and $s$, and PPL outperforms Takacs-Fiksel estimation for $\gamma$.

As anticipated, the weight choice $p/(1-p)$ did not work really well for any of the processes considered in the simulation study.
However, the choice between the fixed weight $p$ and weight estimation using \eqref{e:WeightEst} is not straightforward, especially since the latter is more 
computationally expensive and 
time consuming. 
Usually, the weight estimation yields more stable results, in terms of how the MSE is varying with $p$, 
while,
generally, the weight $p$ seems to work better for smaller values of $p$ than for larger values of $p$. 

As we have seen, the choice of loss function also affects the results. 
For the Poisson, hard-core and Strauss processes, the $\Loss_1$ and $\Loss_2$ loss functions performed better than the $\Loss_3$ loss function. 
However, for the Geyer saturation process, the $\Loss_3$ loss function yielded better results than both the $\Loss_1$ and $\Loss_2$ loss functions. 
Specifically, comparing i) the Geyer saturation process with the $\Loss_3$ loss function to ii) the Poisson, hard-core and Strauss processes with the $\Loss_1$ loss function, using the weight estimates $p$ in both i) and ii), we observe that the squared bias was lower for i) than ii), and the variance was higher for i) than ii).
It is not completely clear to us why we here have this difference in performance.

Regarding giving a general recommendation based on the observations made in the simulation study, we suggest using the fixed weight $p$, retention probability $p \in [0.3,0.4]$
and 
letting $k=100$ 
in the Monte-Carlo cross-validation. It should be stated that we have observed that PPL performs well also for smaller values for $k$. 
Moreover, if a fixed choice has to be made for the loss function, our suggestion would be to stick to $\Loss_1$ or $\Loss_2$, as they performed the best for all models expect the ``attractive'' Geyer saturation model, where $\Loss_3$ had a slightly better performance.

\section{Discussion}
\label{sec:discuss}

Point Process Learning (PPL) is a cross-validation-based statistical methodology for point processes, which was recent introduced by \citet{cronie2023cross} for the purpose of estimating parameters in (Gibbs) point processes models, via their Papangelou conditional intensity functions. 
Intuitively, since both Takacs-Fiksel estimation and the prediction errors of PPL are based on the Georgii-Nguyen-Zessin formula, it seems like Takacs-Fiksel estimation, which has pseudolikelihood estimation as a special case \citep{Coeurjolly2019understanding}, should be a special (limit) case of PPL.
In the main results of this paper we show that this is indeed the case. 
More specifically, we show that by letting the cross-validation scheme on which PPL is based tend to leave-one-out cross-validation, it follows that PPL based on a certain loss function, converges in probability to the loss function used in Takacs-Fiksel estimation.

As we were unable to theoretically show that the general PPL setup produces estimators of better mean square error (MSE) quality than Takacs-Fiksel estimation, we conducted a simulation study,
employing four common model types. Our simulation study shows that we can make hyperparameter choices in PPL such that it outperforms Takacs-Fiksel estimation in terms of MSE for all model types. In certain cases, the MSE was even halved with respect to Takacs-Fiksel estimation. 
The simulation study covered a variety of cases: an inhomogeneous Poisson process (complete spatial randomness), a hard-core process (repulsive), a Strauss process (repulsive and intermediate between a hard-core process and a Poisson process), and lastly, a Geyer saturation process with $\gamma>1$ (point patterns show tendencies of clustering). 
As with most (if not all) simulation studies, there are limitations in ours. In particular, one could explore more variations of the parameters of the underlying models (we used one set of parameters per model), as well more actual models. 
For example, there are infinitely many ways to specify a Poisson process. 
Moreover, it would be interesting to vary the interaction parameter $\gamma$ in both the Strauss and Geyer saturation processes, since here $\gamma$ was quite close to one, which means that the processes are rather close to homogeneous Poisson processes. Specifically, for the Geyer saturation model it would be interesting to explore the case when $\gamma < 1$, while varying the saturation threshold $s$. 
However, given the already extensive scope of our simulation study, with the range of different hyperparameter choices evaluated, we believe that the main message of the performance of PPL get through here. 

Regarding the computational aspects, PPL is by definition more expensive than Takacs-Fiksel estimation. In particular, our current implementation of PPL is more time-consuming to run than the state-of-the-art methods in \textsc{spatstat}. 
However, our current implementation can be optimised to run faster. 

Although we managed to illustrate that we can make hyperparameter choices such that PPL outperforms Takacs-Fiksel estimation, optimal selection of the hyperparameters is still an open question. 
In Monte-Carlo cross-validation, which we used for PPL in our simulation study, the hyperparameter $p$ governs the proportion of points that end up in the validation set, and thus the proportion in that training set. Hence, $p$ essentially governs which part of the data we use to try to predict the rest of the data. Making a good choice for $p$ is hard and needs further study. The same holds for exploring other cross-validation regimes to be used use in PPL, such as the ones suggested in Definition \ref{def:Kfold}. 
Cross-validation parameters are, in turn, strongly connected to another hyperparameter, which is the PPL-weight function; correct specification of this function ensures that the perdition errors in PPL have expectation zero under the right model. We here proposed a few simple ways to estimate the PPL-weights, either based on theoretical bounds or based on basic averaging ideas. As indicated in the paper, there are potentially more sophisticated ways to estimate PPL-weights, but this requires further study. 
Another crucial hyperparameter choice to be made is the
test function, which weights the contribution of each validation set point to the prediction error; 
in the current study we fixed it to the so-called Stoyan-Grabarnik 
test function but there definitely more optimal ways to select it (in a data-driven fashion). 
Although we here made educated but somewhat subjective choices for the hyperparameters, if a good approach to selecting them can be devised, we have high hope that the MSE performance of PPL can be improved even further.

\subsection*{Acknowledgements}
The authors would like to thank Marie-Colette van Lieshout for providing them with a counterexample on attractiveness/repulsiveness of Geyer saturation processes. 
The authors are also grateful to Aila Särkkä, Mike Pereira and Rasmus Waagepetersen for fruitful discussions and helpful insights.
They are further grateful to the developers of the \textsc{spatstat} package, especially Adrian
Baddeley and Ege Rubak for kindly answering implementation-related questions. 
The simulation study was enabled by resources provided by Chalmers e-Commons at Chalmers.
Ottmar Cronie has been supported by the Swedish Research Council (2023-03320).

\pagebreak

 \end{bibunit}

\pagebreak
\clearpage

\appendix
\pagenumbering{arabic}
\onecolumn
\begin{bibunit}
\begin{center}
\textbf{
Electronic Supplementary Material to the paper \\
``Comparison of Point Process Learning and its special case Takacs-Fiksel estimation''\\
Julia Jansson\footnote{juljans@chalmers.se} 
and Ottmar Cronie\footnote{ottmar@chalmers.se}
}
\end{center}
\section{Electronic Supplementary Material}

\thispagestyle{empty}

\subsection{Proofs}\label{s:Proofs}

\subsubsection{Proof of Theorem \ref{thm:TF} 
}
\label{s:proofTF}
To prove Theorem \ref{thm:TF} we need a dominated convergence theorem for sequences of random variables converging in probability. This result can be found stated in different places\footnote{\url{https://terrytao.wordpress.com/2015/10/23/275a-notes-3-the-weak-and-strong-law-of-large-numbers/}, Exercise 2(ix)} and its proof exploits that such sequences have subsequences which converge a.s.\ to the same limit. 

\begin{lemma}[Dominated convergence in probability]\label{DOM}
     Consider a probability space $(\Omega,\mathcal{F},\P)$ carrying a non-negative random variable $Z$, with $\E[Z]<\infty$, and a sequence of random variables which satisfy $Y_n\stackrel{p}{\to}Y$, as $n\to\infty$, and 
$|Y_n|\leq Z$ a.s.\ 
     for all $n$.
We then have $\E[Y_n]\to\E[Y]$ and $\E[|Y_n-Y|]\to0$ as $n\to\infty$. 
\end{lemma}

We are now ready to prove Theorem \ref{thm:TF}. 

\begin{proof}[Proof of Theorem \ref{thm:TF}]

Let $A\subseteq\Sm$ be arbitrary and bounded. 
Consider a 
sequence $a_k$, $k\geq2$,  
and a sequence $p_k\in(0,1)$, $k\geq1$, which decreases to $0$ as $k\to\infty$. 
We want to make choices for these sequences such that 
\[
a_k
\sum_{i=1}^k
\I_{\xi_{\theta}}^{h_{\theta}}(A;X_i^V(p_k),X_i^T(p_k)) 
\stackrel{p}{\to}
\I_{\lambda_{\theta}}^{h_{\theta}}(A;X,X)
\]
as $k\to\infty$. As in the statement of the theorem, for all $k\geq2$, we let $\{(X_i^T(p_k),X_i^V(p_k))\}_{i=1}^k$ and $\{(\x_i^T(p_k),\x_i^V(p_k))\}_{i=1}^k$ be Monte-Carlo cross-validation rounds for $X$ and $\x$, respectively, which have been generated by the retention probability $p_k$. 
Now, consider 
\begin{align*}
\Delta(k;\x)=&
a_k
\sum_{i=1}^k
\I_{\xi_{\theta}}^{h_{\theta}}(A;\x_i^V(p_k),\x_i^T(p_k)) 
- 
\I_{\lambda_{\theta}}^{h_{\theta}}(A;\x,\x)
\\
=&
a_k
\sum_{i=1}^k
\left(
\sum_{x\in \x_i^V(p_k) \cap A}
h_{\theta}(x;\x_i^T(p_k)\setminus\{x\})
-
\int_{A}
h_{\theta}(u;\x_i^T(p_k))
\xi_{\theta}(u;\x_i^T(p_k))
\de u
\right)
\\
&- 
\Bigg(
\sum_{x\in \x \cap A}
h_{\theta}(x;\x\setminus\{x\})
-
\int_{A}
h_{\theta}(u;\x)
\lambda_{\theta}(u|\x)
\de u
\Bigg)
\\
=&
\sum_{x\in \x\cap A}
a_k
\sum_{i=1}^{k}
\1\{x\in \x_i^V(p_k)\}
h_{\theta}(x;\x_i^T(p_k)\setminus\{x\})
-
h_{\theta}(x;\x\setminus\{x\})
\\
&-
\left(
a_k
\sum_{i=1}^{k}
\int_{A}
h_{\theta}(u;\x_i^T(p_k))
\xi_{\theta}(u;\x_i^T(p_k))
\de u
-
\int_{A}
h_{\theta}(u;\x)
\lambda_{\theta}(u|\x)
\de u
\right)
\\
=&
\Delta_1(k;\x) - \Delta_2(k;\x)
.
\end{align*}

\subsubsection*{Convergence of $\Delta_1(k;X)$}

Starting with $\Delta_1(k;\x)$, given any $x\in\x$, we note that
\begin{align*}
& 
\min_{j=1,\ldots,k}
h_{\theta}(x;\x_j^T(p_k)\setminus\{x\})
\sum_{i=1}^{k}
a_k
\1\{x\in \x_i^V(p_k)\}
\\
\leq&
\sum_{i=1}^{k}
a_k
\1\{x\in \x_i^V(p_k)\}
h_{\theta}(x;\x_i^T(p_k)\setminus\{x\})
\\
\leq&
\max_{j=1,\ldots,k}
h_{\theta}(x;\x_j^T(p_k)\setminus\{x\})
\sum_{i=1}^{k}
a_k
\1\{x\in \x_i^V(p_k)\}
.
\end{align*} 
Next, we will show that $\min_{j=1,\ldots,k}
h_{\theta}(x;\x_j^T(p_k)\setminus\{x\})$ and $\max_{j=1,\ldots,k}
h_{\theta}(x;\x_j^T(p_k)\setminus\{x\})$ tend to $h_{\theta}(x;\x\setminus\{x\})$ in probability and that $\sum_{i=1}^{k}
a_k
\1\{x\in \x_i^V(p_k)\}\to1$ in probability, as $k\to\infty$, provided that we make the choice $a_k=p_k=1/\sqrt{k}$. Having shown this, by appealing to Slutsky's lemma \citep[Theorem 6']{ferguson}, as $k\to\infty$, we obtain that both the upper and the lower bound tend to $h_{\theta}(x;\x\setminus\{x\})$ in probability, which in turn implies that $\sum_{i=1}^{k}
a_k
\1\{x\in \x_i^V(p_k)\}
h_{\theta}(x;\x_i^T(p_k)\setminus\{x\})\to 0$ in probability, whereby 
$\Delta_1(k;\x)\to0$ in probability when $k\to \infty$. Hence, as this holds for any realisation $\x$ of $X$, we have that $\Delta_1(k;X)\to0$ in probability as $k\to\infty$.

For a fixed $x\in\x$, let
$S_k=a_k
\sum_{i=1}^{k}
\1\{x\in \x_i^V(p_k)\}$ and note that $\E[S_k] = a_k k \E[\1\{x\in \x_i^V(p_k)\}] = a_k k p_k$. We want to have that $\lim_{k\to\infty}a_k k p_k = 1$ and 
we want to show that 
\begin{align*}
S_k - \E[S_k] = a_k
\sum_{i=1}^{k}
\1\{x\in \x_i^V(p_k)\} - a_k k p_k
=
a_k k p_k
\left(
\frac{1}{k}
\sum_{i=1}^{k}
\frac{\1\{x\in \x_i^V(p_k)\}
}{p_k} - 1\right)
\end{align*}
tends to 0 in probability, as $k\to\infty$. 
Note that all $\1\{x\in \x_i^V(p_k)\}$ are independent Bernoulli random variables with mean $\E[\1\{x\in \x_i^V(p_k)\}] = p_k$ and variance $\Var(\1\{x\in \x_i^V(p_k)\}) = p_k(1-p_k)$. 
By Markov's inequality, 
\begin{align*}
\varepsilon\P(|S_k - \E[S_k]|>\varepsilon)
&\leq \E[(S_k - \E[S_k])^2]
= \Var(S_k)
\\
&= \Var\left(a_k
\sum_{i=1}^{k}
\1\{x\in \x_i^V(p_k)\}\right)
\\
&= a_k^2 k \Var\left(
\1\{x\in \x_i^V(p_k)\}\right)
=
a_k^2 k p_k
(1-p_k)
\end{align*}
for any $\varepsilon>0$, 
which we want to show tends to 0. 
We thus want to have both $\lim_{k\to\infty}a_k k p_k = 1$ and $\lim_{k\to\infty}a_k^2 k p_k
(1-p_k) = (\lim_{k\to\infty}a_k k p_k) (\lim_{k\to\infty} a_k
(1-p_k)) = 0$, which may be achieved by letting 
$$
a_k=p_k=1/\sqrt{k}
,
$$
yielding that the upper bound above is given by $k^{-1/2}(1-k^{-1/2})$. Hence, since $\varepsilon>0$ was arbitrary, with $a_k=p_k=1/\sqrt{k}$ we obtain that $\lim_{k\to\infty}S_k = \lim_{k\to\infty}\E[S_k]=1$ in probability. 

We next turn to the convergence of $\min_{j=1,\ldots,k}
h_{\theta}(x;\x_j^T(p_k)\setminus\{x\})$ and $\max_{j=1,\ldots,k}
h_{\theta}(x;\x_j^T(p_k)\setminus\{x\})$. 
For any $k\geq2$ and any $i \in \{1,\ldots,k\}$, let $m_{i,k}(y)\in\{0,1\}$, $y\in\x$, be the corresponding marking, which yields a sequence of iid random variables with $\P(m_{i,k}(y)=0)=1-p_k$ and $\P(m_{i,k}(y)=1)=p_k$. We have that $\x_i^T(p_k)=\{y\in\x:m_{i,k}(y)=0\}$. Now, for any $\varepsilon>0$, by the law of total probability, 
\begin{align*}
&
\P(|h_{\theta}(x;\x_i^T(p_k)\setminus\{x\}) - h_{\theta}(x;\x\setminus\{x\})|>\varepsilon) 
\\
=&
\P(|h_{\theta}(x;\x_i^T(p_k)\setminus\{x\}) - h_{\theta}(x;\x\setminus\{x\})|>\varepsilon|\x_i^T(p_k)=\x)\P(\x_i^T(p_k)=\x)
\\
&+
\P(|h_{\theta}(x;\x_i^T(p_k)\setminus\{x\}) - h_{\theta}(x;\x\setminus\{x\})|>\varepsilon|\x_i^T(p_k)\neq\x)
(1 - \P(\x_i^T(p_k)=\x))
\\
=&
\P(|h_{\theta}(x;\x_i^T(p_k)\setminus\{x\}) - h_{\theta}(x;\x\setminus\{x\})|>\varepsilon|\x_i^T(p_k)\neq\x)(1 - \P(\x_i^T(p_k) = \x))
\\
=&
\P(|h_{\theta}(x;\x_i^T(p_k)\setminus\{x\}) - h_{\theta}(x;\x\setminus\{x\})|>\varepsilon|\x_i^T(p_k)\neq\x)
\\
&\times
\left(1 - \binom{\#\x}{\#\x}p_k^0(1-p_k)^{\#\x}\right)
\\
=&
\P(|h_{\theta}(x;\x_i^T(p_k)\setminus\{x\}) - h_{\theta}(x;\x\setminus\{x\})|>\varepsilon|\x_i^T(p_k)\neq\x)(1 - (1-p_k)^{\#\x}).
\end{align*} 
Since $p_k\downarrow0$ as $k\to\infty$, we obtain $1 - (1-p_k)^{\#\x}\to0$, whereby the expression above tends to $0$. Since $\varepsilon>0$ was arbitrary, $h_{\theta}(x;\x_i^T(p_k)\setminus\{x\})$ tends to $h_{\theta}(x;\x\setminus\{x\})$ in probability, as $k\to\infty$. Since we have shown this for any $i\in\{1,\ldots,k\}$, it also holds for $\min_{j=1,\ldots,k}
h_{\theta}(x;\x_j^T(p_k)\setminus\{x\})$ and $\max_{j=1,\ldots,k}
h_{\theta}(x;\x_j^T(p_k)\setminus\{x\})$.

One may hope to strengthen the convergence above to hold a.s., which could be achieved by applying a combination of the Borel-Cantelli lemma and \citet[Lemma 1]{ferguson}, provided that the right-hand side of $\sum_{k\geq2}\P(|h_{\theta}(x;\x_i^T(p_k)\setminus\{x\}) - h_{\theta}(x;\x\setminus\{x\})|>\varepsilon) \leq \sum_{k\geq2}(1 - (1-p_k)^{\#\x})$ is finite. However, $1+\sum_{k\geq2}(1 - (1-k^{-b})) = \zeta(b)$, the Riemann zeta function evaluated in $b$, is finite only if $b>1$. In other words, we would have had to have $p_k=k^{-b}$ for some $b>1$, as opposed to $p_k=k^{-1/2}$.

\subsubsection*{Convergence of $\Delta_2(k;X)$}

We next want to show that $\lim_{k\to\infty}\Delta_2(k;X)\stackrel{p}{=}0$. 
Having fixed $a_k=p_k=1/\sqrt{k}$, 
\begin{align*}
&\Delta_2(k;X)
\\
=&
a_k
\sum_{i=1}^{k}
\int_{A}
h_{\theta}(u;X_i^T(p_k))
\xi_{\theta}(u;X_i^T(p_k))
\de u
-
\int_{A}
h_{\theta}(u;X)
\lambda_{\theta}(u|X)
\de u
\\
=&
a_k
\sum_{i=1}^{k}
\int_{A}
h_{\theta}(u;X_i^T(p_k))
p_k
\left.
\E\left[
\frac{\lambda_{\theta}(u|X)}{\lambda_{\theta}(u| X_i^T(p_k))}
\right|
X_i^T(p_k)
\right]
\lambda_{\theta}(u| X_i^T(p_k))
\de u
\\
&-
\int_{A}
h_{\theta}(u;X)
\lambda_{\theta}(u|X)
\de u
\\
=&
\frac{1}{k}
\sum_{i=1}^{k}
\int_{A}
\E\left[
\left.
h_{\theta}(u;X_i^T(k^{-1/2}))
\lambda_{\theta}(u|X)
\right|
X_i^T(k^{-1/2})
\right]
\de u
-
\int_{A}
h_{\theta}(u;X)
\lambda_{\theta}(u|X)
\de u
\\
=&
\E\left[
\left.
\int_{A}
\frac{1}{k}
\sum_{i=1}^{k}
h_{\theta}(u;X_i^T(k^{-1/2}))
\lambda_{\theta}(u|X)
\de u
\right|
X_i^T(k^{-1/2})
\right]
-
\int_{A}
h_{\theta}(u;X)
\lambda_{\theta}(u|X)
\de u
,
\end{align*}
where the last equality follows from the Fubini-Tonelli theorem for conditional expectations, which requires that the Fubini-Tonelli theorem holds for the unconditional version of the statement.
If the variance of $\Delta_2(k;X)$ tends to 0 then we obtain the required result as a consequence of Markov's inequality.

We first show that $\lim_{k\to\infty}\E[\Delta_2(k;X)]=0$. Note that 
\begin{align*}
\E[\Delta_2(k;X)]
=&
\int_{A}
\frac{1}{k}
\sum_{i=1}^{k}
\E[
\E[
h_{\theta}(u;X_i^T(k^{-1/2}))
\lambda_{\theta}(u|X)
|
X_i^T(k^{-1/2})
]
]
\de u
\\
&-
\int_{A}
\E[h_{\theta}(u;X)
\lambda_{\theta}(u|X)]
\de u
\\
=&
\int_{A}
\E[
h_{\theta}(u;X_i^T(k^{-1/2}))
\lambda_{\theta}(u|X)
]
-
\E[h_{\theta}(u;X)
\lambda_{\theta}(u|X)]
\de u
\\
=&
\E\left[
\sum_{x\in X\cap A}
h_{\theta}(x;X_i^T(k^{-1/2})\setminus\{x\})
-
h_{\theta}(x;X\setminus\{x\})
\right]
\end{align*}
by the Fubini-Tonelli theorem, the law of total expectation and the GNZ formula.
We already know that for any $\x$ and $x\in\x$ the deviation $h_{\theta}(x;\x_i^T(p_k)\setminus\{x\})-h_{\theta}(x;\x\setminus\{x\})$ tends to 0 in probability as $k\to\infty$. Therefore, we obtain $\lim_{k\to\infty} \sum_{x\in X\cap A}
h_{\theta}(x;X_i^T(k^{-1/2})\setminus\{x\}) - h_{\theta}(x;X\setminus\{x\}) \stackrel{p}{=} 0$; note that, by definition, $\#(X\cap A)<\infty$ a.s.\ for all bounded $A\subseteq\Sm$. 
Next, we want to apply Lemma \ref{DOM} here to obtain that $\lim_{k\to\infty} \E[\Delta_2(k;X)] = \E[\lim_{k\to\infty} \Delta_2(k;X)] = 0$. To do so, we need to ensure that 
$|\sum_{x\in X\cap A}
h_{\theta}(x;X_i^T(k^{-1/2})\setminus\{x\}) - h_{\theta}(x;X\setminus\{x\})|$ is bounded by an integrable random variable for each $k\geq2$. 
Since $h_{\theta}$ is bounded, $|h_{\theta}(x;\x) - h_{\theta}(x;\y)| \leq 2\max\{|h_{\theta}(x;\x)|, |h_{\theta}(x;\y)|\}<\infty$ for all $\x,\y$. We thus obtain the bounding random variable $2\sum_{x\in X\cap A} \max\{|h_{\theta}(x;X_i^T(k^{-1/2})\setminus\{x\})|, |h_{\theta}(x;X\setminus\{x\})|\}$, 
which has finite expectation since $\#(X\cap A)<\infty$ a.s.\ by the local finiteness of $X$. 
Hence, $\lim_{k\to\infty}\Var(\Delta_2(k;X)) = \lim_{k\to\infty}\E[\Delta_2(k;X)^2]$.

Turning to the second moment, we have
\begin{align*}
&\E[\Delta_2(k;X)^2]
\\
=&
\E\left[\E\left[
\left.
\int_{A}
\frac{1}{k}
\sum_{i=1}^{k}
h_{\theta}(u;X_i^T(k^{-1/2}))
\lambda_{\theta}(u|X)
\de u
\right|
X_i^T(k^{-1/2})
\right]^2
\right]
\\
&-2
\E\left[
\int_{A^2}
\frac{1}{k}
\sum_{i=1}^{k}
\E\left[
\left.
h_{\theta}(u;X_i^T(k^{-1/2}))
\lambda_{\theta}(u|X)
\right|
X_i^T(k^{-1/2})
\right]
h_{\theta}(v;X)
\lambda_{\theta}(v|X)
\de u
\de v
\right]
\\
&+
\E\left[
\int_{A^2}
h_{\theta}(u;X)
\lambda_{\theta}(u|X)
h_{\theta}(v;X)
\lambda_{\theta}(v|X)
\de u\de v
\right]
\\
\leq&
\E\left[
\left(
\int_{A}
\frac{1}{k}
\sum_{i=1}^{k}
h_{\theta}(u;X_i^T(k^{-1/2}))
\lambda_{\theta}(u|X)
\de u
\right)^2
\right]
\\
&-2
\int_{A^2}
\frac{1}{k}
\sum_{i=1}^{k}
\E\left[
\E\left[
\left.
h_{\theta}(u;X_i^T(k^{-1/2}))
\lambda_{\theta}(u|X)
\right|
X_i^T(k^{-1/2})
\right]
h_{\theta}(v;X)
\lambda_{\theta}(v|X)
\right]
\de u
\de v
\\
&+
\int_{A^2}
\E\left[
h_{\theta}(u;X)
h_{\theta}(v;X)
\lambda_{\theta}(u|X)
\lambda_{\theta}(v|X)
\right]
\de u\de v
\\
=& E_1(k) - 2E_2(k) + E_3
,
\end{align*}
where the inequality is a consequence of conditioning being a contractive projection of $L^p$ spaces.

Now, we first want to show that $E_1(k) \to E_3$ when $k \to \infty$. We have
\begin{align*}
    E_1(k) =&
    \E\left[
\left(
\int_{A}
\frac{1}{k}
\sum_{i=1}^{k}
h_{\theta}(u;X_i^T(k^{-1/2}))
\lambda_{\theta}(u|X)
\de u
\right)^2
\right]\\
=&
\E\left[
\int_{A^2}
\frac{1}{k}
\sum_{i=1}^{k}
h_{\theta}(u;X_i^T(k^{-1/2}))
\lambda_{\theta}(u|X)
\frac{1}{k}
\sum_{j=1}^{k}
h_{\theta}(v;X_j^T(k^{-1/2}))
\lambda_{\theta}(v|X)
\de u \de v
\right]\\
=&
\int_{A^2}
\E\left[
\frac{1}{k}
\sum_{i=1}^{k}
h_{\theta}(u;X_i^T(k^{-1/2}))
\lambda_{\theta}(u|X)
\frac{1}{k}
\sum_{j=1}^{k}
h_{\theta}(v;X_j^T(k^{-1/2}))
\lambda_{\theta}(v|X)
\right]
\de u \de v\\
=&
\int_{A^2}
\frac{1}{k^2}
\sum_{i=1}^{k}
\sum_{j=1}^{k}
\E\left[
h_{\theta}(u;X_i^T(k^{-1/2}))
\lambda_{\theta}(u|X)
h_{\theta}(v;X_j^T(k^{-1/2}))
\lambda_{\theta}(v|X)
\right]
\de u \de v\\
=&
\int_{A^2}
\E\left[
h_{\theta}(u;X_i^T(k^{-1/2}))
\lambda_{\theta}(u|X)
h_{\theta}(v;X_j^T(k^{-1/2}))
\lambda_{\theta}(v|X)
\right]
\de u \de v,
\end{align*}
for any $i, j \in \{1,\ldots,k\}$.
Now, we already know that $h_{\theta}(u;\x_i^T(k^{-1/2}))-h_{\theta}(u;\x)$ and $h_{\theta}(v;\x_j^T(k^{-1/2}))-h_{\theta}(v;\x)$ both tend to 0 in probability as $k\to\infty$. 
Since this holds for all $\x$ it also holds for $X$ in probability. Further, by Lemma \ref{DOM} and Slutsky's lemma \citep[Theorem 6']{ferguson} we get 
\begin{align*}
&\lim_{k\to\infty}E_1(k)
\\
    =&
\int_{A^2}
\E\left[\lim_{k\to\infty}
(h_{\theta}(u;X_i^T(k^{-1/2}))
\lambda_{\theta}(u|X)
h_{\theta}(v;X_j^T(k^{-1/2}))
\lambda_{\theta}(v|X))
\right]
\de u \de v\\
=&
\int_{A^2}
\E\left[\big(\lim_{k\to\infty}
h_{\theta}(u;X_i^T(k^{-1/2}))
\lambda_{\theta}(u|X)\big)
\big(\lim_{k\to\infty} h_{\theta}(v;X_j^T(k^{-1/2}))
\lambda_{\theta}(v|X)\big)
\right]
\de u \de v\\
=&
\int_{A^2}
\E\left[
h_{\theta}(u;X)
\lambda_{\theta}(u|X)
h_{\theta}(v;X)
\lambda_{\theta}(v|X)
\right]
\de u \de v = E_3
.
\end{align*}
Note that Lemma \ref{DOM} is applicable here since $\lambda_{\theta}$ and $h_{\theta}$ are bounded by assumption. 

Focusing on $E_2(k)$, by the self-adjointness of conditional expectations we have 
\begin{align*}
&E_2(k)
\\
=&
\int_{A^2}
\E\left[
\frac{1}{k}
\sum_{i=1}^{k}
\E\left[
\left.
h_{\theta}(u;X_i^T(k^{-1/2}))
\right|
X_i^T(k^{-1/2})
\right]
h_{\theta}(v;X)
\lambda_{\theta}(u|X)
\lambda_{\theta}(v|X)
\right]
\de u
\de v
,
\end{align*}
which we want to show tends to $E_3$ in probability.

We start by showing that $\E[
h_{\theta}(u;X_i^T(k^{-1/2}))
|
X_i^T(k^{-1/2})
]$ tends to $h_{\theta}(u;X)$ in probability. We do so using Markov's inequality:
\begin{align*}
&\varepsilon^2\P\left(\left|\E\left[
\left.
h_{\theta}(u;X_i^T(k^{-1/2}))
\right|
X_i^T(k^{-1/2})
\right] - h_{\theta}(u;X)
\right|>\varepsilon\right)
\\
\leq&
\E\left[\left(
\E\left[
\left.
h_{\theta}(u;X_i^T(k^{-1/2}))
\right|
X_i^T(k^{-1/2})
\right] 
- 
h_{\theta}(u;X)\right)^2\right]
\\
=&
\E\left[\left(
\E\left[
\left.
h_{\theta}(u;X_i^T(k^{-1/2}))
\right|
X_i^T(k^{-1/2})
\right] 
-
h_{\theta}(u;X_i^T(k^{-1/2}))
+
h_{\theta}(u;X_i^T(k^{-1/2}))
- 
h_{\theta}(u;X)\right)^2\right]
\\
\leq &
\E\left[\left(
\E\left[
\left.
h_{\theta}(u;X_i^T(k^{-1/2}))
\right|
X_i^T(k^{-1/2})
\right] 
-
h_{\theta}(u;X_i^T(k^{-1/2}))
\right)^2\right]
\\
&+
\E\left[\left(
h_{\theta}(u;X_i^T(k^{-1/2}))
- 
h_{\theta}(u;X)\right)^2\right]
.
\end{align*}
We already know from before that $\lim_{k\to\infty} 
h_{\theta}(x;X_i^T(k^{-1/2})) - h_{\theta}(x;X) \stackrel{p}{=} 0$.
Thus, by applying Lemma \ref{DOM} the second term goes to 0.
We continue with the first term:
\begin{align*}
    &\E\left[\left(
\E\left[
\left.
h_{\theta}(u;X_i^T(k^{-1/2}))
\right|
X_i^T(k^{-1/2})
\right] 
-
h_{\theta}(u;X_i^T(k^{-1/2}))
\right)^2\right]\\
=&
 \E\Big[
\E\left[
h_{\theta}(u;X_i^T(k^{-1/2}))
|
X_i^T(k^{-1/2})
\right]^2
+
h_{\theta}(u;X_i^T(k^{-1/2}))^2\\
&-2\E\left[
h_{\theta}(u;X_i^T(k^{-1/2}))
|
X_i^T(k^{-1/2})
\right]
h_{\theta}(u;X_i^T(k^{-1/2}))
\Big]\\
=&
 \E\Big[
\E\left[
h_{\theta}(u;X_i^T(k^{-1/2}))
|
X_i^T(k^{-1/2})
\right]^2\Big]
+
 \E\Big[h_{\theta}(u;X_i^T(k^{-1/2}))^2\Big]\\
&-2 \E\Big[\E\left[
h_{\theta}(u;X_i^T(k^{-1/2}))
|
X_i^T(k^{-1/2})
\right]
h_{\theta}(u;X_i^T(k^{-1/2}))
\Big]
.
\end{align*}
Now we know that $\E\Big[
\E\left[
h_{\theta}(u;X_i^T(k^{-1/2}))
|
X_i^T(k^{-1/2})
\right]^2\Big] \leq \E\Big[
h_{\theta}(u;X_i^T(k^{-1/2}))
^2\Big]$
due to conditioning being a contractive projection of $L^p$ spaces. 
Further, by the 'taking out what is known' property and the law of total expectation,
\begin{align*}
    &\E\Big[\E\left[
h_{\theta}(u;X_i^T(k^{-1/2}))
|
X_i^T(k^{-1/2})
\right]
h_{\theta}(u;X_i^T(k^{-1/2}))
\Big]\\
=& \E\Big[\E\left[
h_{\theta}(u;X_i^T(k^{-1/2}))^2
|
X_i^T(k^{-1/2})
\right]
\Big]
\\
=& \E\Big[
h_{\theta}(u;X_i^T(k^{-1/2}))^2
\Big]
.
\end{align*}
By putting all of this together we obtain
\begin{align*}
     &\E\Big[
\E\left[
h_{\theta}(u;X_i^T(k^{-1/2}))
|
X_i^T(k^{-1/2})
\right]^2\Big]
+
 \E\Big[h_{\theta}(u;X_i^T(k^{-1/2}))^2\Big]\\
&-2 \E\Big[\E\left[
h_{\theta}(u;X_i^T(k^{-1/2}))
|
X_i^T(k^{-1/2})
\right]
h_{\theta}(u;X_i^T(k^{-1/2}))
\Big]\\
&\leq
\E\Big[
h_{\theta}(u;X_i^T(k^{-1/2}))
^2\Big]+ \E\Big[
h_{\theta}(u;X_i^T(k^{-1/2}))
^2\Big]-2\E\Big[
h_{\theta}(u;X_i^T(k^{-1/2}))
^2\Big] = 0
.
\end{align*}
To summarize, we then have that 
\begin{align*}
&\varepsilon^2\P\left(\left|\E\left[
\left.
h_{\theta}(u;X_i^T(k^{-1/2}))
\right|
X_i^T(k^{-1/2})
\right] - h_{\theta}(u;X)
\right|>\varepsilon\right)
\\
\leq&
\E\left[\left(
\E\left[
\left.
h_{\theta}(u;X_i^T(k^{-1/2}))
\right|
X_i^T(k^{-1/2})
\right] 
-
h_{\theta}(u;X_i^T(k^{-1/2}))
\right)^2\right]
\\
&+
\E\left[\left(
h_{\theta}(u;X_i^T(k^{-1/2}))
- 
h_{\theta}(u;X)\right)^2\right]
\\
&\leq 0 + \E\left[\left(
h_{\theta}(u;X_i^T(k^{-1/2}))
- 
h_{\theta}(u;X)\right)^2\right] \longrightarrow 0
\end{align*}
when $k\to \infty$.
Since $\varepsilon$ is arbitrary, this means that $\E[
h_{\theta}(u;X_i^T(k^{-1/2}))
|
X_i^T(k^{-1/2})
]$ tends to $h_{\theta}(u;X)$ in probability.
Then, we can see that $\lim_{k\to\infty}E_2(k) \stackrel{p}{=} E_3$ by using Lemma \ref{DOM}, which gives us that
\begin{align*}
&\lim_{k\to\infty}\E[\Delta_2(k;X)^2]\leq \lim_{k\to\infty}E_1(k) - 2\lim_{k\to\infty}E_2(k) + E_3 \stackrel{p}{=} E_3-2E_3+E_3 = 0
\end{align*}
which in turn gives us that
$\lim_{k\to\infty}\Var(\Delta_2(k;X)) = \lim_{k\to\infty}\E[\Delta_2(k;X)^2] \stackrel{p}{=} 0$.
Then we obtain $\lim_{k\to\infty}\Delta_2(k;X)\stackrel{p}{=}0$ due to Markov's inequality.
\end{proof}

\subsubsection{Proof of Theorem \ref{thm:TF_Block}}
\label{s:proofTF_Block}

\begin{proof}

We want to show that 
\begin{align*}
\sum_{i=1}^k\I_{\xi_{\theta}}^{h_{\theta}}(A;X_{ik}^V,X_{ik}^T)
=& 
\sum_{i=1}^k\sum_{x\in X_{ik}^V \cap A}
h_{\theta}(x;X_{ik}^T)
-
\sum_{i=1}^k
\int_{A}
h_{\theta}(u;X_{ik}^T)
V_{\theta}(u)
\lambda_{\theta}(u|X_{ik}^T)
\de u
\\
=&
\sum_{i=1}^k\sum_{x\in X \cap A}
\1\{x\in A_{ik}\}
h_{\theta}(x;X\cap (A\setminus A_{ik}))
\\
&-
\sum_{i=1}^k
\int_{A}
h_{\theta}(u;X\cap (A\setminus A_{ik}))
V_{\theta}(u)
\lambda_{\theta}(u|X\cap (A\setminus A_{ik}))
\de u
\end{align*}
converges in probability to
\[
\I_{\lambda_{\theta}}^{h_{\theta}}(A;X\cap A, X\cap A)
=
\sum_{x\in X \cap A}
h_{\theta}(x;X \cap A\setminus\{x\})
-
\int_{A}
h_{\theta}(u;X \cap A)
\lambda_{\theta}(u|X \cap A)
\de u
.
\]
In order to do so, we show that both
$$
\Delta_1(k;X) = 
\sum_{i=1}^k
\sum_{x\in X \cap A}
\1\{x\in A_{ik}\}
h_{\theta}(x;X\cap (A\setminus A_{ik}))
-
\sum_{x\in X \cap A}
h_{\theta}(x;X \cap A\setminus\{x\})
$$
and 
$$
\Delta_2(k;X) = \int_{A}
h_{\theta}(u;X\cap (A\setminus A_{ik}))
V_{\theta}(u)
\lambda_{\theta}(u|X\cap (A\setminus A_{ik}))
\de u-\int_{A}
h_{\theta}(u;X\cap A)
\lambda_{\theta}(u|X\cap A)
\de u
$$
tend to 0 in probability, as $k\to\infty$.

Since $\Sm$ is a metric space (and thereby a Hausdorff space), for any distinct points $x,y\in\Sm$ we can find radii $r_x,r_y>0$ such that $b(x,r_x)\cap b(y,r_y)=\emptyset$. This holds in particular for distinct members $x,y\in\x$ of a point pattern $\x\in\nn$. By the local finiteness of $\nn$, for each $\x\in\nn$ there is a universal $r_{\x,A}>0$ such that $b(x,r_{\x,A})\cap b(y,r_{\x,A})=\emptyset$ for any $x,y\in\x\cap A$. 
Hence, as $\max_{i=1,\ldots,k}|A_{ik}|$ decreases, we can find some $k_{\x,A}\geq2$ such that when $k\geq k_{\x,A}$ we have that $\#(A_{ik}\cap\x)\in\{0,1\}$, i.e.\ each $A_{ik}$ contains at most one element of $\x\cap A$.

Now, for any $u\in A$, let $A_{k}(u)$ be the unique $A_{ik}$ that contains $u$. 
Then
\begin{align}
\label{e:Hfun}
\sum_{i=1}^k\1\{u\in A_{ik}\}h_{\theta}(u;\x\cap (A\setminus A_{ik}))
=
h_{\theta}(u;\x\cap (A\setminus A_{k}(u)))
.
\end{align}
For any $\x\in\nn$, if $u\in\x$ then $\x\cap (A\setminus A_{ik}(u)) \to \x\setminus\{u\}$ as $k\to\infty$ and if $u\notin\x$ then $\x\cap (A\setminus A_{ik}(u)) \to \x$ as $k\to\infty$. 
We thus have 
\begin{align} 
\label{e:LimitIdentity}
&\lim_{k\to\infty}
h_{\theta}(u;\x\cap (A\setminus A_{k}(u)))
\\
=& \1\{u\in\x\cap A\}h_{\theta}(u;\x\cap A\setminus\{u\}) + 
\1\{u\notin\x\cap A\}h_{\theta}(u;\x\cap A)
,
\nonumber
\end{align}
and since this holds for any $\x\in\nn$, we also have 
\begin{align}
\label{e:LimitX}
&\lim_{k\to\infty}
h_{\theta}(u;\x\cap (A\setminus A_{k}(u)))\\
\stackrel{a.s.}{=}&
\1\{u\in X\cap A\}h_{\theta}(u;X\cap A\setminus\{u\}) + 
\1\{u\notin X\cap A\}h_{\theta}(u;X\cap A).\notag
\end{align}

\subsubsection*{Convergence of $\Delta_1(k;X)$}

By \eqref{e:LimitX}, we obtain that
\begin{align*}
&\lim_{k\to\infty}\sum_{i=1}^k
\sum_{x\in X \cap A}
\1\{x\in A_{ik}\}
h_{\theta}(x;X\cap (A\setminus A_{ik}))\\
=&
\sum_{x\in X \cap A}
\lim_{k\to\infty}
h_{\theta}(x;\x\cap (A\setminus A_{k}(u)))
\\
\stackrel{a.s.}{=}&\sum_{x\in X \cap A}
(\1\{x\in X\cap A\}h_{\theta}(x;X\cap A\setminus\{x\}) + 
\1\{x\notin X\cap A\}h_{\theta}(x;X\cap A))\\
=&
\sum_{x\in X \cap A}
h_{\theta}(x;X\cap A\setminus\{x\}),
\end{align*}
whereby 
$\Delta_1(k;X)$
tends to 0 a.s., and thereby in probability, as $k\to\infty$.

\subsubsection*{Convergence of $\Delta_2(k;X)$}

We next want to show that $\lim_{k\to\infty}\Delta_2(k;X)\stackrel{p}{=}0$. This means that for all $\varepsilon>0$, 
\begin{align*}
    \lim_{k\to\infty}P(|\Delta_2(k;X)-0|> \varepsilon)=0.
\end{align*}
By using Markov's inequality, we get
\begin{align*}
    P(|\Delta_2(k;X)|\geq \varepsilon) =  P(\Delta_2(k;X)^2\geq \varepsilon^2) \leq \frac{\E[\Delta_2(k;X)^2]}{\varepsilon^2}.
\end{align*}
If we can show that $\lim_{k\to\infty}\E[\Delta_2(k;X)^2] = 0$ then we have that 
$$\lim_{k\to\infty}\varepsilon^2P(|\Delta_2(k;X)|\geq \varepsilon) \leq \lim_{k\to\infty}\E[\Delta_2(k;X)^2]= 0.$$
Since this then holds for any $\varepsilon>0$, we obtain that $\lim_{k\to\infty}\Delta_2(k;X)\stackrel{p}{=}0$.

We first rewrite the first integral term in $\Delta_2(k;X)$:
\begin{align*}
    &\sum_{i=1}^k\int_{A}
h_{\theta}(u;X\cap (A\setminus A_{ik}))
V_{\theta}(u)
\lambda_{\theta}(u|X\cap (A\setminus A_{ik}))
\de u\\
=& \sum_{i=1}^k\int_{A}
h_{\theta}(u;X\cap (A\setminus A_{ik}))
p_{ik}(u)
\left.
\E\left[
\frac{\lambda_{\theta}(u|X\cap A)}{\lambda_{\theta}(u| X\cap (A\setminus A_{ik}))}
\right|
X\cap (A\setminus A_{ik})
\right]
\lambda_{\theta}(u|X\cap (A\setminus A_{ik}))
\de u\\
=& \int_{A}
\sum_{i=1}^k
h_{\theta}(u;X\cap (A\setminus A_{ik}))
\1\{u\in A_{ik}\}
\left.
\E\left[
\lambda_{\theta}(u|X\cap A)
\right|
X\cap (A\setminus A_{ik})
\right] 
\de u
\\
=& 
\sum_{i=1}^k
\int_{A_{ik}}
\left.
\E\left[
h_{\theta}(u;X\cap (A\setminus A_{ik}))
\lambda_{\theta}(u|X\cap A)
\right|
X\cap (A\setminus A_{ik})
\right]
\de u
.
\end{align*}
Using this, we obtain that 
the second moment satisfies 
\begin{align*}
    &\E[\Delta_2(k;X)^2]\\
    =& \E\Bigg[\Bigg(
    \E\left[\sum_{i=1}^k\int_{A_{ik}}
\left.
h_{\theta}(u;X\cap (A\setminus A_{ik}))
\lambda_{\theta}(u|X\cap A)\de u
\right|
X\cap (A\setminus A_{ik})
\right]\\
&\quad -
\int_{A}h_{\theta}(u;X\cap A)
\lambda_{\theta}(u|X\cap A)
\de u
    \Bigg)^2\Bigg]\\
=&\E\Bigg[
    \E\left[\sum_{i=1}^k\int_{A_{ik}}
\left.
h_{\theta}(u;X\cap (A\setminus A_{ik}))
\lambda_{\theta}(u|X\cap A)\de u
\right|
X\cap (A\setminus A_{ik})
\right]^2\\
&\quad +
\Bigg(\int_{A}h_{\theta}(u;X\cap A)
\lambda_{\theta}(u|X\cap A)
\de u\Bigg)^2
\\
&-
2\E\left[\sum_{i=1}^k\int_{A_{ik}}
\left.
h_{\theta}(u;X\cap (A\setminus A_{ik}))
\lambda_{\theta}(u|X\cap A)\de u
\right|
X\cap (A\setminus A_{ik})
\right]\\
&\qquad \times\int_{A}h_{\theta}(u;X\cap A)
\lambda_{\theta}(u|X\cap A)
\de u 
\Bigg]
\\ 
\leq &
\E\left[\left(\sum_{i=1}^k\int_{A_{ik}}
h_{\theta}(u;X\cap (A\setminus A_{ik}))
\lambda_{\theta}(u|X\cap A)\de u
\right)^2\right]
\\
&+
\int_{A^2}\E[h_{\theta}(u;X\cap A)
\lambda_{\theta}(u|X\cap A)
h_{\theta}(v;X\cap A)
\lambda_{\theta}(v|X\cap A)]
\de v
\de u
\\
&-
2\sum_{i=1}^k\int_{A^2}
\left.
\E\big[\E\left[\1\{u\in A_{ik}\}h_{\theta}(u;X\cap (A\setminus A_{ik}))
\lambda_{\theta}(u|X\cap A)
\right|
X\cap (A\setminus A_{ik})
\right]\\
&\qquad \qquad \quad \times h_{\theta}(v;X\cap A)
\lambda_{\theta}(v|X\cap A)\big]
\de u \de v
\\
=& E_1(k)  + E_3- 2E_2(k),
\end{align*}
where the inequality is a consequence of conditioning being a contractive projection of $L^p$ spaces. 
By showing that $\lim_{k\to\infty}E_1(k)=E_3$ and $\lim_{k\to\infty}E_2(k)=E_3$ we are done, since then 
$$
\lim_{k\to\infty} (E_1(k) - 2E_2(k) + E_3)
= \lim_{k\to\infty}E_1(k) - 2\lim_{k\to\infty}E_2(k) + E_3
= E_3 - 2E_3 + E_3 = 0.
$$

We start with $E_1(k)$. Recalling \eqref{e:Hfun}, we have that  
\begin{align*}
    &E_1(k) = \E\left[\left(\sum_{i=1}^k\int_{A_{ik}}
h_{\theta}(u;X\cap (A\setminus A_{ik}))
\lambda_{\theta}(u|X\cap A)\de u
\right)^2\right]
\\=&
\E\left[\left(\int_{A}\sum_{i=1}^k\1\{u\in A_{ik}\}
h_{\theta}(u;X\cap (A\setminus A_{ik}))
\lambda_{\theta}(u|X\cap A)\de u
\right)^2\right]
\\=&
\E[\int_{A^2}\sum_{i=1}^k\1\{u\in A_{ik}\}
h_{\theta}(u;X\cap (A\setminus A_{ik}))
\lambda_{\theta}(u|X\cap A)
\sum_{j=1}^k\1\{v\in A_{jk}\}\\
&\qquad  \times h_{\theta}(v;X\cap (A\setminus A_{jk}))
\lambda_{\theta}(v|X\cap A)
\de v
\de u
]
\\=&
\E\left[\int_{A^2}h_{\theta}(u;\x\cap (A\setminus A_{k}(u)))
\lambda_{\theta}(u|X\cap A)
h_{\theta}(v;\x\cap (A\setminus A_{k}(v)))
\lambda_{\theta}(v|X\cap A)
\de v
\de u
\right]
\\=&
\int_{A^2}\E\left[h_{\theta}(u;\x\cap (A\setminus A_{k}(u)))
\lambda_{\theta}(u|X\cap A)
h_{\theta}(v;\x\cap (A\setminus A_{k}(v)))
\lambda_{\theta}(v|X\cap A)
\right]
\de v
\de u
.
\end{align*}
By Lemma \ref{DOM}, Slutsky's lemma \citep[Theorem 6']{ferguson}, the observation following \eqref{e:LimitIdentity} and the Fubini-Tonelli theorem, we now get that 
\begin{align*}
&\lim_{k\to\infty} 
E_1(k)
\\
=&
\lim_{k\to\infty}
\int_{A^2}
\E\left[h_{\theta}(u;\x\cap (A\setminus A_{k}(u))) 
\lambda_{\theta}(u|X\cap A) 
h_{\theta}(v;\x\cap (A\setminus A_{k}(v))) 
\lambda_{\theta}(v|X\cap A) 
\right]
\de v
\de u
\\
=&
\int_{A^2}
\lim_{k\to\infty}
\E\left[h_{\theta}(u;\x\cap (A\setminus A_{k}(u)))
\lambda_{\theta}(u|X\cap A)
h_{\theta}(v;\x\cap (A\setminus A_{k}(v)))
\lambda_{\theta}(v|X\cap A)
\right]
\de v
\de u
\\
=&
\int_{A^2}
\E\left[\lim_{k\to\infty}(h_{\theta}(u;\x\cap (A\setminus A_{k}(u)))
\lambda_{\theta}(u|X\cap A)
h_{\theta}(v;\x\cap (A\setminus A_{k}(v)))
\lambda_{\theta}(v|X\cap A))
\right]
\de v
\de u
\\
=&
\int_{A^2}
\E\left[(\lim_{k\to\infty}h_{\theta}(u;\x\cap (A\setminus A_{k}(u))))
\lambda_{\theta}(u|X\cap A)
(\lim_{k\to\infty}h_{\theta}(v;\x\cap (A\setminus A_{k}(v))))
\lambda_{\theta}(v|X\cap A)
\right]
\de v
\de u
\\
=&
\int_{A^2}
\E[(\1\{u\in X\cap A\}  h_{\theta}(u;X\cap A\setminus\{u\}) + \1\{u\notin X\cap A\} h_{\theta}(u;X\cap A))
\lambda_{\theta}(u|X\cap A)
\\
&\times
(\1\{v\in X\cap A\} h_{\theta}(v;X\cap A\setminus\{v\}) + \1\{v\notin X\cap A\}h_{\theta}(v;X\cap A))
\lambda_{\theta}(v|X\cap A)
]
\de v
\de u
\\
=&
\E\Bigg[
\int_{A^2}
(\1\{u\in X\cap A\} h_{\theta}(u;X\cap A\setminus\{u\})+\1\{u\notin X\}h_{\theta}(u;X\cap A))
\lambda_{\theta}(u|X\cap A)
\\
&\times
(\1\{v\in X\cap A\}h_{\theta}(v;X\cap A\setminus\{v\})+\1\{v\notin X\cap A\}h_{\theta}(v;X\cap A))
\lambda_{\theta}(v|X\cap A)
\de v
\de u
\Bigg]
.
\end{align*}
Note that Lemma \ref{DOM} is applicable here since $\lambda_{\theta}$ and $h_{\theta}$ are bounded by assumption. 
Now, for any element $\omega\in\Omega$ of the underlying probability space, by the local finiteness of $\nn$ and the boundedness of $A$, the realisation $X(\omega)\cap A = \x$ is a discrete finite collection of points. 
Thus, recalling that the reference measure $|\cdot|$ is non-atomic, each $X(\omega)\cap A = \x$ is a $|\cdot|$-null set. This implies that the integral over the terms containing $\1\{u\in X\cap A\}$ and $\1\{v\in X\cap A\}$ are 0.
Hence, 
\begin{align*}
\lim_{k\to\infty} E_1(k)
=&
\E\left[
\int_{A^2}
\1\{u,v\notin X\cap A\}h_{\theta}(u;X\cap A)
\lambda_{\theta}(u|X\cap A)
h_{\theta}(v;X\cap A)
\lambda_{\theta}(v|X\cap A)
\de v
\de u
\right].
\end{align*}
The same $|\cdot|$-null set arguments further yield that this integral is indistinguishable from the integral which we obtain by excluding $\1\{u,v\notin X\cap A\}$ above. 
Consequently, 
\begin{align*}
\lim_{k\to\infty} E_1(k)
=\int_{A^2}\E[h_\theta(u,X\cap A)
\lambda_{\theta}(u|X\cap A)h_\theta(v,X\cap A)
\lambda_{\theta}(v|X\cap A)
]
\de v
\de u
= E_3
.
\end{align*}

Turning to $E_2(k)$, we have that 
\begin{align*}
    E_2(k)=&
\int_{A^2}
\sum_{i=1}^k
\E\big[
\E\big[\1\{u\in A_{ik}\}h_{\theta}(u;X\cap (A\setminus A_{ik}))
\lambda_{\theta}(u|X\cap A)
\big|
X\cap (A\setminus A_{ik})
\big]
\\
&\times
h_{\theta}(v;X\cap A)
\lambda_{\theta}(v|X\cap A)\big]
\de u \de v\\
=&
\int_{A^2}
\E\Bigg[
\E\Bigg[
\lambda_{\theta}(u|X\cap A)
\sum_{i=1}^k\1\{u\in A_{ik}\}h_{\theta}(u;X\cap (A\setminus A_{ik}))
\Bigg|
X\cap (A\setminus A_{ik})
\Bigg]
\\
&\times
h_{\theta}(v;X\cap A)
\lambda_{\theta}(v|X\cap A)\Bigg]
\de u \de v
.
\end{align*}

By the self-adjointness of conditional expectations we now obtain 
\begin{align*}
E_2(k)
=&
\int_{A^2}
\E\Bigg[
\lambda_{\theta}(u|X\cap A)
\sum_{i=1}^k
\1\{u\in A_{ik}\}h_{\theta}(u;X\cap (A\setminus A_{ik}))
\\
&\times
\E\big[
h_{\theta}(v;X\cap A)
\lambda_{\theta}(v|X\cap A)
\big|
X\cap (A\setminus A_{ik})
\big]
\Bigg]
\de u \de v
\\
=&
\int_{A^2}
\E\Big[
\lambda_{\theta}(u|X\cap A)
h_{\theta}(u;X\cap (A\setminus A_{k}(u)))
\\
&\times
\E\Big[
h_{\theta}(v;X\cap A)
\lambda_{\theta}(v|X\cap A)
\Big|
X\cap (A\setminus A_{k}(u))
\Big]
\Big]
\de u \de v
.
\end{align*}
As $k\to\infty$, $A_{k}(u)$ is a decreasing sequence of sets, and $A\setminus A_{k}(u)$ is an increasing sequence of sets tending to $A$. This means that the $\sigma$-algebras $\sigma(X\cap (A\setminus A_{k}(u)))$, $k\geq2$, are increasing (in terms of set inclusion), tending to $\sigma(X\cap A)$. Thus, we may apply Martingale convergence \citep[Theorem 5.7]{durrett2019probability} to obtain that $\E[h_{\theta}(v;X\cap A) \lambda_{\theta}(v|X\cap A) | X\cap (A\setminus A_{k}(u))]$ a.s.\ tends to $\E[h_{\theta}(v;X\cap A) \lambda_{\theta}(v|X\cap A) | X\cap A]$ as $k\to\infty$. To see that this is indeed an increasing sequence, note that the finite dimensional distributions of $X\cap (A\setminus A_{k}(u))$ (which characterise its distribution) are all contained in the collection of finite dimensional distributions of $X\cap (A\setminus A_{k+1}(u))$, since $A\setminus A_{k}\subseteq A\setminus A_{k+1}$ by the imposed refinement property. Moreover, by the observation following \eqref{e:LimitIdentity} we have that $h_{\theta}(u;X\cap (A\setminus A_{k}(u)))$ a.s.\ tends to $\1\{u\in X\cap A\}h_{\theta}(u;X\cap A\setminus\{u\}) + 
\1\{u\notin X\cap A\}h_{\theta}(u;X\cap A)$ as $k\to\infty$. Hence, by Lemma \ref{DOM}, Slutsky's lemma \citep[Theorem 6']{ferguson}, the $|\cdot|$-null set arguments above and the law of total expectation, we obtain
\begin{align*}
\lim_{k\to\infty} E_2(k)
=&
\int_{A^2}
\E\Bigg[
\lambda_{\theta}(u|X\cap A)
\lim_{k\to\infty}
h_{\theta}(u;X\cap (A\setminus A_{k}(u)))
\\
&\times
\lim_{k\to\infty}
\E\left[
h_{\theta}(v;X\cap A)
\lambda_{\theta}(v|X\cap A)
\Bigg|
X\cap (A\setminus A_{k}(u))
\right]
\Bigg]
\de u \de v
\\
=&
\int_{A^2}
\left.
\E\left[
h_{\theta}(u;X\cap A)
\lambda_{\theta}(u|X\cap A)
\E\left[
h_{\theta}(v;X\cap A)
\lambda_{\theta}(v|X\cap A)
\right|
X\cap A
\right]
\right]
\de u \de v
\\
=&
\int_{A^2}
\left.
\E\left[
\E\left[
h_{\theta}(u;X\cap A)
\lambda_{\theta}(u|X\cap A)
h_{\theta}(v;X\cap A)
\lambda_{\theta}(v|X\cap A)
\right|
X\cap A
\right]
\right]
\de u \de v
\\
=&
\int_{A^2}
\E\left[
h_{\theta}(u;X\cap A)
\lambda_{\theta}(u|X\cap A)
h_{\theta}(v;X\cap A)
\lambda_{\theta}(v|X\cap A)
\right]
\de u \de v
= E_3
.
\end{align*}
This completes the proof.

\end{proof}

\subsubsection{Proof of Corollary \ref{cor:ConditionsThm}}
\label{s:proofCor}

Here follows the proof of Corollary \ref{cor:ConditionsThm}, which provides sufficient conditions for Theorem \ref{thm:TF} and Theorem \ref{thm:TF_Block} to be satisfied.

\begin{proof}[Proof of Corollary \ref{cor:ConditionsThm}]
The conditional intensity itself is bounded by assumption so we proceed with the test function. 
The case $\alpha=0$ is trivial so we focus on $\alpha>0$. Since $h_{\theta}(u;X^T)=\e^{\alpha(-\log V_{\theta}(u) - \Phi_1(u;\theta) - \Phi_2(u,X^T;\theta))}$, the test function is bounded if $-\log V_{\theta}(u) - \Phi_1(u;\theta) - \Phi_2(u,X^T;\theta)$ is smaller than some finite positive constant. The last two terms are bounded by assumption. In addition, we have that 
$\log V_{\theta}(u) = \log p + 
\log\E[\lambda_{\theta}(u|X) | X^T]$ is smaller than some finite positive constant because $\sup_{p\in(0,1)}\log p=0$ and because $\log\E[\lambda_{\theta}(u|X) | X^T]\leq \log C < \infty$, where $\lambda_{\theta}(u|X)\leq C < \infty$, by the monotonicity of conditional expectations. 
\end{proof}

\subsubsection{Proof of Theorem \ref{thm:PPLweights}}
\label{sec:proofPPLweights}
\begin{proof}[Proof of Theorem \ref{thm:PPLweights}]

Starting with auto-prediction, the result follows immediately by comparing to the innovations of \citet{baddeley2005residual}.

Turning to the Gibbs case, we have
\begin{align*}
V_{\theta}(u,X^T,X^V) 
&=
p(u)
\E
\left.
\left[
\frac{
\lambda_{\theta}(u|X^T\cup X^V)}{\lambda_{\theta}(u|X^T)} 
\right|X^T
\right]
\\
&=
p(u)
\E
\left.
\left[
\frac{
\e^{\Phi_1(u;\theta)+\Phi_2(u,X^T\cup X^V;\theta)}
}{\e^{\Phi_1(u;\theta)+\Phi_2(u,X^T;\theta)}} 
\right|X^T
\right]
\\
&=p(u)
\E[\e^{\Phi_2(u,X^T\cup X^V;\theta)-\Phi_2(u,X^T;\theta)}|X^T],
\end{align*}
and the linearity immediately gives us
\[
V_{\theta}(u,X^T,X^V) 
= 
p(u)
\E[
\e^{\Phi_2(u,X^T;\theta) + \Phi_2(u,X^V;\theta) -\Phi_2(u,X^T;\theta)}
|
X^T
]
.
\]

\end{proof}

\subsubsection{Proof of Lemmas}
\label{sec:lemmas}

\begin{proof}[Proof of Lemma \ref{lemma:HardCore}]
Let $\x\subseteq \y$. We then have 
$
\lambda(u|\x) = \beta \1\{u\notin \bigcup_{x \in \x} b(x,R)\}\geq \beta \1\{u\notin \bigcup_{y \in \y} b(y,R)\} = \lambda(u|\y).
$ 
The inequality holds since if $\x\subseteq\y$, then $\bigcup_{x \in \x} b(x,R) \subseteq \bigcup_{y \in \y} b(y,R)$. This means that $\1\{u\notin \bigcup_{x \in \x} b(x,R)\}\geq \1\{u\notin \bigcup_{y \in \y} b(y,R)\}$, 
since indicators are non-decreasing functions with respect to set inclusion.

For the weight expression, we take the expression for the Strauss process in Lemma \ref{lemma:Strauss}.
When we let $\gamma\to0$, by appealing to the bounded convergence theorem, the continuity of $\gamma\mapsto\gamma^x$, $\gamma\in[0,1]$, and the convention that $0^0=1$, we obtain that PPL-weights for hard-core models satisfy
\begin{align*}
V_{\theta}(u,X^T,X^V) 
&=
p(u)
\E\left[0^{D_R(u,X^V)}\Big|X^T\right]
=
p(u)
\E\bigg[\1\bigg\{u\notin \bigcup_{x \in X^V\setminus\{u\}} b(x,R)\bigg\}
\bigg| X^T\bigg]
\\
&=
p(u)
\P\bigg(u\notin \bigcup_{x \in X^V} b(x,R)\bigg| X^T\bigg)
=
p(u)
\P\bigg(\bigcap_{x \in X^V} \{u\notin b(x,R)\}\bigg| X^T\bigg)
.
\end{align*} 
\end{proof}

\begin{proof}[Proof of Lemma \ref{lemma:Strauss}]

Given $\x\subseteq \y$, 
since 
\begin{align*}
    D_R(u; \x) =& 
    \sum_{x \in \x \setminus \{u\}} \1\{u \in b(x,R)\} 
    \leq 
    \sum_{x \in \x \setminus \{u\}} \1\{u \in b(x,R)\} 
    \\
    &+ \sum_{y \in \y \setminus (\x\cup\{u\})} \1\{u \in b(y,R)\}
    =\sum_{y \in \y \setminus \{u\}} \1\{u \in b(y,R)\} = D_R(u; \y),
\end{align*}
and $\gamma\in[0,1]$, we have  $\gamma^{D_R(u,\x)}\geq \gamma^{D_R(u,\y)}$, whereby 
$$
\lambda(u|\x) = \beta\gamma^{D_R(u,\x)}\geq \beta\gamma^{D_R(u,\y)} = \lambda(u|\y).
$$ 

For the weight expression we use the expression in Theorem \ref{thm:PPLweights} and that the Strauss process is a pairwise interaction model to obtain
\begin{align*}
V_{\theta}(u,X^T,X^V) 
&=
p(u)
\E[\e^{\Phi_2(u,X^V;\theta)}|X^T]
=
p(u)
\E[\e^{D_R(u,X^V)\log \gamma}|X^T]
\\
&=
p(u)
\E[\gamma^{D_R(u,X^V)}|X^T]
,
\end{align*}
where $\theta=(\beta,R,\gamma)\in(0,\infty)\times(0,\infty)\times[0, 1]$.
\end{proof}

\begin{proof}[Proof of Lemma \ref{lem:counter}]
Recall that when $\x\subseteq \y$, $\lambda(u|\x) \leq \lambda(u|\y)$ means that the process is attractive, while $\lambda(u|\x) \geq \lambda(u|\y)$ means that it is repulsive.  We want to show that it does not hold in general that a Geyer model is attractive when $\gamma > 1$, nor that it is repulsive when $\gamma < 1$. We do this by showing an example of a specification of the Geyer saturation process, point patterns $\x$ and $\y$ with $\x\subseteq \y$, and a location $u\in\Sm$, where we have that $\lambda(u|\x) \geq \lambda(u|\y)$ for $\gamma>1$ and $\lambda(u|\x) \leq \lambda(u|\y)$ for $\gamma<1$. This means that it does not hold, for all specifications of the Geyer saturation process, all point patterns $\x$ and $\y$ where $\x\subseteq\y$, and all locations $u\in \Sm$, that $\lambda(u|\x) \leq \lambda(u|\y)$ for $\gamma>1$ and $\lambda(u|\x) \geq \lambda(u|\y)$ for $\gamma<1$. 

Consider the Geyer saturation process with saturation threshold $s = 1$. The other parameters, $\beta$, $R$ and $\gamma$, are arbitrary. Consider two point patterns $\x = \{\eta\}$ and $\y = \{\eta,\zeta\}$ where the points $\eta$ and $\zeta$ satisfy $d(\eta,\zeta)\leq R$. Then, naturally $\x\subseteq \y$. 
Now, let $\xi$ be a point such that $d(\xi,\eta)\leq R$ and $d(\xi,\zeta)\leq R$. This means that $\xi$ is an $R$-close neighbour of both $\eta$ and $\zeta$, which in turn are $R$-close neighbours of each other.
    The conditional intensity is defined as
    $$   
\lambda(u|\z) = \beta\gamma^{\min(s,D_R(u,\z))+\sum_{y\in\z}(\min(s,D_R(y,\z\cup \{u\})-\min(s,D_R(y,\z)))},
$$ 
for any point pattern $\z$ and location $u\in\Sm$.
Now, consider only the term
$$
    \min(s,D_R(u,\z))+\sum_{y\in\z}(\min(s,D_R(y,\z\cup \{u\})-\min(s,D_R(y,\z))),
$$
for $u = \xi$, $\z = \y = \{\eta,\zeta\}$ and $s = 1$, i.e. 
$$
    \min(1,D_R(\xi,\{\eta,\zeta\}))+\sum_{y\in\{\eta,\zeta\}}(\min(1,D_R(y,\{\eta,\zeta\}\cup \{\xi\})-\min(1,D_R(y,\{\eta,\zeta\}))).
$$
    We obtain
    \begin{align*}
    D_R(\xi, \{\eta,\zeta\}) &= 
    \sum_{x \in \{\eta,\zeta\} \setminus \{\xi\}} \1\{d(\xi,x)\leq R\}\\
    &=\1\{d(\xi,\eta)\leq R\}+\1\{d(\xi,\zeta)\leq R\} = 2,\\
    D_R(\eta, \{\eta,\zeta\}) &= 
    \sum_{x \in \{\eta,\zeta\} \setminus \{\eta\}} \1\{d(\eta,x)\leq R\}\\
    &=\1\{d(\eta,\zeta)\leq R\} = 1,\\
    D_R(\zeta, \{\eta,\zeta\}) &= 
    \sum_{x \in \{\eta,\zeta\} \setminus \{\zeta\}} \1\{d(\zeta,x)\leq R\}\\
    &=\1\{d(\zeta,\eta)\leq R\} = 1,\\
    D_R(\eta, \{\eta,\zeta\}\cup \{\xi\}) &= 
    \sum_{x \in \{\eta,\zeta,\xi\} \setminus \{\eta\}} \1\{d(\eta,x)\leq R\}\\
    &=\1\{d(\eta,\zeta)\leq R\} +\1\{d(\eta,\xi)\leq R\} = 2,\\
    D_R(\zeta, \{\eta,\zeta\}\cup \{\xi\}) &= 
    \sum_{x \in \{\eta,\zeta,\xi\} \setminus \{\zeta\}} \1\{d(\zeta,x)\leq R\}\\
    &=\1\{d(\zeta,\eta)\leq R\} +\1\{d(\zeta,\xi)\leq R\}= 2.
    \end{align*}
    We thus have 
    \begin{align*}
        &\min(1,D_R(\xi,\{\eta,\zeta\}))+\sum_{y\in\{\eta,\zeta\}}(\min(1,D_R(y,\{\eta,\zeta\}\cup \{\xi\})-\min(1,D_R(y,\{\eta,\zeta\})))\\
=&\min(1,2)+(\min(1,2)-\min(1,1)) + (\min(1,2)-\min(1,1)) \\
        =&1+(1-1) + (1-1) = 1.\\
    \end{align*}
    We also compute the corresponding term but for $\x$, i.e.
    \begin{align*}
        &\min(1,D_R(\xi,\{\eta\}))+\sum_{y\in\{\eta\}}(\min(1,D_R(y,\{\eta\}\cup \{\xi\})-\min(1,D_R(y,\{\eta\})))\\
        =&\min(1,1)+\sum_{y\in\{\eta\}}(\min(1,D_R(y,\{\eta\}\cup \{\xi\})-\min(1,D_R(y,\{\eta\})))\\
        =&\min(1,1)+(\min(1,1)-\min(1,0)) =1+(1-0) = 2.\\
    \end{align*}
    This means that for $s=1$,
    \begin{align*}
    &\min(s,D_R(\xi,\x))+\sum_{y\in\x}(\min(s,D_R(y,\x\cup \{\xi\})-\min(s,D_R(y,\x)))\\
    \geq&
        \min(s,D_R(\xi,\y))+\sum_{y\in\y}(\min(s,D_R(y,\y\cup \{\xi\})-\min(s,D_R(y,\y))).
    \end{align*}
In this case, if $\gamma < 1$,
    \begin{align*}
    \lambda(\xi|\x) &= \beta\gamma^{\min(s,D_R(\xi,\x))+\sum_{y\in\x}(\min(s,D_R(y,\x\cup \{\xi\})-\min(s,D_R(y,\x)))}\\
    &\leq \beta\gamma^{\min(s,D_R(\xi,\y))+\sum_{y\in\y}(\min(s,D_R(y,\y\cup \{\xi\})-\min(s,D_R(y,\y)))} = \lambda(\xi|\y),
    \end{align*}
    and conversely, if $\gamma > 1$, $\lambda(\xi|\x) \geq \lambda(\xi|\y)$. 
So it does not hold in general that  a Geyer model is attractive when $\gamma > 1$ or that it is repulsive when $\gamma < 1$. 
    This concludes the proof.
\end{proof}

\begin{proof}[Proof of Lemma \ref{lem:geyer}]
In \citep{opitz2009simulation} the conditional intensity is rewritten as
\begin{align*}
    \lambda(u|\x) = \beta\gamma^{\min(s,D_R(u,\x))
    +
    \sum_{y\in\x}\1\{y \in b(u,R)\}
    \1\{1\leq D_R(y,\x\cup \{u\})\leq s\}
    },
\end{align*}
since 
\begin{align*}
    &\sum_{y\in\x}(\min(s,D_R(y,\x\cup \{u\})-\min(s,D_R(y,\x)))\\
    =&
    \sum_{y\in\x}
    \1\{y \in b(u,R)\}
    \1\{1\leq D_R(y,\x\cup \{u\})\leq s\}.
\end{align*}
Using Theorem \ref{thm:PPLweights}, we obtain
\begin{align*}
&V_{\theta}(u,X^T,X^V) 
=
p(u)
\E\left[
\e^{\Phi_2(u,X;\theta)
-
\Phi_2(u,X^T;\theta)
}
|X^T
\right]
,
\end{align*}
where
\begin{align*}
\Phi_2(u,X;\theta)
&=
    \log(\gamma)(\min(s,D_R(u,X))
    \\
&+
\sum_{y\in X}
\1\{y \in b(u,R)\}
\1\{1\leq D_R(y,X\cup \{u\})\leq s\}),\\
\Phi_2(u,X^T;\theta)
&=
\log(\gamma)(\min(s,D_R(u,X^T))
\\
&+
\sum_{y\in X^T}
\1\{y \in b(u,R)\}
\1\{1\leq D_R(y,X^T\cup \{u\})\leq s\}).
\end{align*}
\end{proof}

\section{Implementation}
\label{sec:implementation}
For a given cross-validation round $(\x_i^T,\x_i^V)$, $i=1,\ldots,k$, the prediction errors are computed in the following way:
$$
\I_{\xi_{\theta}}^{h_{\theta}}(\Sm,\x^T_i,\x^V_i)
= 
\sum_{x\in \x^V_i \cap \Sm}
h_{\theta}(x;\x^T_i)
-
\int_{\Sm}
h_{\theta}(u;\x^T_i)
V_{\theta}(u)
\lambda_{\theta}(u|\x^T_i)
\de u.$$
We set $h_{\theta}(x;\y\setminus\{x\}) = f(V_{\theta}(u)
\lambda_{\theta}(u|\x^T_i))$, where $f$ is an arbitrary function, and in the Stoyan-Grabarnik case, $f(x) = 1/x$. 
We now want to numerically evaluate
$$\I_{\xi_{\theta}}^{h_{\theta}}(\Sm,\x^T_i,\x^V_i)
=
\sum_{x\in \x^V_i \cap \Sm}
f(V_{\theta}(x)
\lambda_{\theta}(x|\x^T_i))
-
\int_{\Sm}
f(V_{\theta}(u)
\lambda_{\theta}(u|\x^T_i))
V_{\theta}(u)
\lambda_{\theta}(u|\x^T_i)
\de u.$$
We first find a weight estimate $\widehat V_{\theta}(u)$, as described in Section \ref{sec:w_choice}. 
Given an implementation of the conditional intensity, it is then possible for us to evaluate $\widehat V_{\theta}(u)
\lambda_{\theta}(u|\x^T_i)$ for different 
locations $u$. 
To calculate the integral, following \citet{cronie2023cross}, we approximate it using the Berman–Turner device \citep{berman1992approximating}, resulting in
$$\int_{\Sm}
f(\widehat V_{\theta}(u)
\lambda_{\theta}(u|\x^T_i))
\widehat V_{\theta}(u)
\lambda_{\theta}(u|\x^T_i)
\de u 
\approx 
\sum_{v \in \vv} q_v f(\widehat V_{\theta}(v)
\lambda_{\theta}(v|\x^T_i))
\widehat V_{\theta}(v)
\lambda_{\theta}(v|\x^T_i), $$
where $\vv\in \Sm$ are quadrature points consisting of the validation set points $x \in \x^V_i$ and some dummy points in the window $\Sm$. Further, $\{q_v:v\in\vv\}$ are quadrature weights 
satisfying 
$\sum_{v \in \vv} q_v = |\Sm|$. 
This gives us the approximation  
$$\I_{\xi_{\theta}}^{h_{\theta}}(\Sm,\x^T_i,\x^V_i)
\approx
\sum_{x\in X^V \cap \Sm}
f(\widehat V_{\theta}(x)
\lambda_{\theta}(x|X^T))
-
\sum_{v \in \vv} q_v f(\widehat V_{\theta}(v)
\lambda_{\theta}(v|\x^T_i))
\widehat V_{\theta}(v)
\lambda_{\theta}(v|\x^T_i).
$$
We can rewrite this by introducing an indicator function in the first sum:
\begin{align*}
\I_{\xi_{\theta}}^{h_{\theta}}(\Sm,\x^T_i,\x^V_i)
&\approx
\sum_{v\in \vv}
\1\{v\in \x^V_i\}
f(\widehat V_{\theta}(v)
\lambda_{\theta}(v|X^T))
-
\sum_{v \in \vv} q_v f(\widehat V_{\theta}(v)
\lambda_{\theta}(v|\x^T_i))
\widehat V_{\theta}(v)
\lambda_{\theta}(v|\x^T_i).
\\
&=
\sum_{v\in \vv}
f(\widehat V_{\theta}(v)
\lambda_{\theta}(v|X^T))
(\1\{v\in \x^V_i\}- q_v
\widehat V_{\theta}(v)
\lambda_{\theta}(v|\x^T_i)).
\end{align*}
Then, the loss functions can be computed according to \eqref{e:L1_2} and \eqref{e:L3} in Section \ref{s:PPL}. We minimize the loss functions over the parameter space $\Theta$ using grid search.

\section{Simulation study}
\label{sec:app_sim}
Section \ref{s:Simulations} presents results from the simulation study which compares the performance of PPL and Takacs-Fiksel estimation for Poisson, hard-core, Strauss and Geyer processes. 
Here, we present some additional results.
See Section \ref{s:Simulations} for 
details on 
the simulation procedure.

\subsection{Poisson process}
\label{sec:app_poi}
See Section \ref{s:SimPoisson} for a description of the parameters and grid used for the Poisson process.
In Section \ref{s:SimPoisson}, in Figure \ref{fig:poisson3-L1}, the results for the Poisson process for the $\Loss_1$ loss function for the PPL-weight $p$ are shown. Here, the results for the $\Loss_2$ and $\Loss_3$ loss functions are shown in Figure \ref{fig:poisson3-L2L3}. 
\label{s:plots}
\begin{figure}[!htb]
    \centering
    \includegraphics[width = 0.45\textwidth]{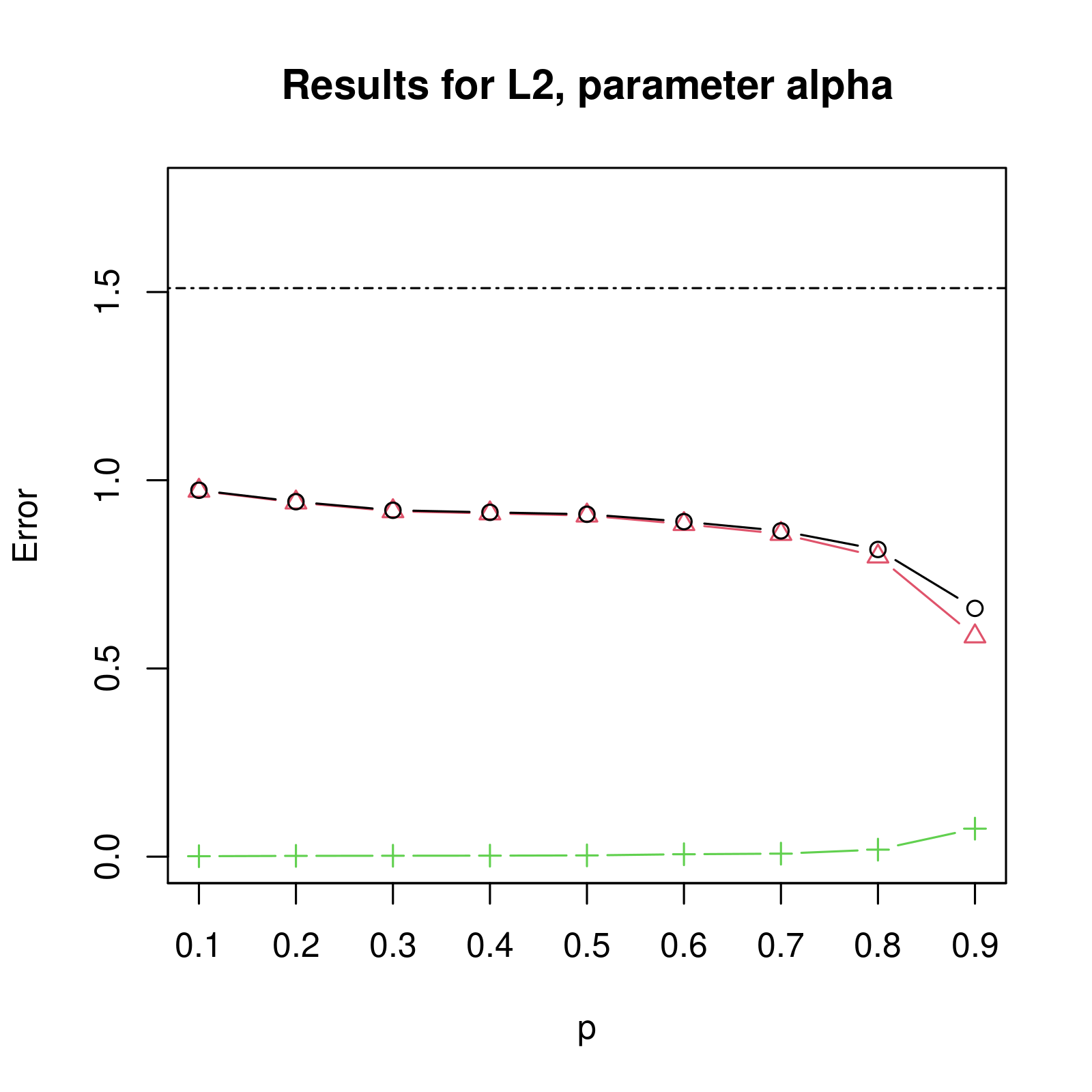}
    \includegraphics[width = 0.45\textwidth]{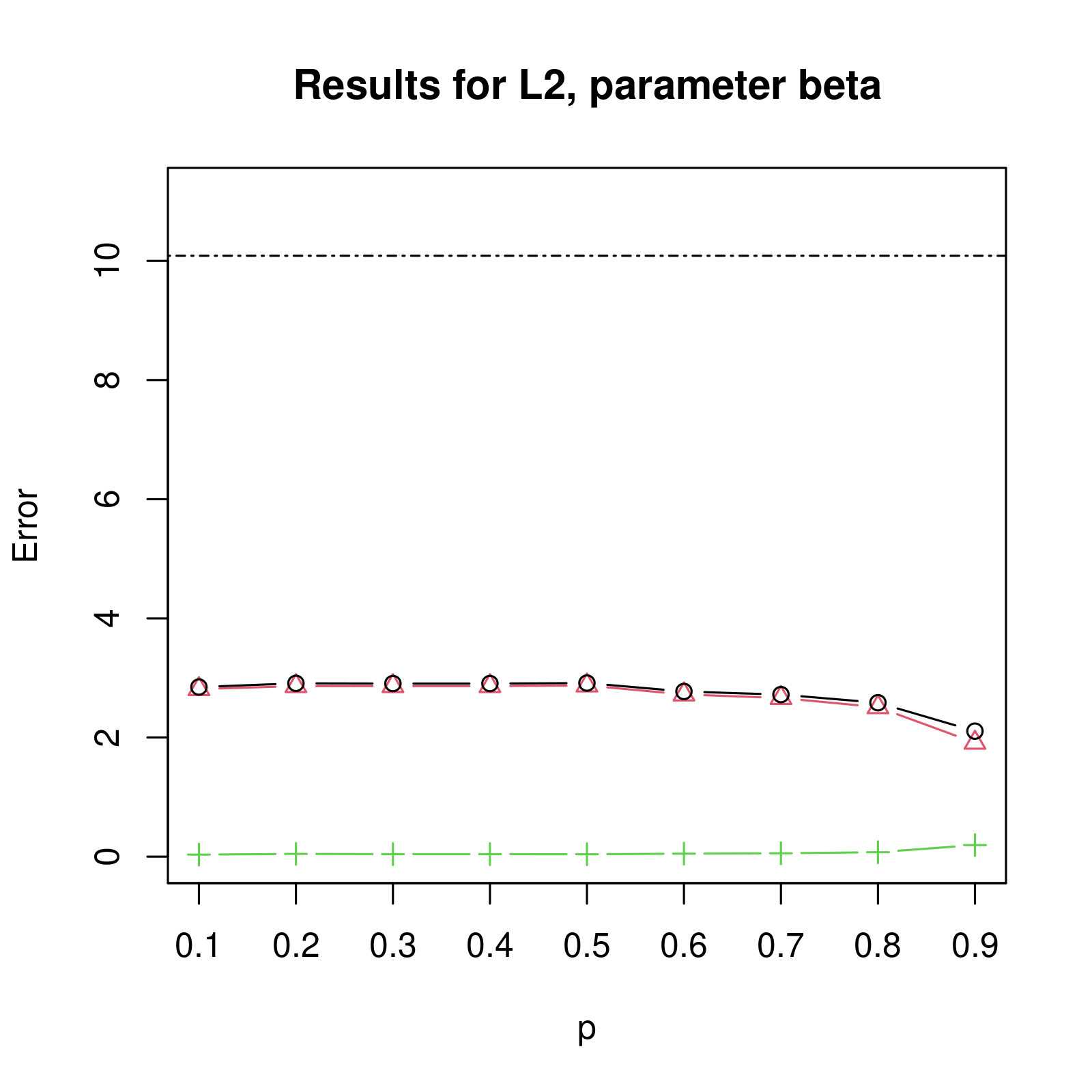}
    \includegraphics[width = 0.45\textwidth]{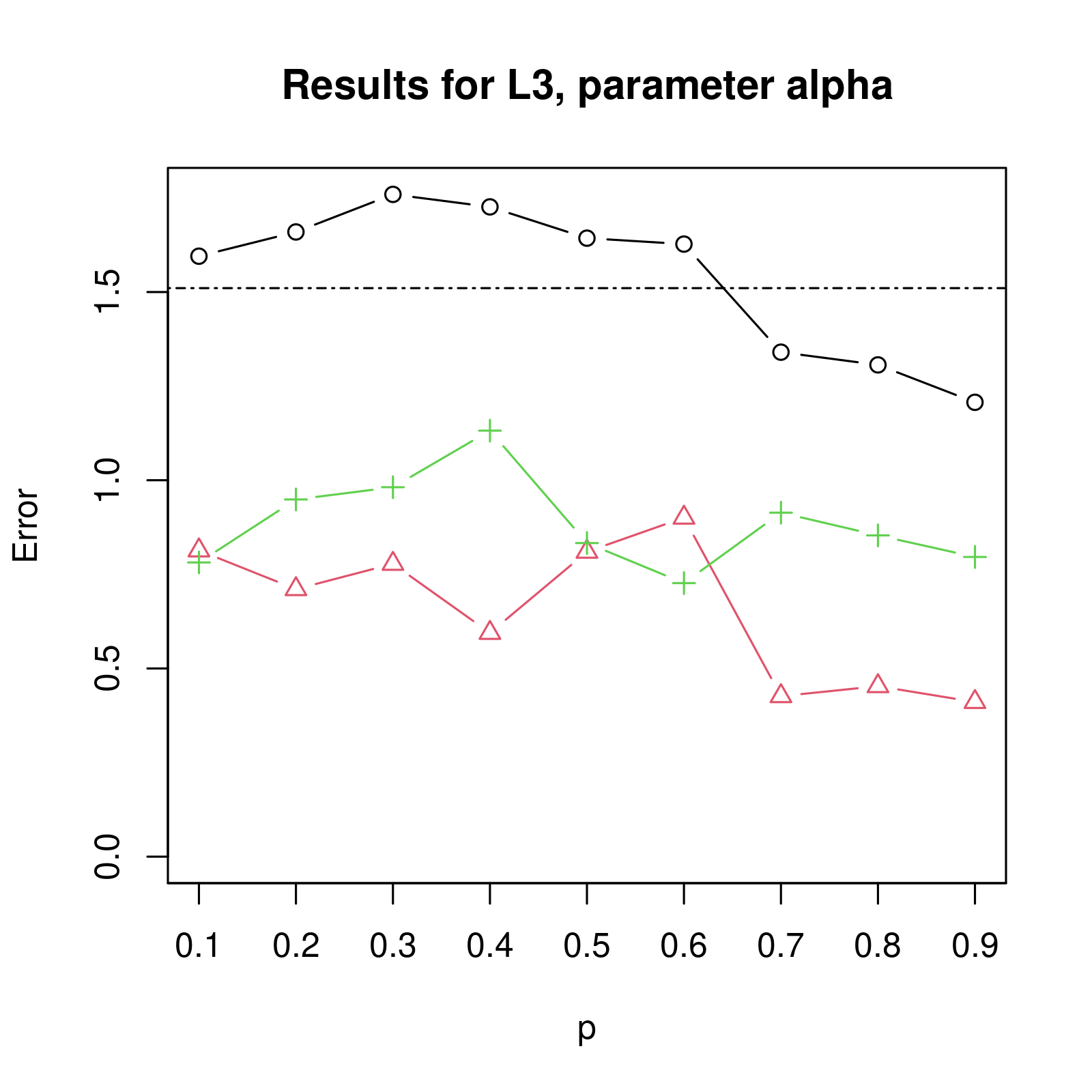}
    \includegraphics[width = 0.45\textwidth]{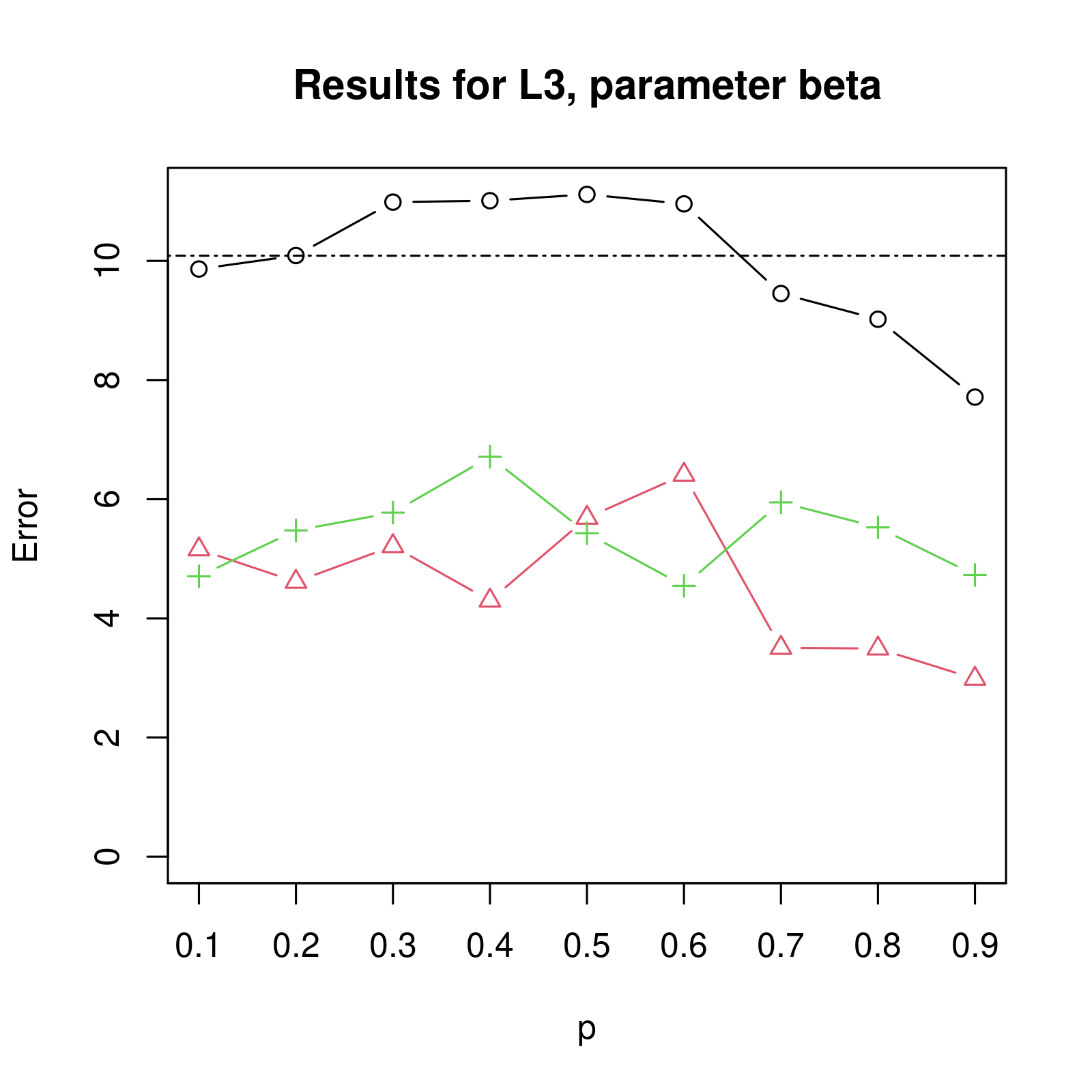}
    \caption{MSE, squared bias and variance for the Poisson process, using PPL with the loss functions $\Loss_2$ and $\Loss_3$, for the parameters $\alpha$ and $\beta$. Here $k = 100$, $N = 100$, the PPL-weight is set to $p$ and $p = 0.1,0.2,\ldots,0.9$. The black lines with circles correspond to MSE, the red lines with triangles correspond to the squared bias and the green lines with the plus signs correspond to the variance. The horizontal black dotted lines correspond to MSE for Takacs-Fiksel estimation.
    }
    \label{fig:poisson3-L2L3}
\end{figure}

\subsection{Hard-core process}
\label{sec:app_hc}
See Section \ref{s:hc_sims} for a description of the parameters and grid used for the hard-core process.
In Section \ref{s:hc_sims}, in Figure \ref{fig:hard-core-p-L1}, the results for the hard-core process for the $\Loss_1$ loss function for the PPL-weight estimate $p$ are shown. 
Here, the corresponding results for the $\Loss_2$ and $\Loss_3$ loss functions can be found in Figure \ref{fig:hard-core-p-L2L3}. 
Further, Figure \ref{fig:hard-core-(1-p)-L1} in Section \ref{s:hc_sims} shows the results for the PPL-weight estimate $p/(1-p)$ for the $\Loss_1$ loss function. 
Again, the results for the $\Loss_2$ and $\Loss_3$ loss functions can be found in Figure \ref{fig:hard-core-(1-p)-L2L3}. 
Lastly, the results using weight estimation according to \eqref{e:WeightEst}, and the $\Loss_1$ loss function can be seen in Figure \ref{fig:hard-core-est-L1} in Section \ref{s:hc_sims}.
The results for the $\Loss_2$ and $\Loss_3$ loss functions can be found in Figure \ref{fig:hard-core-est-L2L3}.

\begin{figure}[!htb]
    \centering
    \includegraphics[width = 0.45\textwidth]{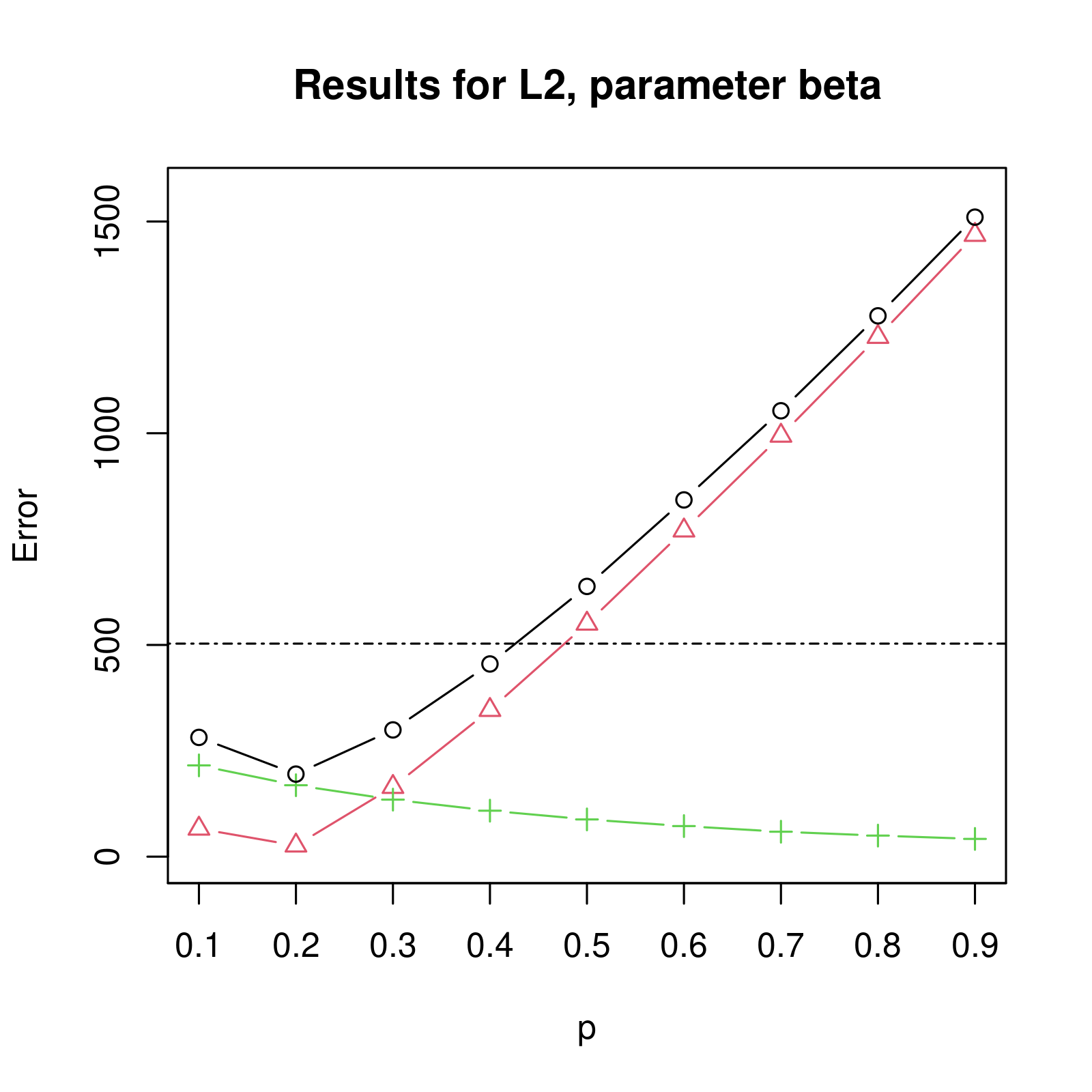}
    \includegraphics[width = 0.45\textwidth]{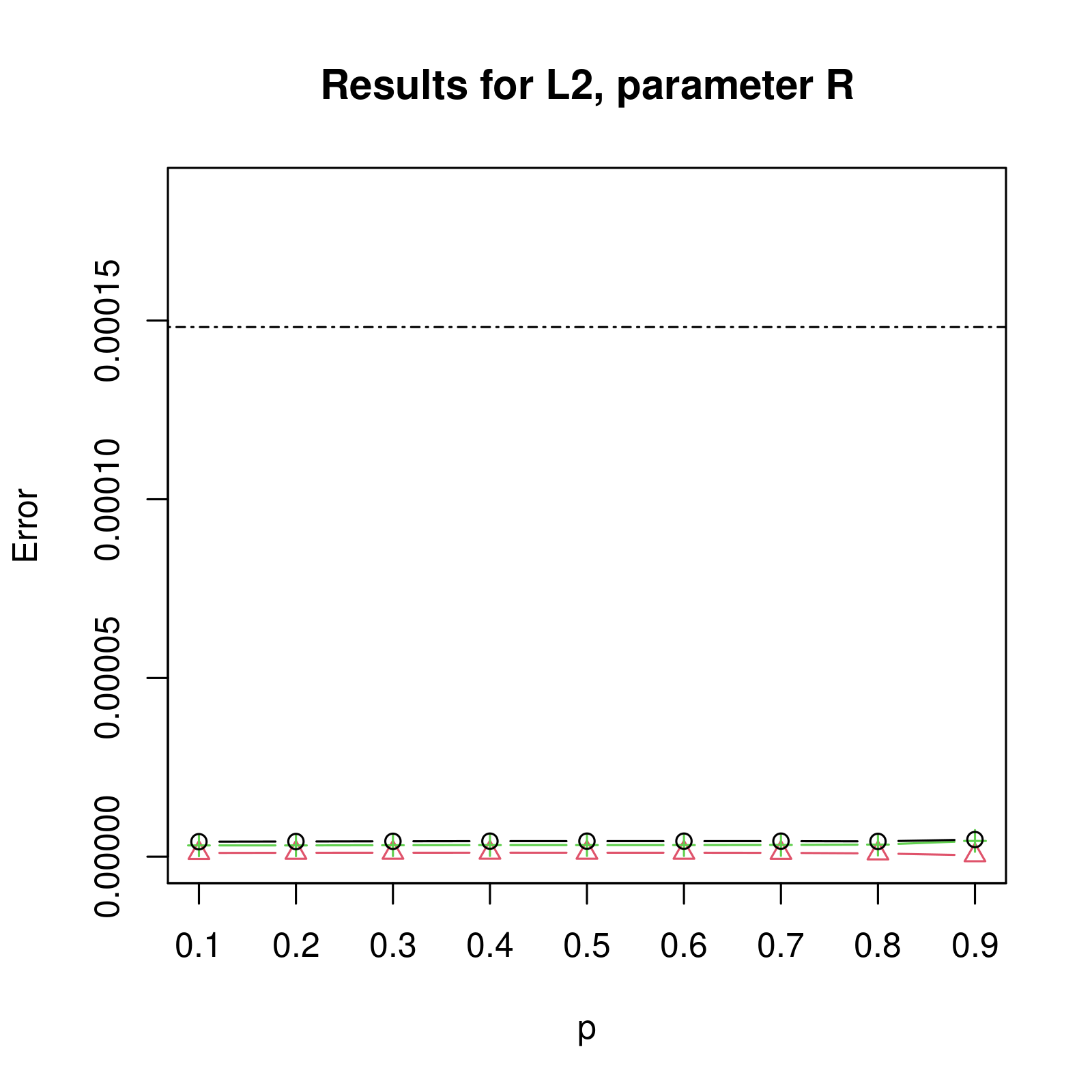}
    \includegraphics[width = 0.45\textwidth]{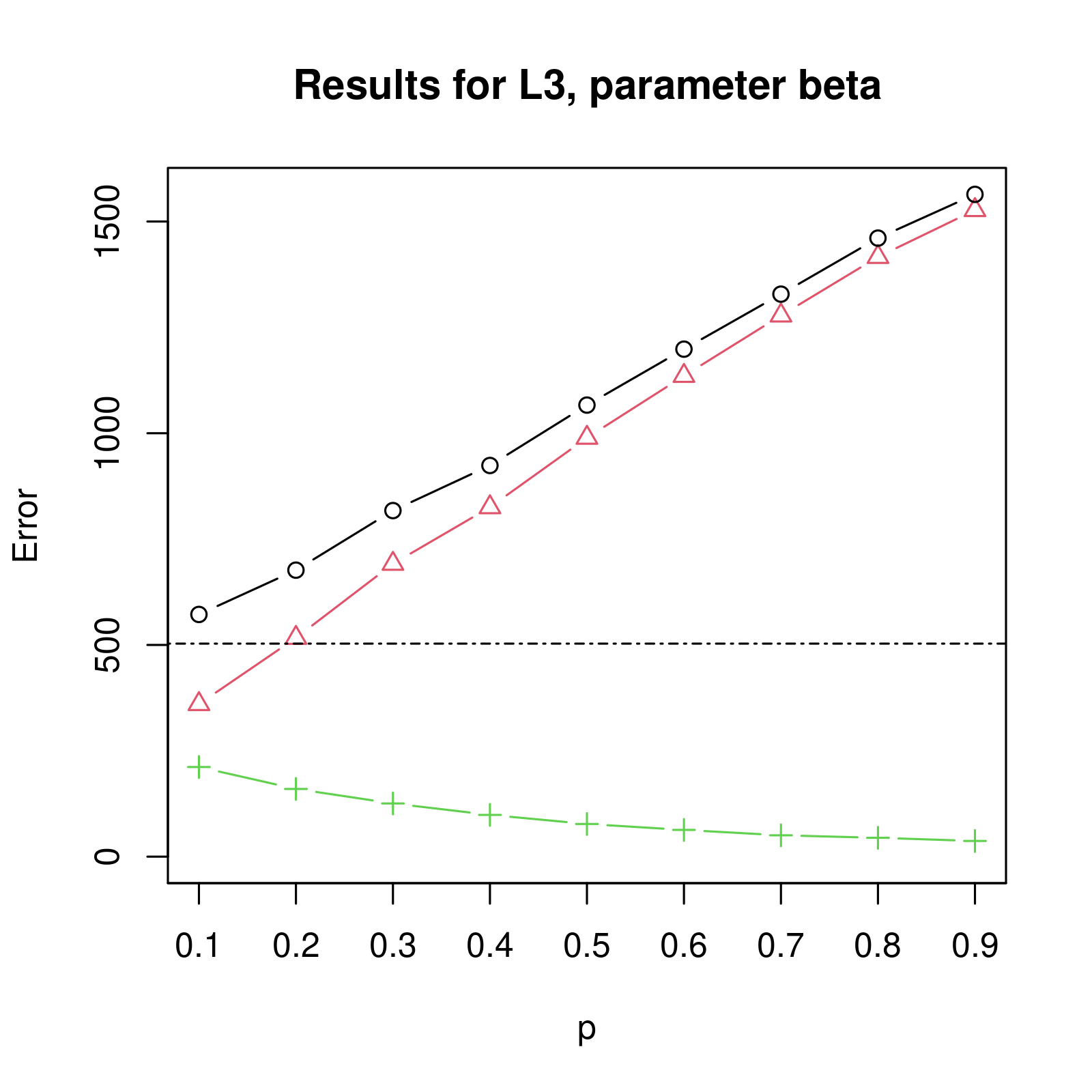}
    \includegraphics[width = 0.45\textwidth]{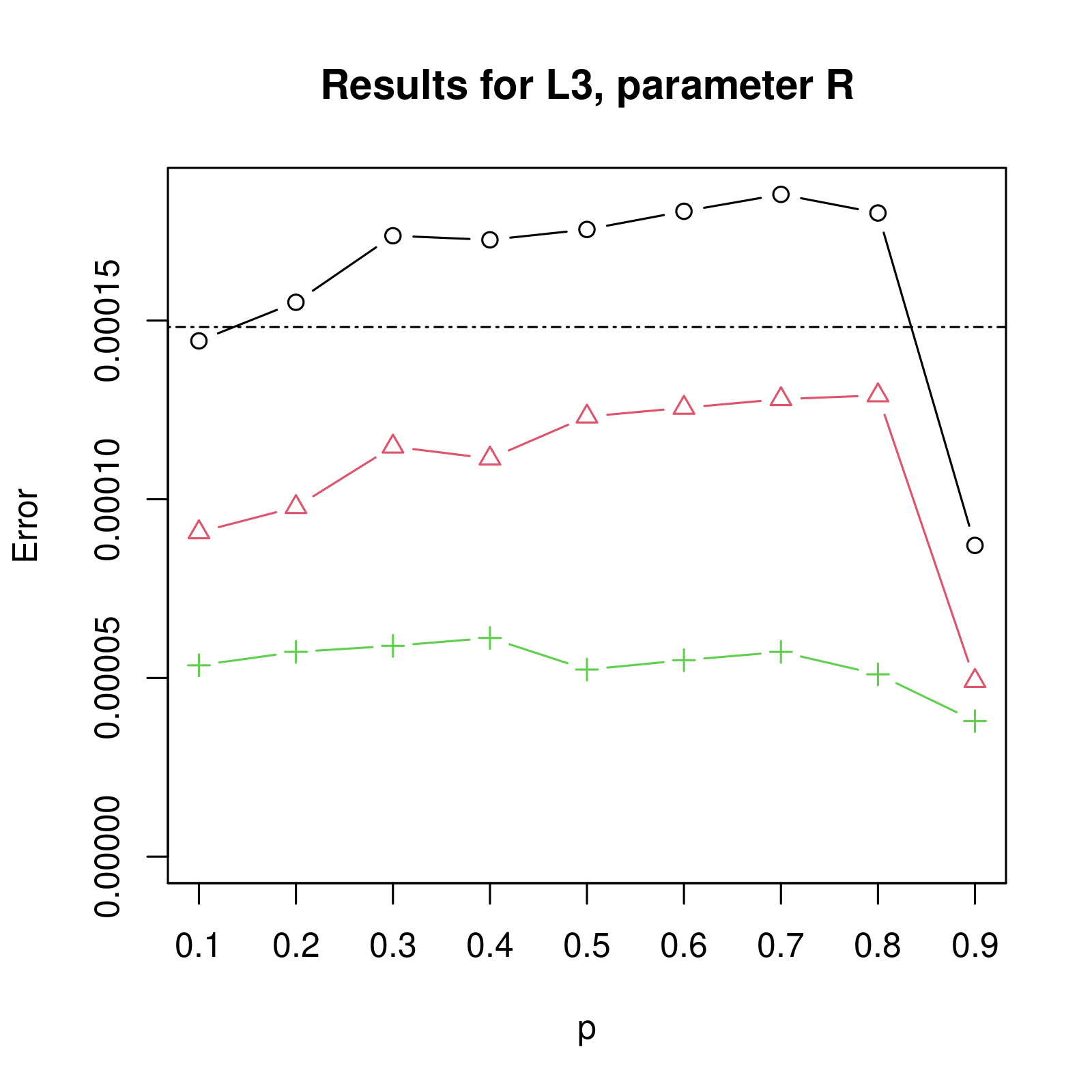}
    \caption{MSE, squared bias and variance for the hard-core model using PPL with the $\Loss_2$ and $\Loss_3$ loss functions, when estimating the  parameters $\beta$ and $R$. Here $k = 100$, $N = 500$,  $p = 0.1,0.2,\ldots,0.9$ and the PPL-weight is set to $p$. The black lines with circles correspond to MSE, the red lines with triangles correspond to squared bias and the green lines with plus signs correspond to variance. The black dotted lines correspond to the Takacs-Fiksel estimates.}
    \label{fig:hard-core-p-L2L3}
\end{figure}

\begin{figure}[!htb]
    \centering
    \includegraphics[width = 0.45\textwidth]{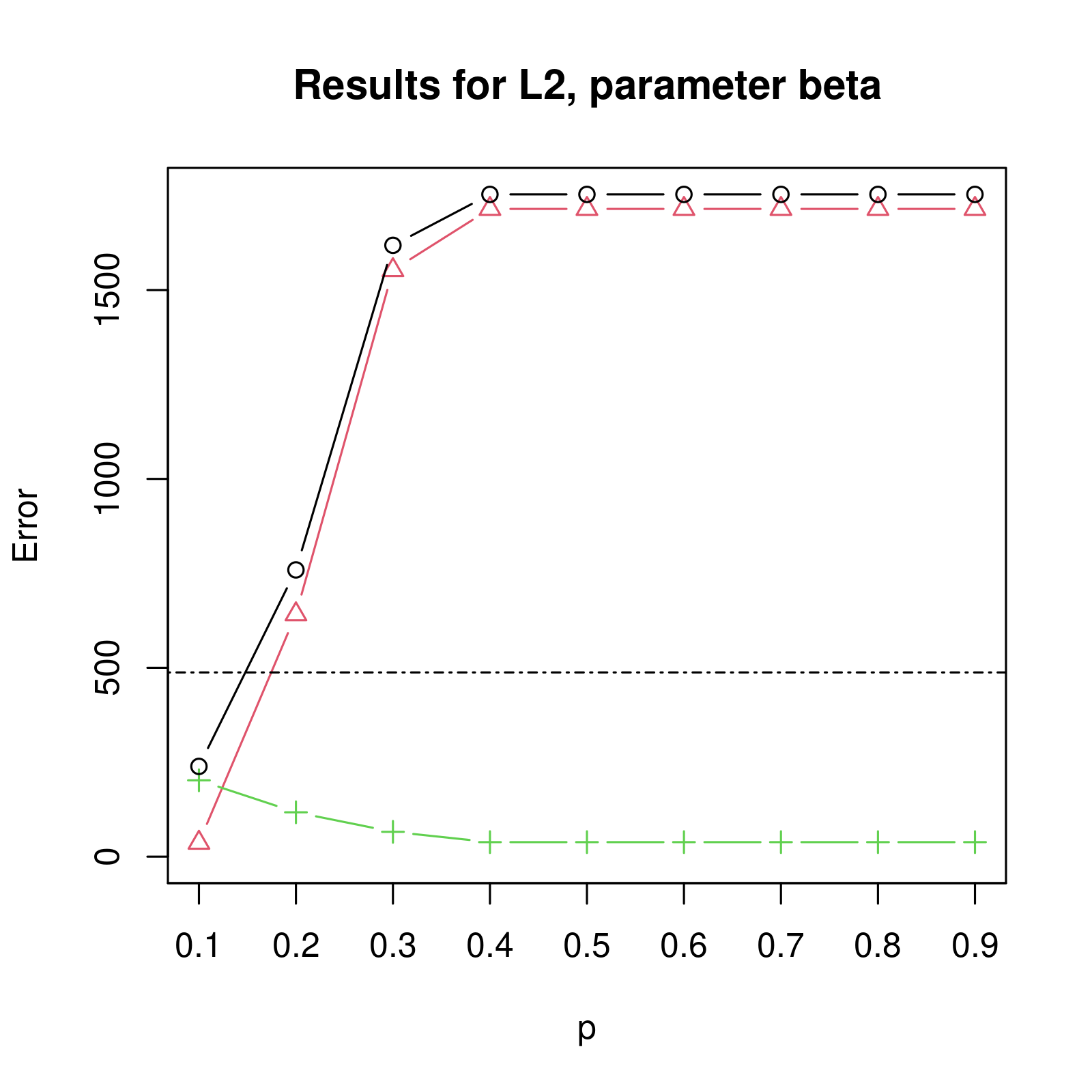}
    \includegraphics[width = 0.45\textwidth]{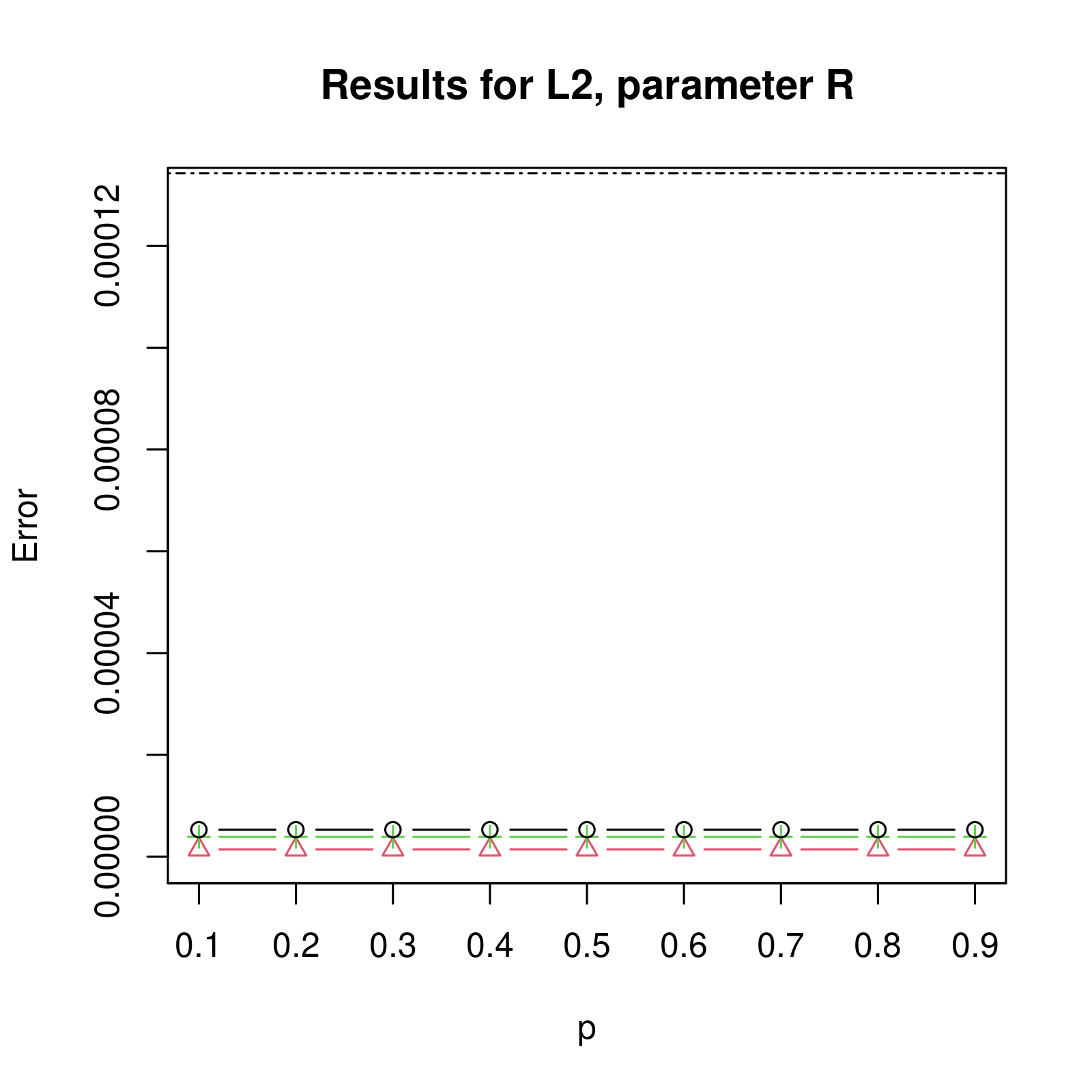}
    \includegraphics[width = 0.45\textwidth]{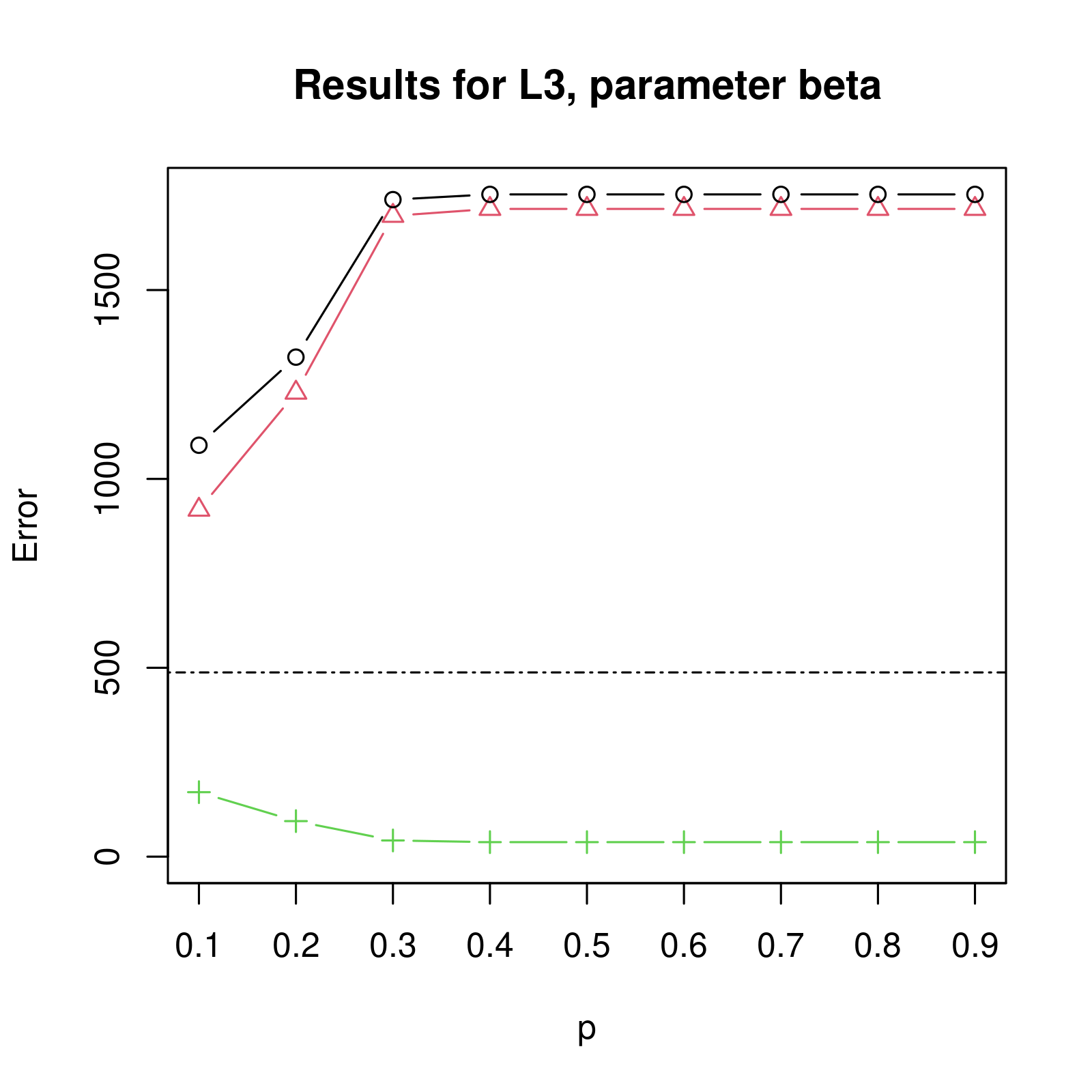}
    \includegraphics[width = 0.45\textwidth]{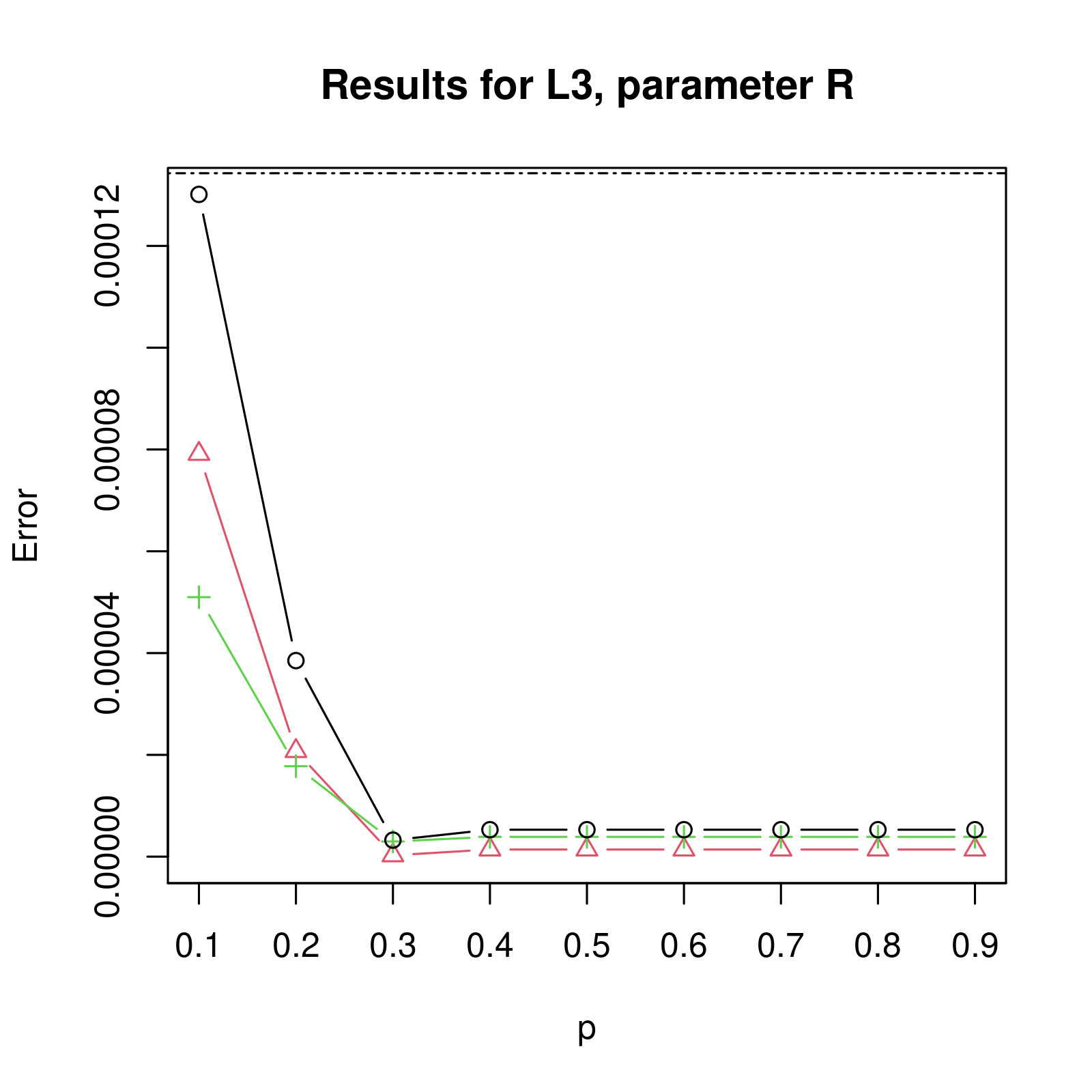}
    \caption{MSE, squared bias and variance for the hard-core model using PPL with the $\Loss_2$ and $\Loss_3$ loss functions, when estimating the  parameters $\beta$ and $R$. Here $k = 100$, $N = 100$,  $p = 0.1,0.2,\ldots,0.9$ and the PPL-weight is set to $p/(1-p)$. The black lines with circles correspond to MSE, the red lines with triangles correspond to squared bias and the green lines with plus signs correspond to variance. The black dotted lines correspond to the Takacs-Fiksel estimates.}
    \label{fig:hard-core-(1-p)-L2L3}
\end{figure}

\begin{figure}[!htb]
    \centering
    \includegraphics[width = 0.45\textwidth]{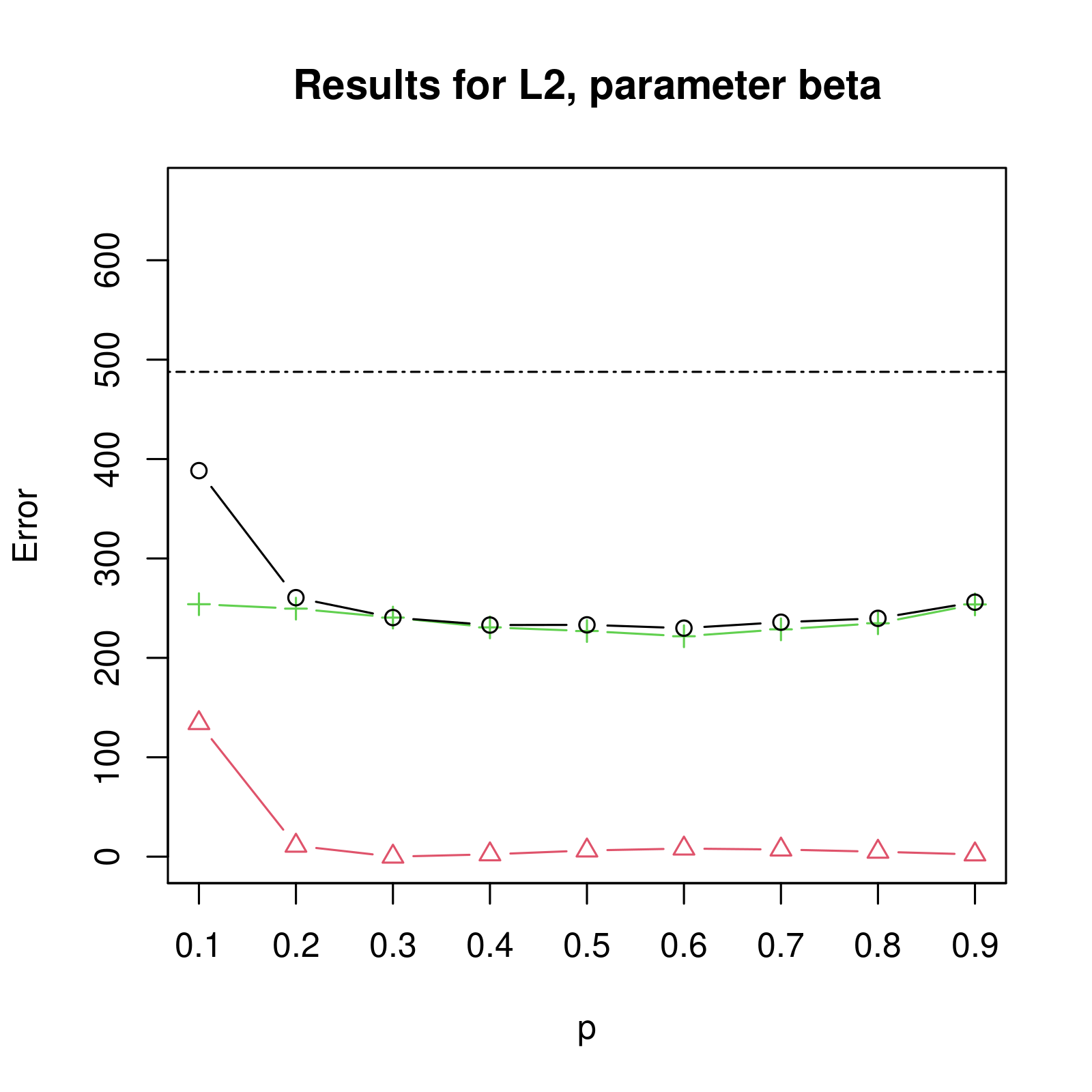}
    \includegraphics[width = 0.45\textwidth]{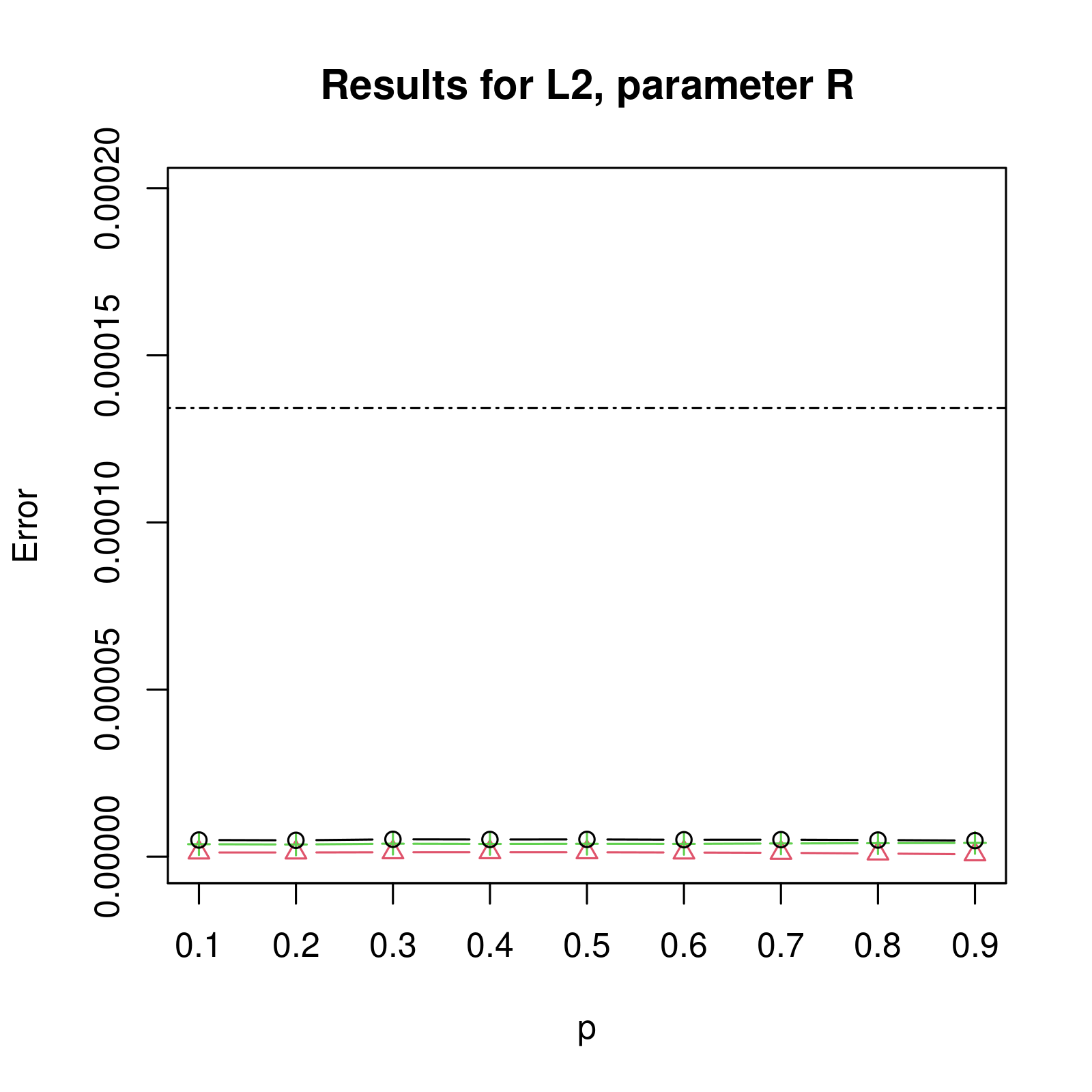}
    \includegraphics[width = 0.45\textwidth]{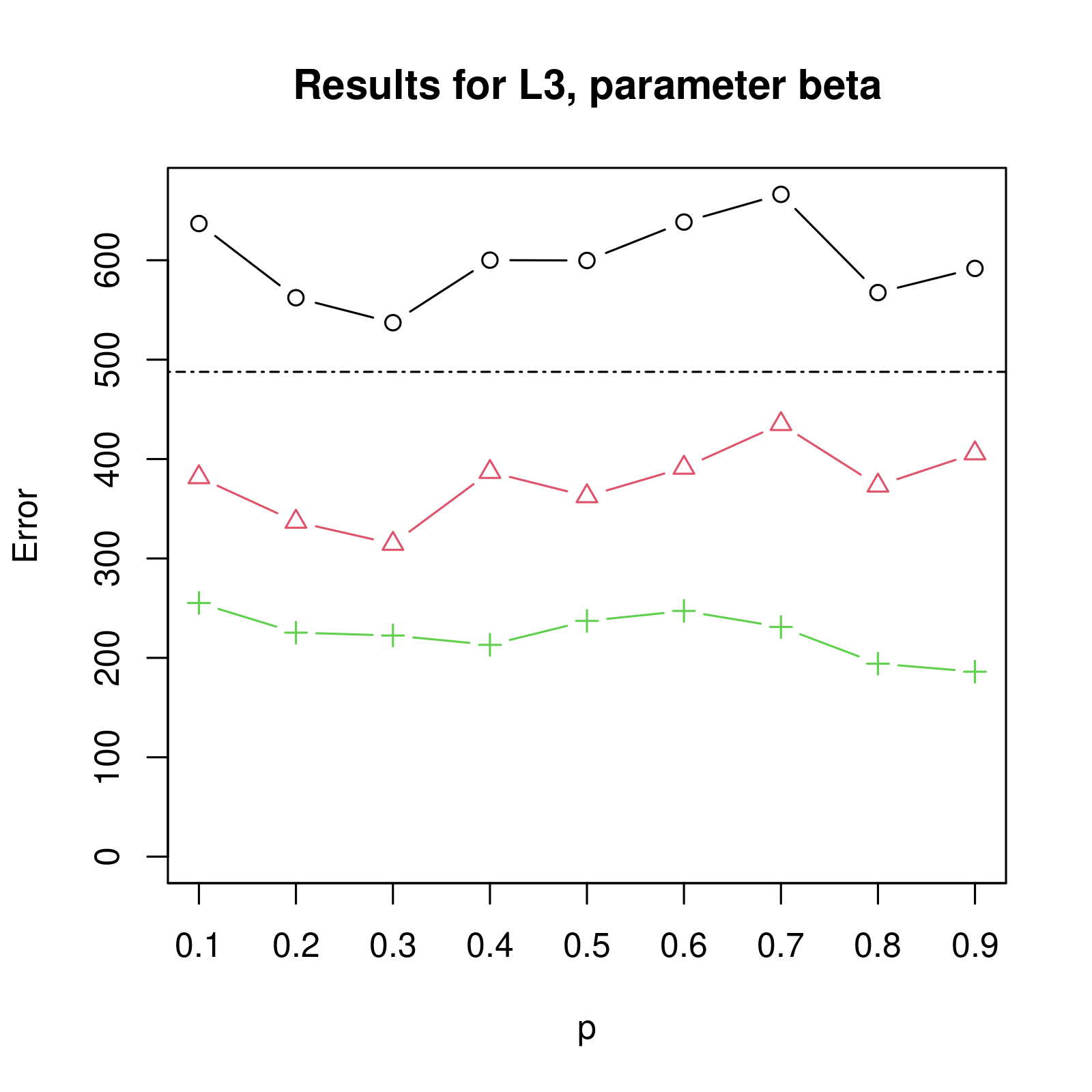}
    \includegraphics[width = 0.45\textwidth]{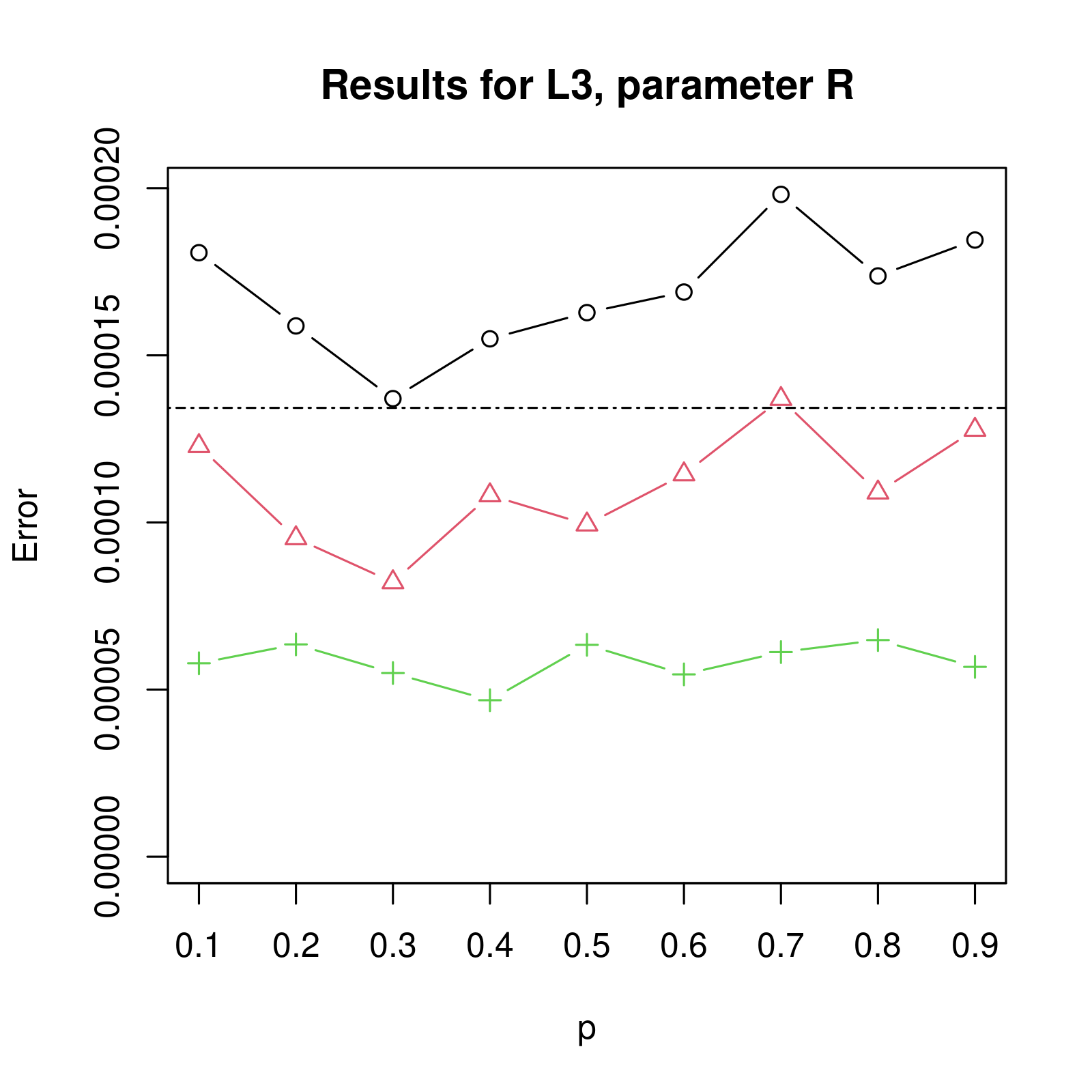}
    \caption{MSE, squared bias and variance for the hard-core model using PPL with the $\Loss_2$ and $\Loss_3$ loss functions, when estimating the parameters $\beta$ and $R$. Here $k = 100$, $N = 100$,  $p = 0.1,0.2,\ldots,0.9$ and the PPL-weight is estimated in accordance with \eqref{e:WeightEst}. The black lines with circles correspond to MSE, the red lines with triangles correspond to squared bias and the green lines with plus signs correspond to variance. The black dotted lines correspond to the Takacs-Fiksel estimates.}
    \label{fig:hard-core-est-L2L3}
\end{figure}

\subsection{Strauss process}
\label{sec:app_strauss}
See Section \ref{sec:strauss_sims} for a description of the parameters and grid used for the Strauss process.
In Section \ref{sec:strauss_sims}, in Figure \ref{fig:strauss-p-L1}, the results for the Strauss process for the $\Loss_1$ loss function for the PPL-weight estimate $p$ are shown. 
Here, the corresponding results for the $\Loss_2$ and $\Loss_3$ loss functions can be found in Figure \ref{fig:strauss-p-L2L3}. 
Further, Figure \ref{fig:strauss-(1-p)-L1} in Section \ref{sec:strauss_sims} shows the results for the PPL-weight estimate $p/(1-p)$ for the $\Loss_1$ loss function. 
Again, the results for the $\Loss_2$ and $\Loss_3$ loss functions can be found in Figure \ref{fig:strauss-(1-p)-L2L3}. 
Lastly, the results using weight estimation according to \eqref{e:WeightEst}, and the $\Loss_1$ loss function can be seen in Figure \ref{fig:strauss-est-L1} in Section \ref{sec:strauss_sims}.
The results for the $\Loss_2$ and $\Loss_3$ loss functions can be found in Figure \ref{fig:strauss-est-L2L3}.

\begin{figure}[!htb]
    \centering
    \includegraphics[width = 0.3\textwidth]{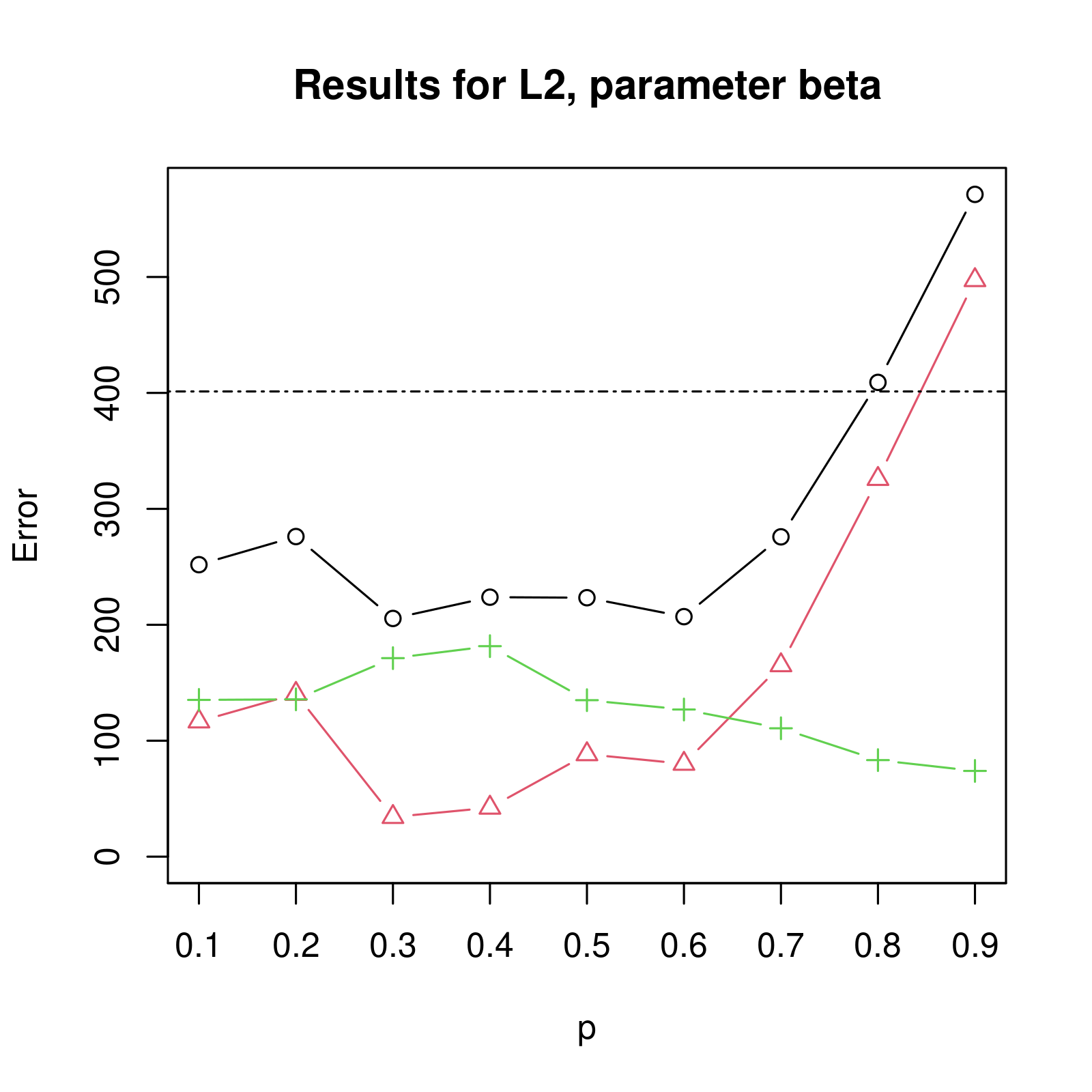}
    \includegraphics[width = 0.3\textwidth]{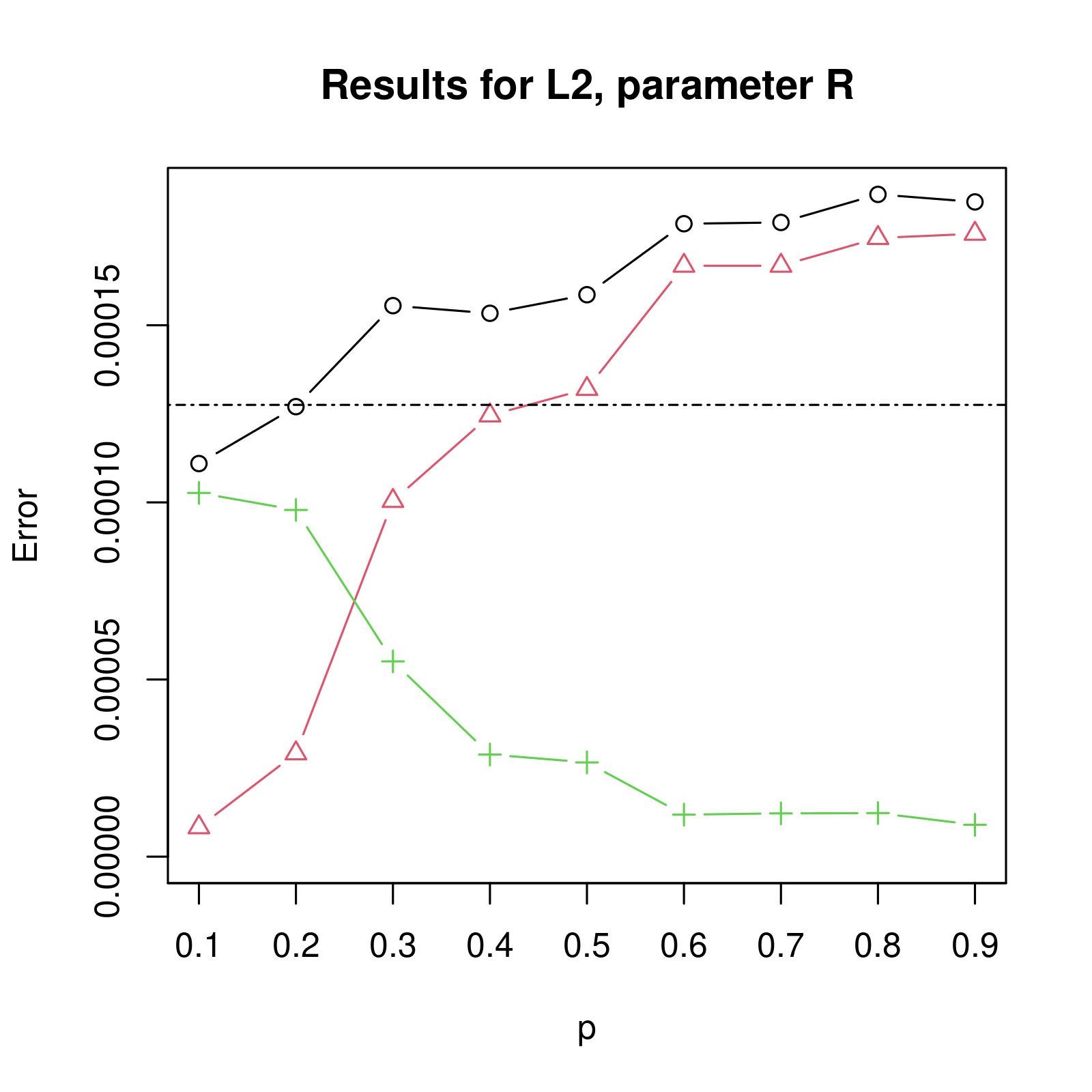}
    \includegraphics[width = 0.3\textwidth]{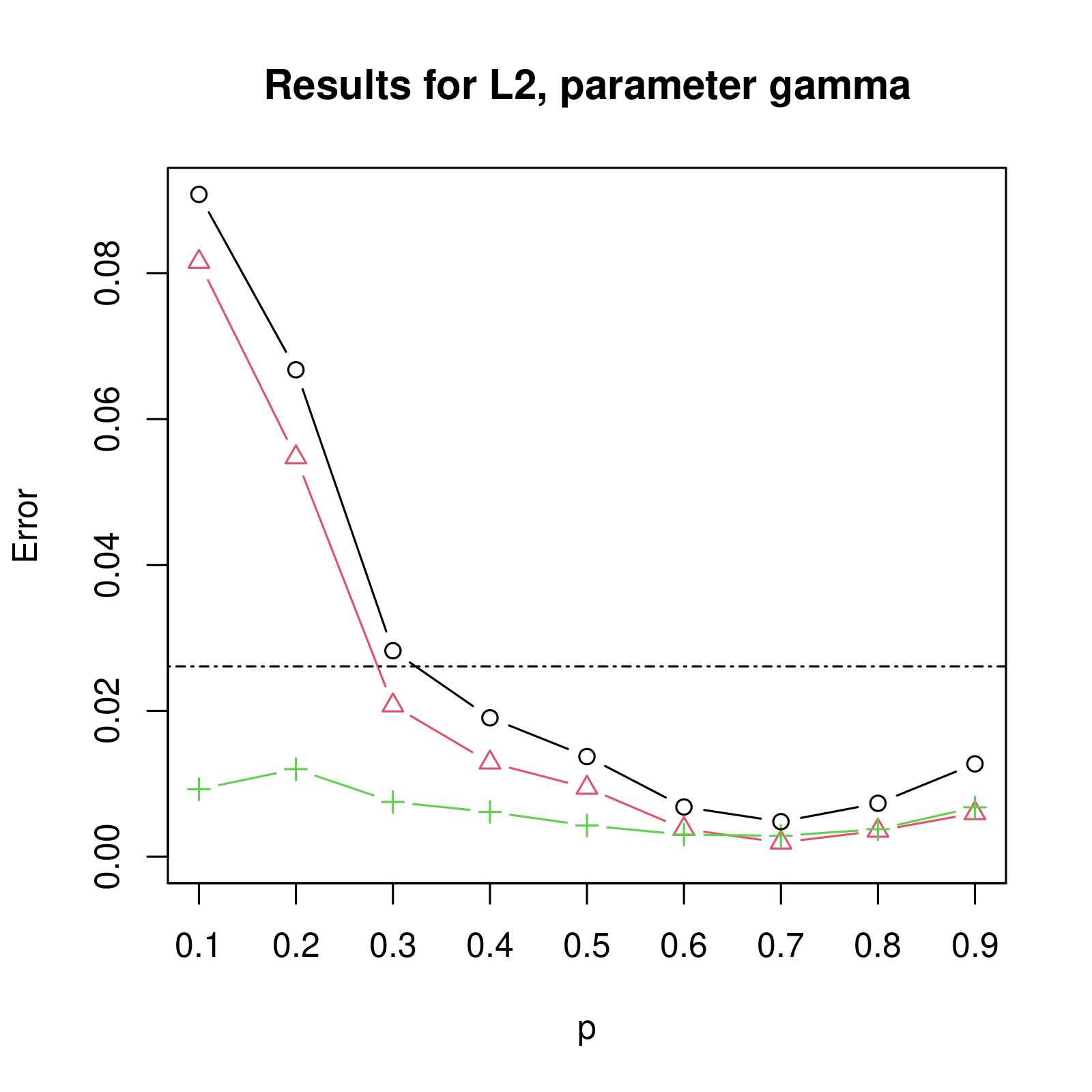}
    \includegraphics[width = 0.3\textwidth]{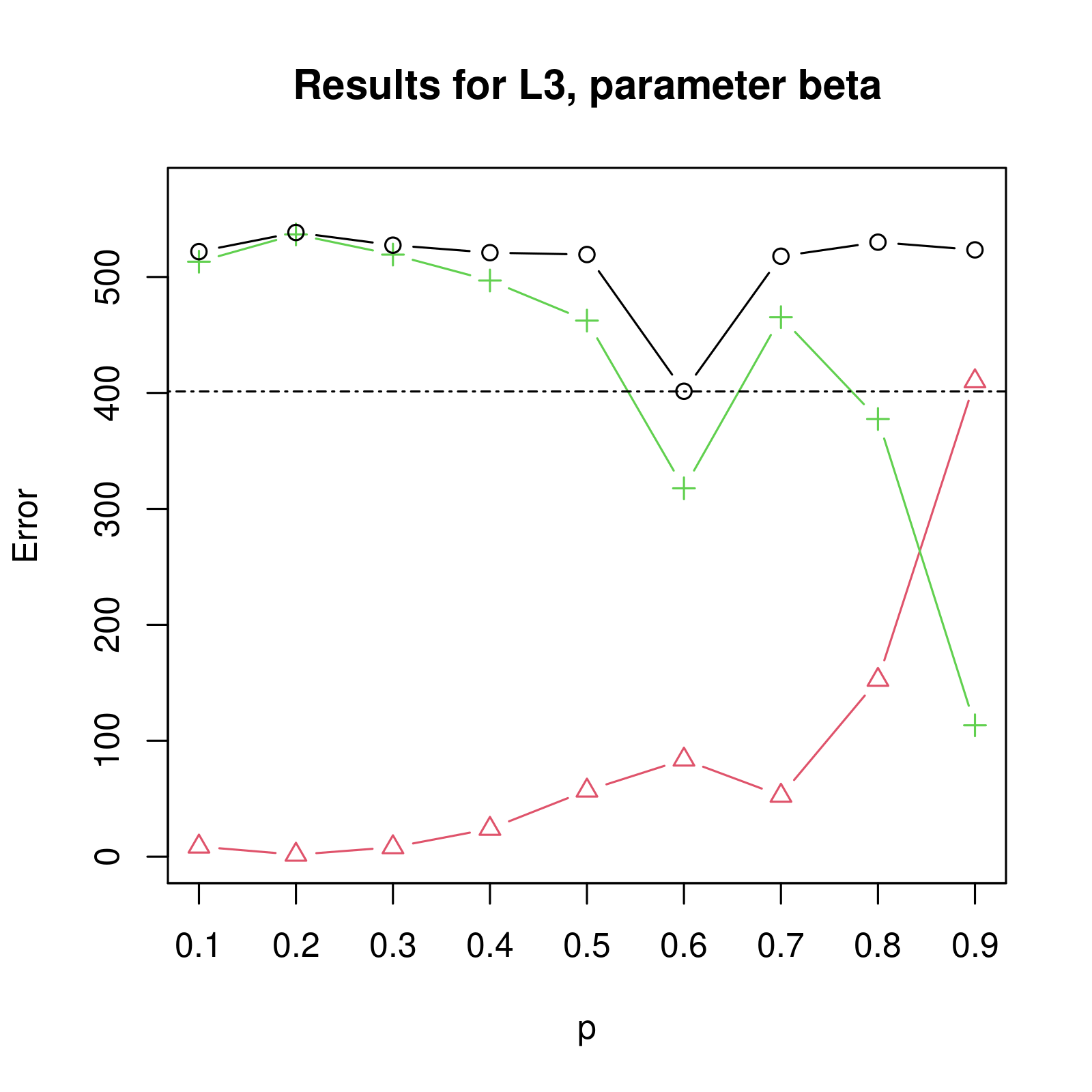}
    \includegraphics[width = 0.3\textwidth]{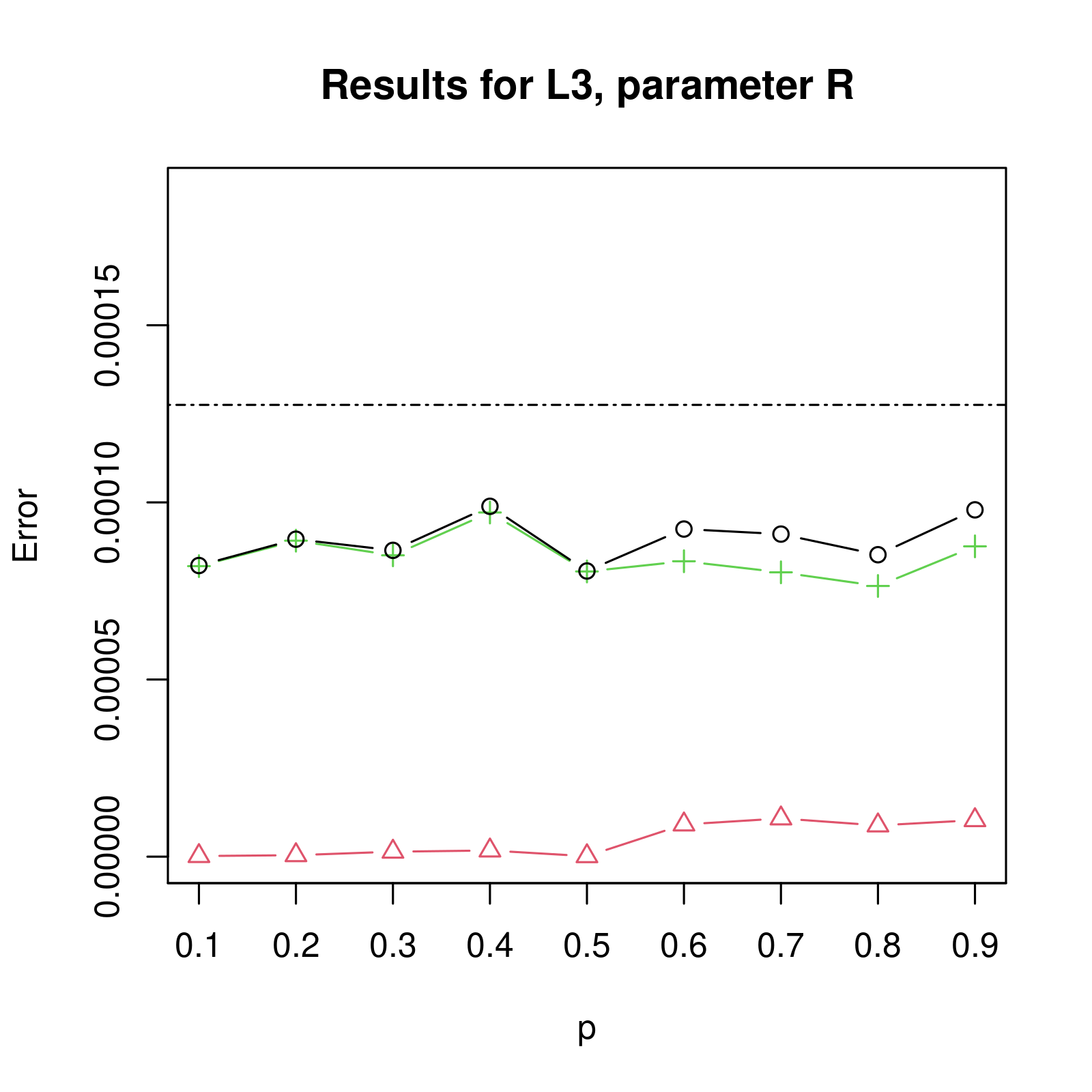}
    \includegraphics[width = 0.3\textwidth]{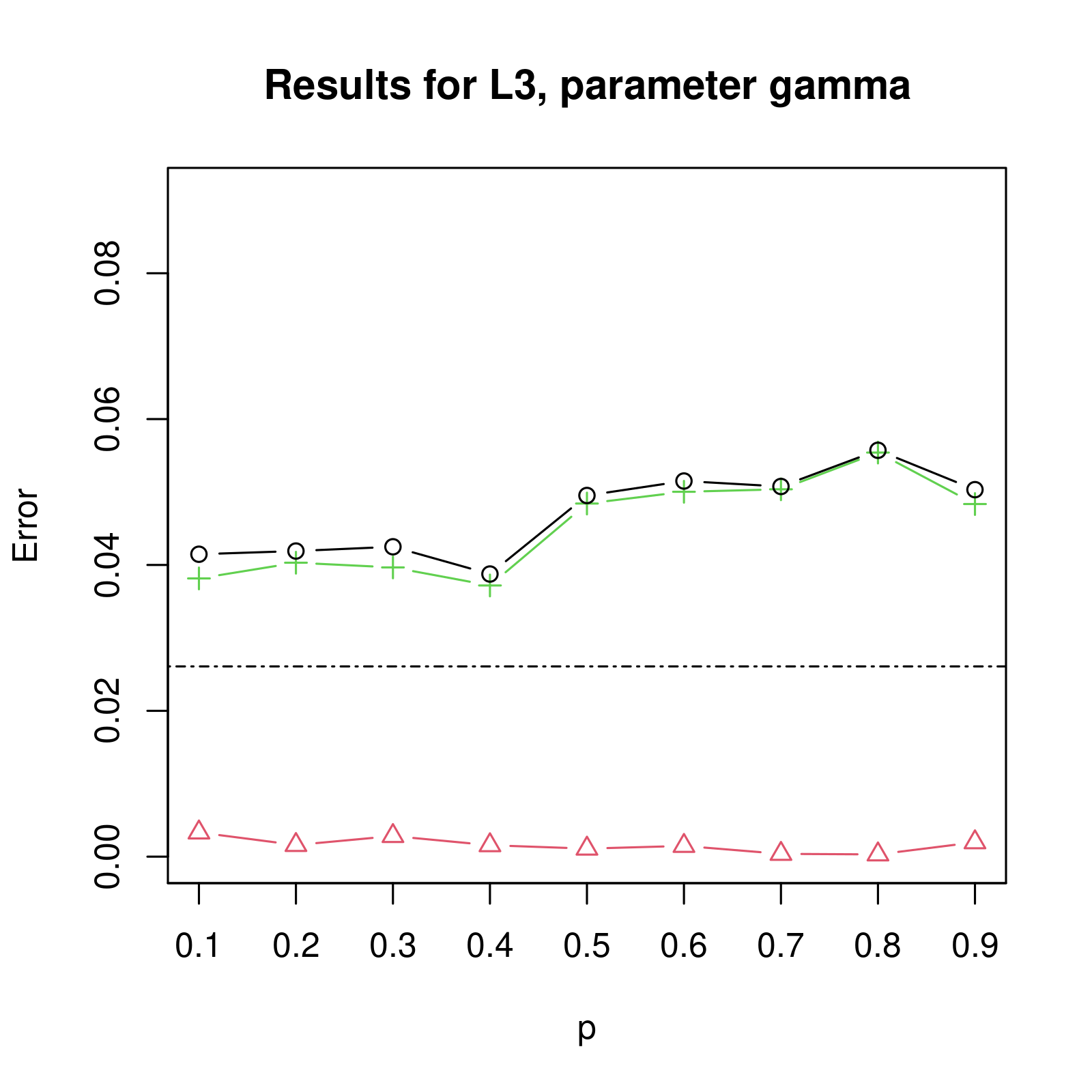}
    \caption{MSE, squared bias and variance for the hard-core model using PPL with the $\Loss_2$ and $\Loss_3$ loss functions, when estimating the parameters 
$\beta$, $R$ and $\gamma$. 
Here $k = 100$, $N = 100$, $p = 0.1,0.2,\ldots,0.9$ and the PPL-weight is set to $p$. 
    The black lines with circles correspond to MSE, the red lines with triangles correspond to squared bias and the green lines with plus signs correspond to variance. The black dotted lines correspond to the Takacs-Fiksel estimates.
}
    \label{fig:strauss-p-L2L3}
\end{figure}

\begin{figure}[!htb]
    \centering
    \includegraphics[width = 0.3\textwidth]{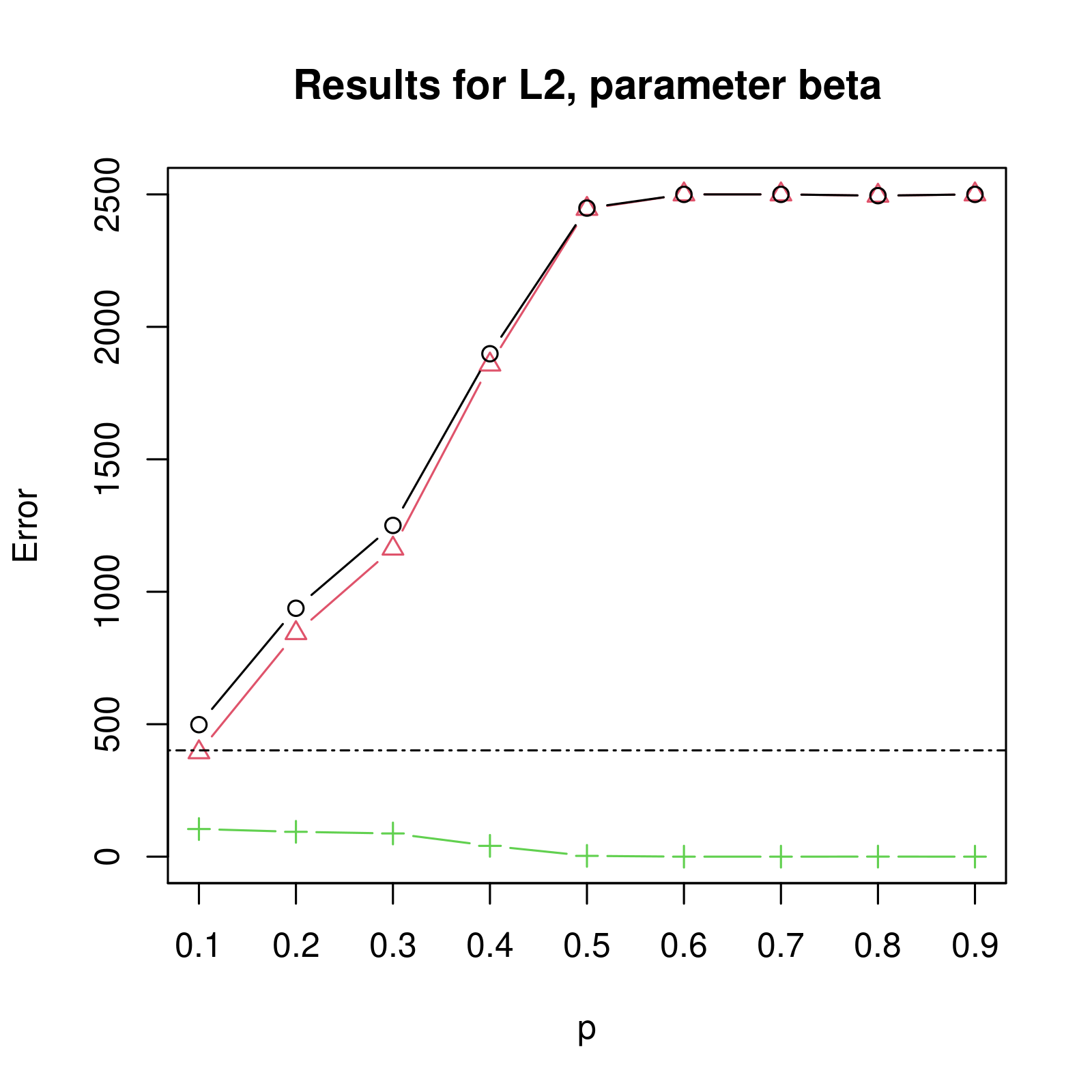}
    \includegraphics[width = 0.3\textwidth]{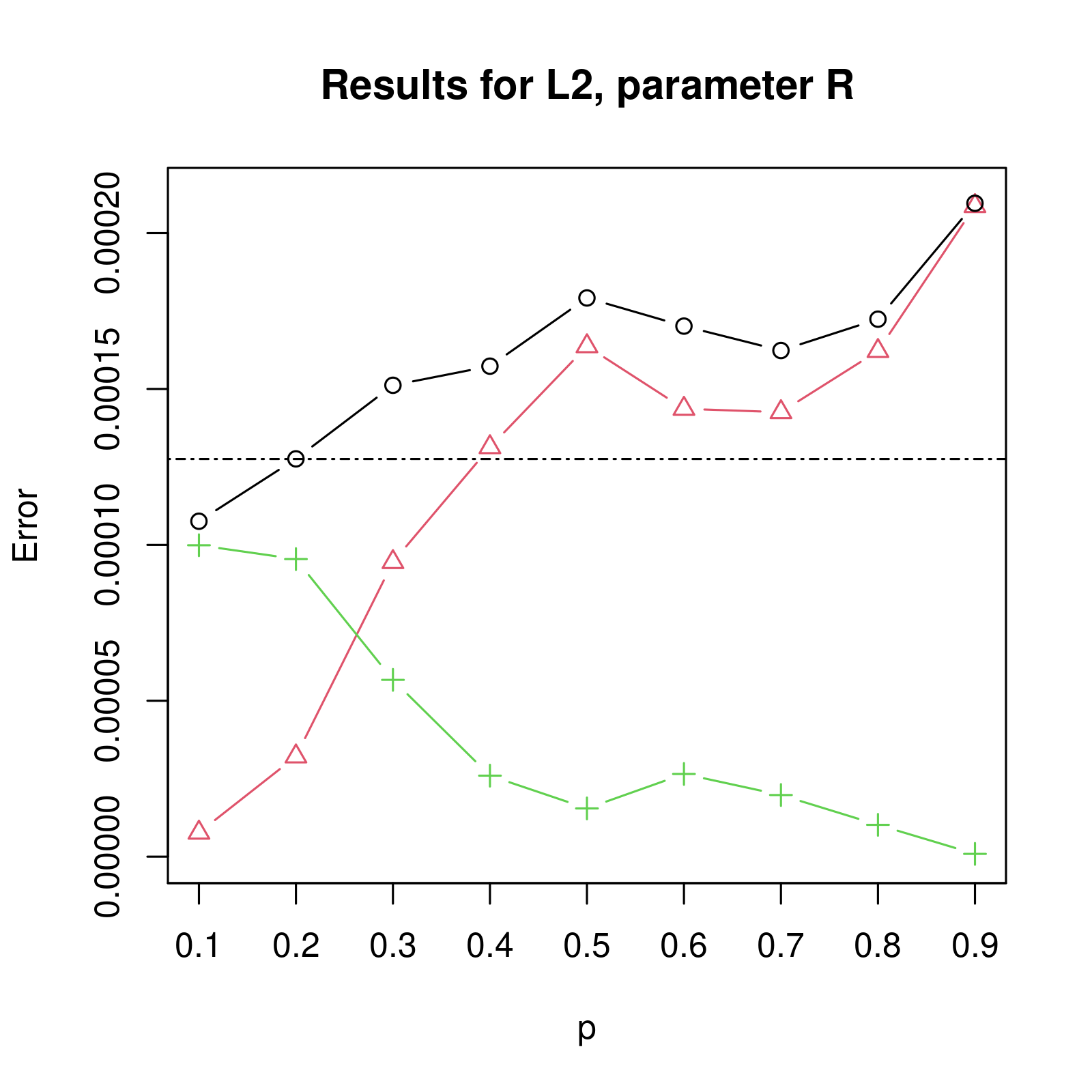}
    \includegraphics[width = 0.3\textwidth]{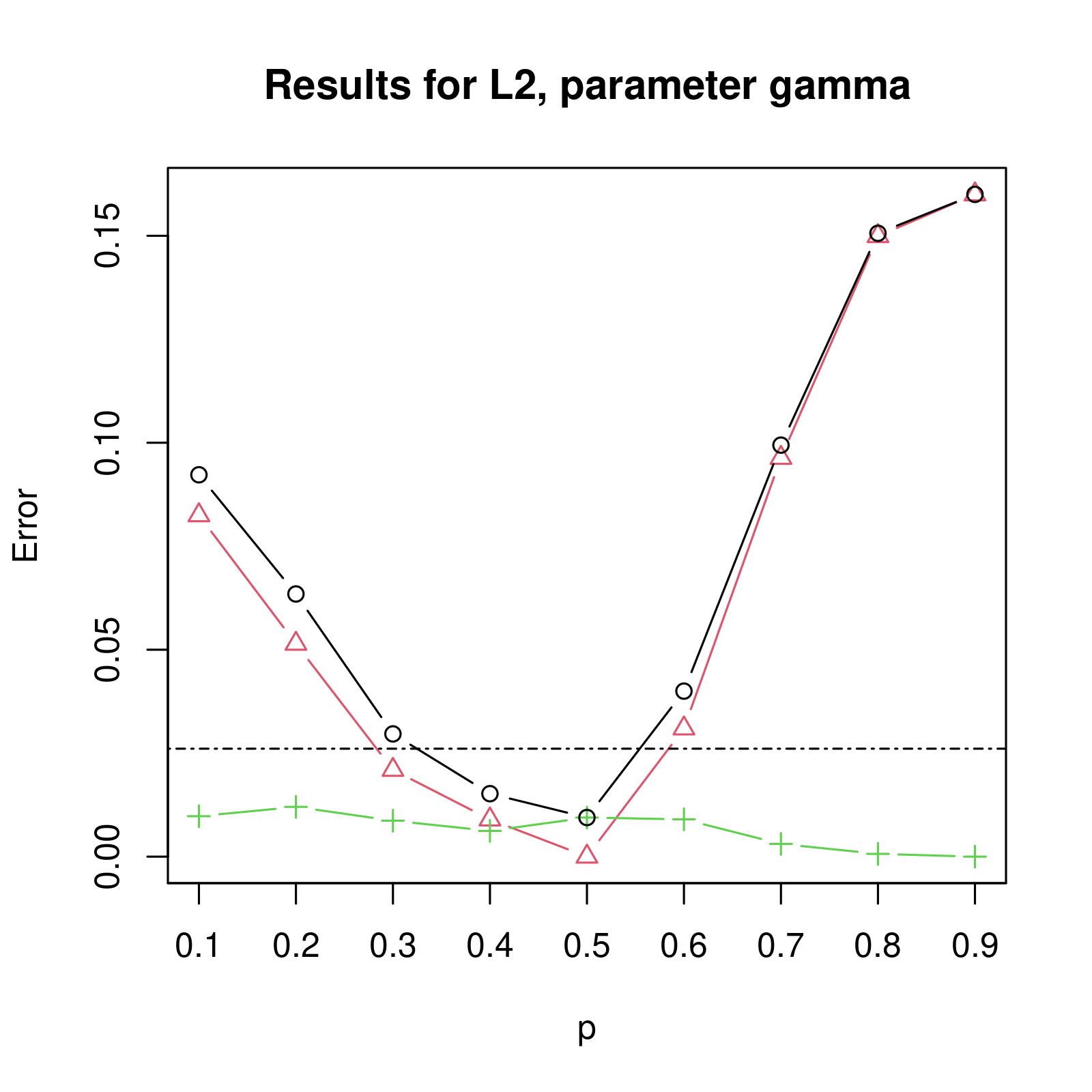}
    \includegraphics[width = 0.3\textwidth]{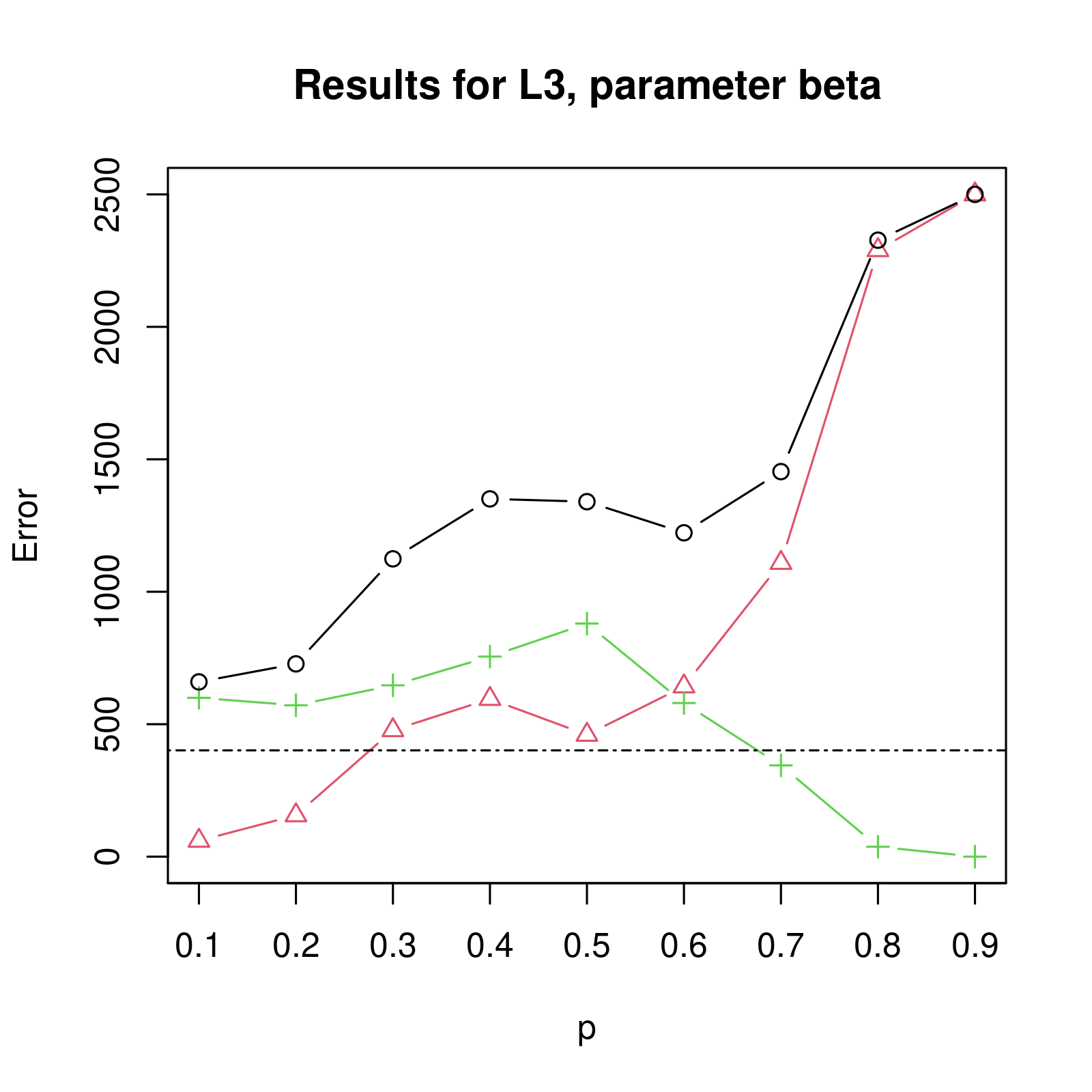}
    \includegraphics[width = 0.3\textwidth]{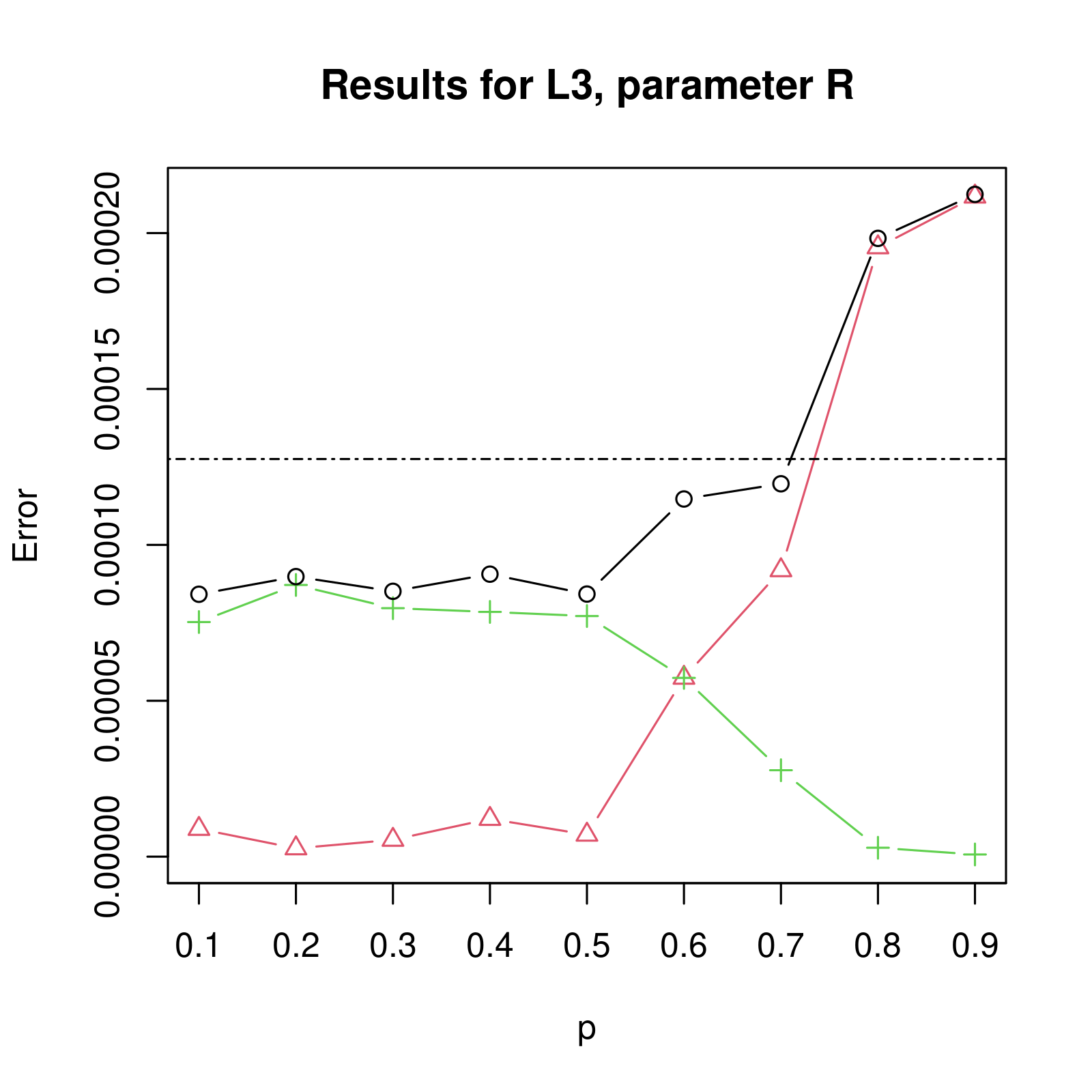}
    \includegraphics[width = 0.3\textwidth]{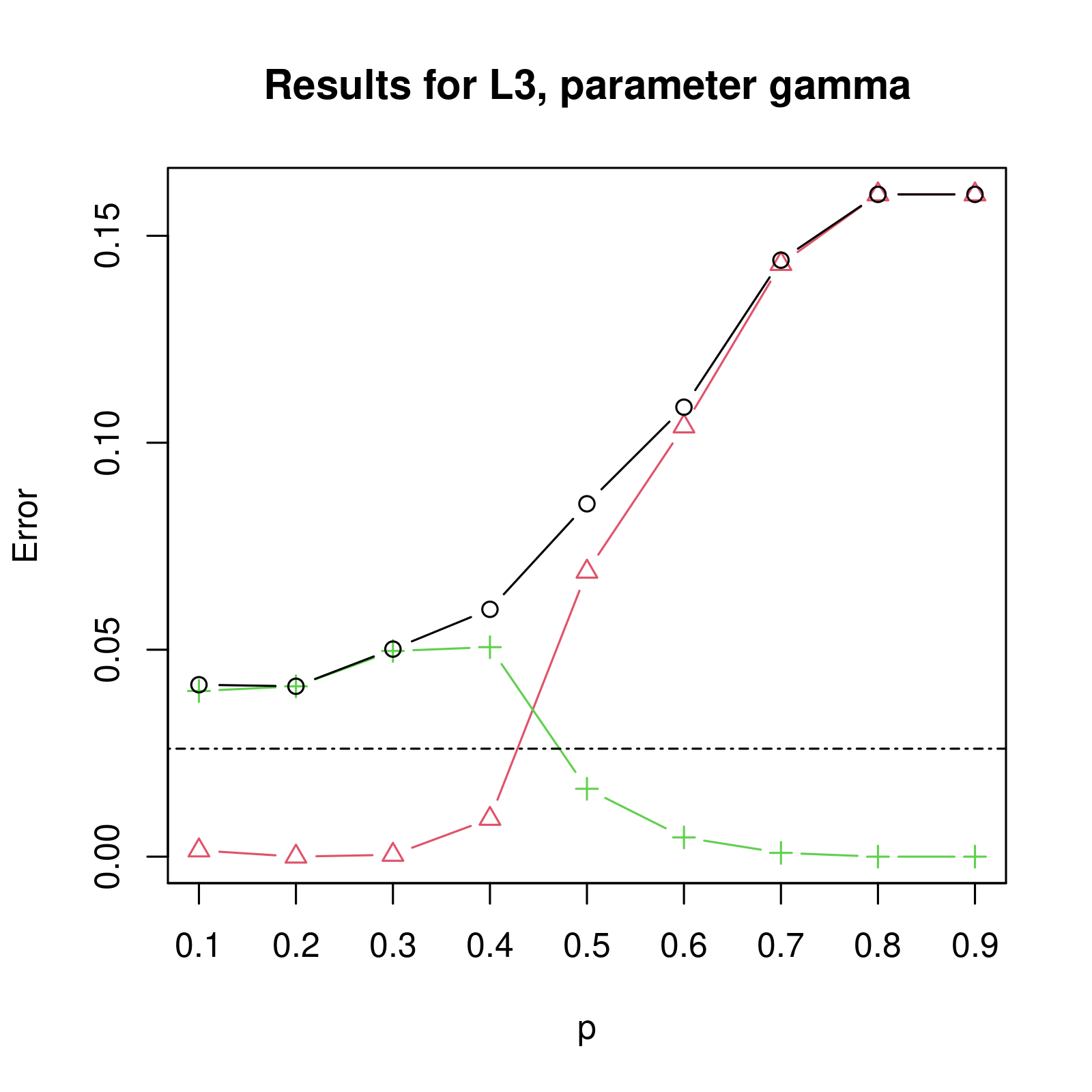}
    \caption{MSE, squared bias and variance for the hard-core model using PPL with the $\Loss_2$ and $\Loss_3$ loss functions, when estimating the parameters 
    $\beta$, $R$ and $\gamma$. 
    Here $k = 100$, $N = 100$, $p = 0.1,0.2,\ldots,0.9$ and the PPL-weight is set to $p/(1-p)$. 
    The black lines with circles correspond to MSE, the red lines with triangles correspond to squared bias and the green lines with plus signs correspond to variance. The black dotted lines correspond to the Takacs-Fiksel estimates.
}
    \label{fig:strauss-(1-p)-L2L3}
\end{figure}

\begin{figure}[!htb]
    \centering
    \includegraphics[width = 0.3\textwidth]{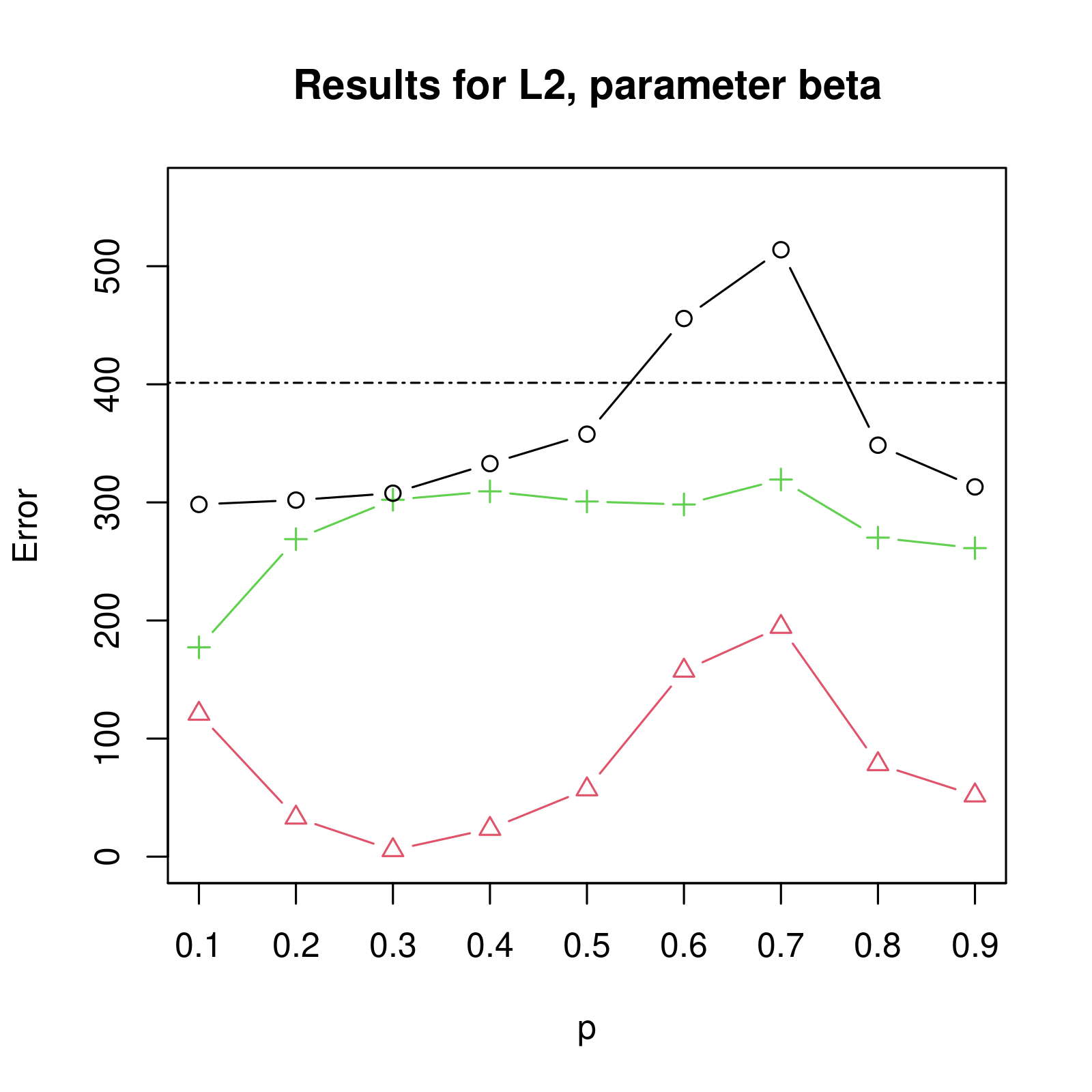}
    \includegraphics[width = 0.3\textwidth]{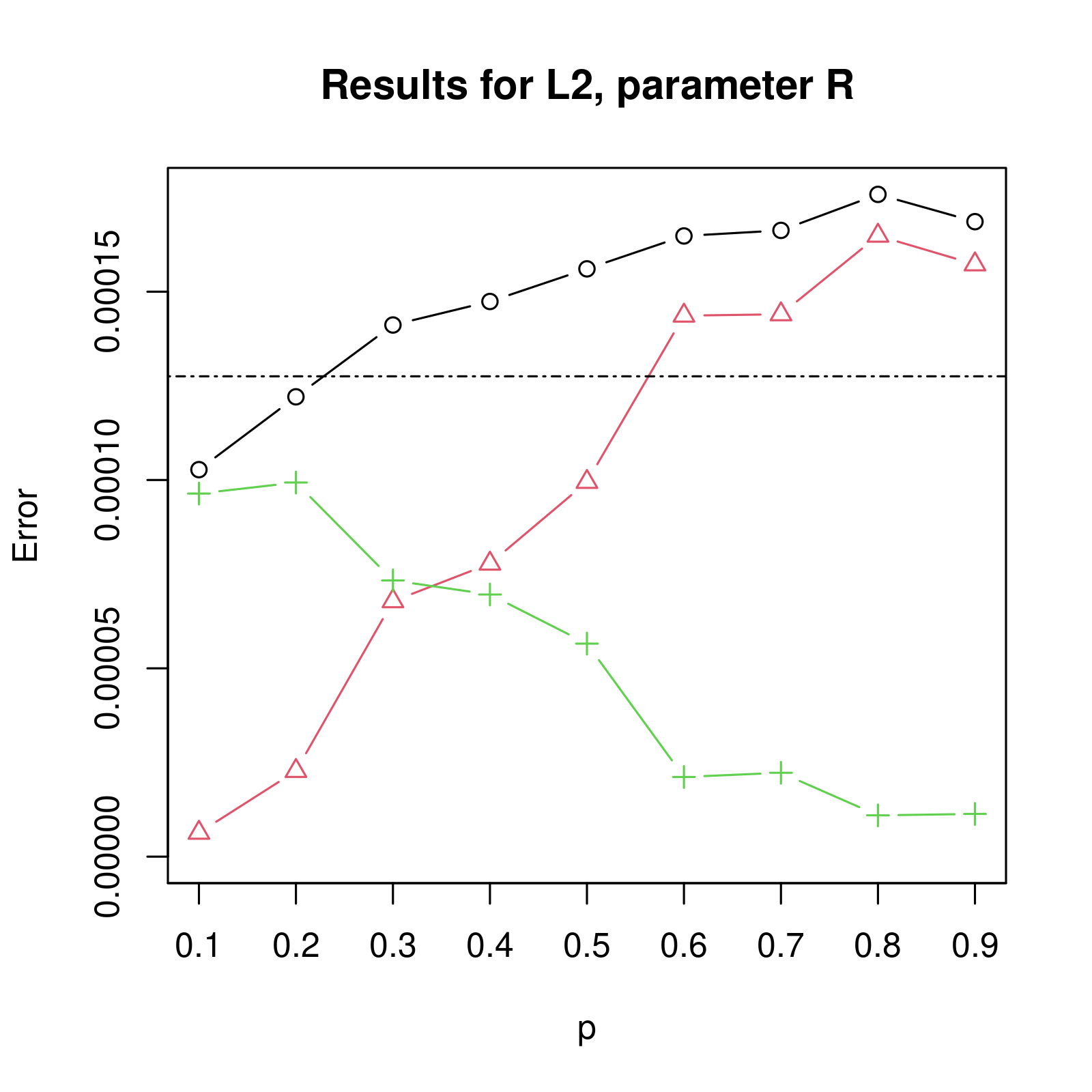}
    \includegraphics[width = 0.3\textwidth]{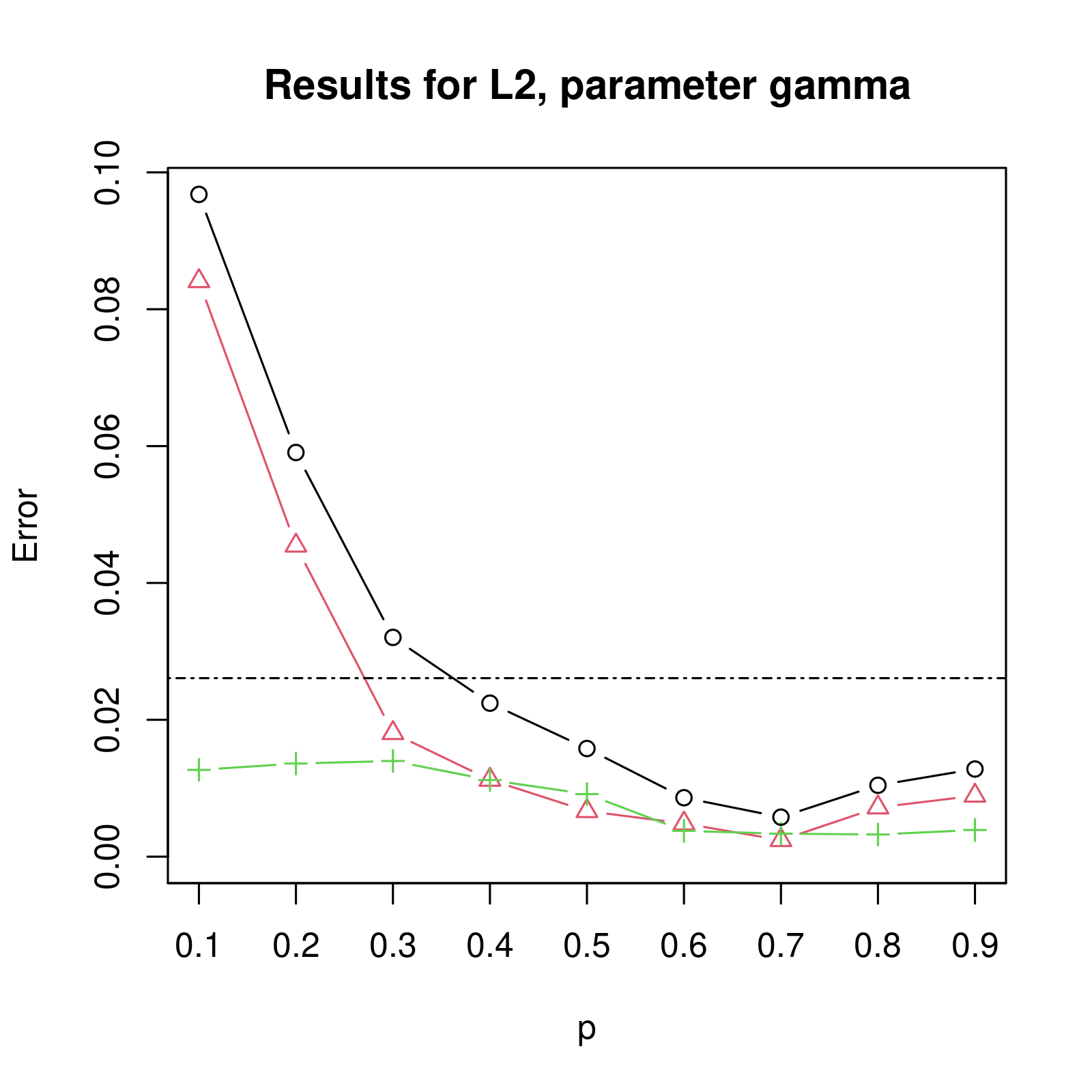}
    \includegraphics[width = 0.3\textwidth]{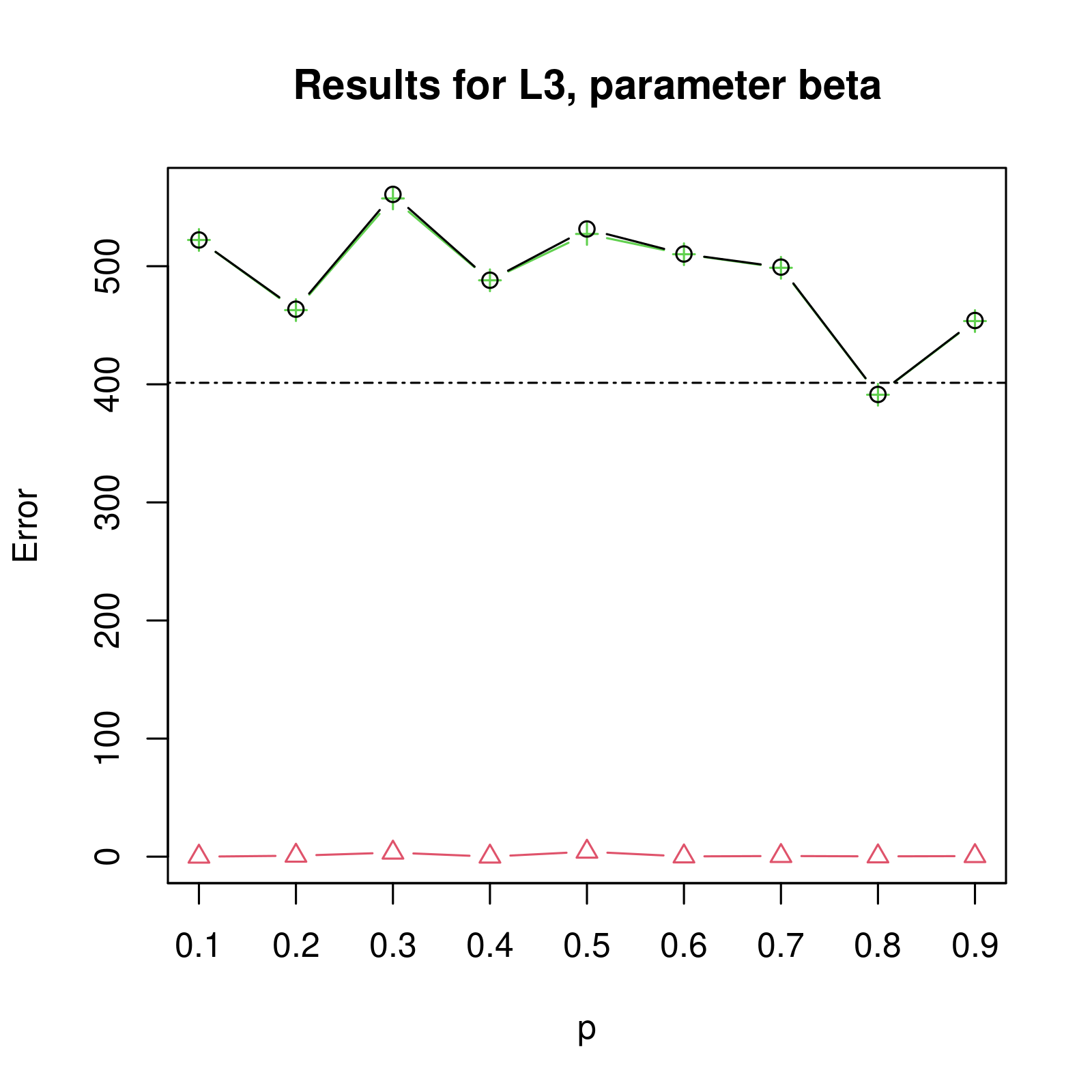}
    \includegraphics[width = 0.3\textwidth]{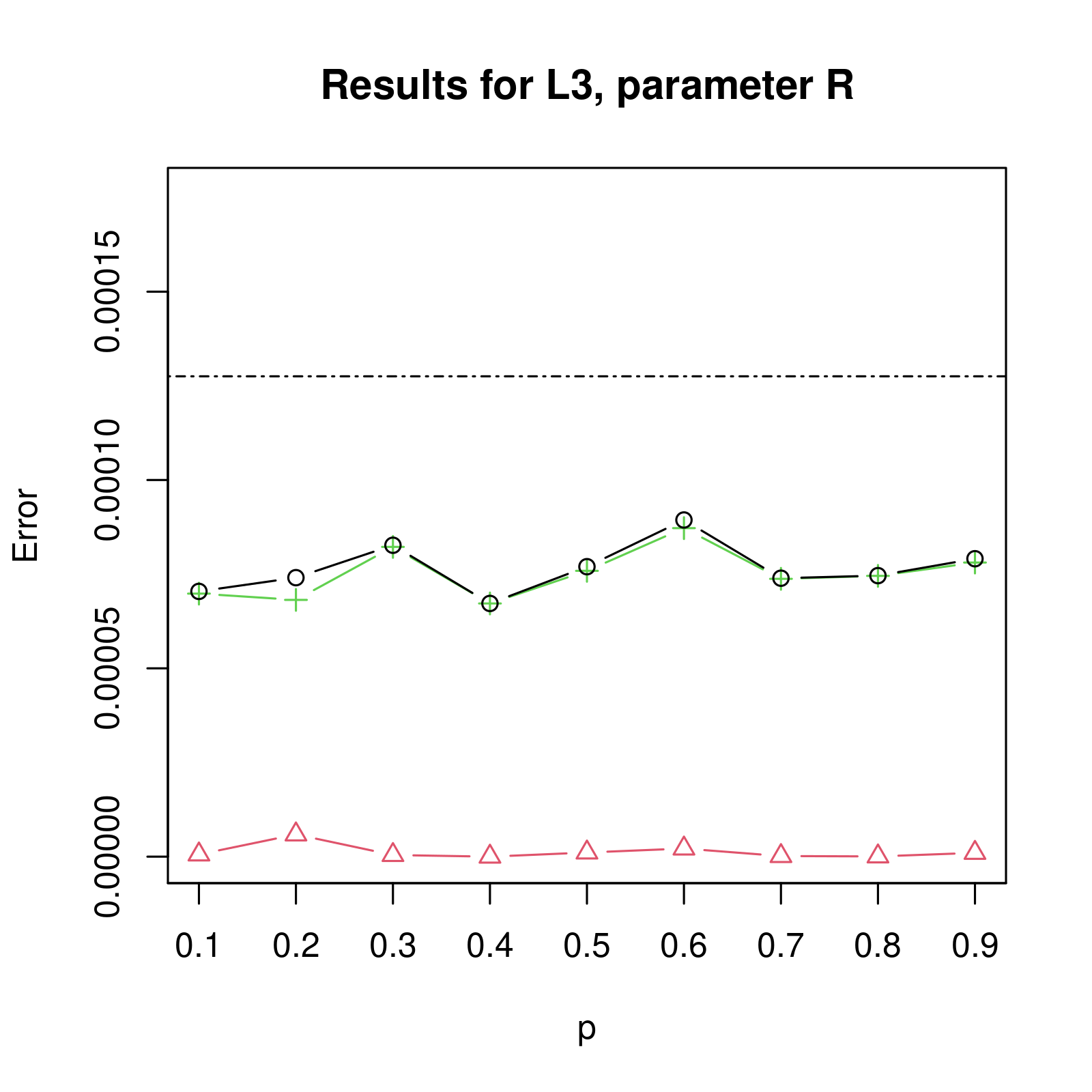}
    \includegraphics[width = 0.3\textwidth]{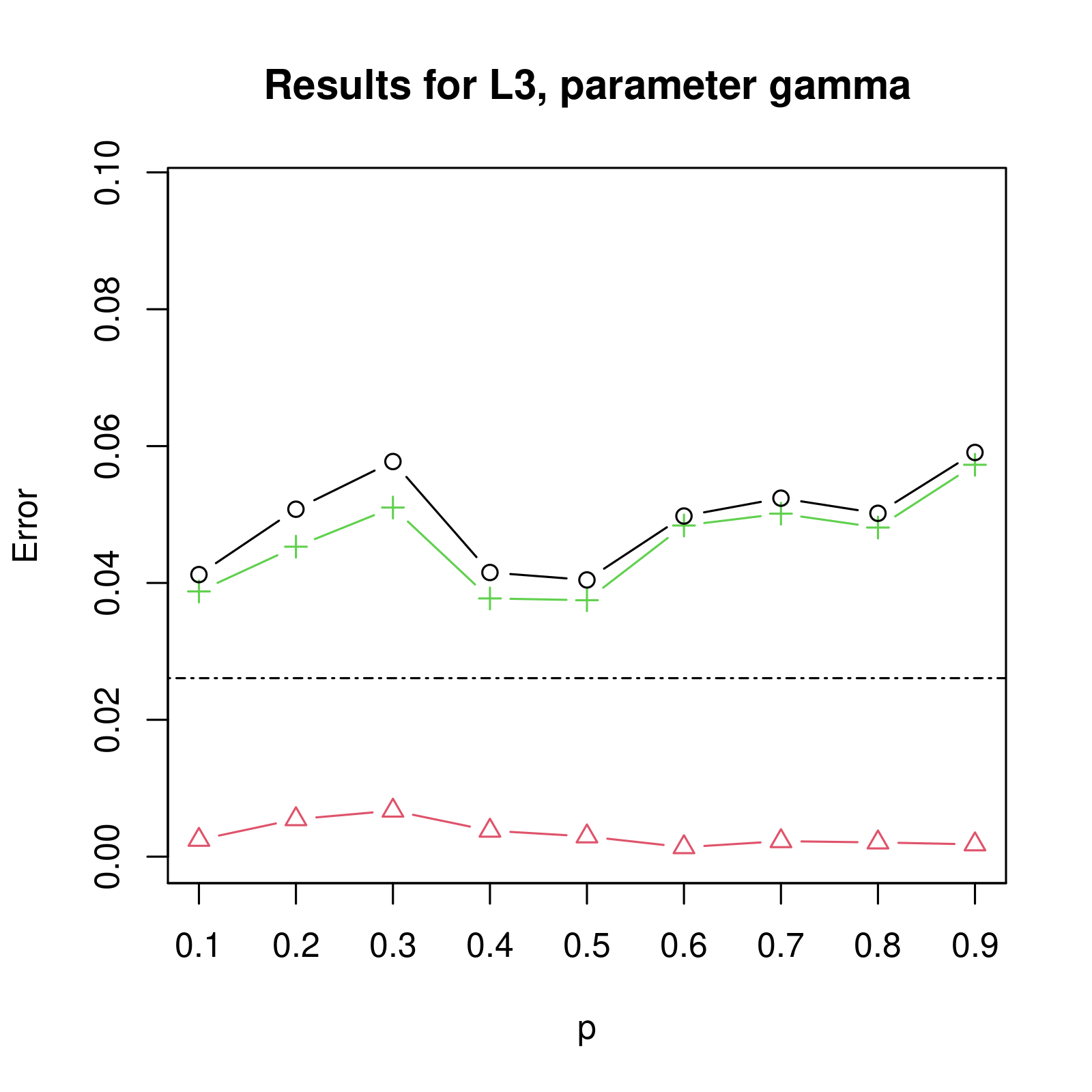}
    \caption{MSE, squared bias and variance for the hard-core model using PPL with the $\Loss_1$ loss function, when estimating the parameters 
    $\beta$, $R$ and $\gamma$. 
    Here $k = 100$, $N = 50$, $p = 0.1,0.2,\ldots,0.9$ and the PPL-weight is estimated in accordance with \eqref{e:WeightEst}. 
    The black lines with circles correspond to MSE, the red lines with triangles correspond to squared bias and the green lines with plus signs correspond to variance. The black dotted lines correspond to the Takacs-Fiksel estimates.
}
    \label{fig:strauss-est-L2L3}
\end{figure}

\subsection{Geyer saturation process}
\label{sec:app_geyer}
See Section \ref{sec:geyer_sims} for a description of the parameters and grid used for the Geyer saturation process.
In Section \ref{sec:geyer_sims}, in Figure \ref{fig:geyer_p-L3}, the results for the Strauss process for the $\Loss_3$ loss function for the PPL-weight estimate $p$ are shown. 
Here, the corresponding results for the $\Loss_1$ and $\Loss_2$ loss functions can be found in Figure \ref{fig:geyer_p-L1L2}. 
Further, Figure \ref{fig:geyer_p/(1-p)-L3} in Section \ref{sec:geyer_sims} shows the results for the PPL-weight estimate $p/(1-p)$ for the $\Loss_3$ loss function. 
Again, the results for the $\Loss_1$ and $\Loss_2$ loss functions can be found in Figure \ref{fig:geyer_p/(1-p)-L1L2}. 
Lastly, the results using weight estimation according to \eqref{e:WeightEst}, and the $\Loss_3$ loss function can be seen in Figure \ref{fig:geyer_w-L3} in Section \ref{sec:geyer_sims}.
The results for the $\Loss_1$ and $\Loss_2$ loss functions can be found in Figure \ref{fig:geyer_w-L1L2}.

\begin{figure}[!htb]
    \centering
    \includegraphics[width = 0.3\textwidth]{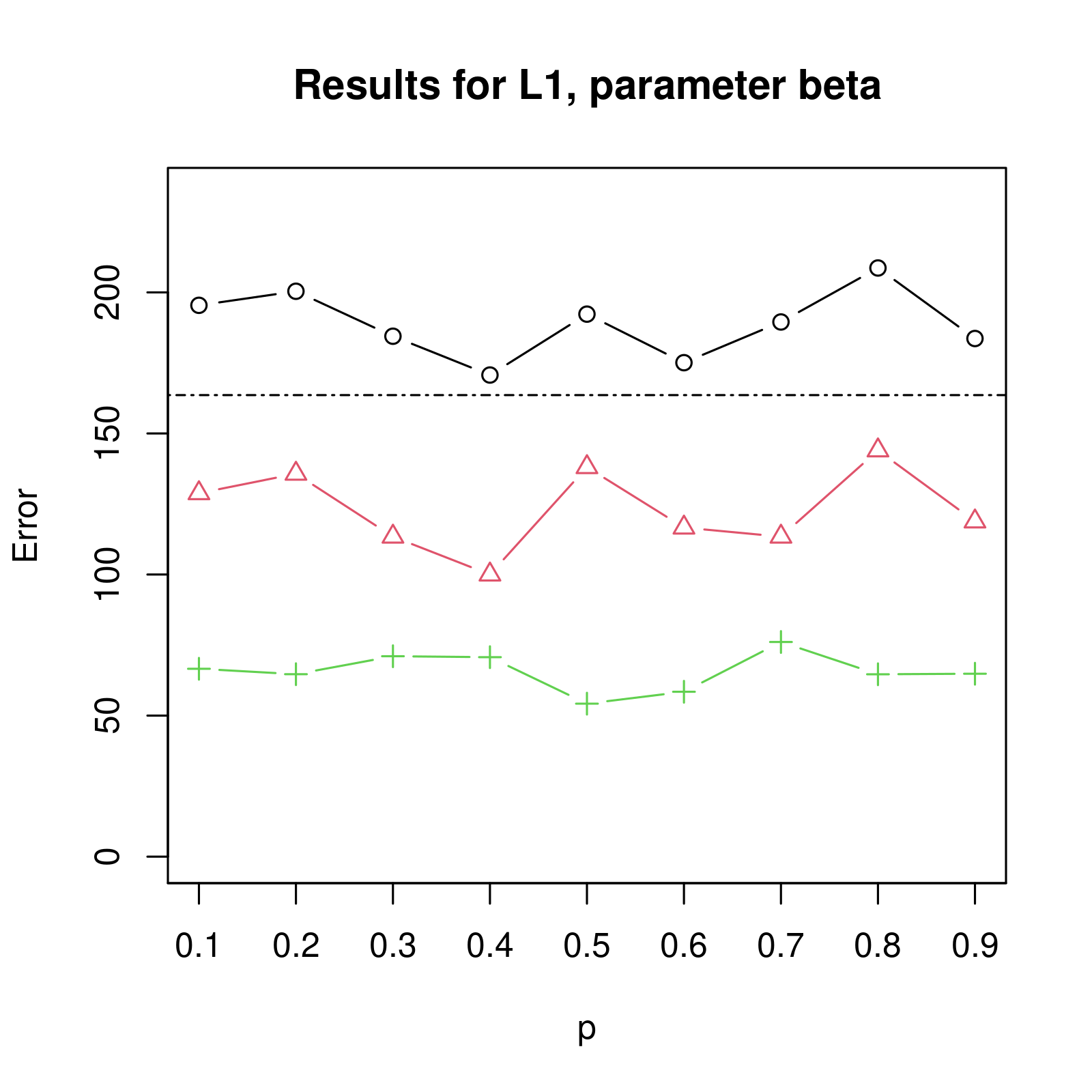}
    \includegraphics[width = 0.3\textwidth]{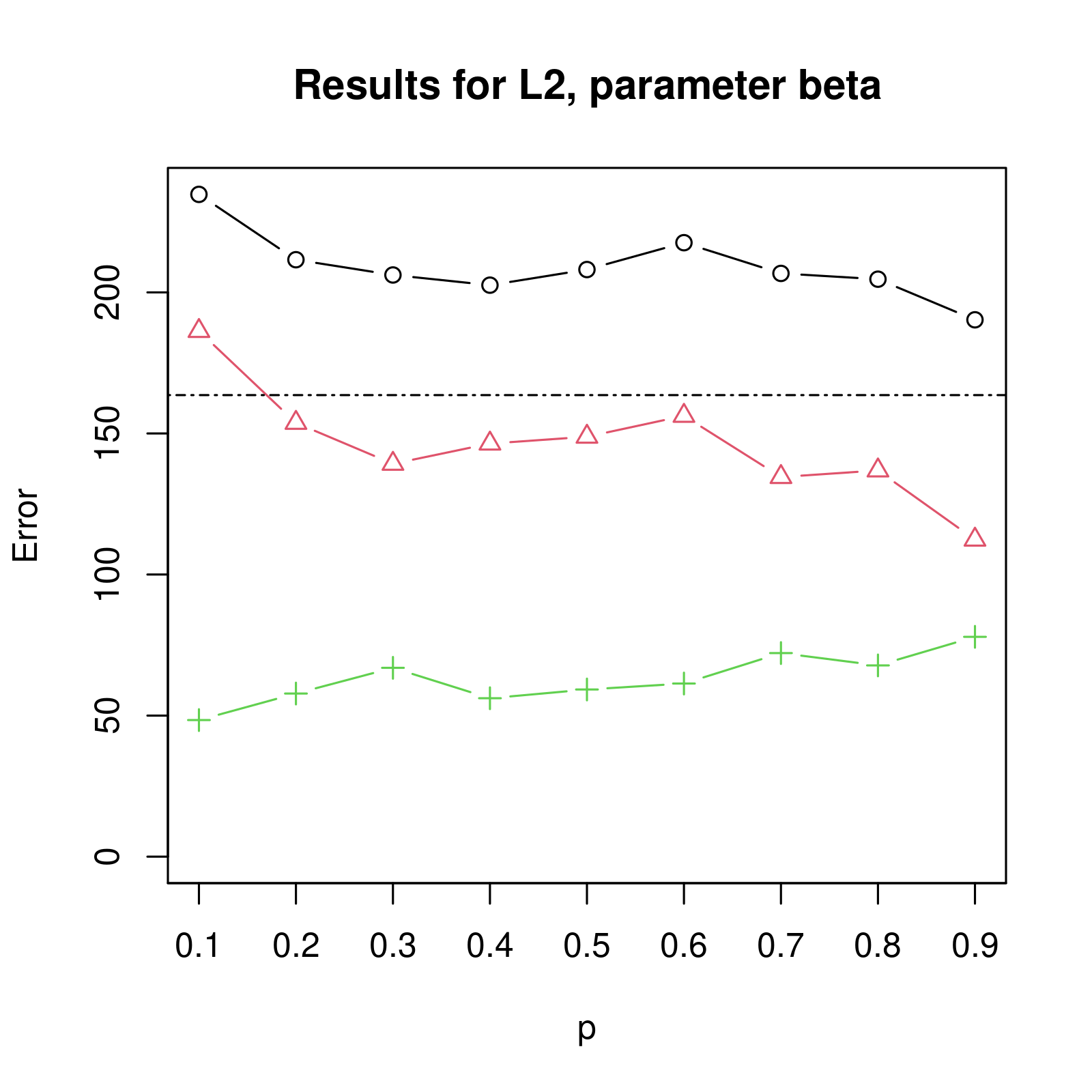}
    
    \includegraphics[width = 0.3\textwidth]{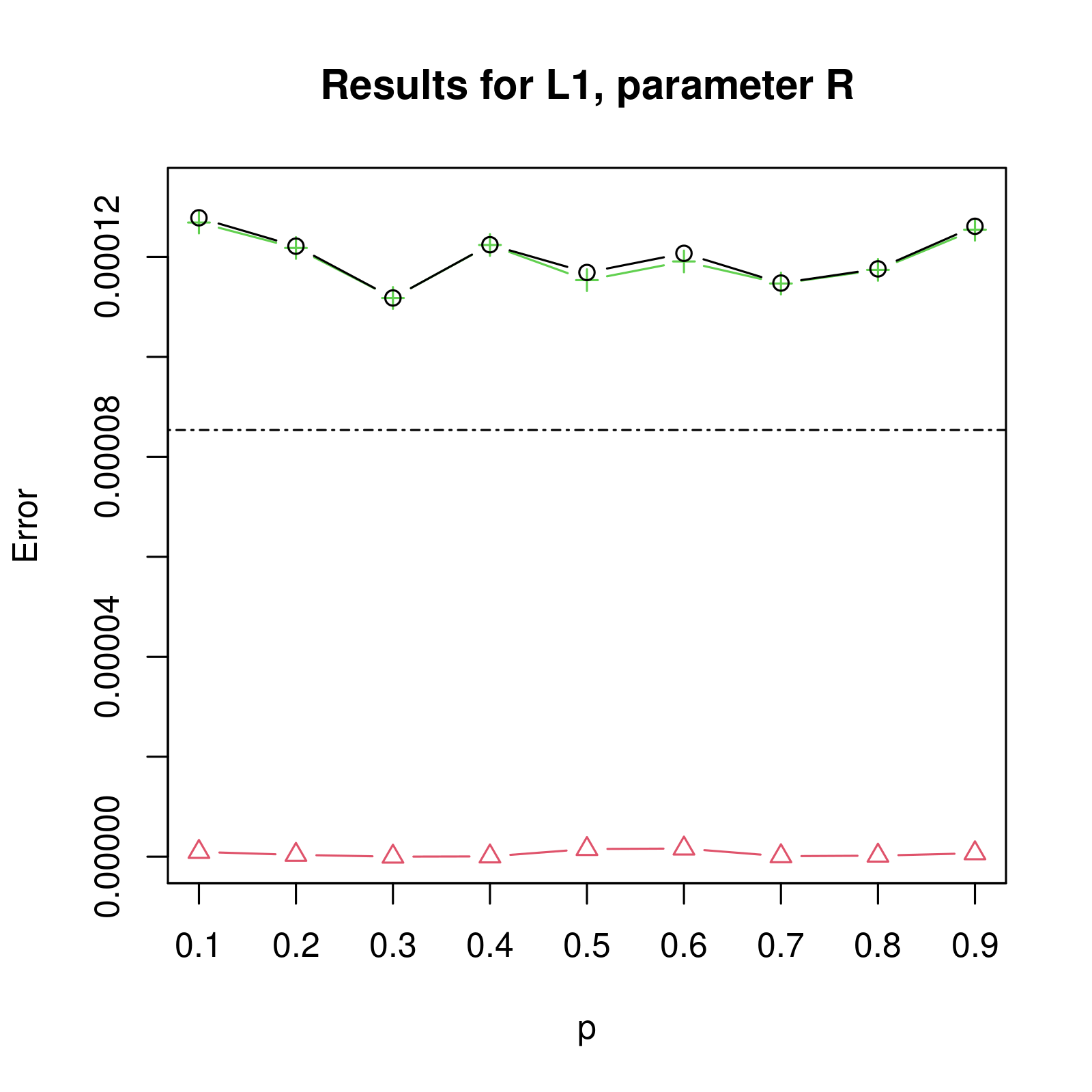}
    \includegraphics[width = 0.3\textwidth]{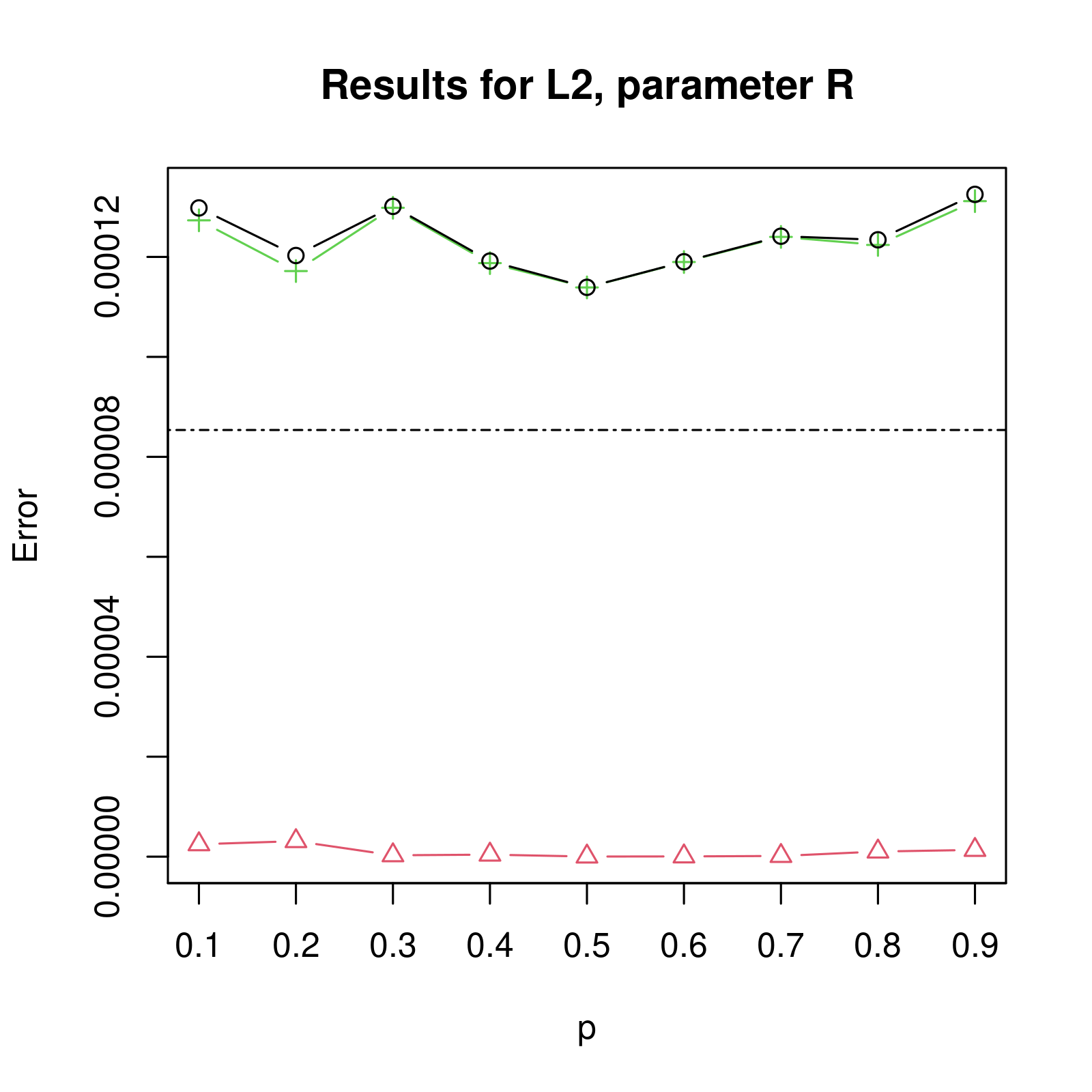}
    
    \includegraphics[width = 0.3\textwidth]{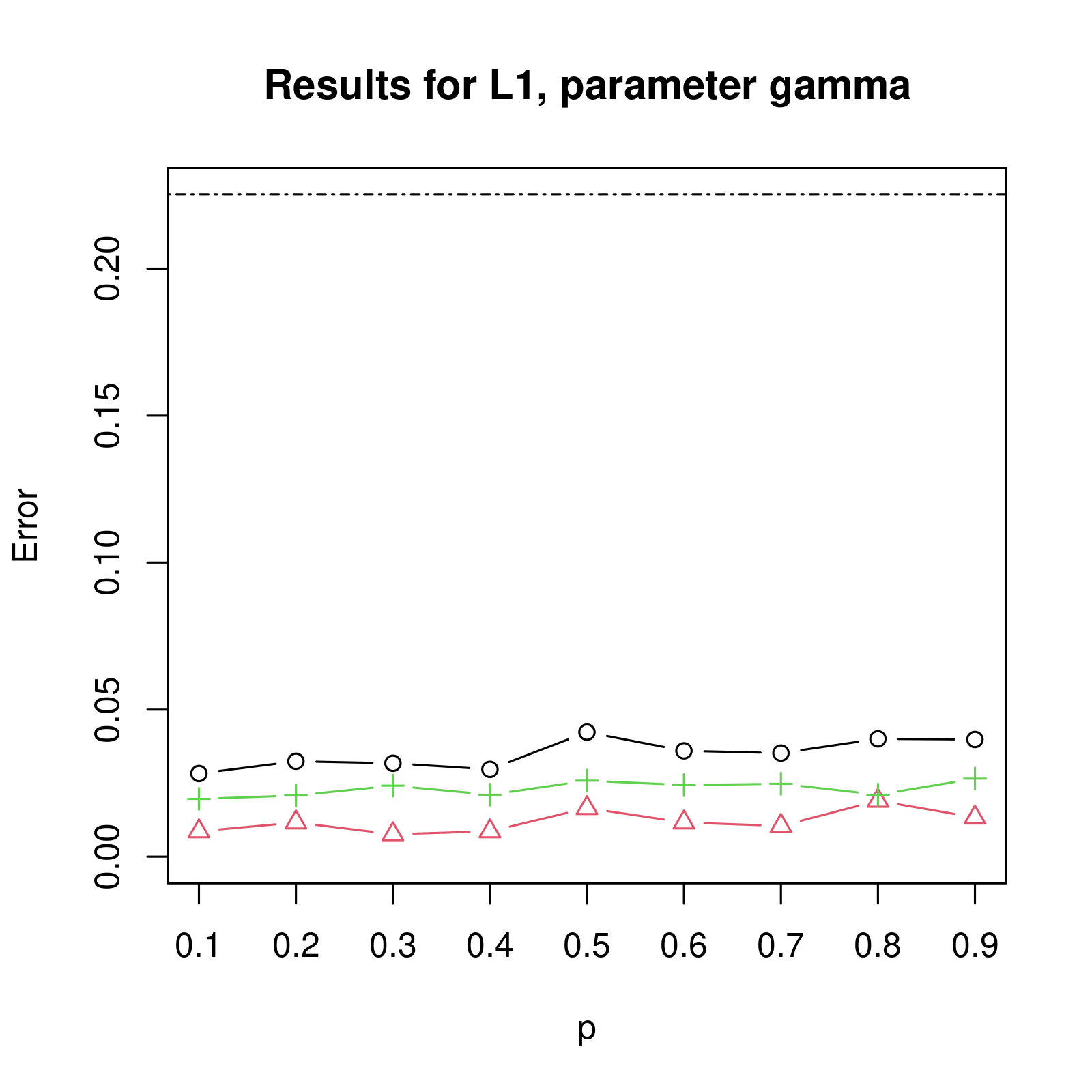}
    \includegraphics[width = 0.3\textwidth]{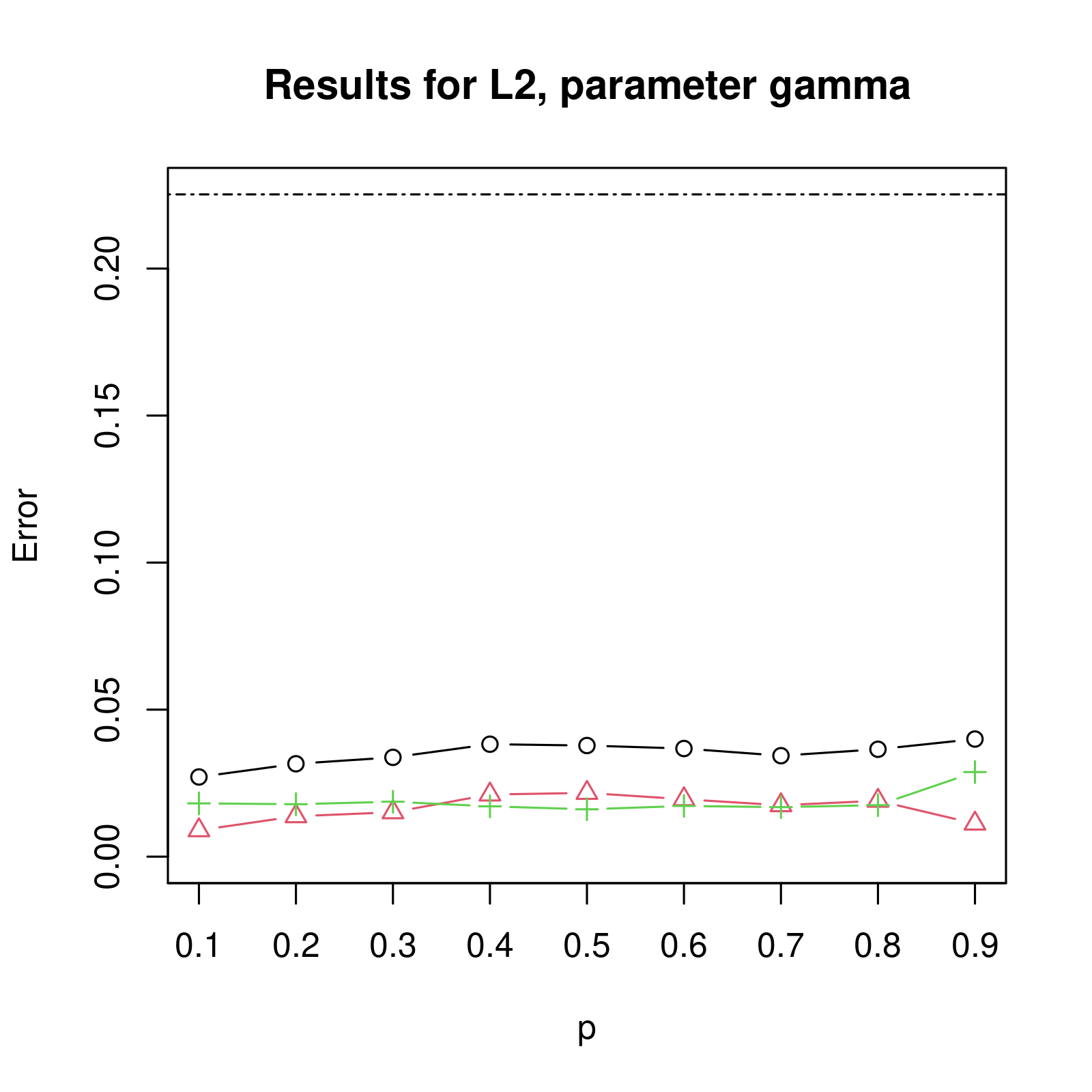}
    
    \includegraphics[width = 0.3\textwidth]{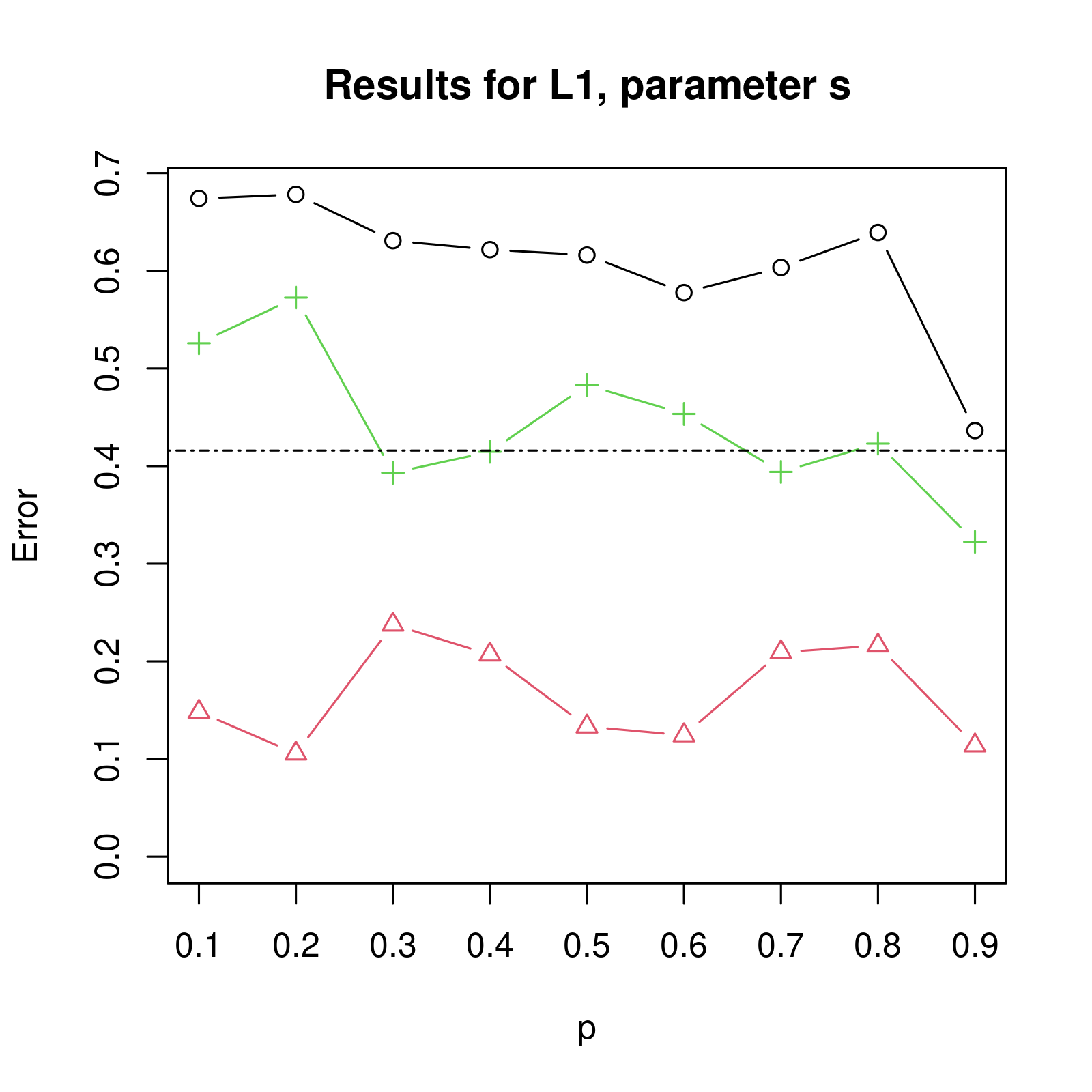}
    \includegraphics[width = 0.3\textwidth]{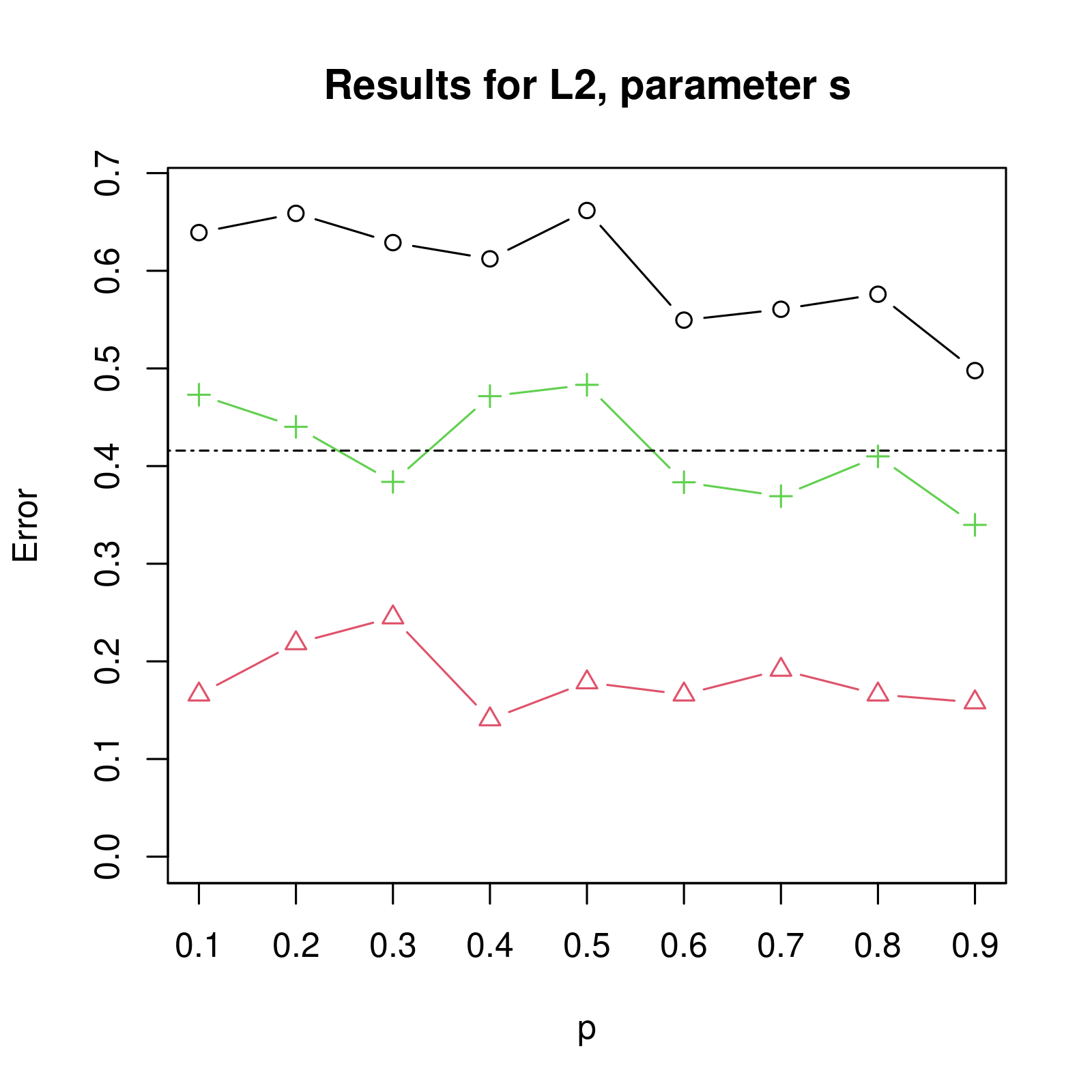}
    \caption{MSE, squared bias and variance for the Geyer saturation model using PPL with the $\Loss_1$ and $\Loss_2$ loss functions, when estimating the parameters 
    $\beta$, $R$, $\gamma$ and $s$. 
    Here $k = 100$, $N = 100$, $p = 0.1,0.2,\ldots,0.9$ and the PPL-weight is set to $p$. 
    The black lines with circles correspond to MSE, the red lines with triangles correspond to squared bias and the green lines with plus signs correspond to variance. The black dotted lines correspond to the Takacs-Fiksel estimates.
}
    \label{fig:geyer_p-L1L2}
\end{figure}

\begin{figure}[!htb]
    \centering
    \includegraphics[width = 0.3\textwidth]{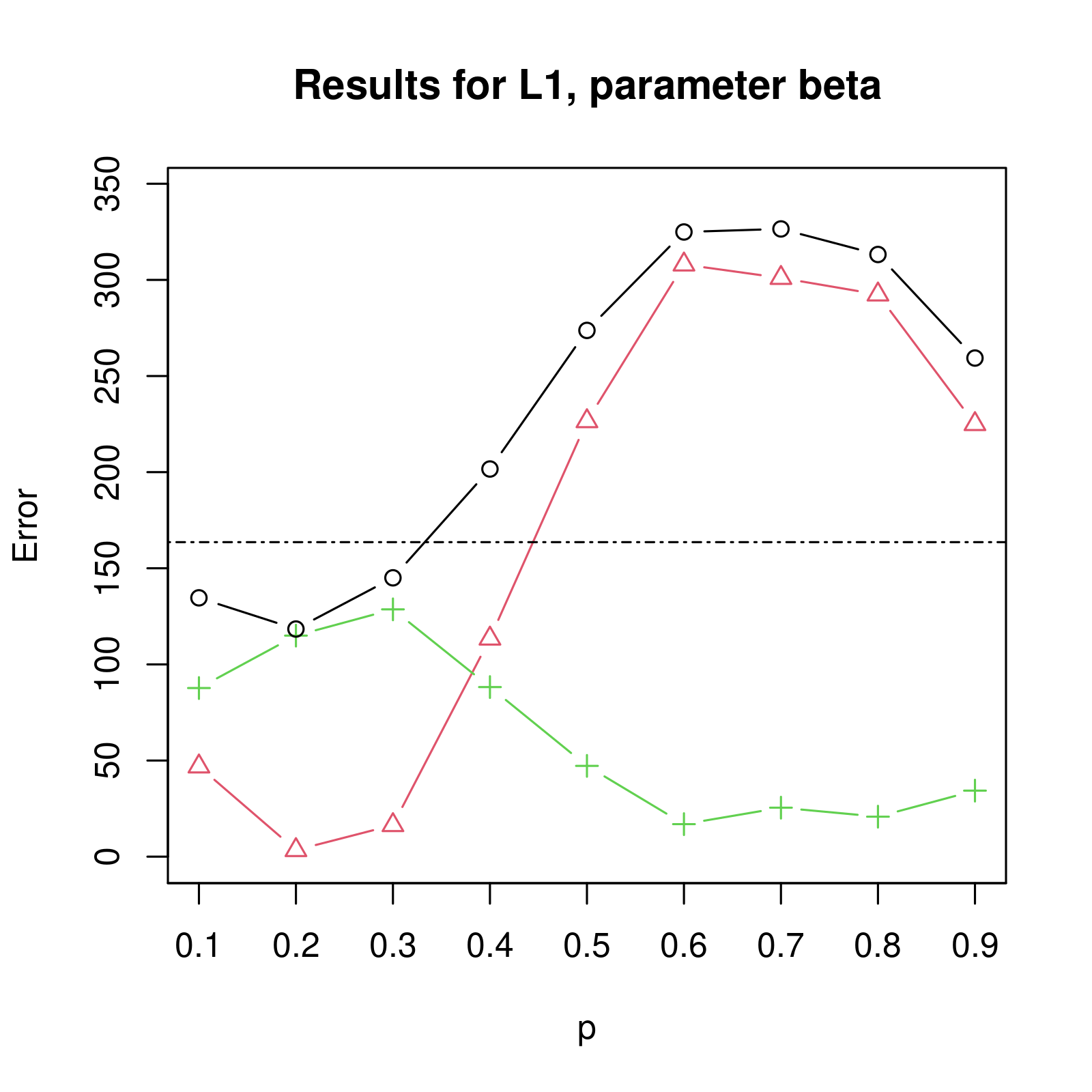}
    \includegraphics[width = 0.3\textwidth]{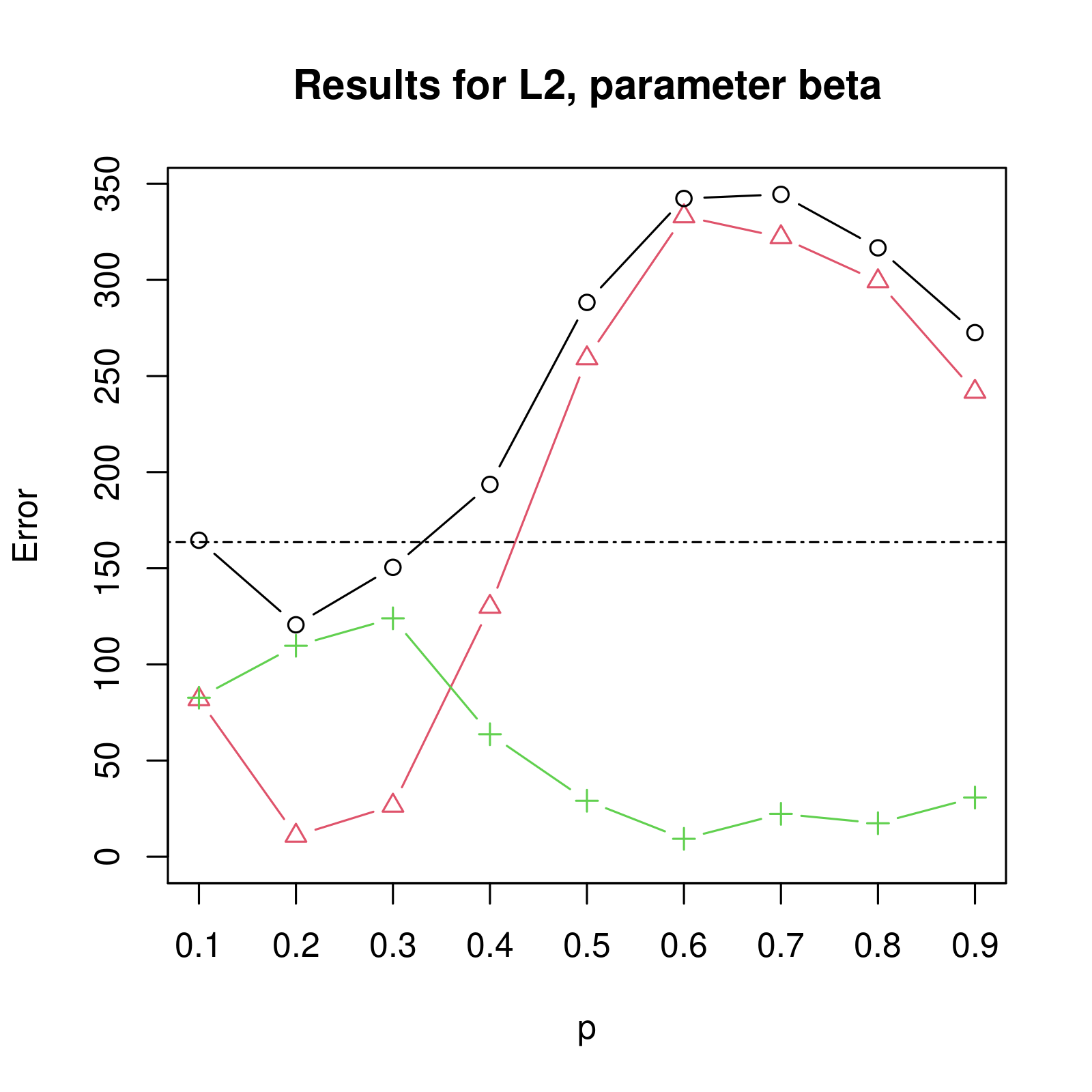}
    
    \includegraphics[width = 0.3\textwidth]{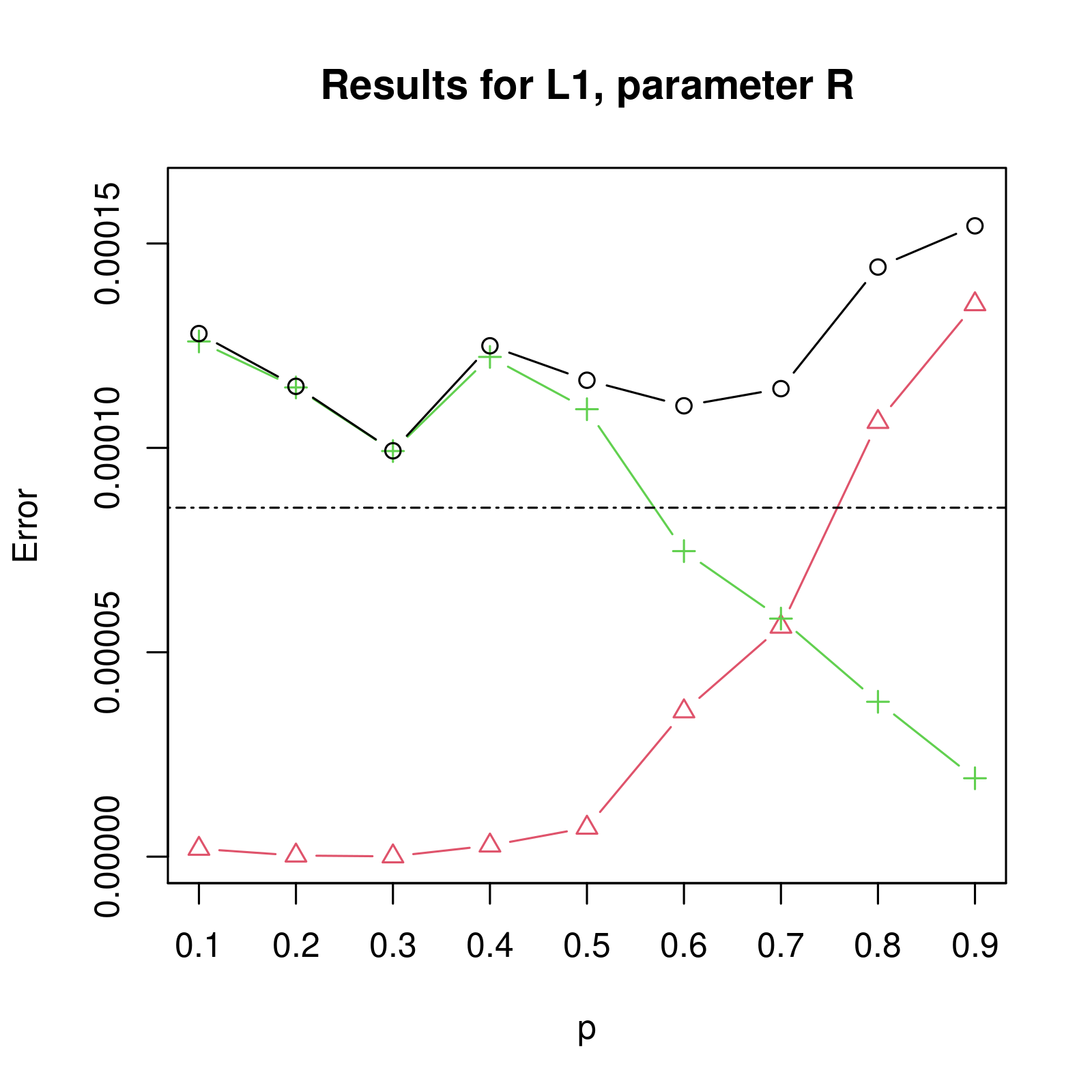}
    \includegraphics[width = 0.3\textwidth]{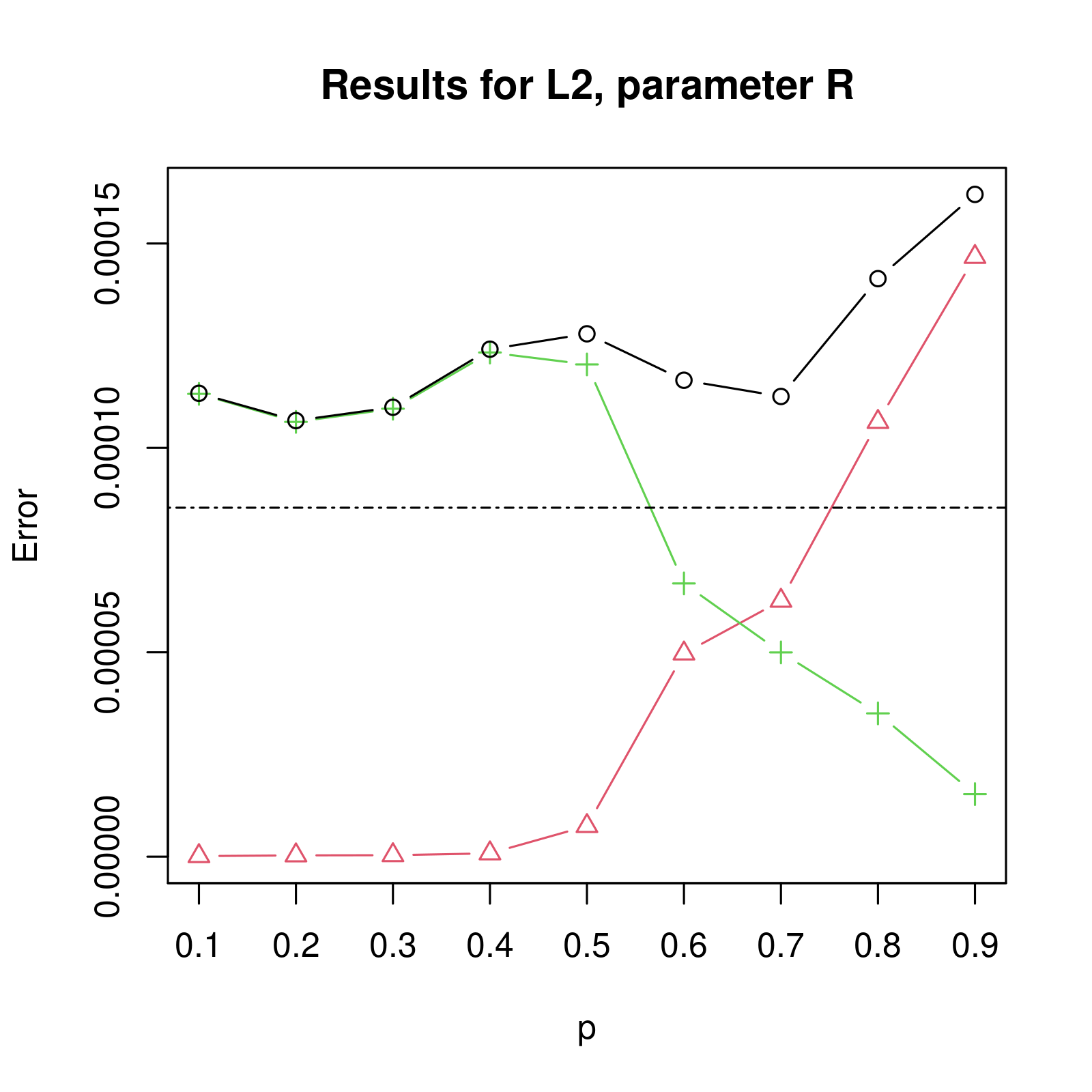}
    
    \includegraphics[width = 0.3\textwidth]{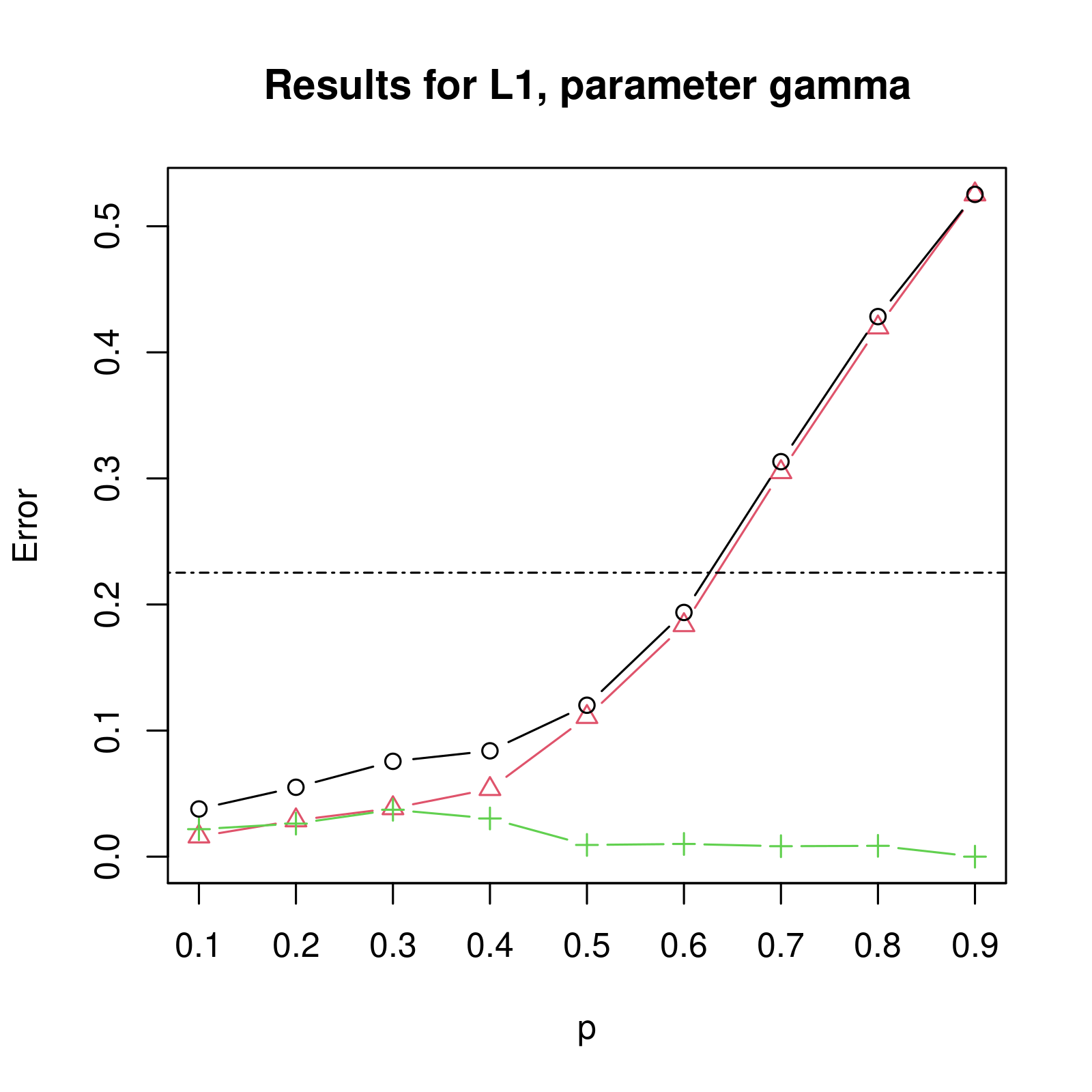}
    \includegraphics[width = 0.3\textwidth]{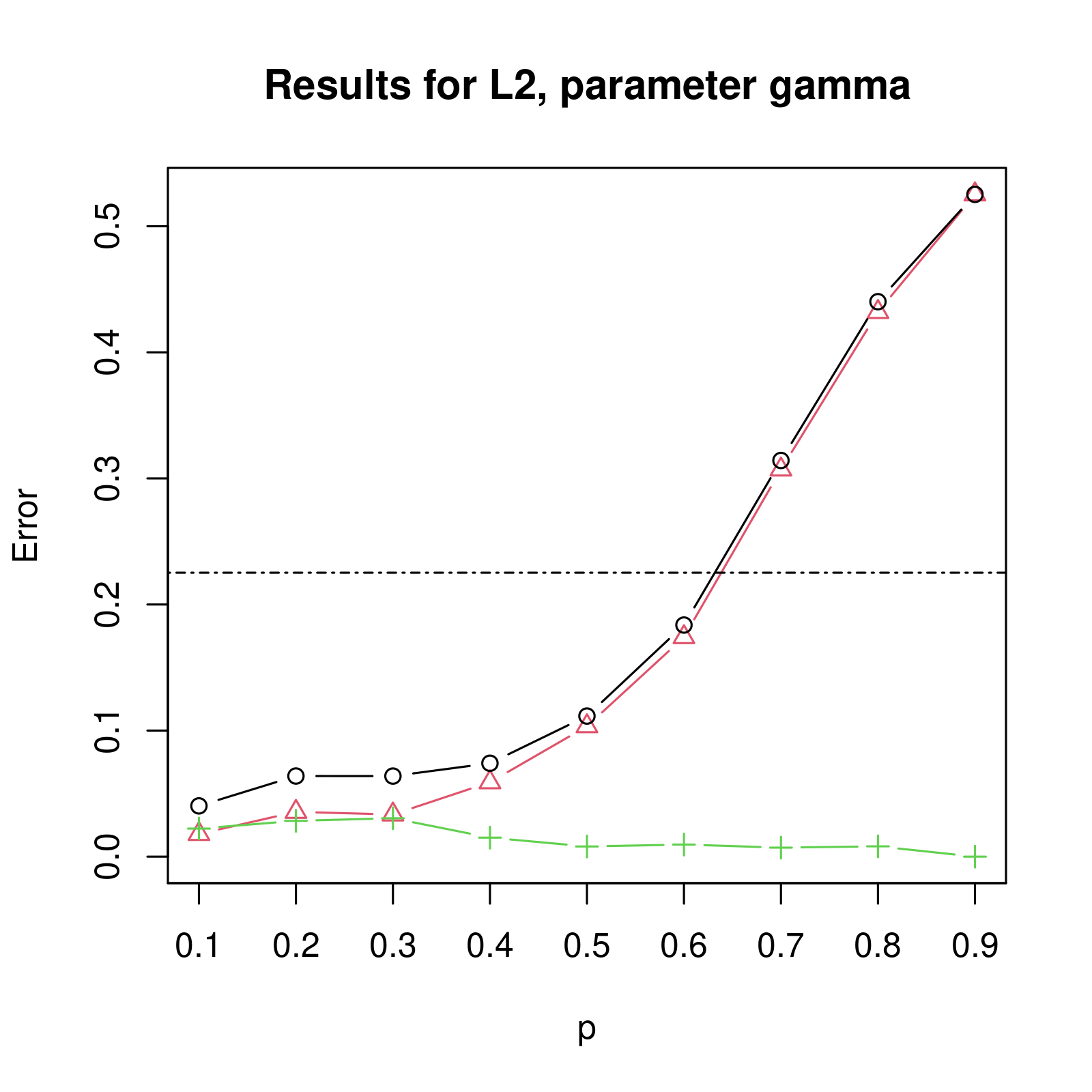}
    
    \includegraphics[width = 0.3\textwidth]{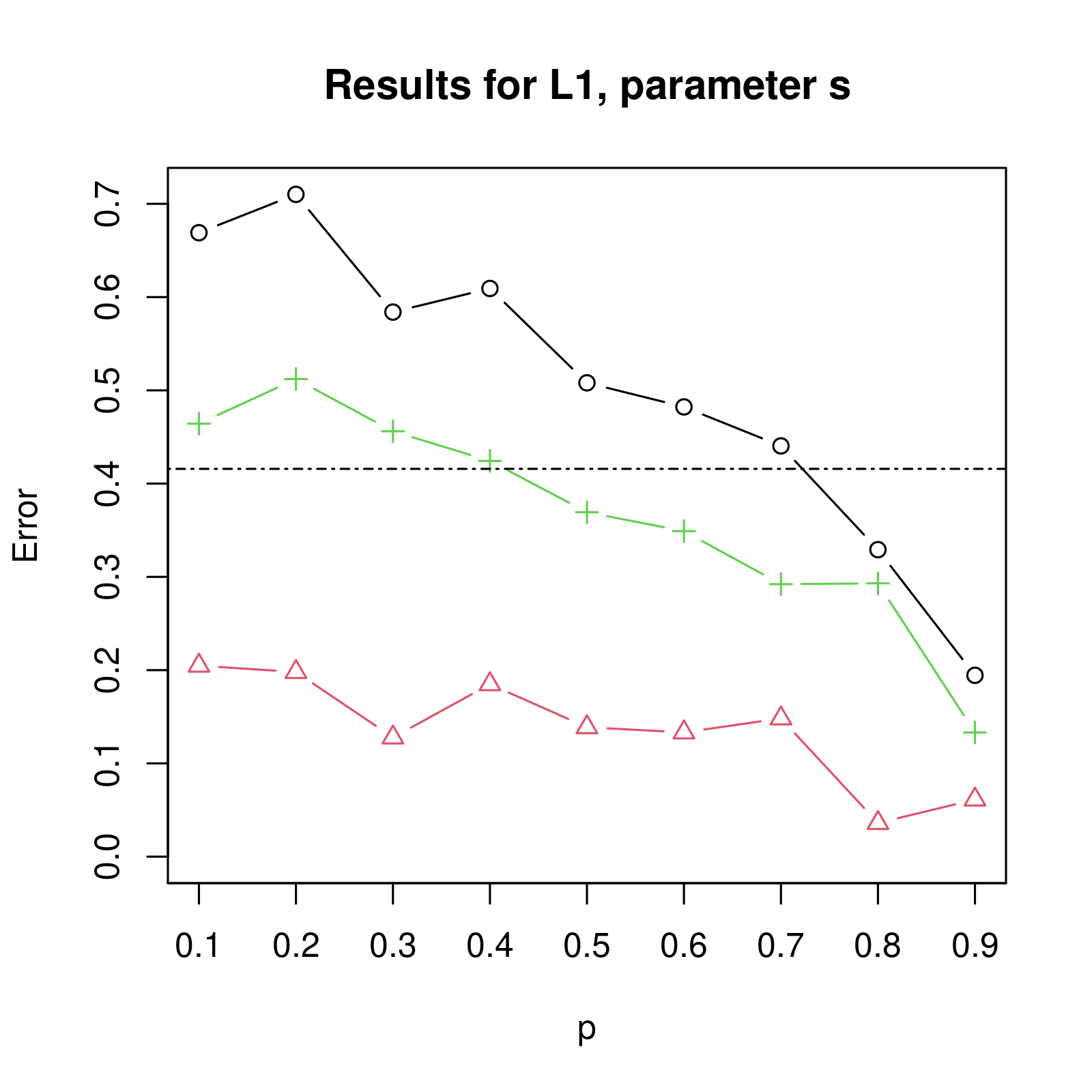}
    \includegraphics[width = 0.3\textwidth]{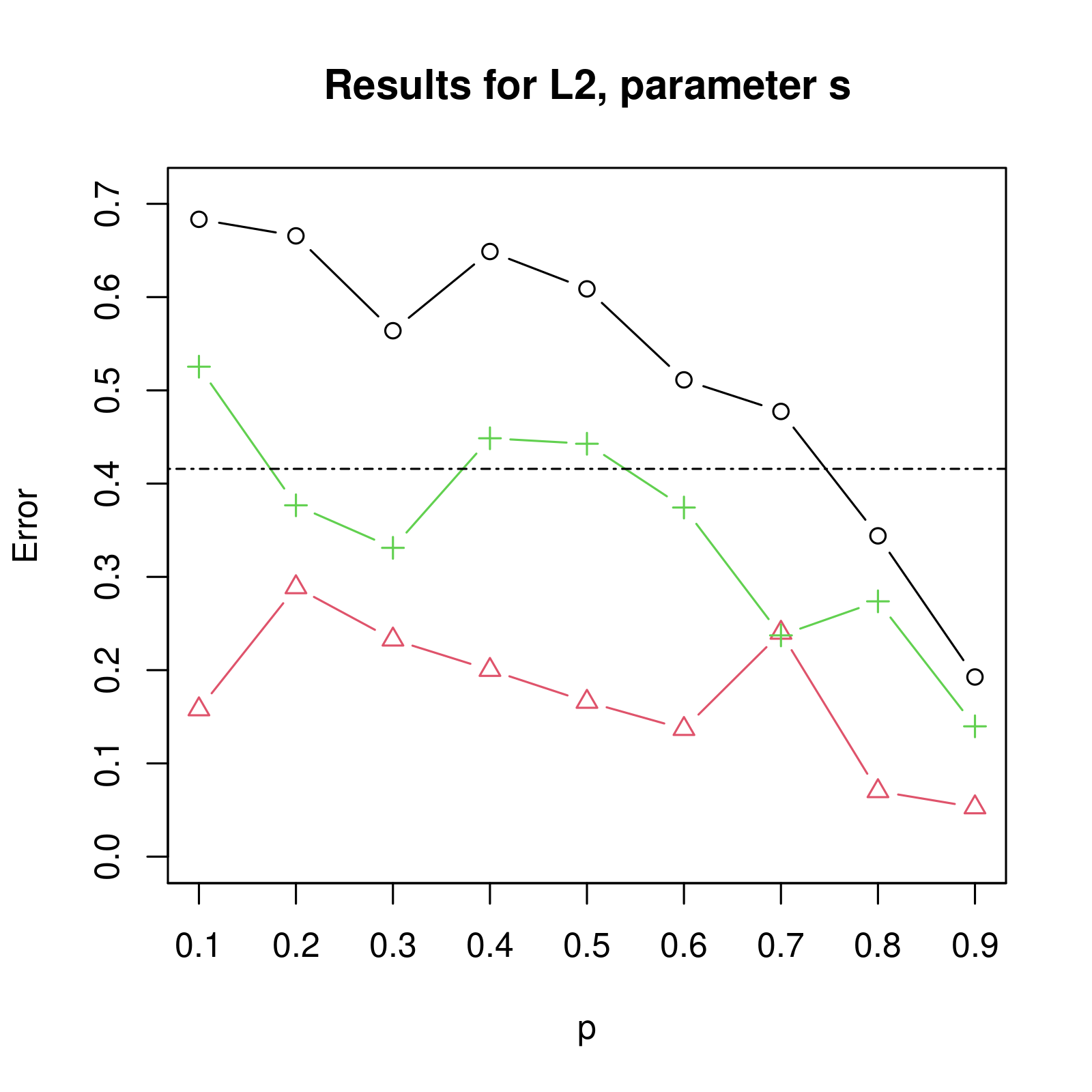}
    \caption{MSE, squared bias and variance for the Geyer saturation model using PPL with the $\Loss_1$ and $\Loss_2$ loss functions, when estimating the parameters 
    $\beta$, $R$, $\gamma$ and $s$. 
    Here $k = 100$, $N = 100$, $p = 0.1,0.2,\ldots,0.9$ and the PPL-weight is set to $p/(1-p)$. 
    The black lines with circles correspond to MSE, the red lines with triangles correspond to squared bias and the green lines with plus signs correspond to variance. The black dotted lines correspond to the Takacs-Fiksel estimates.
}
    \label{fig:geyer_p/(1-p)-L1L2}
\end{figure}

\begin{figure}[!htb]
    \centering
    \includegraphics[width = 0.3\textwidth]{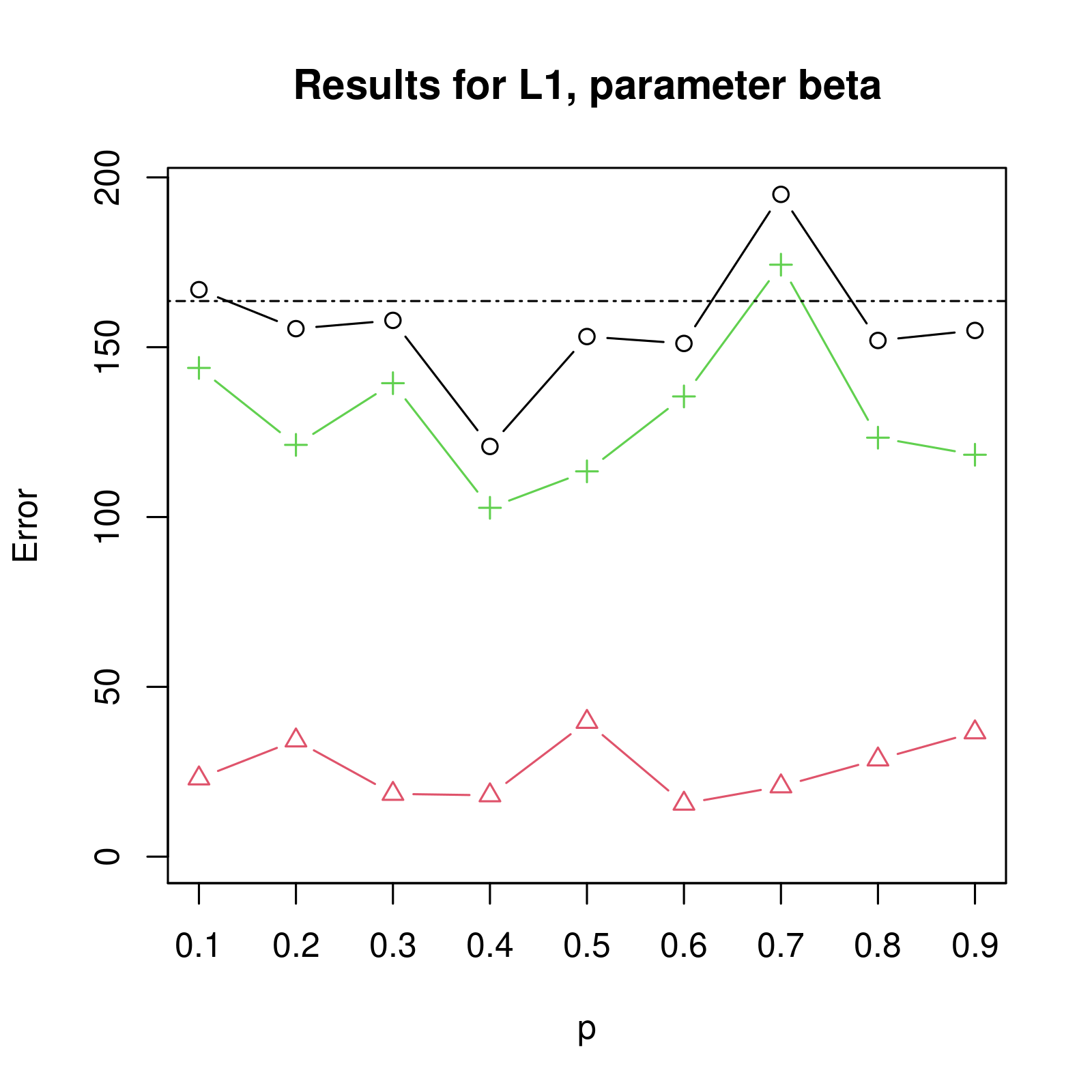}
    \includegraphics[width = 0.3\textwidth]{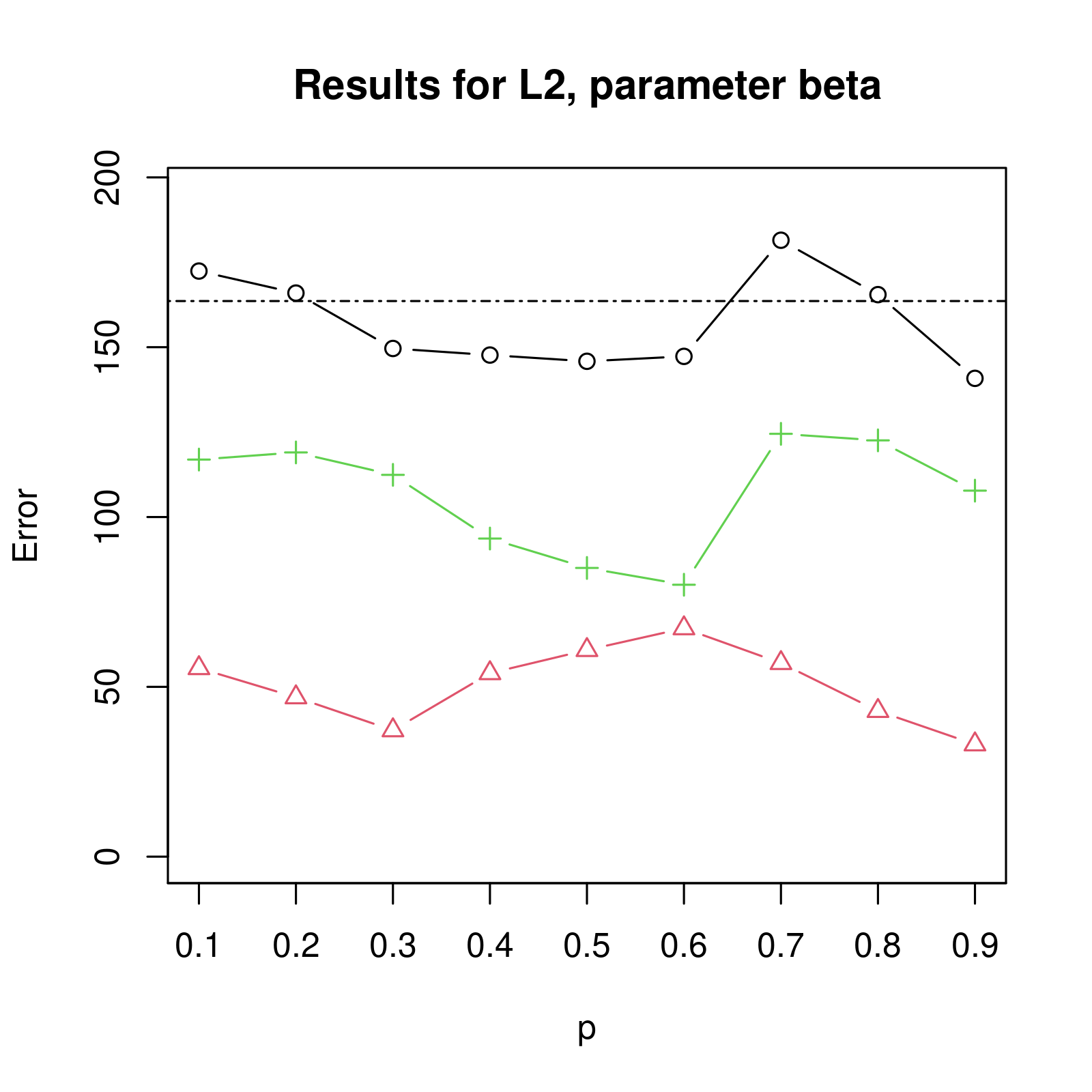}
    
    \includegraphics[width = 0.3\textwidth]{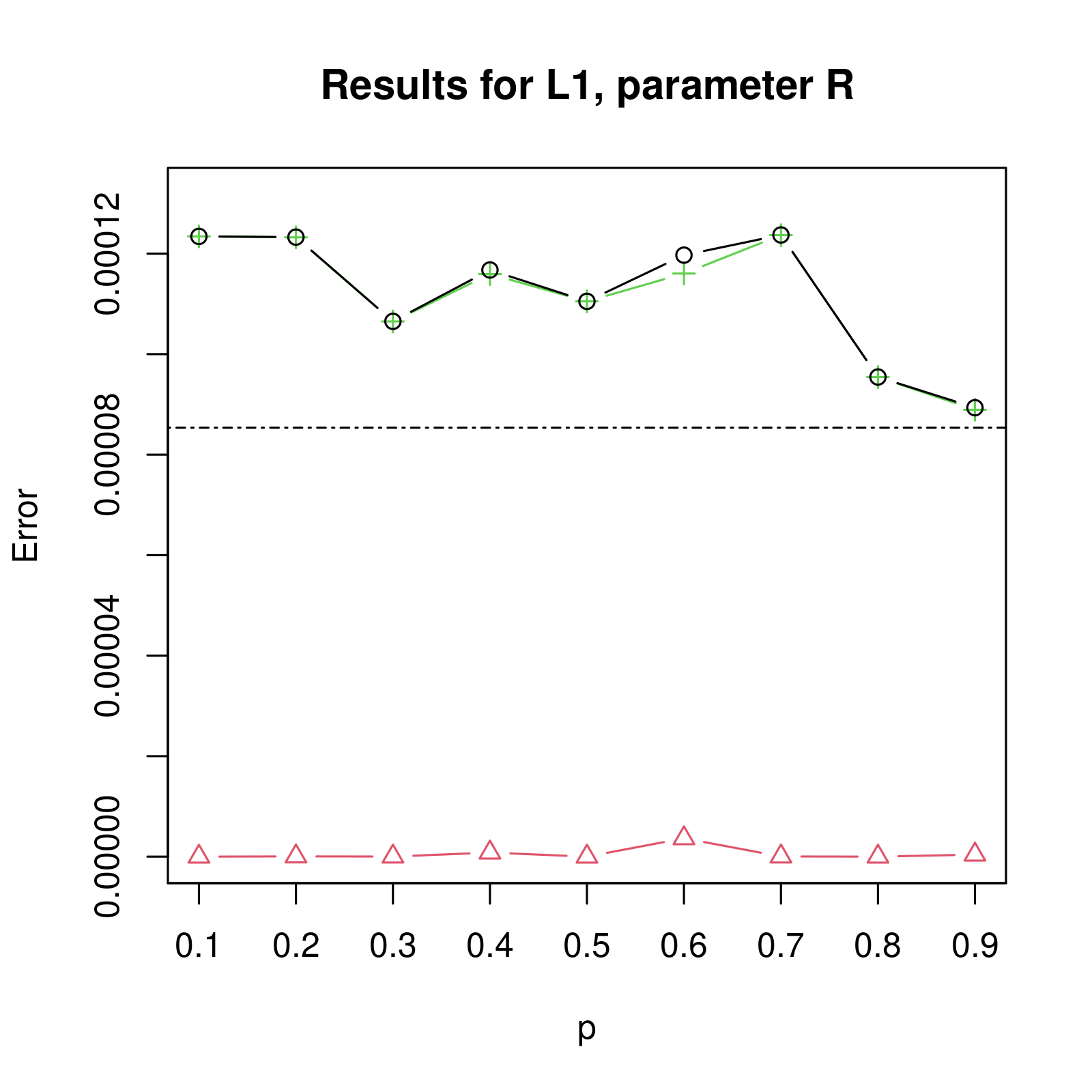}
    \includegraphics[width = 0.3\textwidth]{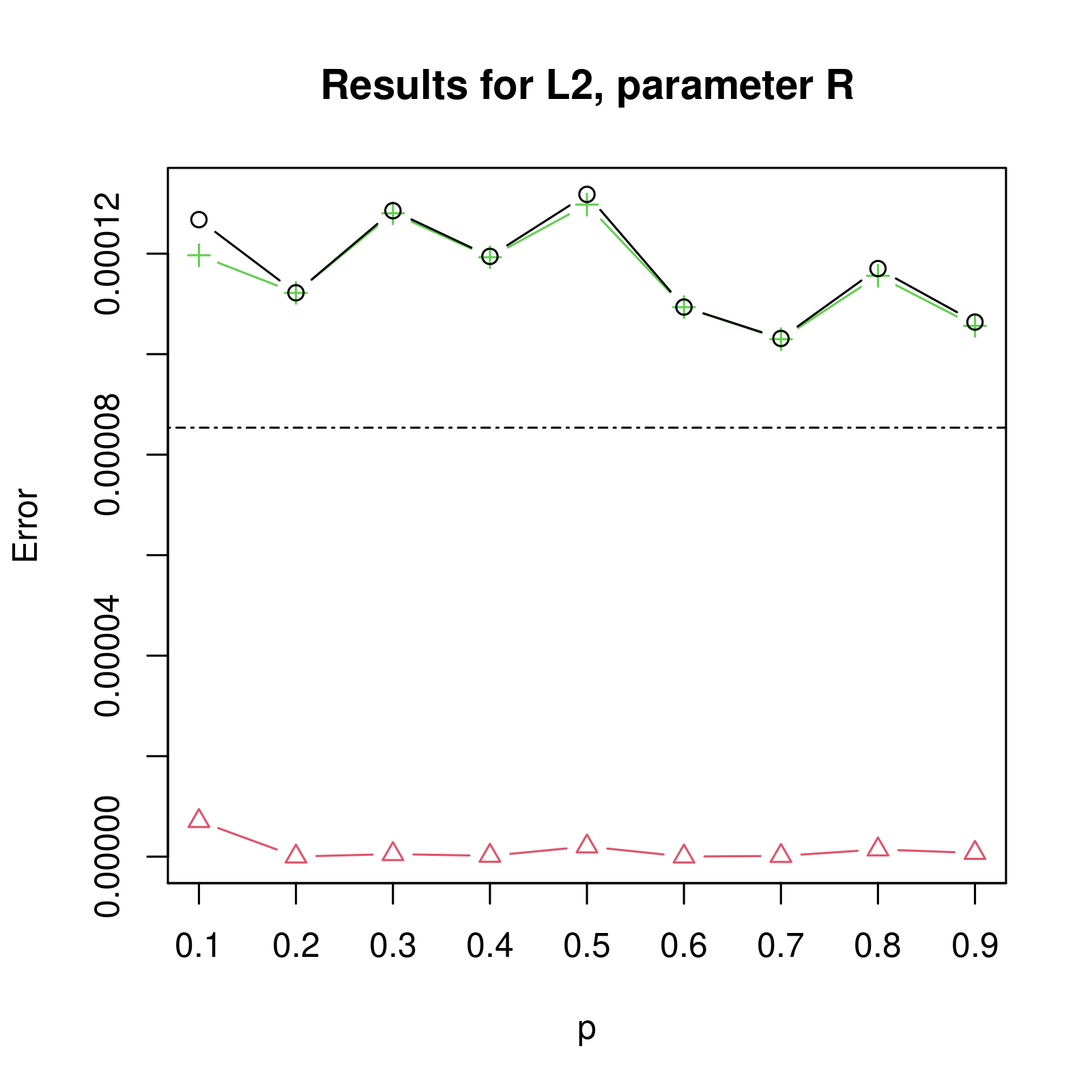}
    
    \includegraphics[width = 0.3\textwidth]{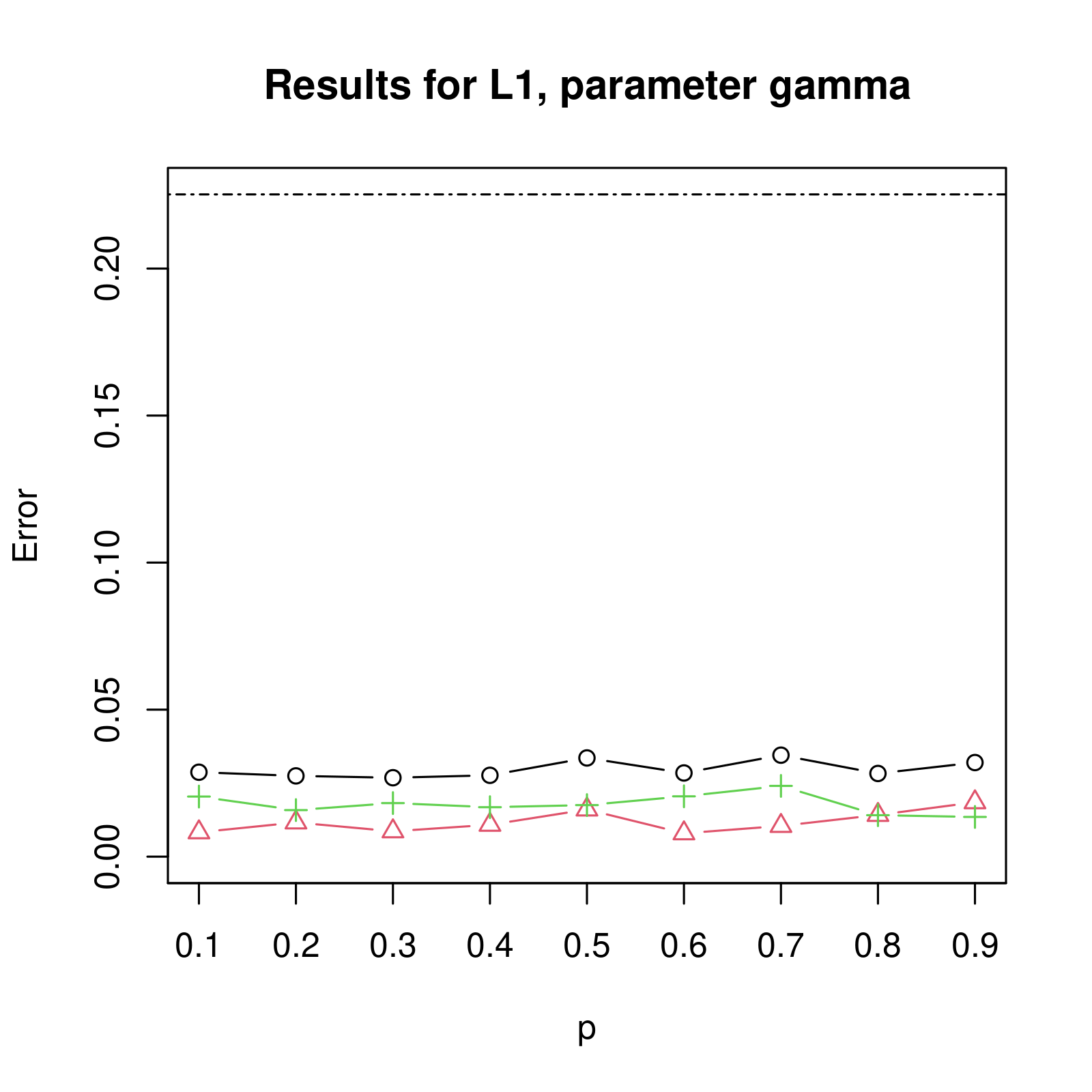}
    \includegraphics[width = 0.3\textwidth]{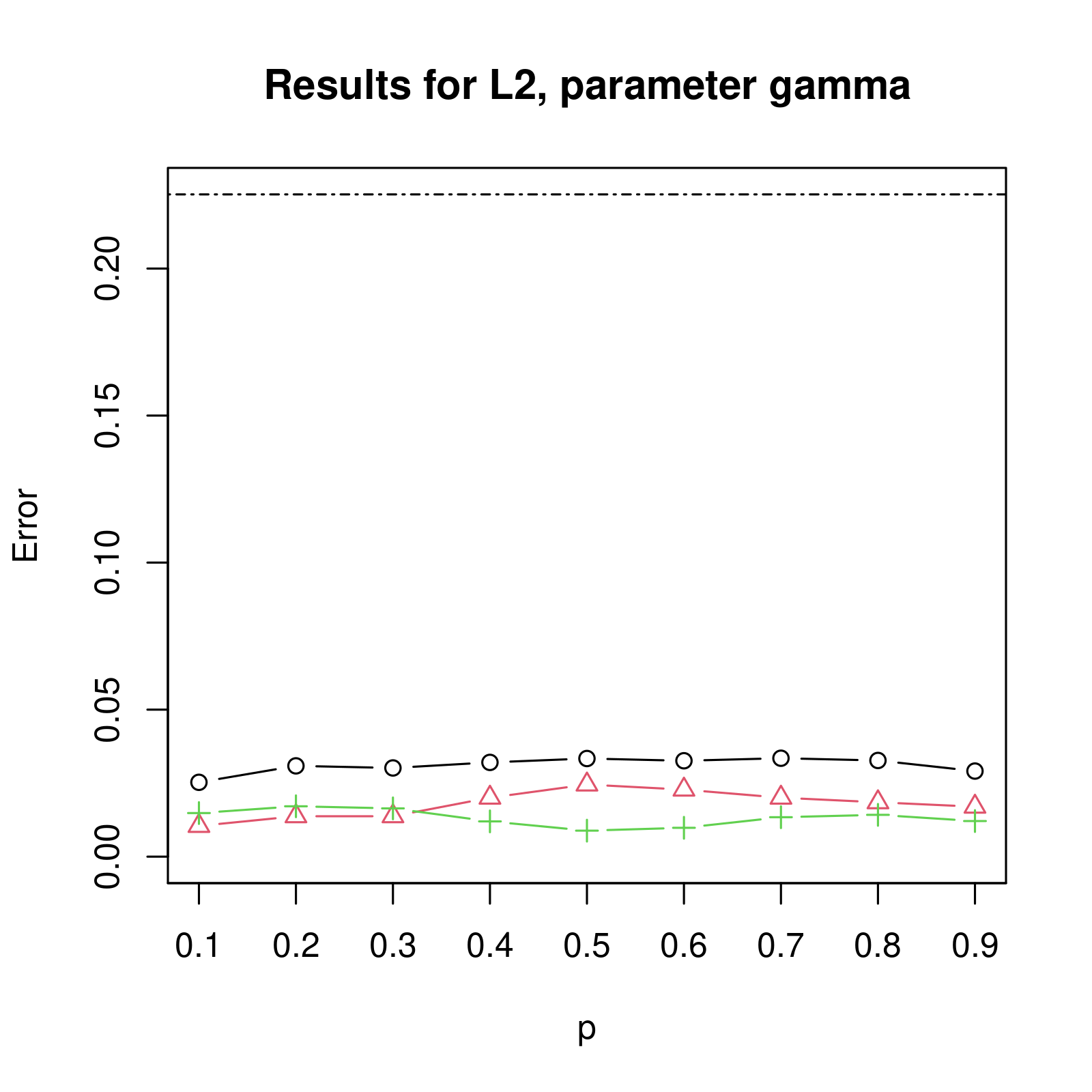}
    
    \includegraphics[width = 0.3\textwidth]{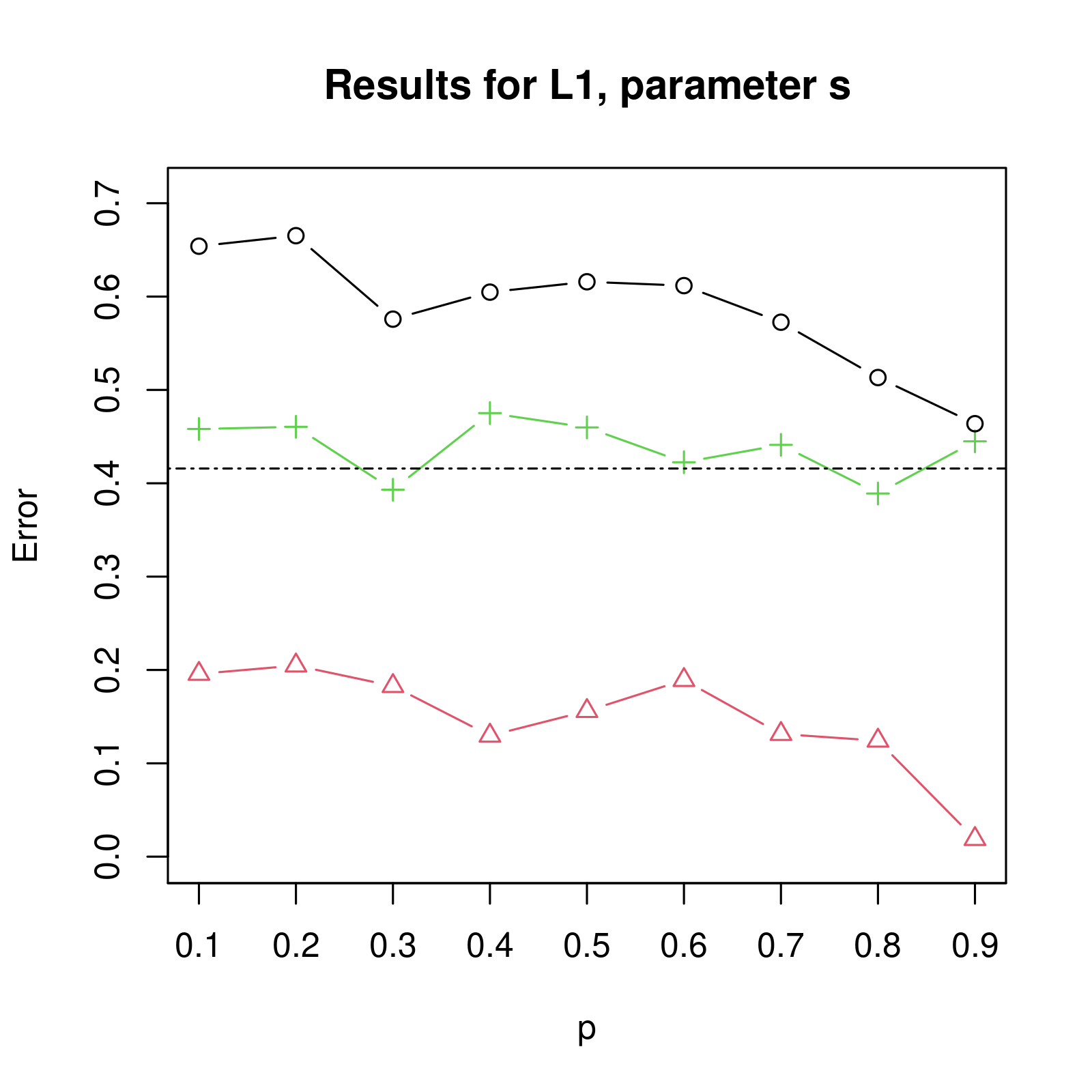}
    \includegraphics[width = 0.3\textwidth]{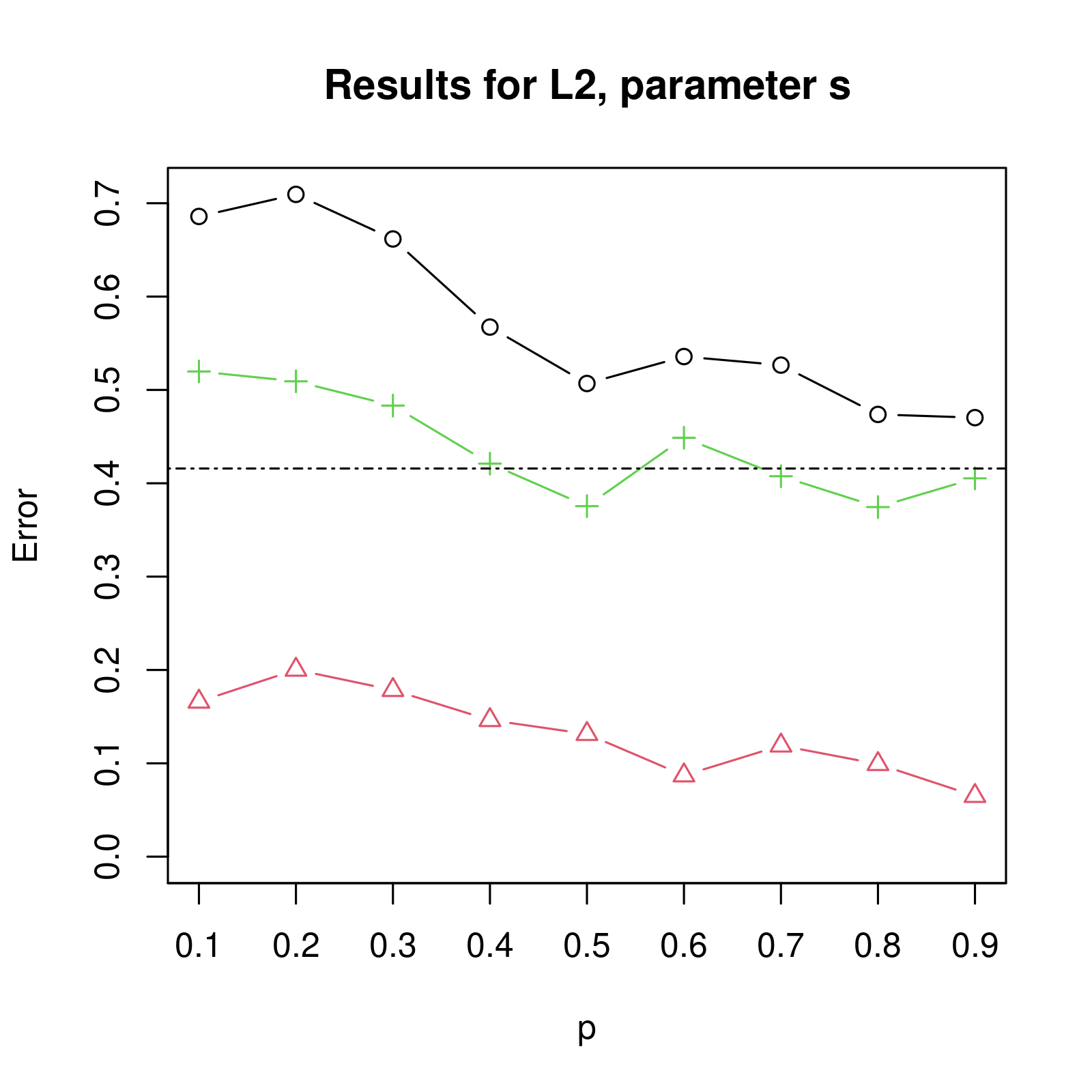}
    \caption{MSE, squared bias and variance for the Geyer saturation model using PPL with the $\Loss_1$ and $\Loss_2$ loss functions, when estimating the parameters 
    $\beta$, $R$, $\gamma$ and $s$. 
    Here $k = 100$, $N = 100$, $p = 0.1,0.2,\ldots,0.9$ and the PPL-weight is estimated in according to \eqref{e:WeightEst}. 
    The black lines with circles correspond to MSE, the red lines with triangles correspond to squared bias and the green lines with plus signs correspond to variance. The black dotted lines correspond to the Takacs-Fiksel estimates.
}
    \label{fig:geyer_w-L1L2}
\end{figure}

 \end{bibunit}

\end{document}